\documentclass[11pt]{article}
\usepackage{enumerate}
\usepackage{hyperref}
\usepackage{amsmath,amssymb,amsthm}
\usepackage[shortcuts]{extdash}
\usepackage{tikz}
\usetikzlibrary{matrix,trees,decorations}

\usepackage{authblk}
\usepackage{algorithm}%
\usepackage{caption}

\DeclareCaptionFormat{algor}{%
\noindent\rule{\textwidth}{0.8pt}%
\par\offinterlineskip\vskip1pt%
  \textbf{#1#2}#3\offinterlineskip\hrulefill}
\DeclareCaptionStyle{algori}{singlelinecheck=off,format=algor,labelsep=space}

\usetikzlibrary{shapes.geometric, arrows,matrix,backgrounds}
\usepackage{algpseudocode,mdframed}
\usepackage[margin=1.1in,letterpaper]{geometry}
\usepackage{subfig}
\usetikzlibrary{shapes.geometric, arrows,matrix,backgrounds}

\hypersetup{
    pdftitle = {Finding branch-decompositions of matroids, hypergraphs, and more},
    pdfauthor=  {Jisu Jeong, Eun Jung Kim, and Sang-il Oum}
   }
\providecommand{\keywords}[1]{\noindent\textbf{\textit{Keywords: }} #1}
\newcommand\abs[1]{\lvert #1\rvert}

\newcommand\spn[1]{\langle #1 \rangle}
\newtheorem{THM}{Theorem}[section]
\newtheorem{LEM}[THM]{Lemma}
\newtheorem*{THMMAIN}{Theorem \ref{thm:summary-brw}}

\newtheorem{COR}[THM]{Corollary}
\newtheorem{PROP}[THM]{Proposition}

\theoremstyle{remark}

\newtheorem*{LEMsum}{Lemma~\ref{lem:sum}} 
\newtheorem*{THMcomputebases}{Theorem~\ref{thm:computebases}}
\newtheorem*{PROPFSroot}{Proposition~\ref{prop:FSroot}}
\newtheorem*{PROPfullset}{Proposition~\ref{prop:fullset}}

\theoremstyle{definition}

\newcommand\rank{\operatorname{rank}}

\newcommand\etabar{\overline{\eta}}
\newcommand\col{\operatorname{col}}
\newcommand\rt{\mathfrak{r}}

\newcommand\F{\mathbb F}
\newcommand\FF{\mathcal F}
\newcommand\FS{\operatorname{FS}_k}
\renewcommand\L{\mathcal L}
\newcommand\V{\mathcal V}

\newcommand\B{\mathfrak B}

\newcommand\T{\mathsf T}
\newcommand\sle\precsim
\newcommand\sge\succsim
\newcommand\seq\sim
\newcommand\tle{\preccurlyeq}
\newcommand\tge\succcurlyeq
\newcommand\teq\cong
\newcommand\ple\sqsubseteq
\newcommand\pge\sqsupseteq

\newcommand \splus \oplus
\newcommand \tplus \oplus

\newcommand \fplus {\operatorname{\oplus}}

\newcommand\trim{\operatorname{trim}}

\newcommand\I{\mathcal I}

\newenvironment{listspec}[1][]{
\mdframed
\textsc{#1}
\description\setlength{\itemsep}{0pt}}{\enddescription\endmdframed}
\newcommand\up{\operatorname{\sf up_k}}

\newcommand\Hlineny{Hlin{\v e}n{\'y}}

\begin{document}
\abovedisplayshortskip=-3pt plus3pt
\belowdisplayshortskip=6pt

\title{Finding branch-decompositions of matroids, hypergraphs,\\ and more\thanks{An extended abstract of this paper appeared in \cite{JKO2018}.}}
\author{Jisu Jeong%
    \thanks{Supported by
      the National Research Foundation of
      Korea (NRF) grant funded by the Korea government (MSIT)
      (No. NRF-2017R1A2B4005020).}%
}
\affil{\small NAVER AI LAB, NAVER Clova, NAVER Corporation, Seongnam, 13561, South Korea}

\author{Eun Jung Kim%
    \thanks{Supported by the French National Research Agency (ANR) under No.:~ANR-17-CE23-0010 
    and No.:~ANR-18-CE40-0025-01.}%
}
\affil{Universit\'{e} Paris-Dauphine, PSL Research University, CNRS, UMR 7243, LAMSADE, 75016, Paris,~France}
\author[$\dagger$34]{Sang-il Oum%
    \thanks{Supported by the Institute for Basic Science (IBS-R029-C1).}}
\affil[3]{Discrete Mathematics Group, Institute for Basic Science (IBS), Daejeon, 34126, South Korea}
\affil[4]{Department of Mathematical Sciences, KAIST, Daejeon, 34141, South Korea}
\affil[ ]{E-mail addresses: \texttt{jisujeong89@gmail.com},
\texttt{eunjungkim78@gmail.com}, \texttt{sangil@ibs.re.kr}
}

\date{\today}
\maketitle

\newcommand\decomposition{de\-com\-po\-si\-tion}

\begin{abstract}
  Given $n$ subspaces of a finite-dimensional vector space over a fixed finite field $\F$,
  we wish to find a ``branch-decomposition''
  of these subspaces of width at most~$k$
  that is a subcubic tree~$T$ with $n$ leaves mapped bijectively to the subspaces
  such that for every edge $e$ of~$T$,
  the sum of subspaces associated to the leaves in one component of $T-e$
  and the sum of subspaces associated to the leaves in the other component
  have the intersection of dimension at most~$k$.
  This problem includes the problems of computing branch-width of
  $\F$-represented matroids,
  rank-width of graphs,
  branch-width of hypergraphs,
  and carving-width of graphs.

  We present a fixed-parameter algorithm to construct
  such a branch-decomposition of width at most~$k$, if it exists,
  for input subspaces of a finite-dimensional vector space over~$\F$.
  Our algorithm is analogous to the algorithm of Bodlaender and Kloks (1996) 
  on tree-width of graphs.
  To  extend their framework to branch-decompositions of vector spaces,
  we developed highly generic tools for branch-decompositions on vector spaces.
  The only known previous fixed-parameter algorithm for
  branch-width of $\F$-represented matroids
  was due to Hlin\v en\'y and Oum (2008) that runs in time $O(n^3)$ 
  where $n$ is the number of elements of the input $\F$-represented matroid.
  But their method is highly indirect.
  Their algorithm uses the nontrivial fact by Geelen et al. (2003) 
  that the number of forbidden minors is finite
  and uses the algorithm of Hlin\v en\'y (2006) on checking monadic second-order formulas on $\F$-represented matroids of small branch-width.
  Our result does not depend on such a fact and is completely self-contained, and yet matches their asymptotic running time for each fixed~$k$.
\end{abstract}

\section{Introduction}\label{sec:intro}

Let $\F$ be a finite field and $r$ be a positive integer.
A \emph{subspace arrangement} $\V$ is a multiset of subspaces of $\F^r$,
which can be represented by an $r\times m$ matrix~$M$
with an ordered partition $\mathcal I=\{I_1,I_2,\ldots,I_n\}$ of $\{1,2,\ldots,m\}$
such that for every $1\le i\le n$
the $i$th element of~$\V$ is 
the column space of the submatrix of~$M$ induced by the columns in $I_i$.
Here, 
an \emph{ordered partition} $\mathcal I=\{I_1,I_2,\ldots,I_n\}$ of $\{1,2,\ldots,m\}$ is
a partition of $\{1,2,\ldots,m\}$ such that $x<y$ for all  $x\in I_i, y\in I_j$ with $i<j$. 
We remark that subspace arrangements can be regarded as representable \emph{partitioned matroids} used in \cite{HO2006}. A partitioned matroid is  a matroid equipped with a partition of its ground set.

Robertson and Seymour~\cite{RS1991} introduced the notion of branch-width for graphs, hypergraphs, and more generally, for connectivity functions.
We are going to define the branch-width of a subspace arrangement as follows.
First, a tree is \emph{subcubic} if every node has degree at most $3$.
We define a \emph{leaf} of a tree as a node of degree at most $1$.
A \emph{branch-decomposition} of~$\V$ is a pair $(T,\L)$ of a subcubic tree $T$ 
with no degree-$2$ nodes and a bijective function $\L$ from the set of all leaves of~$T$ to $\V$.
For a node~$v$ of~$T$ and an edge $e$ incident with~$v$, let us write $A_v(T-e)$ to denote the set of all leaves of~$T$ in the component of $T-e$ containing~$v$.
For a branch-decomposition $(T,\L)$ of~$\V$ and each edge $e=uv$ of~$T$, 
we define the \emph{width} of $e$ to be 
\[\dim \Bigl(\bigl(\sum_{x\in A_u(T-e)} \L(x)\bigr)\cap 
\bigl(\sum_{y\in A_v(T-e)} \L(y)\bigr)\Bigr).\]
The \emph{width} of $(T,\L)$ is the maximum width of all edges of~$T$.
(If $T$ has no edges, then the width of $(T,\L)$ is $0$.)
The \emph{branch-width} of~$\V$ is the minimum $k$ such that 
there exists a branch-decomposition of~$\V$ having width at most~$k$.

Here is our main theorem.
\begin{THM}\label{thm:summary-brw}
Let $\F$ be a finite field, let $r$ be a positive integer, and let $k$ be a nonnegative integer.
Let $\V=\{V_1,V_2,\ldots,V_n\}$ be a subspace arrangement in $\F^r$
where each $V_i$ is given by its spanning set of $d_i$ vectors and $m=\sum_{i=1}^n d_i$.
In time $O(rm^2+(k+1)rmn + k^3n^3 + f(\abs{\F},k) n^2)$ for some function $f$, one can either
 find a branch-decomposition of~$\V$ having width at most~$k$ or confirm that no such branch-decomposition exists.
\end{THM}

Various width parameters of discrete structures have been introduced
and used for algorithmic and structural applications. 
One popular way of creating a width parameter of a discrete structure 
is to define it as the branch-width of some connectivity function
defined on that discrete structure.
Theorem~\ref{thm:summary-brw} immediately gives rise to  analogous algorithms
for many of them, such as 
carving-width of graphs, rank-width of graphs, and branch-width of graphs, hypergraphs, and
matroids. 
We will give a brief overview of each application.
Later Section~\ref{sec:appl} will discuss these applications and necessary definitions in detail.

\begin{itemize}
\item \textbf{Branch-width of matroids represented over a finite field $\F$.}
Let $\V=\{V_1,V_2,\ldots,V_n\}$ be a subspace arrangement in $\F^r$. 
If each $V_i$ is the span of a vector $v_i$ in $\F^r$ for each $i=1,2,\ldots,n$,
then 
$\V$ can be identified with the matroid $M$ represented by the vectors $v_1$, $v_2$, $\ldots$, $v_n$.
Furthermore, branch-width and branch-decompositions of $M$
are precisely branch-width and branch-decompositions of~$\V$, respectively.

\item \textbf{Rank-width of graphs.}
Rank-width, introduced by Oum and Seymour~\cite{OS2004}, is a width parameter of graphs expressing how easy it is to decompose a graph into a tree-like structure, called a \emph{rank-decomposition}, while keeping every edge cut to have a small ``complexity,'' called the \emph{width} of a rank-decomposition, where the complexity is measured by the matrix rank function. 
It will be shown in Section~\ref{sec:appl} that each vertex of a graph $G$ 
can be associated with a subspace of dimension at most $2$
so that the subspace arrangement $\V$ consisting of all subspaces associated with the vertices of~$G$ 
has branch-width $2k$ if and only if $G$ has rank-width~$k$. 
Furthermore, a branch-decomposition of~$\V$ of width $2k$ 
corresponds to a rank-decomposition of~$G$ of width~$k$.

\item \textbf{Branch-width of graphs and hypergraphs.}
Let $\F=GF(2)$ be the binary field and let $\{v_1,v_2,\ldots,v_n\}$ be the standard basis of $\F^n$.
For a hypergraph $G$ with $n$ vertices $v_1,v_2,\ldots,v_n$,
we associate each edge $e$ with the span of the vertices incident with $e$.
Let $\V$ be  the subspace arrangement consisting of all subspaces 
associated with the edges of~$G$.
Then it is straightforward to show that branch-width 
and branch-decomposition of~$G$ are precisely branch-width and branch-decomposition of~$\V$, respectively. 

\item \textbf{Carving-width of graphs.}
Seymour and Thomas~\cite{ST1994} introduced carving-width of graphs.
Let $\F=GF(2)$ be the binary field and let $\{e_1,e_2,\ldots,e_m\}$ be the standard basis of $\F^m$.
For a graph $G$ with edges $e_1,e_2,\ldots,e_m$, 
we associate each vertex $v$ with the span of the edges incident with~$v$.
If $\V$ is the subspace arrangement consisting of all subspaces associated with the vertices of~$G$, then carving-width and carving of~$G$ are precisely 
branch-width and branch-decomposition of~$\V$, respectively.
\end{itemize}

For the first two applications,
the analogous theorems were proved earlier by 
\Hlineny{} and Oum~\cite{HO2006}.
However, their approach was completely indirect; they use a nontrivial fact
shown by Geelen et al.~\cite{GGRW2003a}
that the class of matroids of branch-width at most~$k$
has finitely many forbidden minors, each having at most $O(6^k)$ elements.
Then they use a monadic second-order formula to describe
whether a matroid contains a fixed minor
and use the dynamic programming algorithm to decide a monadic second-order formula
aided by  a given branch-decomposition of bounded width.
So far this describes the decision algorithm of \Hlineny~\cite{Hlineny2004} that decides whether branch-width is at most~$k$.
On top of this algorithm, \Hlineny{} and Oum use a sophisticated reduction to modify the input 
and use the decision algorithm repeatedly to recover a branch-decomposition.
Roughly speaking, this reduction attaches a gadget
to the input matroid
and this step requires extending the underlying finite field to an extension field,
because this gadget is not representable if the underlying field is too small.
As the list of forbidden minors is unknown, 
their algorithm should generate the list of minor-minimal
matroids having branch-width larger than $k$.
Thus, even for small values of~$k$, it would be practically impossible
to implement their algorithm.
Contrary to the previous algorithm,
our algorithm does not depend on the finiteness of obstructions 
and yet matches their asymptotic running time for each fixed $k$.

We do not know any previous analogous theorems for branch-width of hypergraphs.
For branch-width of graphs, 
Thilikos and Bodlaender~\cite{TB2000}
posted a $50$-page technical report in 2000 proving that for every fixed $k$,
one can check in linear time whether a graph has branch-width at most~$k$ and, if so, output a branch-decomposition of minimum width. 
This work was presented at a conference in 1997~\cite{BT1997}. 
For carving-width of graphs, 
the conference paper of 
Thilikos, Serna, and Bodlaender~\cite{TSB2000}
presented 
a linear-time algorithm for each fixed $k$ that 
determines whether the carving-width of an input graph $G$ is at most~$k$
and, if so, constructs a carving of~$G$ with minimum carving-width.

Let us turn our attention to the ingredients of our algorithm.
We aim to solve the following problem.
  \begin{listspec}[Branch-Width]
  \item [Parameters:] A finite field $\F$ and an integer $k$.
  \item [Input:] An $r\times m$ matrix~$M$ over $\F$ with an ordered partition $\mathcal I=\{I_1,I_2,\ldots,I_n\}$ of $\{1,2,\ldots,m\}$ and an integer $k$.
  \item [Output:] A branch-decomposition $(T,\L)$ of width at most~$k$ of a subspace arrangement $\V$ consisting of the column space of the submatrix of $M$ induced by the columns in $I_i$ for each $i$
or a confirmation that the branch-width of~$\V$ is larger than $k$.
  \end{listspec}

We develop
the framework inspired by the approach of Bodlaender and Kloks~\cite{BK1996} on their work on tree-width.
(A similar framework was also given independently by
Lagergren and Arnborg~\cite{LA1991}.)
They created a linear-time algorithm that can find a tree-decomposition of width at most~$k$ or confirm that the tree-width of an input graph is larger than $k$ for each fixed $k$. 
They used dynamic programming based on a given tree-decomposition of bounded width.
For the dynamic programming, they designed a special encoding of all possible tree-decompositions of width at most~$k$ that can arise from certain parts of a graph.

Our algorithm also uses a dynamic
programming approach based on a given branch-\decomposition{} of small width. 
Then, how do we generate a branch-decomposition of small width 
in the first place? For this purpose, we use the technique called the iterative compression, a technique initiated by Reed, Smith, and Vetta~\cite{ReedSV04}. 
This technique reduces a problem to a ``compression version'' of the problem so that iteratively solving the compression version leads to a solution to the original problem. 
For our purpose, we consider the following compression version of the problem \textsc{Branch-Width}.

  \begin{listspec} [Branch-Width Compression] 
  \item [Parameters:] A finite field $\F$ and an integer $k$.
  \item [Input:] An $r\times m$ matrix~$M$ over $\F$ with an ordered partition $\mathcal I=\{I_1,I_2,\ldots,I_n\}$ of $\{1,2,\ldots,m\}$, a branch-decomposition $(T',\L')$ of~$\V$ of width at most $2k$, and  an integer $k$, where $\V$ is a subspace arrangement of the column spaces of the submatrix of $M$ induced by the columns in $I_i$ for each $i$.
  \item [Output:] A branch-decomposition $(T,\L)$ of width at most~$k$ of a subspace arrangement $\V$
or a confirmation that the branch-width of~$\V$ is larger than $k$.
  \end{listspec}

The next proposition is a simple proof of how the iterative compression solves the original problem   \textsc{Branch-Width}. 
The detailed analysis is in the proof of Theorem~\ref{thm:summary-brw}.
\begin{PROP}\label{prop:compression}
  \textsc{Branch-Width} can be solved by solving \textsc{Branch-width Compression} at most $n-2$ times.
\end{PROP}
\begin{proof}[Proof sketch]
  Let $V_i$ be the column space of the submatrix of $M$ induced by $I_i$ for each $i=1,2,\ldots,n$ and 
  let $\V=\{V_1,V_2,\ldots,V_n\}$. 
  We may assume that $n\ge 2$.
  Let $\V_i=\{V_1,V_2,\ldots,V_i\}$ and 
  let $M_i$ be the submatrix of $M$ consisting of the columns in $I_1\cup I_2\cup \cdots \cup I_i$ for each $i=2,3,\ldots,n$. Note that $M_i$ with  a partition $\mathcal I_i=\{I_1,I_2,\ldots,I_i\}$ corresponds to $\V_i$.
  We will later show that we may assume  $\abs{I_i}\le k$  by preprocessing in Subsection~\ref{subsec:preprocessing}. 

  For $i=2$, there is only one  branch-decomposition $(T_2,\L_2)$ of $\V_2$. If it has width larger than $k$, then we can stop, because the branch-width of $\V_i$ is less than or equal to the branch-width of~$\V$.
  So we may assume that $(T_2,\L_2)$ has width at most~$k$.

  Now let us assume that $i\in \{3,4,\ldots,n\}$
  and suppose that we have a branch-decomposition $(T_{i-1},\L_{i-1})$ of $\V_{i-1}$ with width at most~$k$. 
  Let us obtain $T_i'$ from $T_{i-1}$ by adding a leaf $v$ adjacent to a new node created by subdividing one edge arbitrary. We extend $\L_{i-1}$ to create $\L_{i}'$ by mapping $v$ to $V_i$.
  Then $(T_i',\L_i')$ is a branch-decomposition of $\V_i$ and 
  it can be easily shown that $(T_i',\L_i')$ has width at most $2k$, because $\dim V_i\le k$ by the outcome of preprocessing.
  Now we use \textsc{Branch-Width Compression} to find a branch-decomposition $(T_i,\L_i)$ of width at most~$k$ for~$\V_i$, or confirm that the branch-width of $\V_i$ is larger than $k$. If the branch-width of $\V_i$ is larger than $k$, then so is the branch-width of~$\V$ and therefore we can stop. 
  
  We repeat this for all $i\in \{3,4,\ldots,n\}$ and then $(T_n,\L_n)$ is a branch-decomposition of width at most~$k$ for $\V$. We called  \textsc{Branch-Width Compression} at most $n-2$ times to solve \textsc{Branch-Width}.
\end{proof}

Therefore, our algorithm focuses on solving \textsc{Branch-Width Compression}, in which 
we are given a branch-decomposition of width at most $2k$. We use a dynamic programming approach, taking advantage of having a tree-like structure of the given branch-decomposition.
To use a branch-decomposition for dynamic programming, 
we need a concept of a ``boundary,'' that plays the role of a bag in a tree-decomposition. For a branch-decomposition $(T,\L)$ of~$\V$ and an edge $e$ of~$T$, we consider the boundary $B$ as the intersection of the sum of subspaces associated to the leaves in one component of $T-e$ and the sum of subspaces associated to the leaves of the other component of $T-e$. As the branch-width is at most~$k$, the boundary $B$ has dimension at most~$k$. 
Furthermore we restrict our attention to the finite field $\F$ and so the number of subspaces of $B$ is finite. 
For the convenience of dynamic programming on branch-decompositions, 
we define \emph{transcripts} of a branch-decomposition in Section~\ref{sec:keyalg}, 
which is essentially a precomputed list of bases and linear transformations
useful for computing with boundaries.

As usual, we need a data structure to store partial solutions 
that may be extended to a branch-decomposition of width at most~$k$, if it exists.
We have two important aspects here.

First of all, we will restrict our search to a smaller set of branch-decompositions. 
Namely, if $(T',\L')$ is a given branch-decomposition of width at most $2k$ in 
\textsc{Branch-Width Compression},
then the algorithm will find a branch-decomposition of width at most~$k$ which is 
``well-behaved'' with respect to the given decomposition $(T',\L')$, which we call \emph{totally pure with respect to $(T',\L')$}.
To prove that this smaller search space is sufficient to certify that  branch-width is at most~$k$, 
in Section~\ref{sec:pure} we introduce two operations, fork and split,
that can modify a branch-decomposition of width~$k$
into another branch-decomposition of width at most~$k$.
By using these operations, 
we will show that if there is a branch-decomposition of width at most~$k$, 
then there is a branch-decomposition of width at most~$k$ that is ``well-behaved''
(totally pure) with respect to $(T',\L')$.

Second, we will define a \emph{$B$-namu} in Section~\ref{sec:keyalg}, which will be a data 
structure to express partial solutions.
Here, $B$ is going to be the boundary for some edge in $(T',\L')$.
A $B$-namu is, roughly speaking, 
a subcubic tree whose incidences are decorated by subspaces of $B$
and whose edges are labeled by a nonnegative integer 
so that it represents  the ``shadow'' of a branch-decomposition of width at most~$k$ on~$B$.
We will define an operation $\tau$ on $B$-namus that will compress a $B$-namu into a ``compact'' $B$-namu and prove that there are only finitely many compact $B$-namus of width at most~$k$, when $B$ has bounded dimension and $\F$ is a finite field. This operation $\tau$ on $B$-namus contains  the operation to compress an integer sequence introduced by Bodlaender and Kloks~\cite{BK1996} 
for their work on tree-decompositions.

\medskip
In the previous work \cite{JKO2016}, 
the authors found a similar algorithm for path-de\-com\-po\-si\-tions
of a subspace arrangement.
A path-decomposition of a subspace arrangement is a linearized variant of 
a branch-decomposition that restricts the subcubic trees 
to caterpillar trees.
They defined  $B$-trajectories, analogous to $B$-namus of this paper. 
Here are the key technical differences.

First, we improve the running time of an algorithm computing transcripts.
They also use transcripts for the algorithms, though this name was not introduced.
But their algorithm to compute a transcript was slow 
and so inevitably they presented two approaches,
one that is self-contained but only giving the running time $O(n^4)$ based on iterative compressions
and the other giving the running time $O(n^3)$ but using an earlier result of \Hlineny{} and Oum~\cite{HO2006}. 
If we adapt our new method to the result of \cite{JKO2016}, we also get an $O(n^3)$-time algorithm for path-decompositions of subspace arrangements
based on iterative compressions.

Second, the concept of totally pure branch-decompositions was not needed in \cite{JKO2016}.
For our algorithms, we sometimes insert a whole subtree into a branch-decomposition
and we had to make sure that our algorithm finds at least one branch-decomposition
of width at most~$k$ if it exists. That was guaranteed by proving that the operations of fork and split to modify a branch-decomposition do not increase the width.
For path-decompositions, %
there was no need to define such a concept.

\medskip
This paper is organized as follows. Section~\ref{sec:appl} will explain several applications of Theorem~\ref{thm:summary-brw} to various width parameters such as 
branch-width of matroids, graphs and hypergraphs, rank-width of graphs, and carving-width of graphs.
Recall that by Proposition~\ref{prop:compression}, at the heart of designing an efficient algorithm for \textsc{Branch-Width} lies an algorithm for \textsc{Branch-Width Compression}. 
In Section~\ref{sec:keyalg}, we present an overview of the algorithm for \textsc{Branch-Width Compression} along with key technical notions for this algorithm. 
Section~\ref{sec:pure} introduces the notion of totally pure branch-decompositions and proves that totally pure branch-decompositions of optimum width exist.
Section~\ref{sec:namu} proves properties of $B$-namus for a subspace $B$ of $\F^r$. 
Section~\ref{sec:fullset} establishes the correctness of the algorithm in Section~\ref{sec:keyalg} using results from Sections~\ref{sec:pure} and~\ref{sec:namu}. 
Section~\ref{sec:algo} describes our main algorithm for solving \textsc{Branch-Width}. This section encompasses the preprocessing step, manipulating the data structure, computing transcripts, and 
analyzing the running time of the algorithm in Section~\ref{sec:keyalg}
and culminates in the exposition of the algorithm for \textsc{Branch-Width} based on iterative applications of the algorithm in Section~\ref{sec:keyalg} as a subroutine.

\section{Applications}\label{sec:appl}

In this section, we will describe how our main theorem,
Theorem~\ref{thm:summary-brw}, can be applied to various width
parameters of graphs, hypergraphs, and matroids easily.
To introduce various width parameters efficiently, we first define the branch-width of a connectivity function. 
First of all, for a finite set $E$, 
a \emph{connectivity function} $f:2^E\to \mathbb Z$ is a function such that 
\begin{enumerate}[(i)]
\item $f(X)=f(E-X)$ for all $X\subseteq E$, 
\item $f(X)+f(Y)\ge f(X\cap Y)+f(X\cup Y)$ for all $X,Y\subseteq E$, 
\item $f(\emptyset)=0$.
\end{enumerate}
The \emph{branch-decomposition} of $E$ is a pair $(T,\L)$ of 
a subcubic tree without degree-$2$ nodes and a bijection $\L$
from the set of all leaves of~$T$ to $E$.
The \emph{width} of $(T,\L)$ is the maximum of $f(A_e)$
for all edges $e$ of~$T$ and a partition $(A_e,B_e)$ of $E$ induced by leaves in distinct components of $T-e$ with $\L$.
The \emph{branch-width} of $f$ is the minimum width of all possible branch-decompositions of $E$. 
(If $T$ has no edges, then we define its width to be $0$.)

For a set $X$ of vectors, $\spn{X}$ denotes the span of $X$ and 
we may write $\spn{x}$ instead of $\spn{X}$ if $X=\{x\}$ for convenience.
For a set $\V$ of subspaces, 
we denote $\spn{\V}=\sum_{V\in\V} V$.

\subsection{Application to branch-width of matroids}

A \emph{matroid} $M$ is a pair $(E,\mathcal{I})$ of 
a finite set $E$ and a set $\mathcal{I}$ of subsets of $E$ such that 
\begin{enumerate}[({I}1)]
\item $\emptyset\in \mathcal{I}$,
\item if $X\in \mathcal{I}$ and $Y\subseteq X$, then $Y\in \mathcal{I}$,
\item if $X,Y\in \mathcal{I}$ and $\abs{X}<\abs{Y}$,
then $Y-X$ has an element $e$ such that $X\cup \{e\}\in \mathcal{I}$.
\end{enumerate}
We call $E$ a \emph{ground set} and $\mathcal{I}$ a set of \emph{independent sets}.
See the book of Oxley~\cite{Oxley2011a} for more terminology on matroids.
For a subset $X$ of $E$, 
the \emph{rank} $r_{M}(X)$ is the maximum size of an independent subset of $X$
and the \emph{connectivity function} of a matroid $M$ is defined by 
\[\lambda_{M}(X)=r_{M}(X)+r_{M}(E-X)-r_{M}(E).\]
It is easy to verify that $\lambda_M$ is a connectivity function on $E$.

A \emph{branch-decomposition} of a matroid $M=(E,\mathcal{I})$ 
is a branch-decomposition of~$E$.
The \emph{width} of a branch-decomposition $(T,\L)$ of $M$
is measured by $\lambda_M$.
The \emph{branch-width} of a matroid $M$ is the branch-width of $\lambda_M$.

For a matrix $A$ over a field $\F$, 
let $E(A)$ be the set of all column vectors of $A$
and let $\mathcal{I}(A)$ be the set of all linearly independent subsets of $E(A)$.
Then it is known that $M(A)=(E(A),\mathcal{I}(A))$ is a matroid, which is called a \emph{vector matroid}.
We say that a matroid $M$ is \emph{$\F$-representable} 
if %
$M=M(A)$ for some matrix $A$ over a field $\F$, 
and such $A$ is an \emph{$\F$-representation} of $M$.
Let us say that a matroid is \emph{$\F$-represented} if it is given with its $\F$-representation.

A vector matroid $M=M(A)$ on the ground set $E=\{v_1,v_2,\ldots,v_n\}\allowbreak \subseteq \F^r$
is associated with a subspace arrangement $\V_{M}=\{\spn{v_1},\spn{v_2},\ldots,\spn{v_n}\}$.
For a subset $X$ of $E$, %
$r_{M}(X)=\dim \spn{X}$.
Thus, %
for each subset $X$ of $E$, %
\begin{align*}
\lambda_{M}(X)&=\dim \spn{X} +\dim\spn{E-X} - \dim\spn{E} \\
&=\dim (\spn{X} \cap \spn{E - X}).
\end{align*}
Therefore, the branch-width of $M$ is equal to 
the branch-width of $\V_{M}$.

As a corollary of Theorem~\ref{thm:summary-brw},
we have the following.

\begin{THM}\label{thm:main-matroid}
  Let $k$ be an integer.
  Given an input $n$-element matroid represented by an $r\times n$ matrix over a finite field $\F$ with $r\le n$, we can find its branch-decomposition of width at most~$k$
  or confirm that its branch-width is larger than $k$
  in time $f(k,\abs{\F})n^3$ for some function $f$.
\end{THM}

\subsection{Application to rank-width of graphs}\label{subsec:rankwidth}
We review the definition of rank-width.
For a graph $G=(V,E)$, 
let $A_G$ be the adjacency matrix $(m_{ij})_{i,j\in V}$ over the binary field $GF(2)$ such that 
\[
m_{ij}=
\begin{cases}
1 \qquad &\text{if $i$ is adjacent to $j$ in $G$,} \\
0 \qquad &\text{otherwise.}
\end{cases}
\]
For an $\abs{A}\times \abs{B}$ matrix~$M$ and $X\subseteq A$ and $Y\subseteq B$, 
we denote $M[X,Y]$ a submatrix of~$M$
induced by rows in~$X$ and columns in $Y$.
The \emph{cut-rank function} $\rho_G$ of a graph~$G$ is a function on subsets $X$ of $V$ 
such that 
\[\rho_G(X)=\rank (A_G[X,V-X]).\]
Oum and Seymour~\cite{OS2005} showed that the cut-rank function of a graph is 
a connectivity function on its vertex set.
A \emph{rank-decomposition} of a graph is a branch-decomposition of $V$.
The \emph{rank-width} of~$G$ is the branch-width of $\rho_G$.

Let $\F=GF(2)$ be the binary field and 
let $\{e_1,e_2,\ldots,e_n\}$ be the standard basis of $\F^n$. In other words, $e_i$ is the $0$-$1$ vector in $\F^n$ such that the $j$th entry is $1$ if and only if $j=i$.

For a graph $G$ 
with a vertex set $V(G)=\{1,2,\ldots,n\}$
and $i\in V(G)$, 
let $v_i$ be the sum of all vectors $e_j$ for all $j$ adjacent to $i$ in $G$.
In other words, $v_i$ is the $i$th column vector of $A_G$.
For all $i=1,2,\ldots,n$, let $V_i$ be the subspace of $\F^n$ spanned by $\{v_i,e_i\}$.
The following lemma in \cite{JKO2016} gives the 
relation between the rank-width of a graph $G$ 
and the branch-width of a subspace arrangement $\{V_1,V_2,\ldots,V_n\}$.

\begin{LEM}[{\cite[Lemma 8.1]{JKO2016}}]\label{lem:rankwidth}
Let $G$ be a graph and $V_i$ be a subspace defined as above for each $i\in V(G)$. Then
for $X\subseteq V(G)$,
\[
\dim \Bigl(\bigl(\sum_{i\in X} V_i\bigr) \cap \bigl(\sum_{j\in V(G)-X}V_j\bigr)\Bigr)=2\rho_G(X).
\]
\end{LEM}
By Lemma~\ref{lem:rankwidth}, we can deduce the following.
\begin{THM}\label{thm:main-rank-width}
  Given an input $n$-vertex graph $G$, 
  we can find a rank-\decomposition{} of width at most~$k$ 
  or confirm that the rank-width of~$G$ is larger than $k$
  in time $f(k)n^3$ for some function $f$.
\end{THM}
\begin{proof}[Proof assuming Theorem~\ref{thm:summary-brw}]
For a given graph $G$ with a vertex set $V(G)=\{1,2,\ldots,n\}$,
we define a subspace $V_i$ for each $i\in \{1,2,\ldots,n\}$ as above, which can be done in time $O(n^2)$. 
Let $\V=\{V_1,V_2,\ldots,V_n\}$.
By Lemma~\ref{lem:rankwidth},
the branch-width of~$\V$ is precisely twice the rank-width of~$G$.

By applying Theorem~\ref{thm:summary-brw},
we can find 
a branch-decomposition $(T,\L)$ of~$\V$ of width at most $2k$
or confirm that the branch-width of~$\V$ is larger than $2k$
in time $O(n\cdot (2n)^2 + (2k+1) \cdot n \cdot 2n \cdot n+ (2k)^3 \cdot n^3 + f(2,2k) \cdot n^2)=g(k) n^3$ for some function~$g$.

If we have  a branch-decomposition $(T,\L)$ of~$\V$ of width at most $2k$, 
then we can convert it to 
a rank-decomposition $(T,\L')$ of~$G$ of width at most~$k$
by constructing $\L'$ so that 
$\L(v)=V_i$ if and only if $\L'(v)=i$. 
This completes the proof.
\end{proof}

We remark that our method here can be extended naturally for $\F$-rank-width, a concept introduced by Kant\'e and Rao~\cite{KR2013}. 
According to the terminology of \cite{KR2013}, we may state as follows.
\begin{THM}
  Let $\F$ be a finite field and $\F^*=\F\setminus\{0\}$.
  Let $k$ be an integer.
  For an input $n$-vertex $\F^*$-graph, 
  we can either find an $\F$-rank-decomposition of width at most~$k$
or confirm that its $\F$-rank-width  is larger than $k$
 in time $f(\abs{\F}, k) n^3$ for some function $f$.
\end{THM}
\paragraph{Clique-width of graphs}
Rank-width was introduced by Oum and Seymour~\cite{OS2004} to augment the notion of clique-width, another width parameter defined by Courcelle and Olariu~\cite{CO2000}. 
Clique-width is defined as the minimum $k$ such that $G$ can be expressed as a \emph{$k$-expression}, which also describes a tree-like structure of graphs.
 
For a positive integer $k$, a \emph{$k$-expression} is an algebraic expression that consists of the following operations on graphs whose vertices are labeled by integers in $\{1,2,\ldots,k\}$.
\begin{itemize}
 \item $\cdot_i$: a graph with a single vertex labeled by $i\in \{1,2,\ldots,k\}$.
 \item $G_1 \oplus G_2$: the disjoint union of two vertex-labeled graphs $G_1$ and $G_2$.
 \item $\eta_{i,j}(G)$ with $i\neq j$: adding an edge from each pair of a vertex of label $i$ and a vertex of label $j$.
 \item $\rho_{i\to j}(G)$: relabeling all vertices of label $i$ to label $j$.
\end{itemize}
The \emph{clique-width} of a graph $G$ is the minimum $k$ such that there is a $k$-expression representing $G$ after ignoring labels of the vertices \cite{CO2000}.
Fellows, Rosamond, Rotics, and Szeider~\cite{FRRS2009} showed that 
it is NP-hard to compute the clique-width.

Oum and Seymour~\cite{OS2004} designed an algorithm that, in time
$O(n^2)$, 
converts a rank-\decomposition{} of width~$k$ of an input $n$-vertex graph 
into a $(2^{k+1}-1)$-expression and, conversely, 
converts a $k$-expression into a rank-decomposition of width at most~$k$. 
Using this fact, 
we deduce the following corollary from Theorem~\ref{thm:main-rank-width}.
\begin{COR}\label{cor:cwd}
For an input $n$-vertex graph $G$ and a positive integer $k$, 
in time $f(k)n^3$ for some function $f$,
we can find a $(2^{k+1}-1)$-expression of~$G$ %
or confirm that $G$ has clique-width larger than $k$.
\end{COR}

\subsection{Application to branch-width of graphs and hypergraphs}

A \emph{branch-decomposition} of a hypergraph $G=(V,E)$ 
is a branch-decomposition of $E$. 
For a subset $X$ of $E$, we define $\beta_G(X)$ as the number of
vertices incident with both an edge in~$X$ and an edge not in~$X$.
The \emph{branch-width} of~$G$ is the branch-width of $\beta_G$.

\begin{LEM}\label{lem:hypergraphdensity}
  Let $G$ be a hypergraph on $n$ vertices and $m$ edges
  having no parallel edges.
  If the branch-width of~$G$ is at most~$k$,
  then $m\le 2^{2k}n$.
\end{LEM}
\begin{proof}
  We proceed by induction on $n$. 
  We may assume that $G$ has no isolated vertex.
  Let $(T,\L)$ be a branch-decomposition of~$G$ of width at most~$k$.
  For each edge~$e$ of~$T$,
  we have a partition $(A_e,B_e)$ of $E(G)$ induced by the partition of the leaves of $T-e$
  and $\L$.
  Let $M_e$ be the set of vertices of~$G$ incident with both an edge in~$A_e$
  and an edge in $B_e$.
  Let $L_e$ be the set of vertices of~$G$ incident with an edge in~$A_e$ but 
  no edge in $B_e$.
  Let $R_e$ be the set of vertices of~$G$ incident with an edge in~$B_e$ but 
  no edge in~$A_e$.
  As the width of $(T,\L)$ is at most~$k$, $\abs{M_e}\le k$ for all $e\in E(T)$.

  Suppose that $T$ has an edge $e$ 
  such that $1\le \abs{L_e}\le k$.
  Let $v\in L_e$. 
  Then every edge incident with~$v$ is a subset of $L_e\cup M_e$.
  As $G$ has no parallel edges, there are at most $2^{2k}$ such edges containing~$v$ 
  because $\abs{L_e\cup M_e}\le 2k$.
  Let $H$ be the subhypergraph of~$G$ obtained by deleting $v$ and all edges incident
  with~$v$.
  As the branch-width of $H$ is at most~$k$, 
  by the induction hypothesis, $m-2^{2k}\le\abs{E(H)}\le 2^{2k} \abs{V(H)}= 2^{2k}(n-1)$. 
  Thus we may assume that 
  $L_e=\emptyset$ or 
  $\abs{L_e}\ge k+1$
  for each edge $e$ of~$T$.
  By symmetry, we may also assume that 
  $R_e=\emptyset$ or $\abs{R_e}\ge k+1$ 
  for every edge $e$ of~$T$.

  If $L_e=R_e=\emptyset$ for some edge $e$ of~$T$, 
  then $G$ has at most~$k$ vertices
  and therefore $m\le 2^{k}\le 2^{2k} n$.
  So we may assume that $T$ has no such edge $e$.

  Let $e$ be an edge of~$T$ such that
  $A_e$ is minimal if $L_e$ is nonempty
  or $B_e$ is minimal if $R_e$ is nonempty
  among all sets $A_f$ for which $L_f$ is nonempty
  and all sets $B_f$ for which $R_f$ is nonempty.
  By symmetry, we may assume that $L_e$ is nonempty and $A_e$ is minimal.
  Let $x$ be the end of $e$ so that $A_e$ is the set of all edges of~$G$
  corresponding by $\L$ to leaves in the component of $T-e$ containing $x$.

  If $x$ is a leaf of~$T$, then 
  for a vertex $v$ in $L_e$, 
  $G$ has only one edge incident with~$v$. 
  Let $H$ be the subhypergraph of~$G$ obtained by deleting $v$ and
  its unique incident edge.
  Then $H$ has branch-width at most~$k$ and so 
  by the induction hypothesis, $m-1=\abs{E(H)}\le 2^{2k} \abs{V(H)}= 2^{2k}(n-1)$.

  Thus, we may assume that $x$ is not a leaf of~$T$.
  Let $f$, $g$ be two edges incident with $x$ other than $e$.
  Let us assume that $A_f\cup A_g=A_e$
  and $B_f\cap B_g=B_e$.
  As the width of $f$ is at most~$k$, 
  there are at most~$k$ vertices of~$G$ incident with 
  both an edge in $A_f$
  and an edge in $A_g$.
  Note that if a vertex $v$ is $L_e$,
  then either $v\in L_f\cup L_g$
  or $v$ is incident with both an edge in $A_f$ and an edge in $A_g$.
  As $\abs{L_e}\ge k+1$, 
  either $L_f\neq \emptyset$ or $L_g\neq\emptyset$.
  By symmetry, we may assume that $L_f\neq\emptyset$.
  However, $A_f$ is a proper subset of $A_e$, contradicting the minimality assumption.
\end{proof}

Using Theorem~\ref{thm:summary-brw},
we deduce the following. 
\begin{THM}
Let $k$ be an integer.
Given an input hypergraph $G$ with $n$ vertices, %
we can either 
find a branch-decomposition of~$G$ whose width is at most~$k$
or 
confirm that the branch-width of~$G$ is greater than $k$ 
in time $f(k)n^3$ for some function~$f$. 
\end{THM}
\begin{proof}[Proof assuming Theorem~\ref{thm:summary-brw}]
If $G$ has parallel edges $A$, $B$
with $\abs{A}=\abs{B}> k$, then we conclude that the branch-width of~$G$ is greater than $k$.
If $G$ contains parallel edges $A$, $B$
with $\abs{A}=\abs{B}\le k$, then %
we consider $G-B$, which is a hypergraph obtained from $G$ by deleting the edge $B$. 
We may assume that $G$ has at least three edges.
If $G-B$ has branch-width greater than $k$, then the branch-width of~$G$ is also greater than $k$.
If $G-B$ has a branch-decomposition $(T,\L)$ of width at most~$k$,
then %
it is easy to construct a branch-decomposition of~$G$ from $(T,\L)$ 
by 
introducing a node~$v$ with $\L(v)=B$, 
subdividing an edge incident with $\L^{-1}(A)$ in~$T$, and 
joining $v$ and the subdividing node.
Thus, we may assume that $G$ has no parallel edges.

Let $m$ be the number of all edges in $G$.
We may assume that $m\le 2^{2k}n$ because otherwise 
we conclude that the branch-width of~$G$ is greater than $k$ by Lemma~\ref{lem:hypergraphdensity}.

Let $V=\{v_1,v_2,\ldots,v_n\}$ be the vertex set of~$G$ and 
$E=\{E_1,E_2,\ldots,E_m\}$ be the edge set of~$G$.
Let $\F=GF(2)$ be the binary field and let $\{e_1,e_2,\ldots,e_n\}$ be the standard basis of $\F^n$.
Let $\V=\{V_1,V_2,\ldots,V_m\}$ be a subspace arrangement of subspaces of $\F^n$
such that for every $i=1,2,\ldots,m$,
$V_i$ is the span of all $e_j$ such that $v_j$ is incident with $E_i$ and 
also incident with some $E_{i'}$ with $i'\neq i$.
Note that $\V$ can be constructed in time $O(m\cdot n \cdot m)=O(2^{4k}n^3)$.
We may assume that $\dim V_i \le k$ for every $i=1,2,\ldots, m$ because otherwise
$\dim \left(V_i \cap \left(\sum_{j\neq i} V_j\right)\right)= \dim V_i >k$ and then we can conclude that 
the branch-width of~$G$ is greater than $k$.

It is easy to see that for each subset $X$ of $E$,
\[
\dim\Bigl(\bigl(\sum_{E_i\in X} V_i\bigr)\cap \bigl(\sum_{E_i\in E-X} V_i\bigr) \Bigr)= \beta_G(X)
\]
because $V_i$ is the span of some $e_j$'s that represent the vertices in $E_i$ and 
such $e_j$'s are linearly independent.
Thus, the branch-width of~$\V$ is equal to the branch-width of~$G$.
So we can run the algorithm in Theorem~\ref{thm:summary-brw}.
The time complexity is, for some functions $g$ and $f$,  
\[
O(n \cdot (k2^{2k}n)^2 + (k+1) \cdot n \cdot k2^{2k}n \cdot 2^{2k}n + k^3\cdot (2^{2k}n)^3 + g(2,k) (2^{2k}n)^2)
=f(k)n^3. \qedhere
\]
\end{proof}

\subsection{Application to carving-width of graphs}

For a graph $G=(V,E)$ and a subset $X$ of $V$, 
we define  $\gamma_G(X)$ as the number of edges
incident with both a vertex in~$X$ and a vertex not in~$X$.
The \emph{carving} of a graph $G$ is the branch-decomposition of $V$.
The \emph{carving-width} of a graph $G$ 
is the branch-width of $\gamma_G$.

We can deduce the following theorem from Theorem~\ref{thm:summary-brw}.
\begin{THM}
Let $k$ be an integer.
Given an input $n$-vertex graph $G$, %
we can either 
find a carving of~$G$ whose width is at most~$k$
or 
confirm that the carving-width of~$G$ is greater than $k$
in time $f(k)n^3$ for some function~$f$.
\end{THM}
\begin{proof}[Proof assuming Theorem~\ref{thm:summary-brw}]
Let $V=\{v_1,v_2,\ldots,v_n\}$ be the vertex set of~$G$ and 
$E=\{f_1,f_2,\ldots,\allowbreak f_m\}$ be the edge set of~$G$.
We may assume that the degree of each vertex is at most~$k$ so that $m\le kn/2$. 
Otherwise, the carving-width of~$G$ is greater than $k$.
Let $\F=GF(2)$ be the binary field.
Let $\V=\{V_1,V_2,\ldots,V_n\}$ be a subspace arrangement of subspaces of $\F^m$
such that for every $i=1,2,\ldots,n$,
\[V_i=\spn{\{e_j: f_j\text{ is incident with }v_i\}},\]
where $\{e_1,e_2,\ldots,e_m\}$ is the standard basis of $\F^m$.
Then it is easy to see that for each subset $X$ of $V$,
\[
\dim\Bigl(\bigl(\sum_{v_i\in X} V_i\bigr)\cap \bigl(\sum_{v_i\in V-X} V_i\bigr)\Bigr) = \gamma_G(X)
\]
because $V_i$ represents the neighbors of $v_i$.
Thus the carving-width of~$G$ is equal to the branch-width of~$\V$
and the carving of~$G$ of width at most~$k$ corresponds to the
branch-decomposition of~$\V$ of width at most~$k$.
Then we obtain our algorithm from the algorithm of
Theorem~\ref{thm:summary-brw}.
\end{proof}

\section{The key step of the algorithm}\label{sec:keyalg}
This section aims to provide the key step of the algorithm for \textsc{Branch-Width Compression}.
Our algorithm for computing a branch-decomposition of~$\V$ of width~$k$ deploys dynamic programming 
on a precomputed branch-\decomposition{} $(T^b,\L^b)$ of~$\V$ of width at most $\theta$. 
For the dynamic programming, we will need to take advantage of this structure to design
efficient algorithms.

Let $(T^b,\L^b)$ be a branch-decomposition of a subspace arrangement $\V$.
Though we define $T^b$ to have no degree-$2$ nodes in the definition of a branch-decomposition, we will sometimes relax this definition by allowing $T^b$ to have a special node, called the \emph{root}, of degree $2$ or $0$.
To be precise, a \emph{rooted binary tree} is a rooted tree such that 
every nonroot node has degree $1$ or $3$ and the root has degree at most $2$. 
We may assume that $T^b$ is
a rooted binary tree by picking an arbitrary edge $e$ and subdividing~$e$ to create a degree-$2$ root node $\rt$, and  a branch-decomposition $(T^b,\L^b)$ with such tree $T^b$ will be called a \emph{rooted} branch-decomposition. 
We may
orient each edge of $T^b$ toward the root $\rt$ and assume that $T^b$ is a
rooted binary directed tree.

Dynamic programming algorithms on a tree-\decomposition{} by Bodlaender and Kloks~\cite{BK1996} benefit from the  
small width by encoding solutions with respect to the bags. While the bags are explicit in a given tree-decomposition, 
a branch-decomposition of a subspace arrangement does not provide an easy-to-handle metric 
for encoding solutions to our problem. In order to make it more useful, we
will need some extra information. 
For a node~$v$ of $T^b$, let $\V_v$ be the set of all elements of~$\V$
associated with~$v$ and its descendants by $\L$.
The \emph{boundary space} $B_v$ at $v$ is defined as $B_v=\spn{\V_v}\cap \spn{\V-\V_v}$.
We shall encode partial branch-decompositions of $\V_v$ with respect to the boundary spaces $B_v$ of a given branch-decomposition $(T^b,\L^b)$.
In order to describe it, we need to introduce $B$-namus and relevant terminology to manipulate $B$-namus, which will be given in Subsections~\ref{subsec:namu}--\ref{subsec:ksafe}. 

Meanwhile, we need to know $B_v$ for each node~$v$ of a given branch-decomposition $(T^b,\L^b)$ in advance. Moreover, the information relating the boundary space at a node 
and at its parent node will be necessary, which is captured by the notion of a \emph{transcript}, which will be introduced in Subsection~\ref{subsec:transcript}. 
After Subsection~\ref{subsec:transcript}, we are ready to present the key step of our algorithm for \textsc{Branch-Width Compression} in Subsection~\ref{subsec:computing}.
We postpone all proofs of properties to later sections and focus on presenting the algorithm.

\subsection{Namu}\label{subsec:namu}
For a tree $T$, an \emph{incidence} is a pair $(v,e)$ of a node~$v$ of~$T$ and an edge $e$ incident with~$v$.
Let $\I(T)$ be the union of $\{(*,\emptyset),(0,\emptyset)\}$ and the set of all incidences of~$T$. 
Let $\F$ be a field, let $r$ be a positive integer, and let $B$ be a subspace of $\F^r$.
Let $A(T)$ be the set of all leaves of a tree $T$
and recall that %
for an incidence $(u,e)$ of a tree $T$, 
$A_u(T-e)$ is the set of all leaves of~$T$ in the component of $T-e$ containing~$u$.

What is the essence of a branch-decomposition $(T,\L)$ of a subspace arrangement~$\V$ with respect to $B$? From $(T,\L)$, we can define  a function $\alpha$ from $\I(T)$ to the set of all subspaces of $B$, 
a function $\lambda$ from the union of $\{\emptyset\}$ and the set of all edges of~$T$ to the set of integers,
and a subspace $U$ of $B$ as follows:
\begin{itemize}
\item 
for each node~$v$ of~$T$ and an edge $e$ incident with~$v$, 
\[\alpha(v,e)=B\cap \sum_{x\in A_v(T-e)} \L(x),\]
\item for each edge $e=uv$ of~$T$, 
\[\lambda(e)= 
\dim \Bigl( \sum_{x\in A_u(T-e)} \L(x)\cap  \sum_{y\in A_v(T-e)} \L(y)\Bigr),\]
\item $U=B\cap \sum_{x\in A(T)} \L(x)$, and
\item $\alpha(*,\emptyset)=U$, $\alpha(0,\emptyset)=\{0\}$, and $\lambda(\emptyset)=0$.
\end{itemize}
Such a quadruple $(T,\alpha,\lambda,U)$ is called the \emph{canonical $B$-namu} of a branch-\decomposition{} $(T,\L)$ of a subspace arrangement $\V$.
This motivates us to define a $B$-namu. ``Namu'' is a tree in Korean.

A \emph{$B$-namu} $\Gamma$ is a quadruple $(T,\alpha,\lambda, U)$ of 
\begin{itemize}
\item a tree $T$ of maximum degree at most $3$ having at least one node, 
\item a function $\alpha$ from $\I(T)$ to the set of all subspaces of $B$, 
\item a function $\lambda$ from the union of $\{\emptyset\}$ and the set of all edges of~$T$ to the set of integers, and 
\item a subspace $U$ of $B$
\end{itemize}
such that 
\begin{enumerate}[(i)]
\item for every two-edge path $v_0, e_1,v_1,e_2,v_2$ in~$T$, $\alpha(v_0,e_1)$ is a subspace of $\alpha(v_1,e_2)$,
\item for all incidences $(v,e)$ of~$T$, $\alpha(v,e)$ is a subspace of $U$,%

\item $\alpha(*,\emptyset)=U$, $\alpha(0,\emptyset)=\{0\}$, and $\lambda(\emptyset)=0$,
\item for every edge $e=uv$ of~$T$, $\lambda(e)\geq \dim (\alpha(v,e)\cap \alpha(u,e))$.
\end{enumerate}
Our idea of defining a function on incidences of a tree is from the definition of a tree-labeling in Robertson and Seymour~\cite{RS1991}. While a tree-labeling maps an incidence to a subset of the ground set, the function $\alpha$ in the canonical $B$-namu $(T,\alpha,\lambda,U)$ maps an incidence to a subspace of $U$.
Note that if $(T,\alpha,\lambda,U)$ is a canonical $B$-namu of a branch-decomposition $(T,\L)$ of~$\V$, then for every edge $e=uv$ of~$T$, we have 
\[\lambda(e) \geq \dim  \Bigl(B\cap \sum_{x\in A_u(T-e)} \L(x)\cap  \sum_{y\in A_v(T-e)} \L(y)\Bigr)= \dim (\alpha(v,e)\cap \alpha(u,e)),\]
which meets the condition (iv) of the definition of a $B$-namu. 

The \emph{width} of a $B$-namu $\Gamma=(T,\alpha,\lambda, U)$
is the maximum of $\lambda(e)$
over all edges $e=uv$ of~$T$.
(If $T$ has no edges, then the width of $\Gamma$ is defined to be $0$.)
For a $B$-namu $\Gamma=(T,\alpha,\lambda,U)$, we denote
$T(\Gamma)=T$ and call it the \emph{tree in $\Gamma$}. 

Let $\Gamma_1=(T_1,\alpha_1,\lambda_1,U_1)$, $\Gamma_2=(T_2,\alpha_2,\lambda_2,U_2)$ be two $B$-namus.
We say that two $B$-namus $\Gamma_1$, $\Gamma_2$ are \emph{equal}, denoted as $\Gamma_1=\Gamma_2$,
if $T_1$ is equal to $T_2$,
$\alpha_1(v,e)=\alpha_2(v,e)$ for every incidence $(v,e)$ in $T_1$, 
$\lambda_1(e)=\lambda_2(e)$ for every edge $e$ in $T_1$, 
and $U_1=U_2$.

Two $B$-namus $\Gamma_1=(T_1,\alpha_1,\lambda_1,U_1)$ and $\Gamma_2=(T_2,\alpha_2,\lambda_2,U_2)$
are \emph{isomorphic}
if $U_1=U_2$ and 
there is an isomorphism $\phi$ from $T_1$ to $T_2$
such that
$\alpha_1(v,e)=\alpha_2(\phi(v),\phi(e))$
and $\lambda_1(e)=\lambda_2(\phi(e))$ 
for every incidence $(v,e)$ in $T_1$.
A $B$-namu $(T',\alpha',\lambda',U)$ is 
a \emph{subdivision} of a $B$-namu $\Gamma=(T,\alpha,\lambda,U)$  
if $T'$ is a subdivision of~$T$,
$\alpha'(v',e')=\alpha(v,e)$, and $\lambda'(e')=\lambda(e)$
for every incidence $(v',e')$ of $T'$ and its corresponding incidence $(v,e)$ of~$T$.

For two $B$-namus $\Gamma_1=(T_1,\alpha_1,\lambda_1,U_1)$ and $\Gamma_2=(T_2,\alpha_2,\lambda_2,U_2)$,
we say that $\Gamma_1\le\Gamma_2$
if 
\begin{itemize}
\item $T_1=T_2$, $\alpha_1=\alpha_2$, $U_1=U_2$ and
\item $\lambda_1(e)\le\lambda_2(e)$ for every edge $e$ of $T_1$.
\end{itemize}
For two $B$-namus $\Gamma_1$ and $\Gamma_2$, 
we say that $\Gamma_1\tle\Gamma_2$
if 
there exist a $B$-namu $\Gamma_1'$ isomorphic to a subdivision of $\Gamma_1$ and 
a $B$-namu $\Gamma_2'$ isomorphic to a subdivision of $\Gamma_2$
such that $\Gamma_1'\le\Gamma_2'$.
For two subspaces $B$ and $B'$ with $B' \subseteq B$, we define a $B'$-namu obtained from a $B$-namu as follows. 
We define the \emph{projection} $\Gamma|_{B'}$ of a $B$-namu $\Gamma=(T,\alpha,\lambda,U)$ on $B'$
as the $B'$-namu $(T,\alpha',\lambda',U')$ such that
\begin{itemize}
\item $U'=U\cap B'$, 
\item $\alpha'(v,e)=\alpha(v,e)\cap B'$ for all incidences $(v,e)$ of~$T$, and
\item $\lambda'(e)=\lambda(e)$ for all edges $e$ of~$T$.
\end{itemize}
Note that the width of $\Gamma|_{B'}$ is equal to the width of $\Gamma$.

\subsection{Compact $\boldsymbol B$-namus}\label{subsec:compact}

In the previous subsection, we defined the data structure $B$-namu to store partial solutions at a node~$v$ 
of a given branch-decomposition $(T^b,\L^b)$. Naively, we can store all partial solutions that can be potentially extended to an optimal branch-decomposition of~$\V$. 
However, if we use the canonical $B$-namu to represent a partial solution at a node~$v$, 
then this representation can be arbitrarily large and there can be too many canonical $B$-namus to be stored. 
In order to bound the size and the number of partial solutions at each node~$v$ of $T^b$, we deploy two strategies; first we narrow down the search space of all 
branch-decompositions of $\V_v$ to the ``well-behaved'' branch-decompositions, and second we compress $B$-namus of such well-behaved branch-decompositions 
to so-called compact $B$-namus. The former is captured by the notion of totally pure branch-decomposition with respect to $(T^b,\L^b)$, which will be explored in the next section. 
In this subsection, we focus on the latter notion.

For a $B$-namu, we define two operations.
First, we introduce the notion of \emph{compressing} a $B$-namu $\Gamma=(T,\alpha,\lambda,U)$.
If there exists a two-edge path $v_0,e_1,v_1,e_2,v_2$ in~$T$ such that 
\begin{enumerate}[(i)]
\item the degree of $v_1$ is $2$ in~$T$,
\item $\alpha(v_{0},e_1)=\alpha(v_{1},e_{2})$, 
$\alpha(v_{1},e_1)=\alpha(v_{2},e_{2})$, and
\item $\lambda(e_1)=\lambda(e_2)$,
\end{enumerate}
then \emph{compressing $\Gamma$ by the path $v_0,e_1,v_1,e_2,v_2$} is an operation to obtain 
a new $B$-namu $\Gamma'=(T',\alpha',\lambda',U)$ such that 
\begin{itemize}
\item $T'$ is a tree obtained from $T$ by contracting $e_1$,
\item $\alpha'(v,e)=\alpha(v,e)$ for each incidence $(v,e)$ in~$T$ with $e\neq e_1$, and 
\item $\lambda'(e)=\lambda(e)$ for each edge $e\neq e_1$ of~$T$.
\end{itemize}
If there exists a path $v_0,e_1,v_1,e_2,v_2,\ldots, e_n,v_n$ in~$T$ with $n\ge3$ such that
\begin{enumerate}[(i)]
\item the degree of $v_i$ is $2$ for every $1\le i\le n-1$,
\item $\alpha(v_{i-1},e_i)=\alpha(v_{i},e_{i+1})$, 
$\alpha(v_{i},e_i)=\alpha(v_{i+1},e_{i+1})$ for every $1\le i\le n-1$, and 
\item $\lambda(e_1)\le\lambda(e_{j})\le\lambda(e_{n})$ for every $2\le j \le n-1$,
\end{enumerate}
then \emph{compressing $\Gamma$ by the path $v_0,e_1,v_1,e_2,v_2,\ldots, e_n,v_n$} is an operation to obtain 
a new $B$-namu $\Gamma'=(T',\alpha',\lambda',U)$ such that
\begin{itemize}
\item $T'$ is a tree obtained from $T$ by contracting all edges $e_2,\ldots,e_{n-1}$, %
\item $\alpha'(v,e)=\alpha(v,e)$ for each incidence $(v,e)$ in~$T$ with $e\notin \{e_2,\ldots,e_{n-1}\}$,
\item $\lambda'(e)=\lambda(e)$ for each edge $e$ of~$T$ with $e\notin \{e_2,\ldots,e_{n-1}\}$.
\end{itemize}

Second, we introduce the notion of \emph{trimming} a $B$-namu $\Gamma=(T,\alpha,\lambda,U)$.
In order to  define trimming, we introduce a few terminology. 
Let $\Gamma=(T,\alpha,\lambda,U)$ be a $B$-namu.
An edge $uv$ of~$T$ is called \emph{degenerate} in $\Gamma$ if \[\alpha(u,uv)=\alpha(v,uv).\]
We say an edge $uv$ of~$T$ \emph{guards} an end $v$ in $\Gamma$ if
\[\alpha(v,uv)\subsetneq \alpha(u,uv).\]
An edge of~$T$ is \emph{guarding} in $\Gamma$ if it guards one of its ends.
A node $w$ of~$T$ is \emph{blocked} by an edge $e$ in $\Gamma$ if 
$w$ is not an end of $e$ and 
$e$ guards its end closer to $w$ in~$T$.
An edge~$f$ of~$T$ is \emph{blocked} by an edge~$e$ 
if at least one of the ends of $f$ is blocked by $e$ in~$\Gamma$.
A path $xyz$ of length $2$ in~$T$ is called a \emph{blocking path} in~$\Gamma$ if 
\[\alpha(x,xy)=\alpha(y,yz) \text{ and }\alpha(z,zy)=\alpha(y,yx)\]
and 
neither $xy$ nor $yz$ is degenerate or guarding in $\Gamma$.
A node~$v$ of~$T$ is \emph{blocked} by a blocking path $xyz$ in $\Gamma$
if $v\neq y$ and $T-x-z$ has a path from $v$ to $y$.
An edge $e$ of~$T$ is \emph{blocked} by a blocking path $xyz$ in $\Gamma$
if %
some end of $e$ is blocked by $xyz$. 
A node~$v$ or an edge $e$ said to be \emph{blocked} in $\Gamma$ if it is blocked by some guarding edge or some blocking path in $\Gamma$. 
We omit ``in $\Gamma$'' if it is clear from the context.

\emph{Trimming $\Gamma$} is an operation to obtain a new $B$-namu $(T',\alpha',\lambda',U)$ by one of the following.
\begin{enumerate}[(1)]
\item If $T$ has degenerate edges, then $T'$ is a tree with a single node $u$.

\item If $T$ has no degenerate edges, then $T'$ is a tree obtained by removing all blocked nodes, and
\begin{itemize}
\item $\alpha'(v,e) = \alpha(v,e)$ for every incidence $(v,e)$ of $T'$, and 
\item $\lambda'(e)=\lambda(e)$ for every edge $e$ of $T'$.
\end{itemize}

\end{enumerate} 
We define the \emph{trim of $\Gamma$}, denoted by $\trim(\Gamma)$, as the $B$-namu 
obtained %
by trimming~$\Gamma$.
We say that a $B$-namu $\Gamma$ is \emph{trimmed} if $\Gamma=\trim(\Gamma)$.
Note that $\trim(\Gamma)$ has no blocked nodes.

We define the \emph{compactification of $\Gamma$}, denoted by $\tau(\Gamma)$, 
as the $B$-namu obtained from $\trim(\Gamma)$ 
by repeatedly applying compressing until no further compressing can 
be applied.

\begin{PROP}\label{prop:compactwelldefined}
  The compactification of a $B$-namu $\Gamma$ is well defined.
\end{PROP}
\begin{proof}
The order of compressing does not change the resulting $B$-namu $\tau(\Gamma)$.
\end{proof}

We say that a $B$-namu $\Gamma$ is \emph{compact} if $\Gamma=\tau(\Gamma)$.
Let $U_k(B)$ be the set of all compact $B$-namus $\Gamma$ of width at most~$k$
such that $V(T(\Gamma))=\{1,2,\ldots,n\}$ for some integer $n$.

\subsection{A sum of $\boldsymbol B$-namus}\label{subsec:sum}
At each internal node of $T^b$, we want to compute partial solutions using the 
partial solutions computed at its two children. In this subsection, we introduce the operation \emph{sum} of $B$-namus 
to combine two partial solutions. 
 
For two subcubic trees $T$ and $T^*$, a \emph{$T$-model in $T^*$} is a subtree $\eta$ of $T^*$ isomorphic to a subdivision of~$T$. %
For brevity we regard $\eta$ as a function from $E(T^*)\cup\{\emptyset\}$ to $E(T)\cup \{\emptyset\}$ such that 
$\eta(\emptyset)=\emptyset$, 
$\eta(e)=\emptyset$ if $e$ is not in $E(\eta)$,
and 
$\eta(e)=f$ if $e$ is on the path of $\eta$ obtained by subdividing $f$.
For a $T$-model $\eta$ in $T^*$, we define $\etabar:V(T)\to V(T^*)$ so that $\etabar(v)=w$ 
if a node $w$ in $\eta$ is a branch node corresponding to $v$.
See Figure~\ref{fig:model} for an illustration.
\begin{figure}
  \centering
  \tikzstyle {v}=[draw,circle,inner sep=1.2pt]
  \tikzstyle {u}=[draw,circle,fill=black,inner sep=1.2pt]
  \begin{tikzpicture}
	\node [u,label=below:$u_1$] at (0,2) (u1){};
	\node [u,label=below right:$u_2$] at (2,2) (u2){};
	\node [u,label=below right:$u_3$] at (4,2) (u3){};
	\node [u,label=below:$u_4$] at (6,2) (u4){};
	\node [u,label=right:$u_5$] at (2,0) (u5){};
	\node [u,label=right:$u_6$] at (4,0) (u6){};
    \draw node [label=below:$T^*$] at (1,2) {};
    \draw (u1)--(u2)--(u3)--(u4);
    \draw (u2)--(u5);
    \draw (u3)--(u6);
    \node [v,label=above:$v_1$] at (0-0.15,2+0.15) (v1){};
    \node [v,label=above:$v_2$] at (4-0.15,2+0.15) (v2){};
    \node [v,label=above:$v_3$] at (6-0.15,2+0.15) (v3){};
    \node [v,label=left:$v_4$] at (4-0.15,0+0.15) (v4){};
	\draw node [label=above:$T$] at (3,2) {};
    \draw[dashed] (v1)--(v2)--(v3);
    \draw[dashed] (v2)--(v4);
  \end{tikzpicture}
  \caption{A function $\eta$ is a $T$-model in $T^*$ if
$\eta(u_1u_2)=\eta(u_2u_3)=v_1v_2$, $\eta(u_3u_4)=v_2v_3$, $\eta(u_3u_6)=v_2v_4$, 
and $\eta(u_2u_5)=\eta(\emptyset)=\emptyset$. Nodes $u_1,u_3,u_4,u_6$ are branch nodes in $\eta$.
Dashed lines represent $T$ and solid lines represent $T^*$.}
  \label{fig:model}
\end{figure}

A $T$-model $\eta$ in $T^*$ induces a function $\vec\eta:\I(T^*)\to \I(T)$ as follows.
\begin{itemize}
\item For an incidence $(v,e)$ of $T^*$, if $\eta(e)\neq\emptyset$, then let $\bar v$ be a branch node with the minimum distance to $v$ in $T^*-e$ and let $v'$ be the node of~$T$ such that $\etabar(v')=\bar v$. Then
	$\vec\eta(v,e)=(v',\eta(e))$.
\item For an incidence $(v,e)$ of $T^*$, if $\eta(e)=\emptyset$, then let $C_v$ be the component $T^*-e$ containing~$v$ and 
\begin{align*}
\vec\eta(v,e)=
\begin{cases}
(0, \emptyset) \qquad &%
\text{if $C_v$ has no node in $\etabar(V(T))$,}\\
(*, \emptyset) \qquad &\text{otherwise.}
\end{cases}
\end{align*}
\item $\vec\eta(0,\emptyset)=(0,\emptyset)$ and $\vec\eta(*,\emptyset)=(*,\emptyset)$.
\end{itemize}
Let $\eta^{-1}(e)=\{e':\eta(e')=e\}$ be the preimage of $e$. Then for an edge $e$ of~$T$, $\eta^{-1}(e)$ is a path in $T^*$.
A $B$-namu $\Gamma^*=(T^*,\alpha^*,\lambda^*,U^*)$ is 
an \emph{extension} of a $B$-namu $\Gamma=(T,\alpha,\lambda,U)$
if there exists a $T$-model $\eta$ in $T^*$ such that 
$U^*=U$, $\alpha^*(v,e)=\alpha(\vec\eta(v,e))$, $\lambda^*(e)=\lambda(\eta(e))$ 
for every incidence $(v,e)$ of $T^*$ with $\eta(e)\neq\emptyset$.
In this case, we say that $\eta$ \emph{extends} $\Gamma$ to $\Gamma^*$.

For three subcubic trees $T_1$, $T_2$, and $T^+$, a \emph{$(T_1,T_2)$-model in $T^+$} is a pair $(\eta_1,\eta_2)$ of a $T_1$-model $\eta_1$ and a $T_2$-model $\eta_2$ in $T^+$ satisfying the following:
\begin{enumerate}[(i)]
\item For each $i=1,2$, if two edges $e\in E(\eta_i)$ and $f\in E(T^+)\setminus E(\eta_i)$ share a node~$v$ of $T^+$, then $v$ is a subdividing node of $\eta_i$. %
\item $A(T^+)$ is the disjoint union of $\etabar_1(A(T_1))$ and $\etabar_2(A(T_2))$.
\item If $v$ has degree at most $2$ in $T^+$, then it is a branch node in exactly one of $\eta_1$ and $\eta_2$.
\end{enumerate}
Figure~\ref{fig:2model} illustrates the definition of a $(T_1,T_2)$-model in $T^+$.

\begin{figure}
  \centering
  \tikzstyle {u}=[draw,inner sep=2pt]
  \tikzstyle {v}=[draw,fill=black,inner sep=2pt]
  \tikzstyle {w}=[draw,circle,fill=black,inner sep=1pt]
  \tikzstyle {s}=[draw,circle,inner sep=1.2pt]
  \begin{tikzpicture}[xscale=0.85]
    \node [u] at (0,0) (u7){};
    \node [u] at (0,2) (u6){};  
    \node [u] at (3,2) (u5){};  
    \node [u] at (3,0) (u4){};  
    \node [u] at (2,1) (u3){};  
    \node at (2+0.04,1-0.3) {$u_3$};
    \node [u] at (1.5,1) (u2){};  
    \node at (1.5+0.04,1-0.3) {$u_2$};
    \node [u] at (1,1) (u1){};
    \node at (1+0.04,1-0.3) {$u_1$};
    \draw (u7)--(u1)--(u6);  
    \draw (u4)--(u3)--(u5);  
    \draw (u1)--(u2)--(u3); 
    \draw node [label=below:$T_1$] at (1.5,0-0.5) {};
    \node [v] at (2.5+5,0.5) (v8){};
    \node [v] at (0+5,0) (v7){};
    \node [v] at (0+5,2) (v6){};  
    \node [v] at (3+5,2) (v5){};  
    \node [v] at (3+5,0) (v4){};  
    \node [v] at (2+5,1) (v3){};  
    \node at (2+5+0.04,1-0.3) {$v_3$};
    \node [v] at (1.5+5,1) (v2){};  
    \node at (1.5+5+0.04,1-0.3) {$v_2$};
    \node [v] at (1+5,1) (v1){};  
    \node at (1+5+0.04,1-0.3) {$v_1$};
    \draw (v7)--(v1)--(v6);  
    \draw (v4)--(v8)--(v3)--(v5);  
    \draw (v1)--(v2)--(v3);  
    \draw node [label=below:$T_2$] at (1.5+5,0-0.5) {};
    \node [v] at (0+11-0.5,0) (w7v){};
    \node [u] at (0+11,0-0.5) (w7u){};
    \node [s] at (0+11,0) (w7){};
    \node [v] at (0+11,2+0.5) (w6v){};
    \node [u] at (0+11-0.5,2) (w6u){};
    \node [s] at (0+11,2) (w6){};  
    \node [v] at (3+11,2+0.5) (w5v){};
    \node [u] at (3+11+0.5,2) (w5u){};
    \node [s] at (3+11,2) (w5){};  
    \node [v] at (3+11+0.5,0) (w4v){};
    \node [v] at (3+11+0.25,0) (w8v){};
    \node [u] at (3+11,0-0.5) (w4u){}; 
    \node [s] at (3+11,0) (w4){}; 
    \node [w] at (2+11,1) (w3){};
    \node at (2+11+0.04,1-0.3) {$w_4$};
    \node [v] at (1.7+11,1) (w2v){};
    \node at (1.7+11+0.05,1+0.25) {$w_3$};
    \node [u] at (1.3+11,1) (w2u){};
    \node at (1.3+11+0.04,1-0.3) {$w_2$};
    \node [w] at (1+11,1) (w1){};
    \node at (1+11+0.05,1+0.25) {$w_1$};
    \draw (w7v)--(w7)--(w7u);  
    \draw (w6v)--(w6)--(w6u);  
    \draw (w5v)--(w5)--(w5u);  
    \draw (w4v)--(w4)--(w4u); 
    \draw (w7)--(w1)--(w6);   
    \draw (w4)--(w3)--(w5);
    \draw (w1)--(w2u)--(w2v)--(w3);
    \draw node [label=below:$T^+$] at (1.5+11,0-0.5) {};
  \end{tikzpicture}
  \caption[Trees $T_1$, $T_2$, and $T^+$ that give a $(T_1,T_2)$-model $(\eta_1,\eta_2)$ in $T^+$.]{Trees $T_1$, $T_2$, and $T^+$ that give a $(T_1,T_2)$-model $(\eta_1,\eta_2)$ in $T^+$
  such that $\eta_1(w_1w_2)=u_1u_2$, $\eta_1(w_2w_3)=\eta_1(w_3w_4)=u_2u_3$, $\eta_2(w_1w_2)=\eta_2(w_2w_3)=v_1v_2$, $\eta_2(w_3w_4)=v_2v_3$.
   In $T^+$, the nodes represented by \tikz \node [s] {}; are the subdividing nodes described in (i).}
  \label{fig:2model}
\end{figure}

\begin{LEMsum}
  Given two $B$-namus $\Gamma_1=(T_1,\alpha_1,\lambda_1,U_1)$, $\Gamma_2=(T_2,\alpha_2,\lambda_2,U_2)$, 
  and a $(T_1,T_2)$-model $(\eta_1,\eta_2)$ in a tree $T^+$,
  we define $\alpha^+$, $\lambda^+$, and $U^+$ as follows:
  \begin{enumerate}[(i)]
  \item $U^+=U_1+U_2$,
  \item $\alpha^+=\alpha_1\circ \vec\eta_1 +\alpha_2\circ \vec\eta_2$, 
  \item $\lambda^+(\emptyset)=0$, and 
  \item for all $e=uv\in E(T)$, %
  \begin{align*}
  \lambda^+(e)&=\lambda_1\circ\eta_1(e) + \lambda_2\circ\eta_2(e) 
  -\dim (\alpha_1(\vec\eta_1(v,e))\cap\alpha_2(\vec\eta_2(v,e))) 
  \\
  &\quad -\dim (\alpha_1(\vec\eta_1(u,e))\cap\alpha_2(\vec\eta_2(u,e)))
  +\dim (U_1 \cap U_2).
  \end{align*}
  \end{enumerate}
  Then $\Gamma^+=(T^+,\alpha^+,\lambda^+,U^+)$ is a $B$-namu.
  \end{LEMsum}

  Such $\Gamma^+$ is called the \emph{sum} of $\Gamma_1$ and $\Gamma_2$ by $(\eta_1,\eta_2)$, 
and we will denote by $\Gamma^+=\Gamma_1+_{(\eta_1,\eta_2)}\Gamma_2$.
We say that $(\eta_1,\eta_2)$ \emph{co-extends} $\Gamma_1$ and $\Gamma_2$ to $\Gamma^+$.
A $B$-namu $\Gamma^+$ is simply called a \emph{sum} of $\Gamma_1$ and $\Gamma_2$
if there exists a $(T(\Gamma_1),T(\Gamma_2))$-model $(\eta_1,\eta_2)$ in $T(\Gamma^+)$ 
such that $\Gamma^+=\Gamma_1+_{(\eta_1,\eta_2)}\Gamma_2$.
Given $B$-namus $\Gamma_1$ and $\Gamma_2$, 
let us denote by $\Gamma_1\tplus \Gamma_2$ the set of all sums of $\Gamma_1$ and $\Gamma_2$.

\subsection{$\boldsymbol k$-safe extensions of a $\boldsymbol B$-namu}\label{subsec:ksafe}
For a nonnegative integer $k$, a $B$-namu $\Gamma'=(T',\alpha',\lambda',U')$ is 
a \emph{$k$-safe} extension of a $B$-namu $\Gamma=(T,\alpha,\lambda,U)$
if there exists a $T$-model $\eta$ in $T'$
such that 
\begin{itemize}
\item
$U'=U$, $\alpha'(v,e)=\alpha(\vec\eta(v,e))$, $\lambda'(e)=\lambda(\eta(e))$ 
for every incidence $(v,e)$ of $T'$ with $\eta(e)\neq\emptyset$ and 
\item 
$\lambda'(e) + \dim U'- \dim \alpha'(v,e) \le k$ 
for every incidence $(v,e)$ of $T'$ with $\vec\eta(v,e)=(*,\emptyset)$.
\end{itemize}

\subsection{Transcript}\label{subsec:transcript}
A \emph{transcript} of $(T,\L)$ is a pair $\Lambda=(\{\B_v\}_{v\in V(T)},
\allowbreak
\{\B_v'\}_{v\in V(T)})$
of sets of ordered bases 
$\B_v$ and $\B_v'$ of subspaces $B_v=\spn{\B_v}$ and $B_v'=\spn{\B_v'}$ of $\F^r$,
respectively, such that %
\begin{itemize}
\item %
the first $\abs{\B_v}$ elements of $\B_v'$ are precisely $\B_v$ for each node~$v$,
\item $\spn{\B_v}=\spn{\V_v}\cap \spn{\V-\V_v}$ for each node~$v$,
\item $\spn{\B_v'}=\spn{\B_{w_1}}+\spn{\B_{w_2}}$  for each node~$v$ having two children $w_1$
and $w_2$,
\item $\spn{\B_v'}=\spn{\B_v}$  for each leaf $v$.
\end{itemize}
If a node $w$  of~$T$ is a parent of a node~$v$ of~$T$, 
then
there exists the unique $\abs{\B_w'}\times \abs{\B_v}$ matrix $T_v$ over
$\F$
such that 
\[
T_v [x]_{\B_v}=[x]_{\B_w'}
\]
for all $x\in \B_v$.
This matrix $T_v$ is called the \emph{transition matrix} of $\Lambda$
at a node~$v$. 
(For the root node $\rt$, let $T_{\rt}$ be the null matrix.)
The \emph{order} of a transcript $\Lambda$ is the maximum of
$\abs{\B_v'}$ for all nodes $v$.

The next theorem states that a transcript of a given branch-decomposition can be efficiently computed, 
which is a prerequisite to preparing the input of the  algorithm in the next subsection. 
\begin{THMcomputebases}
  Let $\V$ be a subspace arrangement of $\F^r$ represented by an $r\times m$ matrix~$M$ in reduced row echelon form with no zero rows
  such that each $V\in \V$ has dimension at most $\theta$.
  Let $n=\abs{\V}$.  
  Given a rooted branch-decomposition $(T,\L)$ of~$\V$,
  in time $O(\theta^3n^2)$, Algorithm~\ref{alg:bases} correctly computes a basis of 
  $ \spn{\V_v}\cap \spn{\V-\V_v}$
  for all nodes $v$ of~$T$ 
  or confirms that $(T,\L)$ has width larger than $\theta$.
  In addition, if $(T,\L)$ has width at most $\theta$, then  we can compute the transcript $\Lambda=(\{\B_v\},\{\B_v'\})$ of $(T,\L)$ with its transition matrices in time $O(\theta^3n^2)$.
\end{THMcomputebases}
\subsection{The main algorithm}\label{subsec:computing}

We present the key step of the algorithm for \textsc{Branch-Width Compression} assuming the following are given:
\begin{itemize}
\item a rooted branch-decomposition $(T^b, \L^b)$ of~$\V$ of width at most $\theta$, 
\item a set of transition matrices $\{T_v\}_{v\in V(T^b)}$ of some transcript $\Lambda$ of $(T^b,\L^b)$.  
\end{itemize}
The algorithm does not need to look at $\Lambda$.
We note that the above input can be represented in  $O(\theta^2 \cdot \abs{V(T^b)} \cdot \log\abs{\F})$ space.

\begin{algorithm}[t]
  \caption{Constructing the full sets}      \label{alg:fullset}
\begin{algorithmic}[1]
\Procedure{full-set}{$\V,k,(T^b,\L^b),\{T_v\}_{v\in V(T^b)}$}
\Repeat
\State choose an unmarked node~$x$ of $T^b$ farthest from the root
\If{$x$ is a leaf} 
\State let $\Delta_x=(T,\alpha,\lambda,U)$ be the $B_x$-namu such that $T$ is a tree having a single node and $U=B_x$
\State we represent $U$ by the $\abs{\B_x}\times \abs{\B_x}$ identity matrix
\State $\FF_x \gets \{\Delta_x\}$ \label{line:initial1} \Comment{Initialization}
\ElsIf{$x$ is an internal node with two children $x_1$ and $x_2$}
\State convert representations of subspaces in $\FF_{x_1}$ and $\FF_{x_2}$ to representations 
with respect to $\B_x'$ by using $T_{x_1}$ and $T_{x_2}$
\State $\FF_x^+\gets \FF_{x_1}\fplus \FF_{x_2}$ \Comment{Join} \label{line:join}
\State $\FF_x^S\gets \FF_x^+|_{B_x}$ \Comment{Shrink}  \label{line:shrink}
\State $\FF_x^T\gets \{\trim(\Gamma):\Gamma\in\FF_x^S\text{ and }\Gamma\text{ is a $k$-safe extension of }\trim(\Gamma)\}$ \Comment{Trim} \label{line:trim}
\State $\FF_x\gets \{\Delta\in U_k(B_x):\Gamma'\tle\Delta\text{ for some $\Gamma'\in\FF_x^T$}\}$ \Comment{Compare} \label{line:compare}
\EndIf
\State mark $x$
\Until{all nodes in~$T$ are marked}
\EndProcedure
\end{algorithmic}
\end{algorithm}

For a leaf $\ell$ of $T^b$ and a subspace $V=\L^b(\ell)$,
a branch-decomposition $(T,\L)$ of $\V_\ell$ is unique up to isomorphism of~$T$.
Let $\Delta_{\ell}=(T,\alpha, \lambda, U)$ be 
the canonical $B_{\ell}$-namu of $(T,\L)$ of $\V_\ell=\{V\}$.
Since $T$ has no edge, it is clear that 
$\alpha(*,\emptyset)=U$, $\alpha(0,\emptyset)=\{0\}$, $\lambda(\emptyset)=0$, and $U=B_{\ell}$.
Let $\FF_\ell=\{\Delta_\ell\}$. 

Let $x$ be a node of $T^b$ with two children $x_1$, $x_2$ such that $\FF_{x_1}$ and $\FF_{x_2}$ are already computed. We will describe how to compute $\FF_x$.
For two sets ${\mathcal R}_1$ and ${\mathcal R}_2$ of $B$-namus, we define ${\mathcal R}_1 \fplus {\mathcal R}_2$ as the set \[\bigcup_{\Gamma_1 \in {\mathcal R}_1,\Gamma_2 \in {\mathcal R}_2}\Gamma_1 \tplus \Gamma_2.\]
Note that for two children $x_1$ and $x_2$ of a node~$x$ in $T^b$,
when we compute $\FF_{x_1}\oplus \FF_{x_2}$
from a set $\FF_{x_1}$ of $B_{x_1}$-namus and a set $\FF_{x_2}$ of $B_{x_2}$-namus, 
we regard $\FF_{x_1}$ and $\FF_{x_2}$ as the sets of $(B_{x_1}+B_{x_2})$-namus.
Thus, $\FF_{x_1}\fplus\FF_{x_2}$ is well defined.

If $B'$ is a subspace of $B$, then 
we define ${\mathcal R}|_{B'}$ as the set of projections $\Gamma|_{B'}$ for all $\Gamma \in {\mathcal R}$. 
Let $\FF_x^T$ be the set of all $B_x$-namus $\Gamma'=\trim(\Gamma)$ 
for all $\Gamma\in(\FF_{x_1}\fplus \FF_{x_2})|_{B_x}$ such that $\Gamma$ is a $k$-safe extension of $\Gamma'$.
Let \[\FF_x= \{\Delta\in U_k(B_x):\Gamma'\tle\Delta\text{ for some $\Gamma'\in\FF_x^T$}\}.\]

We summarize the procedure to compute $\FF_x$ for all nodes $x$ of $T^b$ in Algorithm~\ref{alg:fullset}.
The following proposition implies that Algorithm~\ref{alg:fullset} correctly decides whether there exists a branch-decomposition of an input subspace arrangement $\V$ having width at most~$k$.

\begin{PROPFSroot}
  Let $k$ be a nonnegative integer. 
  Let $(T^b,L^b)$ be a rooted branch-decomposition of a subspace arrangement $\V$, and for each node~$x$ of $T^b$, let $\FF_x$ be the set computed by Algorithm~\ref{alg:fullset}.
  The branch-width of~$\V$ is at most~$k$ 
  if and only if $\FF_{r}\neq \emptyset$ at the root node $r$ of~$T^b$.
\end{PROPFSroot}

Though Proposition~\ref{prop:FSroot} holds without any assumptions on the input size, to achieve the running time claimed in Theorem~\ref{thm:summary-brw}, we will proprocess the input to make sure that all the subspaces given in the input have small dimension. This will be discussed in Subsection~\ref{subsec:preprocessing}.

\section{Pure branch-decompositions}\label{sec:pure}
Let us start this section with a few terminology on branch-decompositions.
We are going to assume that 
a subspace arrangement $\V$ and its rooted branch-decomposition $(T^b,\L^b)$ are given.
Generally $(T^b,\L^b)$ would have width larger than $k$, as our goal is to find a branch-decomposition of width at most~$k$.

For two nodes $x$, $y$ of $T^b$, %
we say that $x\le y$ if 
either $x=y$ or 
$x$ is a descendant of $y$. Clearly $\le$ is a partial order on $V(T^b)$.
We write $x<y$ if $x\le y$ and $x\neq y$.
For a node~$x$ of $T^b$, 
let $\V_x$ be the set of all subspaces $\L^b(\ell)$ where $\ell$ is a leaf of $T^b$ with $\ell \le x$ and 
let $B_x=\spn{\V_x}\cap \spn{\V-\V_x}$. 

Let $(T,\L)$ be a branch-decomposition of $\V_0\subseteq \V$. 
For an edge $uv$ of a tree $T$ and a node $w$ of~$T$, we say that the ordered pair $(u,v)$ \emph{points toward} 
$w$ if the unique path of~$T$ from $u$ to $w$ visits $v$. Note that this implies $w\neq u$.
We define  $\L(T,u,v)$ to be the set of all $\L(w)$ for each leaf $w$ of~$T$ 
such that $(u,v)$ points toward $w$.
In addition, if $\V_x\subseteq \V_0$ for a node~$x$ of $T^b$, then 
we write $\L_x(T,u,v)=\L(T,u,v)\cap \V_x$ 
and $\lambda_x^{(T,\L)}(uv)=\dim \spn{\L_x(T,u,v)}\cap \spn{\L_x(T,v,u)}$.

\begin{figure}[t]
  \begin{center}
  \tikzstyle{v}=[circle,draw,fill=black,inner sep=0pt,minimum width=2pt]
  \tikzstyle{box}=[rectangle,draw=black,thick, minimum size=3mm]
  \begin{tikzpicture}[level distance=0.5cm,
    level 1/.style={sibling distance=3cm},
    level 2/.style={sibling distance=1.5cm},    level 3/.style={sibling distance=0.5cm}]
    \node [v]{}
      child { node[v,label=$x$](x){}
        child { node[v,label=$y$](y){}
          child { node[v,label=below:$1$](z1){}}
          child { node[v,label=below:$2$](z2){}}}
        child { node[v,label=$y'$](yy){}
         child { node[v,label=below:$3$](z3){}}
          child { node[v,label=below:$4$](z4){}}}}
      child { node[v,label=$x'$](xx){}
        child { node[v,label=$y''$](yyy){}
          child { node[v,label=below:$5$](z5){}}
          child { node[v,label=below:$6$](z6){}}}
        child { node[v]{}
          child { node[v,label=below:$7$](z7){}}
          child { node[v,label=below:$8$](z8){}}}};
    \draw (z1)
    +(0,-7mm) node[box,fill=red]{}
    +(0,-10mm) node[box,fill=red]{}
    +(0,-13mm) node[box,fill=red]{}
    +(0,-16mm) node[box,fill=red]{};
    \draw (z2)
    +(0,-7mm) node[box]{}
    +(0,-10mm) node[box]{}
    +(0,-13mm) node[box,fill=red]{}
    +(0,-16mm) node[box,fill=red]{};
    \draw (z3)
    +(0,-7mm) node[box,fill=red]{}
    +(0,-10mm) node[box]{}
    +(0,-13mm) node[box]{}
    +(0,-16mm) node[box]{};
    \draw (z4)
    +(0,-7mm) node[box]{}
    +(0,-10mm) node[box]{}
    +(0,-13mm) node[box,fill=red]{}
    +(0,-16mm) node[box]{};
    \draw (z5)
    +(0,-7mm) node[box,fill=red]{}
    +(0,-10mm) node[box,fill=red]{}
    +(0,-13mm) node[box]{}
    +(0,-16mm) node[box]{};
    \draw (z6)
    +(0,-7mm) node[box]{}
    +(0,-10mm) node[box,fill=red]{}
    +(0,-13mm) node[box,fill=red]{}
    +(0,-16mm) node[box]{};
    \draw (z7)
    +(0,-7mm) node[box,fill=red]{}
    +(0,-10mm) node[box]{}
    +(0,-13mm) node[box,fill=red]{}
    +(0,-16mm) node[box]{};
    \draw (z8)
    +(0,-7mm) node[box]{}
    +(0,-10mm) node[box,fill=red]{}
    +(0,-13mm) node[box,fill=red]{}
    +(0,-16mm) node[box]{};  
    \draw (x) +(-.1,+.1) [draw=blue,bend right,arrows=->] to +(-3,0) node  (bx){};
    \draw (y) +(-.1,+.1) [draw=blue,bend right,arrows=->] to +(-1,.5) node  (by){};;
    \draw (yy) +(.1,.1) [draw=blue,bend left,arrows=->] to +(6,1) node  (byy){};;
    \draw (yyy) +(.1,.1) [draw=blue,bend left,arrows=->] to +(3,0) node  (byyy){};;
    \draw (xx) +(.1,.1) [draw=blue,bend left,arrows=->] to +(3,.5) node  (bxx){};;
    \draw (bx)   node [label=below:{$B_x$},inner sep=0pt,minimum width=2pt] {}
    +(0,-7mm) node[box,fill=red]{}
    +(0,-10mm) node[box,fill=red]{}
    +(0,-13mm) node[box]{}
    +(0,-16mm) node[box]{}
    +(3mm,-7mm) node[box]{}
    +(3mm,-10mm) node[box,fill=red]{}
    +(3mm,-13mm) node[box,fill=red]{}
    +(3mm,-16mm) node[box]{}    
    ;
    \draw (by)   node [label=below:{$B_y$},inner sep=0pt,minimum width=2pt] {}
    +(0,-7mm) node[box,fill=red]{}
    +(0,-10mm) node[box,fill=red]{}
    +(0,-13mm) node[box]{}
    +(0,-16mm) node[box]{} ;   
    \draw (byy)   node [label=below:{$B_{y'}$},inner sep=0pt,minimum width=2pt] {}
    +(0,-7mm) node[box,fill=red]{}
    +(0,-10mm) node[box]{}
    +(0,-13mm) node[box,fill=red]{}
    +(0,-16mm) node[box]{};   
    \draw (byyy)   node [label=below:{$B_{y''}$},inner sep=0pt,minimum width=2pt] {}
    +(0,-7mm) node[box,fill=red]{}
    +(0,-10mm) node[box,fill=red]{}
    +(0,-13mm) node[box]{}
    +(0,-16mm) node[box]{}
    +(3mm,-7mm) node[box]{}
    +(3mm,-10mm) node[box,fill=red]{}
    +(3mm,-13mm) node[box,fill=red]{}
    +(3mm,-16mm) node[box]{}    
    ;       
    \draw (bxx)   node [label=below:{$B_{x'}$},inner sep=0pt,minimum width=2pt] {}
    +(0,-7mm) node[box,fill=red]{}
    +(0,-10mm) node[box,fill=red]{}
    +(0,-13mm) node[box]{}
    +(0,-16mm) node[box]{}
    +(3mm,-7mm) node[box]{}
    +(3mm,-10mm) node[box,fill=red]{}
    +(3mm,-13mm) node[box,fill=red]{}
    +(3mm,-16mm) node[box]{}    
    ;             
  \end{tikzpicture}
\end{center}
  \caption{An example of a rooted  branch-decomposition {$(T^b,\mathcal{L}^b)$}
  of $\V=\{V_1,V_2,\ldots,V_8\}$
  and the boundary spaces $B_x$, $B_y$, $B_{y'}$. 
  Each grid represents the column space of the $0$-$1$ matrix over the binary field in which the red box corresponds to $1$
  and the white box corresponds to $0$.
  }\label{fig:bd}
\end{figure}
\begin{figure}
  \tikzstyle{v}=[circle,draw,fill=black,inner sep=0pt,minimum width=2pt]
  \begin{center} 
  \subfloat[$x$\=/mixed pairs (red arrows) and $x$\=/crossing edges (blue thick edges)]{
    \begin{tikzpicture}[grow cyclic,level distance=1cm,sibling angle=120,scale=.8]
      \node [v](u) {}
      child { 
        node [v,label=$6$] (z6) {}
      }
      child {
        node [v,label={\color{red}$1$}] (z1) {}
      }
      child {
        node [v] (v) {}
          child {
            node [v] (w) {}
              child {
                node [v,label={\color{red}$4$}] (z4) {}
              }
              child {
                node [v,label=left:$5$] (z5) {}
              }
          }
          child {
            node [v] (w2) {}
              child {
                node [v,label=$7$] (z7) {}
              }
              child {
                node [v] (w3) {}
                  child {
                    node [v,label={\color{red}$2$}] (z2){}
                  }
                  child {
                    node [v,label={\color{red}$3$}] (z3){}
                  }
              }
          }
      };
      \foreach \s/\t in {v/w,z6/u,w3/w2}
        \draw [->,red] ([shift=(60:4mm),shift=(150:1mm)]\s.center) -- ([shift=(60:-4mm),shift=(150:1mm)]\t.center);
      \foreach \s/\t in {w/v}      
        \draw [<-,red] ([shift=(60:4mm),shift=(150:-1mm)]\t.center) -- ([shift=(60:-4mm),shift=(150:-1mm)]\s.center);
      \foreach \s/\t in {v/w2,z4/w,z1/u}
        \draw [->,red] ([shift=(180:4mm),shift=(150+120:1mm)]\s.center) -- ([shift=(180:-4mm),shift=(150+120:1mm)]\t.center);
      \foreach \s/\t in {w2/v,z2/w3}      
        \draw [<-,red] ([shift=(180:4mm),shift=(150+120:-1mm)]\t.center) -- ([shift=(180:-4mm),shift=(150+120:-1mm)]\s.center); 
      \foreach \s/\t in {v/u,z5/w,z7/w2}
        \draw [->,red] ([shift=(-60:4mm),shift=(-150:-1mm)]\s.center) -- ([shift=(-60:-4mm),shift=(-150:-1mm)]\t.center);
      \foreach \s/\t in {u/v,z3/w3}      
        \draw [<-,red] ([shift=(-60:4mm),shift=(-150:1mm)]\t.center) -- ([shift=(-60:-4mm),shift=(-150:1mm)]\s.center);
      \foreach \s/\t in {v/w,v/u,v/w2}
        \draw [very thick,blue] (\s)--(\t);
     \end{tikzpicture}}
     \quad
     \subfloat[Edges cutting $\V_x$ (red thick edges)]{
      \begin{tikzpicture}[grow cyclic,level distance=1cm,sibling angle=120,scale=.8]
        \node [v](u) {}
        child { 
          node [v,label=$6$] (z6) {}
        }
        child {
          node [v,label={\color{red}$1$}] (z1) {}
        }
        child {
          node [v] (v) {}
            child {
              node [v] (w) {}
                child {
                  node [v,label={\color{red}$4$}] (z4) {}
                }
                child {
                  node [v,label=left:$5$] (z5) {}
                }
            }
            child {
              node [v] (w2) {}
                child {
                  node [v,label=$7$] (z7) {}
                }
                child {
                  node [v] (w3) {}
                    child {
                      node [v,label={\color{red}$2$}] (z2){}
                    }
                    child {
                      node [v,label={\color{red}$3$}] (z3){}
                    }
                }
            }
        };
        \foreach \s/\t in {v/w,v/u,v/w2,w/z4,u/z1,w2/w3,w3/z3,w3/z2}
        \draw [very thick,red] (\s)--(\t);
       \end{tikzpicture}}
       \quad
       \subfloat[(proper) $x$\=/degenerate edges]{
        \begin{tikzpicture}[grow cyclic,level distance=1cm,sibling angle=120,scale=.8]
          \node [v](u) {}
          child { 
            node [v,label=$6$] (z6) {}
          }
          child {
            node [v,label={\color{red}$1$}] (z1) {}
          }
          child {
            node [v] (v) {}
              child {
                node [v] (w) {}
                  child {
                    node [v,label={\color{red}$4$}] (z4) {}
                  }
                  child {
                    node [v,label=left:$5$] (z5) {}
                  }
              }
              child {
                node [v] (w2) {}
                  child {
                    node [v,label=$7$] (z7) {}
                  }
                  child {
                    node [v] (w3) {}
                      child {
                        node [v,label={\color{red}$2$}] (z2){}
                      }
                      child {
                        node [v,label={\color{red}$3$}] (z3){}
                      }
                  }
              }
          };
          \foreach \s/\t in {w3/w2,w2/v}
            \draw [very thick,blue] (\s)--(\t);
         \end{tikzpicture}}  
       \quad
       \subfloat[(improper) $y$\=/degenerate edges]{
        \begin{tikzpicture}[grow cyclic,level distance=1cm,sibling angle=120,scale=.8]
          \node [v](u) {}
          child { 
            node [v,label=$6$] (z6) {}
          }
          child {
            node [v,label={\color{red}$1$}] (z1) {}
          }
          child {
            node [v] (v) {}
              child {
                node [v] (w) {}
                  child {
                    node [v,label={$4$}] (z4) {}
                  }
                  child {
                    node [v,label=left:$5$] (z5) {}
                  }
              }
              child {
                node [v] (w2) {}
                  child {
                    node [v,label=$7$] (z7) {}
                  }
                  child {
                    node [v] (w3) {}
                      child {
                        node [v,label={\color{red}$2$}] (z2){}
                      }
                      child {
                        node [v,label={$3$}] (z3){}
                      }
                  }
              }
          };
          \foreach \s/\t in {z2/w3,w3/w2,w2/v,v/u,u/z1}
            \draw [very thick,blue] (\s)--(\t);
         \end{tikzpicture}}   
       \quad
         \subfloat[(improper) $y'$-degenerate edges]{
          \begin{tikzpicture}[grow cyclic,level distance=1cm,sibling angle=120,scale=.8]
            \node [v](u) {}
            child { 
              node [v,label=$6$] (z6) {}
            }
            child {
              node [v,label={$1$}] (z1) {}
            }
            child {
              node [v] (v) {}
                child {
                  node [v] (w) {}
                    child {
                      node [v,label={\color{red}$4$}] (z4) {}
                    }
                    child {
                      node [v,label=left:$5$] (z5) {}
                    }
                }
                child {
                  node [v] (w2) {}
                    child {
                      node [v,label=$7$] (z7) {}
                    }
                    child {
                      node [v] (w3) {}
                        child {
                          node [v,label={$2$}] (z2){}
                        }
                        child {
                          node [v,label={\color{red}$3$}] (z3){}
                        }
                    }
                }
            };
            \foreach \s/\t in {z3/w3,w3/w2,w2/v,v/w,w/z4}
              \draw [very thick,blue] (\s)--(\t);
           \end{tikzpicture}}      
           \quad
           \subfloat[$x$\=/guarding edges. (All $x$\=/guarding edges are proper in this example.)]{
    \begin{tikzpicture}[grow cyclic,level distance=1cm,sibling angle=120,scale=.8]
      \node [v](u) {}
      child { 
        node [v,label=$6$] (z6) {}
      }
      child {
        node [v,label={\color{red}$1$}] (z1) {}
      }
      child {
        node [v] (v) {}
          child {
            node [v] (w) {}
              child {
                node [v,label={\color{red}$4$}] (z4) {}
              }
              child {
                node [v,label=left:$5$] (z5) {}
              }
          }
          child {
            node [v] (w2) {}
              child {
                node [v,label=$7$] (z7) {}
              }
              child {
                node [v] (w3) {}
                  child {
                    node [v,label={\color{red}$2$}] (z2){}
                  }
                  child {
                    node [v,label={\color{red}$3$}] (z3){}
                  }
              }
          }
      };
      \foreach \s/\t in {v/w,w/z4,w/z5,u/z6,u/z1,w3/z2,w3/z3,w2/z7,v/u}
        \draw [very thick,red,->] (\s)--(\t);
     \end{tikzpicture}}     
  \end{center}
  \caption{The first example of $x$\=/mixed pairs, $x$\=/crossing edges, edges cutting $\V_x=\{V_1,V_2,V_3,V_4\}$, and $x$\=/guarding edges when $\V_0=\{V_1,V_2,\ldots,V_7\}$ from Figure~\ref{fig:bd}.
  All $y$\=/degenerate edges are improper $y$\=/degenerate
  and 
  all $y'$-degenerate edges are improper $y'$-degenerate.
  But all $x$\=/degenerate edges are proper $x$\=/degenerate.
  In this example, $(T,\mathcal L)$ is $y$\=/degenereate and $y'$-degenerate, but not $x$\=/degenerate. 
  No node of~$T$ is $x$\=/blocked by an $x$\=/guarding edge. 
  There is no $x$\=/blocking path.
  }
  \label{fig:mixed}
\end{figure}
\begin{figure}
  \tikzstyle{v}=[circle,draw,fill=black,inner sep=0pt,minimum width=2pt]
  \begin{center} 
  \subfloat[$x$\=/mixed pairs (red arrows) and $x$\=/crossing edges (blue thick edges)]{
    \begin{tikzpicture}[grow cyclic,level distance=1cm,sibling angle=120,scale=.8]
      \node [v](u) {}
      child { 
        node [v,label=$6$] (z6) {}
      }
      child {
        node [v,label={\color{red}$1$}] (z1) {}
      }
      child {
        node [v] (v) {}
          child {
            node [v] (w) {}
              child {
                node [v,label={\color{red}$2$}] (z2) {}
              }
              child {
                node [v,label=left:$5$] (z5) {}
              }
          }
          child {
            node [v] (w2) {}
              child {
                node [v,label=$7$] (z7) {}
              }
              child {
                node [v] (w3) {}
                  child {
                    node [v,label={\color{red}$4$}] (z4){}
                  }
                  child {
                    node [v,label={\color{red}$3$}] (z3){}
                  }
              }
          }
      };
      \foreach \s/\t in {v/w,z6/u,w3/w2}
        \draw [->,red] ([shift=(60:4mm),shift=(150:1mm)]\s.center) -- ([shift=(60:-4mm),shift=(150:1mm)]\t.center);
      \foreach \s/\t in {w/v}      
        \draw [<-,red] ([shift=(60:4mm),shift=(150:-1mm)]\t.center) -- ([shift=(60:-4mm),shift=(150:-1mm)]\s.center);
      \foreach \s/\t in {v/w2,z2/w,z1/u}
        \draw [->,red] ([shift=(180:4mm),shift=(150+120:1mm)]\s.center) -- ([shift=(180:-4mm),shift=(150+120:1mm)]\t.center);
      \foreach \s/\t in {w2/v,z4/w3}      
        \draw [<-,red] ([shift=(180:4mm),shift=(150+120:-1mm)]\t.center) -- ([shift=(180:-4mm),shift=(150+120:-1mm)]\s.center); 
      \foreach \s/\t in {v/u,z5/w,z7/w2}
        \draw [->,red] ([shift=(-60:4mm),shift=(-150:-1mm)]\s.center) -- ([shift=(-60:-4mm),shift=(-150:-1mm)]\t.center);
      \foreach \s/\t in {u/v,z3/w3}      
        \draw [<-,red] ([shift=(-60:4mm),shift=(-150:1mm)]\t.center) -- ([shift=(-60:-4mm),shift=(-150:1mm)]\s.center);
      \foreach \s/\t in {v/w,v/u,v/w2}
        \draw [very thick,blue] (\s)--(\t);
     \end{tikzpicture}}
     \quad
     \subfloat[Edges cutting $\V_x$ (red thick edges)]{
      \begin{tikzpicture}[grow cyclic,level distance=1cm,sibling angle=120,scale=.8]
        \node [v](u) {}
      child { 
        node [v,label=$6$] (z6) {}
      }
      child {
        node [v,label={\color{red}$1$}] (z1) {}
      }
      child {
        node [v] (v) {}
          child {
            node [v] (w) {}
              child {
                node [v,label={\color{red}$2$}] (z2) {}
              }
              child {
                node [v,label=left:$5$] (z5) {}
              }
          }
          child {
            node [v] (w2) {}
              child {
                node [v,label=$7$] (z7) {}
              }
              child {
                node [v] (w3) {}
                  child {
                    node [v,label={\color{red}$4$}] (z4){}
                  }
                  child {
                    node [v,label={\color{red}$3$}] (z3){}
                  }
              }
          }
      };
        \foreach \s/\t in {v/w,v/u,v/w2,w/z2,u/z1,w2/w3,w3/z3,w3/z4}
        \draw [very thick,red] (\s)--(\t);
       \end{tikzpicture}}
       \quad
       \subfloat[(improper) $y$\=/degenerate edges]{
        \begin{tikzpicture}[grow cyclic,level distance=1cm,sibling angle=120,scale=.8]
          \node [v](u) {}
      child { 
        node [v,label=$6$] (z6) {}
      }
      child {
        node [v,label={\color{red}$1$}] (z1) {}
      }
      child {
        node [v] (v) {}
          child {
            node [v] (w) {}
              child {
                node [v,label={\color{red}$2$}] (z2) {}
              }
              child {
                node [v,label=left:$5$] (z5) {}
              }
          }
          child {
            node [v] (w2) {}
              child {
                node [v,label=$7$] (z7) {}
              }
              child {
                node [v] (w3) {}
                  child {
                    node [v,label={$4$}] (z4){}
                  }
                  child {
                    node [v,label={$3$}] (z3){}
                  }
              }
          }
      };
          \foreach \s/\t in {v/u,u/z1,v/w,w/z2}
            \draw [very thick,blue] (\s)--(\t);
         \end{tikzpicture}}   
         \quad
         \subfloat[(improper) $y'$-degenerate edges]{
          \begin{tikzpicture}[grow cyclic,level distance=1cm,sibling angle=120,scale=.8]
            \node [v](u) {}
        child { 
          node [v,label=$6$] (z6) {}
        }
        child {
          node [v,label={\color{red}$1$}] (z1) {}
        }
        child {
          node [v] (v) {}
            child {
              node [v] (w) {}
                child {
                  node [v,label={\color{red}$2$}] (z2) {}
                }
                child {
                  node [v,label=left:$5$] (z5) {}
                }
            }
            child {
              node [v] (w2) {}
                child {
                  node [v,label=$7$] (z7) {}
                }
                child {
                  node [v] (w3) {}
                    child {
                      node [v,label={\color{red}$4$}] (z4){}
                    }
                    child {
                      node [v,label={\color{red}$3$}] (z3){}
                    }
                }
            }
        };
        \foreach \s/\t in {z3/w3,w3/z4}
              \draw [very thick,blue] (\s)--(\t);
           \end{tikzpicture}}  
           \quad
           \subfloat[(proper) $x$\=/blocking path]{
    \begin{tikzpicture}[grow cyclic,level distance=1cm,sibling angle=120,scale=.8]
      \node [v](u) {}
      child { 
        node [v,label=$6$] (z6) {}
      }
      child {
        node [v,label={\color{red}$1$}] (z1) {}
      }
      child {
        node [v] (v) {}
          child {
            node [v] (w) {}
              child {
                node [v,label={\color{red}$2$}] (z2) {}
              }
              child {
                node [v,label=left:$5$] (z5) {}
              }
          }
          child {
            node [v] (w2) {}
              child {
                node [v,label=$7$] (z7) {}
              }
              child {
                node [v] (w3) {}
                  child {
                    node [v,label={\color{red}$4$}] (z4){}
                  }
                  child {
                    node [v,label={\color{red}$3$}] (z3){}
                  }
              }
          }
      };
      \foreach \s/\t in {w3/w2,w2/v}
        \draw [very thick,red] (\s)--(\t);
     \end{tikzpicture}}
           \quad
             \subfloat[(proper) $y''$-blocking paths]{
              \begin{tikzpicture}[grow cyclic,level distance=1cm,sibling angle=120,scale=.8]
                \node [v](u) {}
            child { 
              node [v,label={\color{red}$6$}] (z6) {}
            }
            child {
              node [v,label={$1$}] (z1) {}
            }
            child {
              node [v] (v) {}
                child {
                  node [v] (w) {}
                    child {
                      node [v,label={$2$}] (z2) {}
                    }
                    child {
                      node [v,label=left:{\color{red}$5$}] (z5) {}
                    }
                }
                child {
                  node [v] (w2) {}
                    child {
                      node [v,label=$7$] (z7) {}
                    }
                    child {
                      node [v] (w3) {}
                        child {
                          node [v,label={$4$}] (z4){}
                        }
                        child {
                          node [v,label={$3$}] (z3){}
                        }
                    }
                }
            };
            \foreach \s/\t in {v/u,u/z6}
                  \draw [very thick,green] (\s)--(\t);
            \foreach \s/\t in {z5/w,w/v}
                  \draw [very thick,red] (\s)--(\t);            
            \foreach \s/\t in {w/v,v/u}
                  \draw [very thick,blue] ([xshift=-1mm]\s.center)--([xshift=-1mm]\t.center);                  
               \end{tikzpicture}}
     \quad
     \subfloat[(proper) $x$\=/guarding edges. (All $x$\=/guarding edges are proper in this example.)]{
    \begin{tikzpicture}[grow cyclic,level distance=1cm,sibling angle=120,scale=.8]
      \node [v](u) {}
      child { 
        node [v,label=$6$] (z6) {}
      }
      child {
        node [v,label={\color{red}$1$}] (z1) {}
      }
      child {
        node [v] (v) {}
          child {
            node [v] (w) {}
              child {
                node [v,label={\color{red}$2$}] (z2) {}
              }
              child {
                node [v,label=left:$5$] (z5) {}
              }
          }
          child {
            node [v] (w2) {}
              child {
                node [v,label=$7$] (z7) {}
              }
              child {
                node [v] (w3) {}
                  child {
                    node [v,label={\color{red}$4$}] (z4){}
                  }
                  child {
                    node [v,label={\color{red}$3$}] (z3){}
                  }
              }
          }
      };
      \foreach \s/\t in {w3/z4,w3/z3,w2/z7,v/w,w/z2,w/z5,u/z1,u/z6,v/u}
        \draw [very thick,red,->] (\s)--(\t);
     \end{tikzpicture}}  
     \quad
     \subfloat[$x$\=/blocked edges and nodes ($x$\=/blocked by the $x$\=/blocking path)]{
\begin{tikzpicture}[grow cyclic,level distance=1cm,sibling angle=120,scale=.8]
\node [v](u) {}
child { 
  node [v,label=$6$] (z6) {}
}
child {
  node [v,label={\color{red}$1$}] (z1) {}
}
child {
  node [v] (v) {}
    child {
      node [v] (w) {}
        child {
          node [v,label={\color{red}$2$}] (z2) {}
        }
        child {
          node [v,label=left:$5$] (z5) {}
        }
    }
    child {
      node [v] (w2) {}
        child {
          node [v,label=$7$] (z7) {}
        }
        child {
          node [v] (w3) {}
            child {
              node [v,label={\color{red}$4$}] (z4){}
            }
            child {
              node [v,label={\color{red}$3$}] (z3){}
            }
        }
    }
};
  \draw [very thick,blue] (w2)--(z7);
  \node [circle,draw,inner sep=0pt,minimum width=4pt] at (z7){};
\end{tikzpicture}}
\quad
     \subfloat[$y''$-blocked edges and nodes ($y''$-blocked by an $y''$-blocking path)]{
\begin{tikzpicture}[grow cyclic,level distance=1cm,sibling angle=120,scale=.8]
\node [v](u) {}
child { 
  node [v,label=$6$] (z6) {}
}
child {
  node [v,label={\color{red}$1$}] (z1) {}
}
child {
  node [v] (v) {}
    child {
      node [v] (w) {}
        child {
          node [v,label={\color{red}$2$}] (z2) {}
        }
        child {
          node [v,label=left:$5$] (z5) {}
        }
    }
    child {
      node [v] (w2) {}
        child {
          node [v,label=$7$] (z7) {}
        }
        child {
          node [v] (w3) {}
            child {
              node [v,label={\color{red}$4$}] (z4){}
            }
            child {
              node [v,label={\color{red}$3$}] (z3){}
            }
        }
    }
};
  \foreach \s/\t in {w/z2,u/z1,v/w2,w2/z7,w2/w3,w3/z4,w3/z3}
    \draw [very thick,blue] (\s)--(\t);
  \foreach \s in {z3,z4,z7,z2,z1,w3,w2}
    \node [circle,draw,inner sep=0pt,minimum width=4pt] at (\s){};
\end{tikzpicture}}  
  \end{center}
  \caption{The second example of various terms on a branch-decomposition $(T,\mathcal L)$ where $\V_x=\{V_1,V_2,V_3,V_4\}$ and $\V_0=\{V_1,V_2,\ldots,V_7\}$ from Figure~\ref{fig:bd}.
  No edge is $x$\=/degenerate. 
  In this example, $(T,\mathcal L)$ is $y$\=/degenerate and $y'$-degenerate, but not $x$\=/degenerate. All $y$\=/degenerate edges are improper $y$\=/degenerate
  and all $y'$-degenerate edges are improper $y'$-degenerate.
  Every $x$\=/guarding edge is proper.
  Also $(T,\mathcal L)$ is $y'$-disjoint but not $y$\=/disjoint.  
  By definition, $(T,\mathcal L)$ is $x$\=/pure and $y'$-pure, but not $y$\=/pure, and therefore it is not totally pure with respect to $(T^b,\mathcal L^b)$.
  }
  \label{fig:mixed2}
\end{figure}

For an edge $uv$ of~$T$ and a node~$x$ of $T^b$ with $\V_x\subseteq\V_0$, we say that the ordered pair $(u,v)$ is \emph{$x$\=/mixed} if $\L(T,u,v)\cap \V_x\neq \emptyset$ 
and $\L(T,u,v)\cap (\V_0- \V_x)\neq \emptyset$.
For $\mathcal{S}\subseteq \V_0$, an edge $uv$ of~$T$ is said to \emph{cut} $\mathcal{S}$ if $\L(T,u,v)\cap \mathcal{S}\neq \emptyset$ and 
$\L(T,v,u)\cap \mathcal{S}\neq \emptyset$. Note that $uv$ does \emph{not} cut $\mathcal{S}$ if and only if $\mathcal{S}\subseteq \L(T,u,v)$ 
or $\mathcal{S} \subseteq \L(T,v,u)$. 
We say that $uv$ is \emph{$x$\=/crossing} if both $(u,v)$ and $(v,u)$ are $x$\=/mixed, or equivalently $uv$ cuts both $\V_x$ and $\V_0-\V_x$.

For a node~$x$ of $T^b$ with $\V_x\subseteq \V_0$,
we say that an edge $uv$ of~$T$ is \emph{$x$\=/degenerate} if
\[ \spn{\L_x(T,u,v)}\cap B_x= \spn{\L_x(T,v,u)}\cap B_x.\]
We define the concept of $x$\=/degenerate branch-decompositions recursively as follows.
We say that $(T,\L)$ is \emph{$x$\=/degenerate} if 
$T$ has an $x$\=/degenerate edge $uv$ 
such that 
\begin{itemize}
\item $uv$ cuts $\V_x$ and 
\item for all $y<x$, if $(T,\L)$ is $y$\=/degenerate, then $uv$ does not cut $\V_y$.   
\end{itemize}
For an $x$\=/degenerate branch-decomposition, such an edge $uv$ is said to be \emph{improper} $x$\=/degenerate.
An $x$\=/degenerate edge is \emph{proper} if it is not improper.
Note that if $x$ is a leaf of $T^b$, then $(T,\L)$ is not $x$\=/degenerate, because $T$ has no edge cutting $\V_x$.
We say that $(T,\L)$ is \emph{$x$\=/disjoint}
if $\V_0=\V_x$ or 
$T$ has an edge $uv$ such that 
$\L(T,u,v)=\V_x$
and $v$ is incident with an improper $x$\=/degenerate edge.

We say that an edge $uv$ of~$T$ \emph{$x$\=/guards its end $v$}
if
\[
\spn{\L_x(T,u,v)}\cap B_x\subsetneq \spn{\L_x(T,v,u)}\cap B_x.
\]
We say an edge is \emph{$x$\=/guarding} if it $x$\=/guards one of its ends.
Note that $\L_x(T,v,u)\neq\emptyset$.
An edge  $uv$ that $x$\=/guards $v$ is called \emph{improper} $x$\=/guarding
if $v$ has two neighbors $v_1$, $v_2$ in $T-uv$
such that 
both  $\L_x(T,v,v_1)$ and $\L_x(T,v,v_2)$ are nonempty
and %
$(u,v)$ is $x$\=/mixed. 
An $x$\=/guarding edge is \emph{proper} if it is not improper.
A node $w$ of~$T$ is \emph{$x$\=/blocked} by an $x$\=/guarding edge $e$ 
if 
(1) $e$ $x$\=/guards an end $v$ having two neighbors $v_1$, $v_2$ in $T-e$ such that 
both $\L_x(T,v,v_1)$ and $\L_x(T,v,v_2)$ are nonempty,
(2) $w\neq v$, and 
(3) the component of $T-e$ containing $w$ contains $v$.
An edge $f$ of~$T$ is \emph{$x$\=/blocked} by an $x$\=/guarding edge $e$
if at least one end of $f$ is $x$\=/blocked by $e$.

A $2$-edge path $uvw$ of~$T$ is an \emph{$x$\=/blocking path} if 
  \begin{align*}
\spn{\L_x(T,u,v)}\cap B_x&=\spn{\L_x(T,v,w)}\cap B_x,\\
\spn{\L_x(T,w,v)}\cap B_x&=\spn{\L_x(T,v,u)}\cap B_x, 
  \end{align*}
and neither $uv$ nor $vw$ is $x$\=/degenerate or $x$\=/guarding.
We say that $v$ is the \emph{center} of the path $uvw$.
An $x$\=/blocking path $uvw$ is \emph{improper}
if $v$ has a neighbor $t$ in $T-u-w$ such that 
$(v,t)$ is $x$\=/mixed
and $\L_x(T,v,u)$, $\L_x(T,v,w)$, $\L_x(T,v,t)$ are nonempty.
An $x$\=/blocking path is \emph{proper} if it is not improper. 
A node~$z$ of~$T$ is called \emph{$x$\=/blocked} by an $x$\=/blocking path $uvw$
if $z\neq v$ and the path from $v$ to $z$ has neither $u$ nor $w$.
An edge~$f$ of~$T$ is \emph{$x$\=/blocked} by an $x$\=/blocking path $uvw$
if at least one end of $f$ is $x$\=/blocked by $uvw$.
A node or an edge of~$T$ is said to be \emph{$x$\=/blocked}
if it is $x$\=/blocked by an $x$\=/blocking path or an $x$\=/guarding edge.

We observe that 
if $uvw$ is an $x$\=/blocking path in~$T$ and $t$ is the neighbor of $v$ other than $u$ and $w$, then
$vt$ $x$\=/guards $t$ as 
\begin{align*}
\spn{\L_x(T,v,t)}\cap B_x &= \spn{\L_x(T,u,v)\cap \L_x(T,w,v)}\cap B_x \\
                          &\subseteq \spn{\L_x(T,u,v)}\cap \spn{\L_x(T,w,v)}\cap B_x\\
                            &= \spn{\L_x(T,v,u)}\cap \spn{\L_x(T,v,w)}\cap B_x\\ 
&\subseteq \spn{\L_x(T,v,u)}\cap B_x + \spn{\L_x(T,v,w)}\cap B_x \\
&\subseteq \spn{\L_x(T,v,u)\cup \L_x(T,v,w)}\cap B_x = \spn{\L_x(T,t,v)}\cap B_x
\end{align*}
and equality holds if and only if $\spn{\L_x(T,v,u)}\cap B_x \subseteq \spn{\L_x(T,v,w)}\cap B_x$ or 
$\spn{\L_x(T,v,w)}\cap B_x \subseteq \spn{\L_x(T,v,u)}\cap B_x$, which is impossible since 
neither $uv$ nor $vw$ is $x$\=/degenerate or $x$\=/guarding.
We remark that if $v$ is a node of~$T$ having three neighbors $v_1$, $v_2$, $v_3$,
then at most one of $v_1vv_2$, $v_2vv_3$, $v_3vv_1$ is an $x$\=/blocking path.
Otherwise, say $v_1vv_2$ and $v_2vv_3$ are $x$\=/blocking.
Then $\spn{\L_x(T,v_2,v)}\cap B_x=\spn{\L_x(T,v,v_3)}\cap B_x\subseteq \spn{\L_x(T,v_1,v)}\cap B_x=\spn{\L_x(T,v,v_2)}\cap B_x$, which implies that an edge $vv_2$ is $x$\=/degenerate or $x$\=/guards $v$.

We say that $(T,\L)$ is \emph{$x$\=/pure} 
if  the following hold.
\begin{itemize}
\item 
If  $(T,\L)$ is not $x$\=/degenerate, then
all $x$\=/blocking paths and all $x$\=/guarding edges of~$T$ are proper.

\item 
If $(T,\L)$ is $x$\=/degenerate, then $(T,\L)$ is $x$\=/disjoint. 
\end{itemize}
We will see later that if $(T,\L)$ is $x$\=/degenerate, then every edge of~$T$ is either $x$\=/degenerate or $x$\=/guarding.

We say that a branch-decomposition $(T,\L)$ of $\V_0$ is \emph{totally pure with respect to $(T^b,\L^b)$}
if $(T,\L)$ is $x$\=/pure for all nodes $x$ of $T^b$ with $\V_x\subseteq \V_0$.
We omit ``with respect to $(T^b,\L^b)$''
if it is clear from the context.
See Figures~\ref{fig:mixed} and \ref{fig:mixed2} for examples.
Our goal is to prove that if the branch-width of a subspace arrangement is at most~$k$, then there exists a totally pure branch-decomposition of the subspace arrangement whose width is at most~$k$. 

Here are useful operations that would transform a branch-decomposition of width~$k$ into an $x$\=/pure branch-decomposition of width at most~$k$.

\subsection{Fork}
Suppose that $(T,\L)$ is a branch-decomposition of $\V_0$,
$x$ is a node of $T^b$ with $\V_x\subseteq \V_0$,
and $v$ is a node of~$T$ having three neighbors $v_1$, $v_2$, $v_3$
such that 
$v_1vv_2$ is an improper $x$\=/blocking path, and 
$\lambda_x^{(T,\L)}(vv_1)\ge \lambda_x^{(T,\L)}(vv_2)$.

Now we describe the operation called \emph{forking}. 
First let $T^{vv_3}$ be the component of $T-vv_3$ containing $v_3$, regarded as a rooted binary tree with the root $v_3$.
\emph{Smoothing} a degree-$2$ node $w$ is to contract one of the edges incident with $w$.
Let $T^{vv_3}_x$ be the minimal subtree of $T^{vv_3}$ 
that contains every leaf of $T^{vv_3}$ mapped to a subspace in $\V_x$ by $\L$.
Let the root of $T_x^{vv_3}$ be the node of $T_x^{vv_3}$ closest to $v_3$ in $T^{vv_3}$.
Let $T^{v,x}$ be the rooted binary tree obtained from $T^{vv_3}_x$ 
by smoothing degree-$2$ nonroot nodes if necessary.
Let $T^{vv_3}_{\bar{x}}$ be the minimal subtree of $T^{vv_3}$ 
that contains every leaf of $T^{vv_3}$ mapped to a subspace in $\V_0-\V_x$ by $\L$.
Let the root of $T^{vv_3}_{\bar{x}}$ be the node of $T^{vv_3}_{\bar{x}}$ closest to $v_3$ in $T^{vv_3}$
such that every subspace in $\V_0-\V_x$ corresponds to a leaf of $T^{vv_3}_{\bar{x}}$
and the root of $T^{vv_3}_{\bar{x}}$ is the closest node to $v_3$ in $T^{vv_3}$.
Let $T^{v,\bar{x}}$ be the rooted binary tree obtained from $T^{vv_3}_{\bar{x}}$ 
by smoothing degree-$2$ nonroot nodes if necessary.

Then, \emph{forking at $v$ by $\V_x$} is an operation
to obtain a new branch-decomposition $(T',\L)$ from $(T,\L)$ 
by  
deleting all nodes $x$\=/blocked by $v_1vv_2$, 
subdividing the edge $vv_2$ to create a degree-$2$ node $v'$,
making the disjoint union with $T^{v,x}$ and $T^{v,\bar x}$,
and adding edges joining $v$ with the root of $T^{v,x}$
and $v'$ with the root of $T^{v,\bar{x}}$.
See Figure~\ref{fig:fork2}. 
\begin{figure}
  \centering
  \tikzstyle{v}=[circle,draw,fill=black,inner sep=0pt,minimum width=2pt]
  \tikzstyle{r}=[draw,fill=black,inner sep=2pt]
  \tikzstyle{b}=[draw,circle,inner sep=2pt]
  \subfloat[$T$]{\begin{tikzpicture}
      \draw node [v,label=below:$v_1$] (v1) at (0,0) {};
      \draw node [v,label=below:$v$] (v2) at (2,0) {};
      \draw node [v,label=below:$v_2$] (v3) at (4,0) {};
      \draw node [v] (v4) at (-0.7,-0.7) {};
      \draw node [v] (v5) at (-0.7,0.7) {};
      \draw node [v] (v6) at (2,1) {};
      \draw node [v] (v7) at (4.7,0.7) {};
      \draw node [v] (v8) at (1,1.4) {};
      \draw node [v] (v9) at (1,2.1) {};
      \draw node [v] (v10) at (3,1.4) {};
      \draw node [v] (v11) at (2.4,2.1) {};
      \draw node [v] (v12) at (3.7,2.1) {};
      \draw node [r] (r1) at (-1.4,-0.7) {};
      \draw node [r] (r2) at (-1.4,0.7) {};
      \draw node [r] (r3) at (0.5,2.5) {};
      \draw node [r] (r4) at (1.5,2.5) {};
      \draw node [r] (r5) at (2.1,2.5) {};
      \draw node [r] (r8) at (2.7,2.5) {};
      \draw node [r] (r6) at (4.7,1.4) {};
      \draw node [r] (r7) at (4.7,-0.7) {};
      \draw node [b] (b1) at (-0.7,-1.4) {};
      \draw node [b] (b2) at (-0.7,1.4) {};
      \draw node [b] (b3) at (0.3,1.4) {};
      \draw node [b] (b4) at (3.4,2.5) {};
      \draw node [b] (b5) at (4,2.5) {};
      \draw node [b] (b6) at (4+1.4,0.7) {};
      \draw (r1)--(v4)--(b1);
      \draw (r2)--(v5)--(b2);
      \draw (v5)--(v1)--(v4);
      \draw (v1)--(v2)--(v3);
      \draw (r7)--(v3)--(v7);
      \draw (r6)--(v7)--(b6);
      \draw (b3)--(v8)--(v9);
      \draw (r3)--(v9)--(r4);
      \draw (v8)--(v6)--(v10);
      \draw (v11)--(v10)--(v12);
      \draw (r5)--(v11)--(r8);
      \draw (b4)--(v12)--(b5); 
      \draw (v2)--(v6);
      \draw [dashed] plot [smooth cycle] coordinates {(2,0.8) (0.2,1) (0.1,2.8) (4.1,2.8) (4,1.6) } 
      node at (2,3.2)  {$T^{vv_3}$};
    \end{tikzpicture}}
    \qquad
  \subfloat[$T'$]{\begin{tikzpicture}
      \draw node [v,label=below:$v_1$] (v1) at (0,0) {};
      \draw node [v,label=below:$v$] (v2) at (1.1,0) {};
      \draw node [v,label=below:$v'$] (w2) at (2.7,0) {};
      \draw node [v,label=below:$v_2$] (v3) at (4,0) {};
      \draw node [v] (v4) at (-0.7,-0.7) {};
      \draw node [v] (v5) at (-0.7,0.7) {};
      \draw node [v] (v6) at (1.1,1) {};
      \draw node [v] (w6) at (2.7,1) {};
      \draw node [v] (v9) at (0.4,1.8) {};
      \draw node [v] (v10) at (1.8,1.8) {};
      \draw node [v] (v7) at (4.7,0.7) {};
      \draw node [v] (v8) at (3.5,1.8) {};
      \draw node [r] (r1) at (-1.4,-0.7) {};
      \draw node [r] (r2) at (-1.4,0.7) {};
      \draw node [r] (r3) at (0,2.5) {};
      \draw node [r] (r4) at (0.8,2.5) {};
      \draw node [r] (r5) at (1.5,2.5) {};
      \draw node [r] (r8) at (2.1,2.5) {};
      \draw node [r] (r6) at (4.7,1.4) {};
      \draw node [r] (r7) at (4.7,-0.7) {};
      \draw node [b] (b1) at (-0.7,-1.4) {};
      \draw node [b] (b2) at (-0.7,1.4) {};
      \draw node [b] (b3) at (2.3+0.1,1.4) {};
      \draw node [b] (b4) at (3+0.1,2.5) {};
      \draw node [b] (b5) at (4+0.05,2.5) {};
      \draw node [b] (b6) at (4+1.4,0.7) {};
      \draw (r1)--(v4)--(b1);
      \draw (r2)--(v5)--(b2);
      \draw (v5)--(v1)--(v4);
      \draw (v1)--(v2)--(w2)--(v3);
      \draw (r7)--(v3)--(v7);
      \draw (r6)--(v7)--(b6);
      \draw (v2)--(v6);
      \draw (w2)--(w6);
      \draw (r3)--(v9)--(r4);
      \draw (v9)--(v6)--(v10);
      \draw (b3)--(w6)--(v8);
      \draw (b4)--(v8)--(b5);
	  \draw (r5)--(v10)--(r8);
      \draw [dashed] plot [smooth cycle] coordinates {(1.1,0.8) (0.18,1.55) (-0.2,2.8) (2.3,2.8) (2,1.6) } 
      node at (1.1,3.2)  {$T^{v,x}$};
      \draw [dashed] plot [smooth cycle] coordinates {(2.7,0.8) (2.15,1.45) (3,2.8) (4.3,2.8) (3.8,1.6) } 
      node at (3.7,3.2)  {$T^{v,\bar{x}}$};
    \end{tikzpicture}}
  \caption[Constructing $T'$ by forking at $v$ by $\V_x$.]{Constructing $T'$ by forking at $v$ by $\V_x$. \tikz \node [r] {}; represents a leaf node mapped to an element of $\V_x$ by $\L$ and \tikz \node [b] {}; represents a leaf node mapped to an element of $\V_0-\V_x$ by~$\L$.} %
  \label{fig:fork2}
\end{figure}

\begin{PROP}\label{prop:fork}
  If a branch-decomposition $(T',\L)$ is obtained from a
  branch-\decomposition{} $(T,\L)$ by forking at $v$ by $\V_x$, 
  then the width of $(T',\L)$ is at most the width of $(T,\L)$.
\end{PROP}

For this proposition, we present the following lemmas.
\begin{LEM}\label{lem:forksub}
  Let $L_1$, $L_2$, $R_1$, $R_2$, $X_1$, $X_2$, $Y_1$, $Y_2$ be subspaces of $\F^r$ 
and $S_1=L_1+R_1+Y_1$ and $S_2=L_2+R_2+Y_2$. 
If 
$(X_1+Y_1)\cap (S_2+X_2)\subseteq 
(L_1+R_1)\cap (S_2+X_2)$,
then 
\begin{align*}
\dim (X_1\cap (S_1+S_2+X_2))&\le \dim ((X_1+X_2)\cap (S_1+S_2)),\\
\dim (X_2\cap (S_1+S_2+X_1))&\le \dim ((X_1+X_2)\cap (S_1+S_2)).
\end{align*}
\end{LEM}
\begin{proof}
  For the first inequality, 
  $\dim (X_1\cap (S_1+S_2+X_2))-\dim ((X_1+X_2)\cap (S_1+S_2))
  =\dim X_1+\dim (S_1+S_2)+\dim X_2-\dim(S_1+S_2)\cap X_2
  -\dim (X_1+X_2)-\dim (S_1+S_2)
  = \dim (X_1\cap X_2)-\dim ((S_1+S_2)\cap X_2)\le 0$
  because 
  $ X_1\cap X_2 
  \subseteq (X_1+Y_1)\cap (S_2+X_2) \cap X_2
  \subseteq    (L_1+R_1)\cap (S_2+X_2)\cap X_2
  \subseteq S_1\cap X_2\subseteq (S_1+S_2)\cap X_2$.

  For the second inequality, 
  $\dim (X_2\cap (S_1+S_2+X_1))-\dim ((X_1+X_2)\cap (S_1+S_2))
  =\dim (X_2\cap X_1)-\dim ((S_1+S_2)\cap X_1)\le 0$
  because $X_1\cap X_2
  \subseteq X_1\cap (X_1+Y_1) \cap (S_2+X_2)
  \subseteq X_1\cap (L_1+R_1) \cap (S_2+X_2)
  \subseteq X_1\cap (L_1+R_1)
  \subseteq X_1\cap (S_1+S_2)$.
\end{proof}

\begin{LEM}[{\cite[Lemma 3.18]{JKO2016}}]\label{lem:dim-join}
For finite-dimensional subspaces $X_1,X_2,Y_1,Y_2$, 
\begin{multline*}
  \dim((X_1+X_2)\cap (Y_1+Y_2)) \\
  =\dim (X_1\cap Y_1) + \dim (X_2\cap Y_2)
  -\dim (X_1\cap X_2) - \dim (Y_1\cap Y_2) \\
  + \dim ((X_1+Y_1)\cap (X_2+Y_2)).
\end{multline*}
\end{LEM}

\begin{LEM}\label{lem:forksub2}
  Let $L_1$, $L_2$, $R_1$, $R_2$, $Z_1$, $Z_2$ be subspaces of $\F^r$
  such that
\begin{align*}
  L_1\cap (L_2+R_2+Z_2)&= (L_1+Z_1)\cap (L_2+R_2+Z_2),\\
  R_1\cap (L_2+R_2+Z_2)&= (R_1+Z_1)\cap (L_2+R_2+Z_2),\\
  \dim (L_1\cap (Z_1+R_1))&\ge  \dim ((L_1+Z_1)\cap R_1).
\end{align*}
  Then 
  $\dim  ((L_1+L_2)\cap (R_1+R_2+Z_1+Z_2))
  \ge  \dim ((L_1+L_2+Z_1)\cap (R_1+R_2+Z_2))$.
\end{LEM}
\begin{proof}
  By the hypothesis, 
  we have 
  $(L_1+Z_1)\cap L_2
  = (L_1+Z_1)\cap (L_2+R_2+Z_2)\cap L_2
  =L_1\cap (L_2+R_2+Z_2)\cap L_2
  = L_1\cap L_2$
  and similarly 
  $R_1\cap (R_2+Z_2)
  = R_1\cap (L_2+R_2+Z_2)\cap (R_2+Z_2)
  = (R_1+Z_1)\cap (L_2+R_2+Z_2)\cap (R_2+Z_2)
  = (R_1+Z_1)\cap (R_2+Z_2)$.
  By Lemma~\ref{lem:dim-join}, we have 
  \begin{align*}
    \lefteqn{\dim ((L_1+Z_1+L_2)\cap (R_1+R_2+Z_2))}\\
    &= \dim ((L_1+Z_1)\cap R_1)
      + \dim (L_2\cap (R_2+Z_2))\\
    &\quad
      -\dim ((L_1+Z_1)\cap L_2)-\dim (R_1\cap (R_2+Z_2))\\
    &\quad +\dim ((L_1+Z_1+R_1)\cap (L_2+R_2+Z_2))\\
    &= \dim ((L_1+Z_1)\cap R_1)
      + \dim (L_2\cap (R_2+Z_2))\\
    &\quad
      -\dim (L_1\cap L_2)-\dim ((R_1+Z_1)\cap (R_2+Z_2))\\
    &\quad +\dim ((L_1+Z_1+R_1)\cap (L_2+R_2+Z_2))\\
    &\le  \dim (L_1\cap (Z_1+ R_1))
      + \dim (L_2\cap (R_2+Z_2))\\
    &\quad
      -\dim (L_1\cap L_2)-\dim ((R_1+Z_1)\cap (R_2+Z_2))\\
    &\quad +\dim ((L_1+Z_1+R_1)\cap (L_2+R_2+Z_2))\\
    &=\dim ((L_1+L_2)\cap (R_1+Z_1+R_2+Z_2)).\qedhere 
  \end{align*}
\end{proof}
\begin{proof}[Proof of Proposition~\ref{prop:fork}]
  Let $v_1$, $v_2$, $v_3$ be the neighbors of $v$ such that $v_1vv_2$ is an improper $x$\=/blocking path and $\lambda_x^{(T,\L)}(vv_1)\ge \lambda_x^{(T,\L)}(vv_2)$.
Let $L_1=\spn{\L_x(T,v,v_1)}$, $L_2=\spn { \L(T,v,v_1)-\V_x}$, 
$R_1=\spn{\L_x(T,v,v_2)}$, $R_2=\spn{\L(T,v,v_2)-\V_x}$, 
$Z_1=\spn{\L_x(T,v,v_3)}$, $Z_2=\spn{\L(T,v,v_3)-\V_x}$.
Let $v'$ be the new node of $T'$ on the path from $v$ to $v_2$.
See Figure~\ref{fig:fork}.
\begin{figure}
  \centering
  \tikzstyle{v}=[circle,draw,fill=black,inner sep=0pt,minimum width=4pt]
  \subfloat[$T$]{\begin{tikzpicture}
      \draw plot [smooth cycle] coordinates {(0,0) (1,0.1) (1.3,0.5) (2,1.4) (1.5,2.5) (0.8,2.5) (0.3,1.2) (-0.2,0.6) } node  [v] (v1) at (2,1.4)  {};
      \draw [xshift=4cm] plot [smooth cycle] coordinates {(0,1.4) (1,0.1) (1.3,0.3) (1.2,1.4) (1.1,2.5) (0.8,2.5) (0.3,1.9)}  node [v] (v2) at (0,1.4) {};
      \draw [xshift=3cm, yshift=1cm] plot [smooth cycle] coordinates {(0,1.4) (1,3.2) (-1,3.4)} 
      node [v] (v3) at (0,1.4) {}
      node at (0,2.6) {$Z_1$, $Z_2$};
      \draw node [v,label=below:$v$] (v) at (3,1.4) {};
      \draw node at (1.2,1.8) {$L_1$,};
      \draw node at (1.2,1.1) {$L_2$\phantom{,}};
      \draw node at (4.8,1.8) {$R_1$,};
      \draw node at (4.8,1.1) {$R_2$\phantom{,}};
      \draw (v1)--(v)--(v2);
      \draw (v)--(v3);      
    \end{tikzpicture}}
    \qquad
  \subfloat[$T'$]{\begin{tikzpicture}
      \draw plot [smooth cycle] coordinates {(0,0) (1,0.1) (1.3,0.5) (2,1.4) (1.5,2.5) (0.8,2.5) (0.3,1.2) (-0.2,0.6) } node  [v] (v1) at (2,1.4)  {};
      \draw [xshift=5cm] plot [smooth cycle] coordinates {(0,1.4) (1,0.1) (1.3,0.3) (1.2,1.4) (1.1,2.5) (0.8,2.5) (0.3,1.9)}  node [v] (v2) at (0,1.4) {};
      \draw [xshift=3cm, yshift=1cm] plot [smooth cycle] coordinates {(0,1.4) (0.2,2.7) (-1,2.4)} 
      node [v] (v3) at (0,1.4) {};
      \draw [xshift=4cm, yshift=1cm] plot [smooth cycle] coordinates {(0,1.4) (-0.25,2.6) (1,2.6)} 
      node [v] (v4) at (0,1.4) {};
      \draw node [v,label=below:$v$] (v) at (3,1.4) {};
      \draw node [v,label=below:$v'$] (v') at (4,1.4) {};
      \draw node at (1.2,1.8) {$L_1$,};
      \draw node at (1.2,1.1) {$L_2$\phantom{,}};
      \draw node at (5.8,1.8) {$R_1$,};
      \draw node at (5.8,1.1) {$R_2$\phantom{,}};
      \draw node at (2.8,3.1) {$Z_1$};
      \draw node at (4.2,3.1) {$Z_2$};
      \draw (v1)--(v)--(v')--(v2);
      \draw (v)--(v3);      
      \draw (v')--(v4);      
    \end{tikzpicture}} %
  \caption{The subspaces from $T$ and $T'$.}
  \label{fig:fork}
\end{figure}

For an edge $e'$ of $T'$,
if $e'$ is %
in $T'-V(T^{v,x})-V(T^{v,\bar{x}})-vv'$, 
then there exists the corresponding edge $e$ of~$T$ such that 
the width of $e'$ is equal to the width of $e$.

Now we show that the width of $vv'$ is at most the width of $vv_1$ in~$T$.
Since $v_1vv_2$ is an $x$\=/blocking path, 
\[(L_1+Z_1)\cap B_x = L_1 \cap B_x\text{ and }
(R_1+Z_1)\cap B_x = R_1\cap B_x.\]
Then, 
by taking the intersection with $L_2+R_2+Z_2$, we have
\begin{align*}
(L_1+Z_1)\cap (L_2+R_2+Z_2)\cap B_x& = L_1 \cap (L_2+R_2+Z_2)\cap B_x\text{ and }\\
(R_1+Z_1) \cap (L_2+R_2+Z_2)\cap B_x &= R_1 \cap (L_2+R_2+Z_2)\cap B_x.
\end{align*}
As $(L_1+Z_1) \subseteq \spn{\V_x}$ and $(L_2+R_2+Z_2) \subseteq \spn{\V_0-\V_x}$,
we deduce that $L_1\cap (L_2+R_2+Z_2) \subseteq (L_1+Z_1)\cap (L_2+R_2+Z_2)\subseteq \spn{\V_x}\cap\spn{\V_0-\V_x} \subseteq B_x$.
Thus, we have
$(L_1+Z_1)\cap (L_2+R_2+Z_2) = (L_1+Z_1)\cap (L_2+R_2+Z_2)\cap B_x 
=L_1 \cap (L_2+R_2+Z_2) \cap B_x = L_1 \cap (L_2+R_2+Z_2)$. 
Similarly we obtain 
\[ (R_1+Z_1) \cap (L_2+R_2+Z_2) = R_1 \cap (L_2+R_2+Z_2). \]
We can also see that \[\dim (L_1\cap (Z_1+R_1)) =\lambda_x^{(T,\L)}(vv_1) \ge \lambda_x^{(T,\L)}(vv_2) = \dim ((L_1 + Z_1) \cap R_1).\]
Thus the conditions for Lemma~\ref{lem:forksub2} hold. %
Then Lemma~\ref{lem:forksub2} implies that
the width of $vv'$ is at most the width of $v_1v$ in~$T$.

We claim that 
for an edge $e'$ in $E(T'[V(T^{v,x})\cup\{v\}])$, 
the width of $e'$ is at most the width of some edge %
in $E(T[V(T^{vv_3})\cup\{v\}])$.
Let $e'=x'y'$ where $z'$ is closer to $v$ than $y'$.
By the construction of $T'$,
there exists an edge $e=zy$, which is $x$\=/blocked by $v_1vv_2$, in $E(T[V(T^{vv_3})\cup\{v\}])$
such that 
$\L_x(T,z,y)=\L_x(T',z',y')$.

Let $X_1=\spn{\L_x(T,z,y)}$, $X_2=\spn{\L(T,z,y)-\V_x}$, 
$Y_1=\spn{\L_x(T,v,v_3)-\L_x(T,z,y)}$, 
$Y_2=\spn{\L(T,v,v_3)-\L(T,z,y)-\V_x}$.
Note that $X_1+Y_1=Z_1$ and $X_2+Y_2=Z_2$.
Then, we observe that 
$Z_1\cap (L_2+R_2+Z_2)\subseteq (L_1+Z_1)\cap (L_2+R_2+Z_2)
= L_1\cap (L_2+R_2+Z_2)$
and similarly $Z_1\cap (L_2+R_2+Z_2)\subseteq R_1\cap (L_2+R_2+Z_2)$.
Thus, 
\[Z_1\cap (L_2+R_2+Z_2)\subseteq 
  L_1\cap R_1\cap (L_2+R_2+Z_2) 
\subseteq (L_1+R_1)\cap (L_2+R_2+Z_2).\]
So the condition for Lemma~\ref{lem:forksub} holds.
And Lemma~\ref{lem:forksub} shows that the width of $e'$ is at most the width of $e$.

Similarly, %
it is easy to show that the claim is true
for an edge $e'$ in $E(T'[V(T^{v,\bar{x}})\cup\{v'\}])$.
\end{proof}

\subsection{Split}
Suppose again that $(T,\L)$ is a branch-decomposition of $\V_0$, 
$x$ is a node of $T^b$ with $\V_x\subseteq \V_0$,
and an edge $uv$ of~$T$ 
$x$\=/guards $v$ improperly 
or $uv$ is improper $x$\=/degenerate.
Under those circumstances, 
if $(u,v)$ is $x$\=/mixed, 
we define an operation called \emph{splitting}. 
Similar to the definitions for forking, 
we define $T^{uv}$ to be the component of $T-uv$ containing~$v$, regarded as a rooted binary tree with the root $v$. 
Let $T^{uv}_x$ be the minimal subtree of $T^{uv}$
that contains every leaf of $T^{uv}$ mapped to a subspace in $\V_x$ by $\L$. 
Let the root of $T_x^{uv}$ be the node of $T_x^{uv}$ closest to $v$ in $T^{uv}$. 
Let $T^{v,x}$ be the rooted binary tree obtained from $T^{uv}_x$ 
by smoothing degree-$2$ nonroot nodes if necessary.
Let $T^{uv}_{\bar{x}}$ be the minimal subtree of $T_{uv}$
that contains every leaf of~$T^{uv}$ mapped to a subspace in $\V_0-\V_x$ by $\L$. 
Let the root of $T_x^{uv}$ be the node of $T_x^{uv}$ closest to $v$ in $T^{uv}$
such that every subspace in $\V_0-\V_x$ corresponds to a leaf of $T^{uv}_{\bar{x}}$
and the root of $T^{uv}_{\bar{x}}$ is the closest node to $v$ in~$T^{uv}$.
Let $T^{v,\bar{x}}$ be the rooted binary tree obtained from $T^{uv}_{\bar{x}}$ 
by smoothing degree-$2$ nonroot nodes if necessary.

\begin{figure}
  \tikzstyle{v}=[circle,draw,fill=black,inner sep=0pt,minimum width=2pt]
\begin{center}
  \begin{tikzpicture}[grow cyclic,level distance=1cm,sibling angle=120,scale=.8]
    \node [v](u) {}
child { 
  node [v,label=$6$] (z6) {}
}
child {
  node [v] (n) {} 
    child {
      node [v,label={\color{red}$1$}] (z1) {}
    }
    child {
      node[v,label={\color{red}$2$}] (z2) {}
    }
}
child {
  node [v] (v) {}
    child {          
          node [v,label=left:$5$] (z5) {}           
    }
    child {
      node [v] (w2) {}
        child {
          node [v,label=$7$] (z7) {}
        }
        child {
          node [v] (w3) {}
            child {
              node [v,label={$4$}] (z4){}
            }
            child {
              node [v,label={$3$}] (z3){}
            }
        }
    }
};
    \draw [very thick, red] (z1)--(n);
   \end{tikzpicture}
  \end{center}
  \caption{A branch-decomposition $(T',\mathcal L')$ obtained from $(T,\mathcal L)$ by splitting an improper $y$\=/degenerate edge in $(T,\mathcal L)$ from 
  Figure~\ref{fig:mixed2}.
  Then $(T',\mathcal L')$ is $y$\=/pure, $y'$-pure, $x$\=/pure, and $y''$-pure and therefore $(T',\mathcal L')$ is totally pure with respect to $(T^b,\mathcal L^b)$.
  }\label{fig:bdsplit}
\end{figure}

\emph{Splitting at $uv$ by $\V_x$} is an operation to transform $(T,\L)$ to another branch-decomposition $(T',\L)$ by 
deleting all nodes in the component of $T-v$ not containing~$u$,
making the disjoint union with $T^{v,x}$ and $T^{v,\bar x}$, 
adding an edge joining $v$ with the root of $T^{v,x}$
and an edge joining $v$ with the root of $T^{v, \bar x}$.
See Figures~\ref{fig:bdsplit} and \ref{fig:split}.

\begin{figure}
  \centering
  \tikzstyle{v}=[circle,draw,fill=black,inner sep=0pt,minimum width=4pt]
  \subfloat[$T$]{\begin{tikzpicture}
      \draw plot [smooth cycle] coordinates {(0,0) (-0.8,1.5) (-1.1,0.2) (-0.8,-1.5)} 
      		node  [v,label=left:$u$] (u) at (0,0)  {};
      \draw [xshift=1.5cm] plot [smooth cycle] coordinates {(0,0) (1,-1.2) (1.8,-1.5) (2,0.1) (1.8,1.5) (1.2,1.5)}  node [v,label=right:$v$] (v) at (0,0) {};
      \draw (u)--(v);      
    \end{tikzpicture}}
    \qquad\quad
  \subfloat[$T'$]{\begin{tikzpicture}
      \draw plot [smooth cycle] coordinates {(0,0) (-0.8,1.5) (-1.1,0.2) (-0.8,-1.5)} 
      		node  [v,label=left:$u$] (u) at (0,0)  {};
      \draw node [v,label=right:$v$] (v) at (1,0) {};
      \draw [xshift=2cm, yshift=1cm] plot [smooth cycle] coordinates {(0,0) (1.2,1) (1.4,-0.1) (0.7,-0.4)} 
      node [v] (v1) at (0,0) {}
      node at (0.8,0.1) {$T^{v,x}$};
      \draw [xshift=2cm, yshift=-1cm] plot [smooth cycle] coordinates {(0,0) (1,0.6) (1,-1)} 
      node [v] (v2) at (0,0) {}
      node at (0.6,-0.1) {$T^{v,\bar{x}}$};
      \draw (u)--(v);
      \draw (v1)--(v)--(v2);     
    \end{tikzpicture}} 
  \caption{Constructing $T'$ by splitting at $uv$ by $\V_x$.}
  \label{fig:split}
\end{figure}

We prove that splitting does not increase the width. 
\begin{PROP}\label{prop:split}
  If a branch-decomposition $(T',\L)$ is obtained from a branch-decomposition $(T,\L)$ by
  splitting $uv$ by $\V_x$, then the width of $(T',\L)$ is less than or
  equal to the width of $(T,\L)$. 
\end{PROP}
\begin{proof}
  This is immediate from Lemma~\ref{lem:forksub}.
\end{proof}

\subsection{Pure branch-decompositions exist}

Our goal is to prove the following proposition.
\begin{PROP}\label{prop:pure}
  Let $(T^b,\L^b)$ be a rooted branch-decomposition of a subspace arrangement $\V$
  and let $\V_0\subseteq \V$.
  If the branch-width of $\V_0$ is at most~$k$, %
  then $\V_0$ has a branch-decomposition of width at most~$k$ that is totally pure with respect to $(T^b,\L^b)$.
\end{PROP}

To prove this proposition, we need some lemmas.
\begin{LEM}\label{lem:Bx-to-By}
  Let $x$, $y$ be nodes of $T^b$ with $x\le y$. 
  If $L$ is a subspace of $\spn{\V_x}$,
  then
  \[
  L\cap B_y \subseteq L\cap B_x.
  \]
\end{LEM}
\begin{proof}
  Since $B_x=\spn{\V_x}\cap \spn{\V-\V_x}$ and $L\subseteq \spn{\V_x} \subseteq \spn{\V_y}$, 
  we see that $L\cap B_x = L\cap \spn{\V-\V_x}$ and $L\cap B_y=L\cap \spn{\V-\V_y}$.
  As $\spn{\V-\V_y}\subseteq \spn{\V-\V_x}$,
  we have $L\cap B_y \subseteq L\cap B_x$.
\end{proof}

\begin{LEM}\label{lem:x-to-y}
  Let $x$, $y$ be nodes of $T^b$ with $x\le y$.
  Let $(T,\L)$ be a branch-\decomposition{} of $\V_0$ with $\V_y\subseteq \V_0$.
  If $\L_y(T,u,v)\subseteq \V_x$
  and an edge $e=uv$ of~$T$ is $x$\=/degenerate or 
  $x$\=/guards $v$, 
  then $e$ is $y$\=/degenerate or $y$\=/guards $v$.
\end{LEM}
\begin{proof}
  Let $L_x=\spn{\L_x(T,u,v)}$, $R_x=\spn{\L_x(T,v,u)}$, $L_y=\spn{\L_y(T,u,v)}$, and $R_y=\spn{\L_y(T,v,u)}$. 
  We will show that if $L_x\cap B_x\subseteq R_x\cap B_x$,
  then $L_y\cap B_y\subseteq R_y\cap B_y$.
  Since $\L_y(T,u,v)\subseteq\V_x$ and $\V_x\subseteq\V_y$, we have $L_x=L_y$ and $R_x\subseteq R_y$.
  By Lemma~\ref{lem:Bx-to-By}, $L_y\cap B_y\subseteq L_x\cap B_x$.
  Since $L_x\cap B_x\subseteq R_x\cap B_x$ and $R_x \subseteq R_y$,
  we have $L_y\cap B_y \subseteq R_y$. Thus, $L_y\cap B_y\subseteq R_y \cap B_y$.
\end{proof}

\begin{LEM}\label{lem:guardfar}
  Let $x$ be a node of $T^b$ and let $(T,\L)$ be a branch-\decomposition{} of~$\V_0$
  such that $\V_x\subseteq \V_0$. %
  Let $v_0v_1\cdots v_\ell$ be a path in~$T$.
  If $v_0v_1$ $x$\=/guards $v_1$ or $v_{0}v_1$ is $x$\=/degenerate,
  then $v_{\ell-1}v_\ell$ $x$\=/guards $v_\ell$ or $v_{\ell-1}v_\ell$ is $x$\=/degenerate.
\end{LEM}
\begin{proof}
It is easy to see that $\L_x(T,v_{\ell-1},v_{\ell})\subseteq \L_x(T,v_0,v_1)$ and $\L_x(T,v_1,v_0)\subseteq \L_x(T,v_{\ell},v_{\ell-1})$.
Thus, \[\spn{\L_x(T,v_{\ell-1},v_{\ell})}\cap B_x \subseteq \spn{\L_x(T,v_{\ell},v_{\ell-1})}\cap B_x\]
because $\spn{\L_x(T,v_{0},v_{1})}\cap B_x \subseteq \spn{\L_x(T,v_{1},v_{0})}\cap B_x$.
\end{proof}

\begin{LEM}\label{lem:easycase}
Let $x$, $y$ be nodes of $T^b$ with $y\nleq x$.
Let $(T,\L)$ be a branch-\decomposition{} of~$\V_0$ with $\V_x\cup\V_y\subseteq \V_0$.
Let $pq$ be an edge of~$T$ such that $\L(T,p,q)\subseteq \V_x$ and it either $x$\=/guards $q$ or is $x$\=/degenerate. 
If $(q,p)$ points toward a node $u$ of~$T$ and 
a branch-decomposition $(T',\L)$ is obtained from $(T,\L)$ by
either forking at $u$ by $\V_y$ or splitting at an edge $uv$ of~$T$ by $\V_y$, 
then $T'$ has an edge $p'q'$ such that $\L(T',p',q')=\L(T,p,q)$. Furthermore, $q'=q$.
\end{LEM}
\begin{proof}
We first claim that for two nodes $u$ and $v$ of~$T$, if $(u,v)$ is $y$\=/mixed, 
then $(u,v)\neq(p,q)$. Suppose not.
If $x<y$, then 
$\L(T,u,v)=\L(T,p,q)\subseteq\V_x\subseteq\V_y$,
which contradicts that $(u,v)$ is $y$\=/mixed.
If $x$ and $y$ are incomparable, then 
$\L(T,u,v)\subseteq\V_x\subseteq\V_0-\V_y$,
which also contradicts that $(u,v)$ is $y$\=/mixed.

Suppose that $(T',\L)$ is obtained from $(T,\L)$ by splitting at $uv$ by $\V_y$. 
Let $T_q$ be the component of $T-pq$ containing $q$.
Note that $(u,v)\neq (p,q)$ because $(u,v)$ is $y$\=/mixed. 
As $(u,v)\neq (p,q)$ and $(q,p)$ points toward $u$,
we know 
that $T_q$ is contained 
either in the component of $T-uv$ containing~$v$ or in the component of $T-uv$ containing~$u$.
In the former case, 
after splitting, $T_q$ is a subtree of $T^{uv}_y$ if $x<y$, and 
$T_q$ is a subtree of $T^{uv}_{\bar{y}}$ if $x,y$ are incomparable.

Then the statement trivially follows from $\L(T,p,q)\subseteq \V_x$. 
In the latter case, note that the component of $T'-uv$ containing~$u$ 
is identical to the component of $T-uv$ containing~$u$, and the statement follows. 

If $(T',\L)$ is obtained by forking at $u$, 
then let $v$ be the unique neighbor of $u$ which does not lie on the $y$\=/blocking path. 
We remark that $(u,v)$ is $y$\=/mixed and $uv$ cuts~$\V_y$. 
Thus similarly we can see that the statement holds. 
\end{proof}

\begin{LEM}\label{lem:disjoint}
Let $x$, $y$ be nodes of $T^b$ with $y\nleq x$.
Let $(T,\L)$ be a branch-\decomposition{} of $\V_0$ with $\V_x\cup\V_y\subseteq \V_0$.
Let $(T',\L)$ be a branch-\decomposition{} obtained from $(T,\L)$ by
either forking at $u$ by $\V_y$
or splitting at $uv$ by $\V_y$.
If $(T,\L)$ is $x$\=/disjoint, then $(T',\L)$ is $x$\=/disjoint.
\end{LEM}
\begin{proof}
If $(T',\L)$ is obtained by forking, 
then let $v_1uv_2$ be an improper $y$\=/blocking path 
and let $v$ be the neighbor of $u$ other than $v_1$ and $v_2$.
Then we observe that $uv$ $y$\=/guards $v$ and $uv$ cuts $\V_y$.
Note that if 
$(T',\L)$ is obtained by splitting at $uv$, 
then $uv$ also cuts $\V_y$.
We also remark that $(u,v)$ is $y$\=/mixed in both cases.
Since $(T,\L)$ is $x$\=/disjoint and $\V_0\neq \V_x$,
there exists an edge $pq$ of~$T$ such that $\L(T,p,q)=\V_x$ and $q$ is 
incident with an improper $x$\=/degenerate edge.
If $(v,u)$ points toward $p$ and $(p,q)$ points toward $v$, then 
$\L(T,u,v)\subseteq \L(T,p,q)=\V_x$. Then either $\L(T,u,v)\subseteq \V_y$ or $\L(T,u,v)\subseteq \V_0-\V_y$, 
depending on whether $x<y$ or $x,y$ are incomparable, contradicting 
that $(u,v)$ is $y$\=/mixed. 
If $(q,p)$ points toward $v$, 
then, by Lemma~\ref{lem:easycase},
there exists an edge $p'q'$ of $T'$ such that $\L(T',p',q')=\V_x$ and $q'$ is 
incident with an improper $x$\=/degenerate edge. That is, $(T',\L)$ is $x$\=/disjoint as desired.

Now we assume that $(u,v)$ points toward $p$ and $(p,q)$ points toward $u$. 
If $x$ and $y$ are incomparable, then we have $\L(T,v,u)\subseteq \L(T,p,q)=\V_x\subseteq \V_0-\V_y$.
This contradicts the fact that $uv$ cuts $\V_y$.
Hence, we assume that $x<y$.
Since $q$ is incident with an improper $x$\=/degenerate edge, if $(T',\L)$ is obtained by forking at $u$, then 
by Lemmas~\ref{lem:x-to-y} and~\ref{lem:guardfar}, $uv_1$ is 
either $y$\=/guarding or $y$\=/degenerate, a contradiction. 
Now we may further assume that $(T',\L)$ is obtained by splitting at $uv$. 
If $(p,q)=(v,u)$, then clearly $(T',\L)$ is $x$\=/disjoint as desired. 
If $(p,q)\neq (v,u)$, then again by applying Lemmas~\ref{lem:x-to-y} and~\ref{lem:guardfar}, 
we know that $uv$ $y$\=/guards $u$ or $uv$ is $y$\=/degenerate. 
Since $uv$ $y$\=/guards $v$, $uv$ can only be $y$\=/degenerate, which must be improper. 
On the other hand,  $uv$ cuts $\V_x$ because $\L(T,p,q)=\V_x$ and $(p,q)\neq (v,u)$.
This is a contradiction to the fact that $uv$ is an improper $y$\=/degenerate edge because $(T,\L)$ is $x$\=/degenerate.
\end{proof}

For a branch-decomposition $(T,\L)$ of a subspace arrangement $\V_0$, 
a subset $\V'$ of $\V_0$ is said to \emph{induce} a branch-decomposition $(T',\L')$ 
if a subdivision of $T'$ is the minimal subtree of~$T$ containing all leaves in $\L^{-1}(\V')$ and $\L'$ is the restriction of $\L$ on $\V'$.
We write $(T',\L')=(T,L)|_{\V'}$ if $(T',\L')$ is induced by $\V'$ from $(T,\L)$.

\begin{LEM}\label{lem:induced}
Let $x$, $y$ be nodes of $T^b$ with %
$y\nleq x$.
Let $(T,\L)$ be a branch-\decomposition{} of $\V_0$ with $\V_x\cup \V_y\subseteq \V_0$
and $(T',\L)$ be a branch-\decomposition{} obtained from $(T,\L)$ by 
either forking at $v$ by $\V_y$ or splitting at $uv$ by $\V_y$.
Then $(T,\L)|_{\V_x}=(T',\L)|_{\V_x}$.
\end{LEM}
\begin{proof}
If $x\le y$, then
by the construction of $T'$, 
subspaces of $\V_0$ mapped by the leaves of $T^{v,\bar y}$ are not in $\V_y$
and thus not in $\V_x$ because $\V_0-\V_y\subseteq\V_0-\V_x$. %
Therefore, it is clear by definition that $(T,\L)|_{\V_x}=(T',\L)|_{\V_x}$.

If $x$ and $y$ are incomparable, then
subspaces of $\V_0$ mapped by the leaves of $T^{v,y}$ are in $\V_y$ 
and thus not in $\V_x$ because $\V_y\subseteq \V_0-\V_x$.
Therefore, $(T,\L)|_{\V_x}=(T',\L)|_{\V_x}$.
\end{proof}

\begin{COR}\label{cor:xdegenerate}
Let $x$, $y$ be nodes of $T^b$ with %
$y\nleq x$.
Let $(T,\L)$ be a branch-decomposition of $\V_0$ with $\V_x\cup \V_y\subseteq \V_0$
and $(T',\L)$ be a branch-decomposition obtained from $(T,\L)$ by 
either forking at $v$ by $\V_y$ or splitting at $uv$ by $\V_y$.
Then $(T,\L)$ is $x$\=/degenerate if and only if 
$(T',\L)$ is $x$\=/degenerate.
\end{COR}
\begin{proof}
It follows immediately from Lemma~\ref{lem:induced} and the definition of $x$\=/degenerate 
branch-decompositions.
\end{proof}

For %
a node~$x$ of $T^b$ with $\V_x\subseteq\V_0$,
if $(T',\L')=(T,\L)|_{\V_x}$, then by definition, 
there exists a canonical injection $\phi_x:V(T')\to V(T)$
such that for every leaf $\ell$ of $T'$, $\L'(\ell)=\L(\phi_x(\ell))$. We say that a node~$v$ of~$T$ is an 
\emph{$x$\=/branch node} if $\phi_x$ maps a node of $T'$ to $v$, and such $v$ is said to \emph{correspond to} $\phi_x^{-1}(v)$. 
An edge $e$ of~$T$ is said to \emph{correspond to} an edge $f=ab$ of $T'$ if $e$ is on the unique path of~$T$ from $\phi_x^{-1}(a)$ to $\phi_x^{-1}(b)$.
Observe that if an edge $uv$ of~$T$ corresponds to $ab$ of $T'$ and $(u,v)$ points toward $\phi_x^{-1}(b)$, then $\L_x(T,u,v)=\L'(T',a,b)$ by definition of $(T',\L')$.

Lemmas~\ref{lem:notdegenpure} and~\ref{lem:degenpure} will show that
if $(T,\L)$ is $x$\=/pure and $(T',\L)$ is obtained from $(T,\L)$ by either forking or splitting,
then $(T',\L)$ is also $x$\=/pure.

\begin{LEM}\label{lem:notdegenpure}
Let $x$, $y$ be nodes of $T^b$ with $y\nleq x$.
Let $(T,\L)$ be a branch-decomposition of $\V_0$ with $\V_x\cup\V_y\subseteq\V_0$.
Let $(T',\L)$ be a branch-decomposition obtained from $(T,\L)$ by
either forking at $u$ by $\V_y$ for the center $u$ of an improper $y$\=/blocking path,
or splitting at $uv$ by $\V_y$ for an edge $uv$ of~$T$ $y$\=/guarding $v$ improperly.
If $(T,\L)$ is $x$\=/pure, then $(T',\L)$ is $x$\=/pure.
\end{LEM}
\begin{proof}
If $(T,\L)$ is $x$\=/degenerate, then $(T,\L)$ is $x$\=/disjoint.
By Lemma~\ref{lem:disjoint} and Corollary~\ref{cor:xdegenerate}, 
$(T',\L)$ is $x$\=/disjoint and $x$\=/degenerate, which implies that $(T',\L)$ is $x$\=/pure.
So we assume that $(T,\L)$ is not $x$\=/degenerate.

Let $v$ be the neighbor of $u$ that is not on the improper $y$\=/blocking path when $(T',\L)$ is 
obtained by forking at $u$.
Then for both cases, forking at $u$ and splitting at~$uv$, 
we observe that 
\begin{itemize}
\item $uv$ $y$\=/guards $v$, 
\item $uv$ cuts $\V_y$, and 
\item $(u,v)$ is $y$\=/mixed, which implies $\L_y(T,u,v)\neq\L(T,u,v)$.
\end{itemize}

First, we will show that every $x$\=/guarding edge in $(T',\L)$ is proper.
Suppose an edge $p'q'$ of $T'$ $x$\=/guards $q'$ improperly.
Let $(T_x,\V_x)=(T',\L)|_{\V_x}$.
Since $p'q'$ is improper, $q'$ is an $x$\=/branch node of $T'$ and thus 
there exists a node $q_x$ of $T_x$ with $\phi'(q_x)=q'$ 
where $\phi'$ is the canonical injection 
from the set of all nodes of $T_x$ to the set of all nodes of $T'$.
Let $p_xq_x$ be the edge of $T_x$ such that 
an edge $p'q'$ of $T'$ corresponds to $p_xq_x$.
By Lemma~\ref{lem:induced}, $(T,\L)|_{\V_x}=(T_x,\V_x)$.
For the canonical injection $\phi$ 
from the set of all nodes of $T_x$ to the set of all nodes of~$T$,
let $q=\phi(q_x)$
and let $pq$ be the edge of~$T$ corresponding to an edge $p_xq_x$ of $T_x$.
Then we observe that $pq$ $x$\=/guards $q$ so that it is proper because $(T,\L)$ is $x$\=/pure.
And thus, we have $\L(T,p,q)\subseteq \V_x$.

We will consider %
cases depending on the locations of $u$, $v$, $p$, $q$ in~$T$.
If an edge~$uv$ is in the component of $T-pq$ containing $p$,
then by Lemma~\ref{lem:easycase}, 
$\L(T',p',q')=\L(T,p,q)\subseteq \V_x$, and so $p'q'$ is proper, which is a contradiction.
Suppose $uv$ is not in the component of $T-pq$ containing $p$.
If $x$ and $y$ are incomparable, then 
it is also a contradiction because $uv$ cuts $\V_y$ and $\L(T,u,v)\subseteq\L(T,p,q)\subseteq \V_x$.
Thus we assume $x<y$.
If $(v,u)$ points toward $p$, then
it is a contradiction because $\L_y(T,u,v)\neq\L(T,u,v)$ and $\L(T,u,v)\subseteq\L(T,p,q)\subseteq \V_x$.
If $(u,v)$ points toward $p$, then by Lemmas~\ref{lem:x-to-y} and~\ref{lem:guardfar},
$vu$ $y$\=/guards $u$ or is $y$\=/degenerate because $pq$ $x$\=/guards~$q$.
It is a contradiction because $uv$ $y$\=/guards $v$.
Therefore, we conclude that $p'q'$ is proper, which implies that every $x$\=/guarding edge in $(T',\L)$ 
is proper.

It remains to show that every $x$\=/blocking path in $(T',\L)$ is proper.
Suppose $p_1'p'p_2'$ is an improper $x$\=/blocking path in $T'$
and $q'$ is a neighbor of $p'$ other than $p_1'$ and $p_2'$.
Let $(T_x,\V_x)=(T',\L)|_{\V_x}$.
Since $p_1'p'p_2'$ is improper, $p'$ is an $x$\=/branch node of $T'$ and thus 
there exists $p_x$ with $\phi'(p_x)=p'$ 
where $\phi'$ is the canonical injection 
from the set of all nodes of $T_x$ to the set of all nodes of $T'$.
Let $p_xq_x$ be the edge of $T_x$ such that 
an edge $p'q'$ of $T'$ corresponds to $p_xq_x$.
By Lemma~\ref{lem:induced}, let $(T,\L)|_{\V_x}=(T_x,\V_x)$.
For the canonical injection $\phi$ 
from the set of all nodes of $T_x$ to the set of all nodes of~$T$,
let $p=\phi(p_x)$
and let $pq$ be the edge of~$T$ corresponding to an edge $p_xq_x$ of $T_x$.
Similarly, $p_1p$ and $pp_2$ are the edges of~$T$ corresponding to the edges of $T_x$ incident with $p_x$.
Then $p_1pp_2$ is an $x$\=/blocking path in~$T$ by Lemma~\ref{lem:induced}.
Since $(T,\L)$ is $x$\=/pure, we observe that $\L(T,p,q)\subseteq \V_x$.
We also observe that $pq$ $x$\=/guards $q$. %
It is a contradiction because $uv$ $y$\=/guards $v$.
Therefore, $p_1'p'p_2'$ is proper, which implies that every $x$\=/blocking path in $(T',\L)$ is proper.
This completes the proof.
\end{proof}

Let $y$ be a node of $T^b$ and $(T,\L)$ be a branch-decomposition of $\V_0$ with $\V_y\subseteq\V_0$.
We say that an improper $y$\=/degenerate edge $uv$ of~$T$ is \emph{$y$\=/acceptable} 
if either $uv$ is $y$\=/crossing, or $\V_0-\V_y\subseteq \L(T,u,v)$ and there is no edge $ab\neq uv$ of~$T$ such that 
$\V_0-\V_y\subseteq \L(T,a,b) \subseteq \L(T,u,v)$.

\begin{LEM}\label{lem:degenpure}
Let $x$, $y$ be nodes of $T^b$ with $y\nleq x$.
Let $(T,\L)$ be a branch-\decomposition{} of $\V_0$ with $\V_x\cup\V_y\subseteq\V_0$.
Let $(T',\L)$ be a branch-decomposition of $\V_0$ obtained from $(T,\L)$ by
splitting at a $y$\=/acceptable edge $uv$ by $\V_y$. 
If $(T,\L)$ is $x$\=/pure, then $(T',\L)$ is $x$\=/pure.
\end{LEM}
\begin{proof}
First, consider the case when $(T,\L)$ is $x$\=/degenerate. 
Then $(T,\L)$ is $x$\=/disjoint and by Lemma~\ref{lem:disjoint}, $(T',\L)$ is $x$\=/disjoint.
Thus $(T',\L)$ is $x$\=/pure because it is $x$\=/degenerate by Corollary~\ref{cor:xdegenerate}.

Now suppose that $uv$ is an edge of $T[V(T_q)\cup \{p\}]$, where $T_q$ is the component 
of $T-pq$ containing $q$. If $(p,q)=(u,v)$ or $uv$ is an edge of $T_q$ such that $(v,u)$ points toward $q$, then 
$\L(T,u,v)\subseteq \L(T,p,q)= \V_x$. Recall that $\V_x \subseteq \V_y$ or $\V_x\cap \V_y=\emptyset$, depending on whether 
$x<y$ or $x,y$ are incomparable nodes of $T^b$, and therefore we have $\L(T,u,v) \subseteq \V_y$ or $\L(T,u,v)\cap \V_y=\emptyset$. 
This, however, contradicts to the assumption that $(u,v)$ is $y$\=/mixed, which is the requirement 
for splitting at $uv$. Therefore, we conclude that either $(p,q)=(v,u)$ or $uv$ is an edge of $T_q$ such that 
$(u,v)$ points toward $q$. Now observe that $\L(T,v,u)\subseteq \L(T,p,q)\subseteq \V_x$. 
This means that $x<y$: indeed, if $x,y$ are incomparable, then $\L(T,v,u)\cap \V_y=\emptyset$, which contradicts 
that $uv$ cuts $\V_y$ as an improper $y$\=/degenerate edge.

Since $uv$ is an improper $y$\=/degenerate edge and $(T,\L)$ is $x$\=/degenerate, $uv$ does not cut $\V_x$. 
Therefore, we deduce that $(p,q)=(v,u)$. Observe that splitting at $uv$ by $\V_y$ keeps $T_q$ intact, that is, $T_q$ 
is identical to a subtree of $T'$. It follows that $(T',\L)$ is $x$\=/disjoint as desired. This establishes the case 
when $(T,\L)$ is $x$\=/degenerate.

Second, we consider the case when $(T,\L)$ is not $x$\=/degenerate. Suppose that the statement does not hold. 
That is, $T'$ has an improper $x$\=/guarding edge 
or an improper $x$\=/blocking path. We take an edge $p'q'$ of $T'$ as follows: in the former case, $p'q'$ is 
an edge of $T'$ improperly $x$\=/guarding $q'$. In the latter case, we take $p'$ so that it is the center of an improper $x$\=/blocking path 
and $p'q'$ is an edge incident with $p'$ that is not on the $x$\=/blocking path. Notice that $p'q'$ $x$\=/guards $q'$ in both cases. 
Since $(p',q')$ is $x$\=/mixed and $\L_x(T',q',p')\neq \emptyset$ in both cases, 
$p'q'$ cuts $\V_x$. Therefore, $p'q'$ corresponds to an edge $f$ of $(T',\L)|_{\V_x}$. 

Now we shall choose an edge of~$T$ corresponding to $f$ of $(T',\L)|_{\V_x}$. 
Clearly, $T$ has an edge corresponding to $f$  because $(T,\L)|_{\V_x}=(T',\L)|_{\V_x}$.
Note that in the case when $p'q'$ is an improper $x$\=/guarding edge, $q'$ is an $x$\=/branch node by definition.
When $p'q'$ is chosen as an edge incident with an improper $x$\=/blocking path, 
$p'$ is an $x$\=/branch node by definition. 
We choose $pq$ so that $q$ and $q'$ correspond to the same node of $(T,\L)|_{\V_x}$ in the former case, and 
$p$ and $p'$ correspond to the same node of $(T,\L)|_{\V_x}$ in the latter case. 
Once $p$ or $q$ is fixed, $pq$ is uniquely determined by the condition that $pq$ and $p'q'$ 
correspond to $f$. 
Observe that $pq$ is $x$\=/guarding in the former case and it must be proper as $(T,\L)$ is $x$\=/pure.
If $p'q'$ is chosen from an improper $x$\=/blocking path of $T'$, then $p$ is the center of an $x$\=/blocking path of~$T$, which 
must be proper,  
and $pq$ is the edge incident with $p$ that is not on this $x$\=/blocking path. 
Therefore, $pq$ $x$\=/guards $q$. 
Since $(p',q')$ is $x$\=/mixed and $\L_x(T',p',q')=\L_x(T,p,q)$ while $(p,q)$ is not $x$\=/mixed in~$T$, it follows that $\L(T,p,q)\subseteq \V_x$. 
This also implies that $q$ is an $x$\=/branch node.

Suppose that $u$ is a node of the component of $T-pq$ containing $p$. 
For $\L(T,p,q)\subseteq \V_x$, we can apply Lemma~\ref{lem:easycase} and deduce that
$T'$ contains an edge $p''q''$ such that $\L(T,p,q)=\L(T',p'',q'')$ and $q''=q$. 
Then $q''$ is an $x$\=/branch node of $T'$ which correspond to 
the same node of $(T,\L)|_{\V_x}$ as $q$, implying $q''=q'$ because $\L(T',p'',q'')\subseteq \V_x$.
It follows that $p''q''=p'q'$ as $\L_x(T',p',q')=\L_x(T,p,q)=\L_x(T',p'',q'')$. 
Now, $\L(T',p',q')=\L(T',p'',q'')\subseteq \V_x$, which contradicts the choice of $p'q'$. 

Therefore, $uv$ is either an edge of $T_q$, where $T_q$ is the component of $T-pq$ containing $q$, or $(p,q)=(v,u)$. 
Especially, if $uv$ is an edge of $T_q$ then $(u,v)$ must point toward $q$ since otherwise $(u,v)$ is not $y$\=/mixed. 
Also similarly as in the previous case, we have $x<y$. 

Let $P=w_0w_1\cdots w_s$ be the unique path of~$T$ from $p=w_0$ to $u=w_s$ and recall that the first edge is $pq$ and the last edge is $uv$. 
We claim that every edge on $P$ is an improper $y$\=/degenerate edge. 
Lemma~\ref{lem:guardfar} and the existence of $uv$ such that $(u,v)$ pointing toward $q$ 
imply that every edge $w_{i-1}w_i$ $y$\=/guards $w_{i-1}$ or $w_{i-1}w_i$ is $y$\=/degenerate. 
On the other hand, $pq=w_0w_1$ $x$\=/guards $q=w_1$. By Lemma~\ref{lem:guardfar}, 
every edge $w_{i-1}w_i$ $x$\=/guards $w_i$ or is $x$\=/degenerate, and thus 
$w_{i-1}w_i$ $y$\=/guards $w_i$ or is $y$\=/degenerate due to Lemma~\ref{lem:x-to-y}.
It follows that every edge on $P$ is $y$\=/degenerate, including $pq$. Now, recall that 
$uv$ is $y$\=/acceptable and $uv$ is not $y$\=/crossing due to $\L(T,v,u)\subseteq \V_x$. 
Therefore, $uv$ is an improper $y$\=/degenerate edge such that $\L(T,u,v)$ is a minimal set containing $\V_0-\V_y$. 
We conclude that $(p,q)=(v,u)$ and $(T',\L)$ is obtained by splitting at $uv=qp$. It is easy to see 
that $q'=q$ and $\L(T',p',q')=\L(T,p,q)\subseteq \V_x$. However, this contradicts 
the choice of $p'q'$.
\end{proof}

\begin{LEM}\label{lem:terminate}
Let $y$ be a node of $T^b$ and 
let $(T,\L)$ be a branch-\decomposition{} of a subspace arrangement $\V_0$ such that $\V_y\subseteq \V_0$ 
and $(T,\L)$ is not $y$\=/degenerate.
By repeatedly applying forking if it has an improper $y$\=/blocking path, 
and applying splitting if it has an improper $y$\=/guarding edge to $(T,\L)$,
we can construct a branch-\decomposition{}
that has neither improper $y$\=/blocking paths nor improper $y$\=/guarding edges.
\end{LEM}
\begin{proof}
Suppose $(T,\L)$ has an improper $y$\=/blocking path or an improper $y$\=/guarding edge.
We are going to show that 
if we apply forking or splitting to $(T,\L)$,
then 
the number of improper $y$\=/blocking paths and improper $y$\=/guarding edges 
must decrease strictly.

Let $v_1vv_2$ be an improper $y$\=/blocking path in $(T,\L)$.
Let $(T',\L)$ be a branch-\decomposition{} obtained from $(T,\L)$ by forking at $v$ by $\V_y$
such that %
$v'$ is the new node in $T'$ created by subdividing $vv_2$.
Let us consider the number of improper $y$\=/blocking paths and improper $y$\=/guarding edges in $T'$ as follows.
\begin{itemize}
\item
Let $V_1$ and $V_2$ be the set of all nodes in the components of $T-v$ containing $v_1$ and $v_2$, respectively. 
The total number of improper $y$\=/blocking paths and improper $y$\=/guarding edges 
in $T'[V_1\cup\{v\}]$ and $T'[V_2\cup\{v'\}]$
is equal to the number of those in $T[V_1\cup\{v\}]$ and $T[V_2\cup\{v\}]$.
\item
By the construction of $T'$, 
all $y$\=/blocking paths and all $y$\=/guarding edges 
in $T'[V(T^{v,\bar{y}})\cup\{v,v',v_2\}]$
are proper 
because a subspace mapped from a nonroot leaf of $T^{v,\bar{y}}$ by $\L$ is not in $\V_y$.
Note that edges $vv'$ and $v'v_2$ of $T'$ are not $y$\=/guarding because $v_1vv_2$ in~$T$ is $y$\=/blocking.
\item
For a $y$\=/blocking path $p_1pp_2$ in $T'[V(T^{v,y})\cup\{v\}]$, 
if $p_3$ is the neighbor of $p$ other than $p_1$ and $p_2$,
then $(p,p_3)$ does not point toward $v$ by Lemma~\ref{lem:guardfar}
because $vv_3$ $y$\=/guards $v_3$.
Thus, every $y$\=/blocking path in $T'[V(T^{v,y})\cup\{v\}]$ is proper
because a subspace mapped from a nonroot leaf of $T^{v,\bar{y}}$ by $\L$ is in~$\V_y$.
Similarly, every $y$\=/guarding edge in $T'[V(T^{v,y})\cup\{v\}]$ is proper.
\item 
As $v_1vv'$ is a $y$\=/blocking path, neither $v_1vv_3$ nor $v'vv_3$ in $T'$ is a $y$\=/blocking path
where $v_3$ is the root of $T^{v,y}$.
\end{itemize}
The number of improper $y$\=/blocking paths and improper $y$\=/guarding edges
should decrease strictly
because $v_1vv_2$ in~$T$ is improper but neither $v_1vv'$ nor $vv'v_2$ is proper.

Similarly,
it is easy to check that the number of improper $y$\=/blocking paths and improper $y$\=/guarding edges 
decreases strictly when we apply splitting to $(T,\L)$.
This completes the proof.
\end{proof}

Finally we give the proof of Proposition~\ref{prop:pure} that is our goal of this section.
\begin{proof}[Proof of Proposition~\ref{prop:pure}]
We may assume that $\abs{\V_0}\ge2$.
For a branch-\decomposition{} $(T,\L)$ of $\V_0$, let
\[S_{(T,\L)}=\{x\in V(T^b) : \V_x\subseteq\V_0 \text{ and } (T,\L) \text{ is $z$-pure for all } z\le x\}.\]
We choose a branch-decomposition $(T,\L)$ of $\V_0$ having width at most~$k$
so that $\abs{S_{(T,\L)}}$ is maximal.
If $(T,\L)$ is totally pure with respect to $(T^b,\L^b)$, then we are done.
If $(T,\L)$ is not totally pure with respect to $(T^b,\L^b)$,
then there exists a node $y$ of $T^b$ with $\V_y\subseteq\V_0$ such that 
$y\notin S_{(T,\L)}$ and two children of $y$ are in $S_{(T,\L)}$.

First let us consider the case when $(T,\L)$ is $y$\=/degenerate.
\begin{itemize}
\item If $T$ has an improper $y$\=/degenerate edge $u'v'$ such that $\V_0-\V_y\subseteq\L(T,u',v')$,
then choose a $y$\=/acceptable edge $uv$ of~$T$.
And then we apply splitting at $uv$ by $\V_y$ to $(T,\L)$ so that 
we obtain a new branch-decomposition~$(T',\L)$.

\item If there is no such edge, 
then there must exist an improper $y$\=/degenerate edge $ab$ that is $y$\=/crossing.
We apply splitting at $ab$ by $\V_y$ to $(T,\L)$ so that we obtain a new branch-decomposition $(T'',\L)$.
Let $c$ be the root of $T^{b,y}$. We claim that $cb$ is an improper $y$\=/degenerate edge such that 
$\V_0-\V_y\subseteq\L(T'',c,b)$. %
Let $d$ be the neighbor of $b$ other than $a$ and $c$, so that $d$ is the root of $T^{b,\bar{y}}$.
Then $\L_y(T'',b,d)=\emptyset$.
Thus, $\L_y(T'',b,a)=\L_y(T'',c,b)$ and $\L_y(T'',a,b)=\L_y(T'',b,c)$. 
This proves the claim because $ab$ is an improper $y$\=/degenerate edge.
And then we choose a $y$\=/acceptable edge $uv$ of $T''$.
Now we apply splitting at $uv$ by $\V_y$ to $(T'',\L)$ so that 
we obtain a new branch-decomposition $(T',\L)$.
\end{itemize}
By the construction of $(T',\L)$, $(T',\L)$ is $y$\=/disjoint.
Note that $(T,\L)$ is $y$\=/degenerate if and only if $(T',\L)$ is $y$\=/degenerate because 
$(T,\L)|_{\V_y}=(T',\L)|_{\V_y}$.
This implies that $(T',\L)$ is $y$\=/pure because $(T',\L)$ is $y$\=/degenerate.
Then $S_{(T',\L)}\supseteq S_{(T,\L)}\cup\{y\}$ by Lemma~\ref{lem:degenpure},
which contradicts the maximality of $S_{(T,\L)}$.
Note that $(T',\L)$ has width at most~$k$ by Proposition~\ref{prop:split}.

If $(T,\L)$ is not $y$\=/degenerate, then 
there exists an improper $y$\=/blocking path or an improper $y$\=/guarding edge.
By repeatedly applying forking if it has an improper $y$\=/blocking path 
and applying splitting if it has an improper $y$\=/guarding edge to $(T,\L)$,
we can construct a new branch-decomposition $(T',\L)$
that has neither improper $y$\=/blocking paths nor improper $y$\=/guarding edges by Lemma~\ref{lem:terminate}.
Note that $(T,\L)$ is $y$\=/degenerate if and only if $(T',\L)$ is $y$\=/degenerate because 
$(T,\L)|_{\V_y}=(T',\L)|_{\V_y}$.
Thus $(T',\L)$ %
is not $y$\=/degenerate and so it is $y$\=/pure.
Also, for a node~$x$ of $T^b$ with $y\nleq x$,
if $(T,\L)$ is $x$\=/pure, then $(T',\L)$ is $x$\=/pure by Lemma~\ref{lem:notdegenpure}.
Note that the resulting branch-decomposition has width at most~$k$
by Propositions~\ref{prop:fork} and~\ref{prop:split}.
Then $S_{(T',\L)}\supseteq S_{(T,\L)}\cup\{y\}$, 
which contradicts the maximality of~$S_{(T,\L)}$.
\end{proof}

The following lemma will be used later.
\begin{LEM}\label{lem:pure}
Let $x_1$ and $x_2$ be two children of a node~$x$ in $T^b$.
Let $(T,\L)$ be a branch-decomposition of $\V_x$. 
If $(T,\L)$ is totally pure with respect to $(T^b,\L^b)$,
then %
$(T,\L)|_{\V_{x_i}}$ is totally pure with respect to $(T^b,\L^b)$ 
for each $i=1,2$.
\end{LEM}
\begin{proof}
By symmetry, it suffices to show that $(T,\L)|_{\V_{x_1}}$ is 
$y$\=/pure for every node~$y$ of $T^b$ with $y\le x_1$.
Trivially, $(T,\L)|_{\V_{x_1}}$ is $x_1$-pure.

Let $(T',\L')=(T,\L)|_{\V_{x_1}}$ and let $y$ be a node of $T^b$ with $y<x_1$.
As $(T',\L')|_{\V_y}$ is equal to $(T,\L)|_{\V_y}$,
$(T',\L')$ is $y$\=/degenerate if and only if $(T,\L)$ is $y$\=/degenerate.

If $(T,\L)$ is $y$\=/degenerate, then $(T,\L)$ is $y$\=/disjoint so that $T$ has an edge $uv$
such that $\L(T,u,v)=\V_y$ and $v$ is incident with an improper $y$\=/degenerate edge.
By the construction of $T'$, 
$T'$ contains an edge $u'v'$,
which corresponds to an edge $uv$ of~$T$,
such that $\L'(T',u',v')=\V_y$ and $v'$ is incident with an improper $y$\=/degenerate edge.
So $(T',\L')$ is $y$\=/disjoint and thus $y$\=/pure because $(T',\L')$ is $y$\=/degenerate.

If $(T,\L)$ is not $y$\=/degenerate, then all $y$\=/blocking paths and all $y$\=/guarding edges of~$T$ are proper
because $(T,\L)$ is $y$\=/pure.
For each $y$\=/blocking path $P'$ in $(T',\L')$, 
there exists the corresponding $y$\=/blocking path $P$ in $(T,\L)$. 
As $P'$ is in $(T',\L')$ and $P$ is proper, if $p$ is the center of $P$ and $q$ is the neighbor of $p$
that is not on $P$, then $(p,q)$ is not $y$\=/mixed.
Thus, we can conclude that $P'$ is proper.
Similarly, it is easy to see that every $y$\=/guarding edge in $(T',\L')$ is proper.
This completes the proof because $(T',\L')$ is not $y$\=/degenerate.
\end{proof}

\section{Properties of namus}\label{sec:namu}

In this section we review the properties of $B$-namus which will be pivotal for proving the correctness of Algorithm~\ref{alg:fullset} and analyzing its time complexity. 

\subsection{Typical sequences}
We review the notion of \emph{typical sequences} 
proposed by Bodlaender and Kloks~\cite{BK1996}.
Let $s=a_1,a_2,\ldots,a_n$ be a sequence of integers.
Let $\tau(s)$ be a subsequence of $s$ obtained
by applying the following two operations as many as possible.
(i) Delete $a_{i+1}, \ldots, a_{j-1}$ 
if there exist $i<j$ such that 
$a_i\le a_k \le a_j$ or $a_i\ge a_k \ge a_j$
for all $k$ with $i\le k \le j$.
(ii) Remove $a_k$ if $a_k=a_{k+1}$ for some $k$.
For example, if $s=1,2,5,3,4,2,4,4$, then $\tau(s)=1,5,2,4$.
We say that a sequence $s$ is \emph{typical} if $s=\tau(s)$.
The following lemmas give upper bounds for the length of a typical sequence and the number of typical sequences consisting of integers in $\{0,\ldots,k\}$.

\begin{LEM}[{\cite[Lemma 3.3]{BK1996}}]\label{lem:lengthtypical}
The length of a typical sequence consisting of integers in $\{0,\ldots , k\}$ is at most $2k+1$.
\end{LEM}

\subsection{Compact $\boldsymbol B$-namus}
In Subsection~\ref{subsec:namu}, $B$-namus are introduced.
We will show that for a vector space $B$ of small dimension over a finite field $\F$, each compact $B$-namu of width at most~$k$ has small size and the number of such $B$-namus is bounded by a function of $k$.

The authors introduced the notion of \emph{$B$-trajectories} in~\cite{JKO2016}, which can be seen as a linear version of $B$-namus related to path-decompositions.
The following lemma can be easily deduced from Lemma~10 in~\cite{JKO2016}, which uses Lemma~\ref{lem:lengthtypical}.

\begin{LEM}\label{lem:trajlength}
Let $\F$ be a field and let $B$ be a subspace of $\F^r$ of dimension $\theta$.
If $\Gamma$ is a compact $B$-namu of width at most~$k$,
then the diameter of $T(\Gamma)$ is at most $(2\theta+1)(2k+1)$.
\end{LEM}

\begin{LEM}\label{lem:numedge}
Let $d$ be an odd integer.
If $T$ is a subcubic tree of diameter at most~$d$,
then the number of edges in~$T$ is at most $2^{(d+3)/2}-3$. %
\end{LEM}

\begin{proof}
We may assume that the diameter of~$T$ is exactly $d$.
Let $e$ be the middle edge in a longest path of~$T$.
There are two components in $T-e$ and
each component is a rooted tree of height $(d-1)/2$ such that 
each vertex has at most two children.
It follows that each component of $T-e$ has at most $2+4+\cdots+2^{(d-1)/2}=2 \times (2^{(d-1)/2}-1)$ edges.
Therefore, the number of edges in~$T$ is at most 
\[4\times \left(2^{(d-1)/2}-1\right)+1=2^{(d+3)/2}-3. \qedhere \] 
\end{proof}

\begin{LEM}\label{lem:ukbtime}
Let $\F$ be a finite field and let $B$ be a subspace of $\F^r$ of dimension~$\theta$. 
If $\Gamma$ is a compact $B$-namu of width at most~$k$,
then the number of edges in $T(\Gamma)$ is at most $2^{2\theta k+\theta+k+2}-3$. %
Furthermore, the set $U_k(B)$ contains at most $f(k,\theta,\abs{\F})$ elements and 
can be generated from $B$ in $g(k,\theta,\abs{\F})$ steps
for some functions $f$ and~$g$. %
\end{LEM}

\begin{proof}
Let $\Gamma=(T,\alpha,\lambda,U)\in U_k(B)$.
By Lemmas~\ref{lem:trajlength} and~\ref{lem:numedge},
the number of edges in~$T$ is at most 
$2^{2\theta k+\theta +k+2}-3$.
Therefore, the number of subcubic trees on the node set $\{1,\ldots,n\}$ for some $n\le 2^{2\theta k+\theta +k+2}-2$
is bounded by some function of $k$ and $\theta$.

Let $q=\abs{\F}$. 
Note that $U$ is a subspace of $B$ and there exist $\sum_{i=0}^{\theta}\binom{\theta}{i}_q$ subspaces of $B$
where $\binom{\theta}{i}_q=\frac{(q^\theta-1)(q^{\theta-1}-1)\cdots(q^{\theta-i+1}-1)}{(q^i-1)(q^{i-1}-1)\cdots (q-1)}$ 
for $i\le \theta$, which is known as Gaussian binomial coefficients. 

Given a fixed ordered basis $\B$ of $B$, we explain how to generate all subspaces of $B$ of dimension $i$ in the claimed running time. 
A subspace $B'$ of $B$ with a fixed ordered basis $\B'$ can be represented a $\theta\times i$ matrix $M_{\B,\B'}$ over $\F$, %
where each column represents the coordinate vector of an element of $\B'$ with respect to $\B$. 
We note that two matrices $M_{\B,\B'}$ and $M_{\B,\B''}$ represent the same subspace if and only if $M_{\B,\B'}=M_{\B,\B''} A$ for some $i\times i$ invertible matrix $A$ over $\F$. %
By enumerating all possible independent columns of $M_{\B,\B'}$ and all $i\times i$ invertible matrices over $\F$, %
we get representations of distinct subspaces exactly as many as the Gaussian binomial coefficients. That is, by sieving out matrices
representing the same subspace, we can obtain all (and distinct) subspaces of $B$ as $\theta\times i$ matrices over $\F$. %

We first fix a subcubic tree $T$ on the node set $\{1,\ldots,n\}$ for some integer $n\le 2^{2\theta k+\theta +k+2}-2$
and fix a subspace $U$ of $B$.
For every incidence $(v,e)$ in~$T$, we assign a subspace $\alpha(v,e)$ of $U$.
For every edge $e$ in~$T$, assign a nonnegative integer $\lambda(e)$ with $\lambda(e)\le k$.
For every two-edge path $xyz$ in~$T$, 
we check whether $\alpha(x,xy)$ is a subspace of $\alpha(y,yz)$.
Lastly, we check whether $\lambda(e)\ge\dim(\alpha(v,e)\cap\alpha(u,e))$ for every edge $e=uv$ of~$T$.
If they hold, then $(T,\alpha,\lambda,U)$ is a $B$-namu.

Since 
the number of incidences in~$T$ and the number of subspaces of $U$
 are at most some function of $k$, $\theta$, $q$, 
we conclude that $\abs{U_k(B)}$ is at most $f(k,\theta,q)$ for some function $f$ of $k$, $\theta$, $q$.
Also, the set $U_k(B)$ can be computed in time $g(k,\theta,q)$ for some function $g$ of $k$, $\theta$, $q$.
\end{proof}

\subsection{A sum of $\boldsymbol B$-namus}

In Subsection~\ref{subsec:sum}, we defined the sum of two $B$-namus by a $(T_1,T_2)$-model in a tree. We provide the missing proof that this sum is well defined.
 
\begin{LEM}\label{lem:sum}
Given two $B$-namus $\Gamma_1=(T_1,\alpha_1,\lambda_1,U_1)$, $\Gamma_2=(T_2,\alpha_2,\lambda_2,U_2)$, 
and a $(T_1,T_2)$-model $(\eta_1,\eta_2)$ in a tree $T^+$,
we define $\alpha^+$, $\lambda^+$, and $U^+$ as follows:
\begin{enumerate}[(i)]
\item $U^+=U_1+U_2$,
\item $\alpha^+=\alpha_1\circ \vec\eta_1 +\alpha_2\circ \vec\eta_2$, 
\item $\lambda^+(\emptyset)=0$, and 
\item for all $e=uv\in E(T)$, %
\begin{align*}
\lambda^+(e)&=\lambda_1\circ\eta_1(e) + \lambda_2\circ\eta_2(e) 
-\dim (\alpha_1(\vec\eta_1(v,e))\cap\alpha_2(\vec\eta_2(v,e)))  
\\&\quad
-\dim (\alpha_1(\vec\eta_1(u,e))\cap\alpha_2(\vec\eta_2(u,e)))
+\dim (U_1 \cap U_2).
\end{align*}
\end{enumerate}
Then $\Gamma^+=(T^+,\alpha^+,\lambda^+,U^+)$ is a $B$-namu.
\end{LEM}

\begin{proof}%
It suffices to verify that $\lambda^+(e)\geq \dim (\alpha^+(u,e)\cap \alpha^+(v,e))$,
which is the fourth condition of the definition of a $B$-namu. 
By Lemma~\ref{lem:dim-join},
\begin{align*}
\lefteqn{\dim (\alpha^+(u,e)\cap \alpha^+(v,e))}\\&= \dim \Bigl(\bigl(\alpha_1(\vec\eta_1(u,e))+\alpha_2(\vec\eta_2(u,e)) \bigr)\cap \bigl(\alpha_1(\vec\eta_1(v,e))+\alpha_2(\vec\eta_2(v,e))\bigr)\Bigr)\\
&=\dim \bigl(\alpha_1(\vec\eta_1(u,e))\cap \alpha_1(\vec\eta_1(v,e))\bigr) + \dim \bigl(\alpha_2(\vec\eta_2(u,e))\cap \alpha_2(\vec\eta_2(v,e))\bigr) \\
& \qquad - \dim (\alpha_1(\vec\eta_1(u,e))\cap \alpha_2(\vec\eta_2(u,e))) - \dim (\alpha_1(\vec\eta_1(v,e))\cap \alpha_2(\vec\eta_2(v,e)))\\
&\qquad 
+\dim  \Bigl(\bigl(\alpha_1(\vec\eta_1(u,e))+\alpha_1(\vec\eta_1(v,e))\bigr) \cap \bigl(\alpha_2(\vec\eta_2(u,e))+\alpha_2(\vec\eta_2(v,e))\bigr)\Bigr),
\end{align*}
and hence,
\begin{align*}
&\lambda^+(e)- \dim (\alpha^+(u,e)\cap \alpha^+(v,e))\\
  &\geq \lambda_1(\eta_1(e)) \\
  &\quad- \dim (\alpha_1(\vec\eta_1(u,e))\cap \alpha_1(\vec\eta_1(v,e))) + \lambda_2(\eta_2(e)) - \dim (\alpha_2(\vec\eta_2(u,e))\cap \alpha_2(\vec\eta_2(v,e)))\\
&\quad +\dim (U_1\cap U_2) - \dim  \bigl((\alpha_1(\vec\eta_1(u,e))+\alpha_1(\vec\eta_1(v,e))) \cap (\alpha_2(\vec\eta_2(u,e))+\alpha_2(\vec\eta_2(v,e)))\bigr)\\
& \geq 0. \qedhere 
\end{align*}
\end{proof}

For algorithms, it is necessary to bound the size of the set $\Gamma_1\tplus \Gamma_2$. For this, the next  lemma places an upper bound on  the size of a sum of two $B$-namus. 
 
\begin{LEM}\label{lem:sumsize}
Let $T_1$, $T_2$, and $T$ be three subcubic trees.
If there exists a $(T_1,T_2)$-model in~$T$, then 
\[\abs{V(T)}=
  \begin{cases}
\abs{V(T_1)}+\abs{V(T_2)}+2 &\text{if } \abs{V(T_1)}\ge 2, ~\abs{V(T_2)}\ge 2,\\
\abs{V(T_1)}+\abs{V(T_2)} &\text{if } \abs{V(T_1)}=1, ~\abs{V(T_2)}=1,\\
\abs{V(T_1)}+\abs{V(T_2)}+1 &\text{otherwise.}
  \end{cases}
\]
\end{LEM}
\begin{proof}
It is trivial if $\abs{V(T_1)}=1$ or $\abs{V(T_2)}=1$. So we may assume $ \abs{V(T_1)}\ge 2$, $\abs{V(T_2)}\ge 2$.
Let $(\eta_1,\eta_2)$ be a $(T_1,T_2)$-model in~$T$.
By condition (ii) of a $(T_1,T_2)$-model, the number of degree-$1$ nodes in~$T$ is 
$\abs{A(T_1)}+\abs{A(T_2)}$. 
Since $T$ is subcubic, $T$~contains $\abs{A(T_1)}+\abs{A(T_2)}-2$ nodes of degree $3$. 
Every degree-$2$ node in~$T$ is a branch node in exactly one of $\eta_1$ and $\eta_2$ by condition (iii), and thus the number of degree-$2$ nodes in~$T$ is $\abs{V(T_1)}-\abs{A(T_1)}-(\abs{A(T_1)}-2)+\abs{V(T_2)}-\abs{A(T_2)}-(\abs{A(T_2)}-2)$. To sum up, we have \[\abs{V(T)}=\abs{V(T_1)}+\abs{V(T_2)}+2.\qedhere  \]
\end{proof}

Using Lemma~\ref{lem:sumsize}, we can bound the size of the set $\Gamma_1\tplus \Gamma_2$ for the sets of compact $B$-namus.
\begin{LEM}\label{lem:generatesum}
Let $B$ be a subspace of $\F^r$ of dimension at most $\theta$ 
and let $\Gamma_1$, $\Gamma_2$ be compact $B$-namus of width at most~$k$. 
Then the set $\Gamma_1\tplus \Gamma_2$ contains at most $2^{2^{2(2\theta k+\theta+k+3)}}$ $B$-namus.
Furthermore, the number of sum operations needed to generate the set $\Gamma_1 \tplus \Gamma_2$ is bounded by a function  depending only on $\theta$ and $k$. %
\end{LEM}
\begin{proof}
  For each $i=1,2$, let $\Gamma_i=(T_i,\alpha_i,\lambda_i,U_i)$ and 
  $n_i$ be the number of leaves of $T_i$. 
  By Lemma~\ref{lem:ukbtime}, $T_i$ has at most $2^{2\theta k+\theta+k+2}-2$ nodes and therefore 
  \[
  n_i \le 2^{2\theta k+\theta+k+1}, 
  \]
  because $2n_i-2\le \abs{V(T_i)}$ as $T_i$ is subcubic.
  If $n_1=1$ or $n_2=1$, then it is trivial.  So let us assume that $n_1>1$ and $n_2>1$.

  To enumerate all $\Gamma=(T,\alpha,\lambda,U)\in \Gamma_1\tplus \Gamma_2$, 
  we will enumerate all subcubic trees~$T$ on $\abs{V(T_1)}+\abs{V(T_2)}+2$ labeled nodes with $n_1+n_2$ labeled leaves.
  Then we map the first $n_1$ leaves of~$T$ to the leaves of $T_1$ in the given order
  and the last $n_2$ leaves of~$T$ to the leaves of $T_2$. 
  If this mapping of leaves induces a $(T_1,T_2)$-model $(\eta_1,\eta_2)$ in~$T$, then compute $\Gamma$ so that $(\eta_1,\eta_2)$ co-extends $\Gamma_1$ and $\Gamma_2$ to $\Gamma$.
  
  Observe that $T$ has precisely $\abs{V(T_1)}+\abs{V(T_2)}+2\le 2^{2\theta k +\theta+k+3}-2$ nodes by Lemma~\ref{lem:sumsize}.
  There are at most  $F(\theta,k)$ labeled subcubic trees with $\abs{V(T_1)}+\abs{V(T_2)}+2$ nodes
  and exactly $n_1+n_2$ leaves
  for some function $F(\theta,k)$. 
  (As the number of $n$-node labeled trees is $n^{n-2}$, $F(\theta,k)\le (2^{2\theta k+\theta+k+3}-2)^{2^{2\theta k+\theta+k+3}-4}$ trivially.)
  We try to map the first $n_1$ leaves (in the ordering of node labels) of~$T$ to the leaves of $T_1$ by~$\eta_1$
  and the other $n_2$ leaves of~$T$ to the leaves of $T_2$ by $\eta_2$
  and check if this induces a $(T_1,T_2)$-model.
  Thus, the number of $B$-namus in $\Gamma_1\tplus \Gamma_2$ is at most 
  \[F(\theta,k) 
  < 2^{(2\theta k+\theta+k+3){2^{2\theta k+\theta+k+3}}}
  \le 2^{2^{2(2\theta k+\theta+k+3)}}.
  \]
  Clearly $\Gamma_1\tplus \Gamma_2$ can be generated in time depending only on $\theta$ and $k$ by following the steps of this proof.
\end{proof}

\begin{LEM}\label{lem:eachwidthk}
Let $\Gamma_1$ and $\Gamma_2$ be two $B$-namus and $\Gamma\in\Gamma_1\tplus\Gamma_2$.
Then the width of~$\Gamma_1$ is at most the width of $\Gamma$.
\end{LEM}
\begin{proof}
Let $\Gamma=(T,\alpha,\lambda,U)$,
$\Gamma_1=(T_1,\alpha_1,\lambda_1,U_1)$, 
and $\Gamma_2=(T_2,\alpha_2,\lambda_2,U_2)$.
We may assume that $\abs{V(T_1)}>1$ or $\abs{V(T_2)}>1$.
Let $(\eta_1,\eta_2)$ be a $(T_1,T_2)$-model in~$T$
co-extending $\Gamma_1$ and $\Gamma_2$ to $\Gamma$.
Then for each edge $f$ of $T_1$, $T$ has an edge $e=uv$ such that $\eta_1(e)=f$.
It is enough to show that $\lambda_1(f)\le\lambda(e)$.
Since $\alpha=\alpha_1\circ\vec\eta_1+\alpha_2\circ\vec\eta_2$,
by Lemma~\ref{lem:dim-join}, 
we obtain
\begin{align*}
  \lefteqn{\dim \alpha(u,e)\cap \alpha(v,e)}\\
  &=\dim \alpha_1(\vec\eta_1(u,e))\cap \alpha_1(\vec\eta_1(v,e)) + \dim \alpha_2(\vec\eta_2(u,e))\cap \alpha_2(\vec\eta_2(v,e)) \\
& \qquad - \dim \alpha_1(\vec\eta_1(u,e))\cap \alpha_2(\vec\eta_2(u,e)) - \dim \alpha_1(\vec\eta_1(v,e))\cap \alpha_2(\vec\eta_2(v,e))\\
&\qquad +\dim  (\alpha_1(\vec\eta_1(u,e))+\alpha_1(\vec\eta_1(v,e))) \cap (\alpha_2(\vec\eta_2(u,e))+\alpha_2(\vec\eta_2(v,e))).
\end{align*}
After rearranging terms, we have
\begin{align*}
\lambda(e)&=\lambda_1(\eta_1(e))+\lambda_2(\eta_2(e)) -\dim \alpha_1(\vec\eta_1(v,e))\cap\alpha_2(\vec\eta_2(v,e)) \\
&\qquad  -\dim \alpha_1(\vec\eta_1(u,e))\cap\alpha_2(\vec\eta_2(u,e))+\dim U_1\cap U_2 \\
&=\lambda_1(\eta_1(e))\\
&\qquad +\dim \alpha(u,e)\cap \alpha(v,e)- \dim \alpha_1(\vec\eta_1(u,e))\cap \alpha_1(\vec\eta_1(v,e)) \\
&\qquad + \lambda_2(\eta_2(e))-  \dim \alpha_2(\vec\eta_2(u,e))\cap \alpha_2(\vec\eta_2(v,e)) +\dim U_1\cap U_2 \\
&\qquad - \dim  (\alpha_1(\vec\eta_1(u,e))+\alpha_1(\vec\eta_1(v,e))) \cap (\alpha_2(\vec\eta_2(u,e))+\alpha_2(\vec\eta_2(v,e)))\\
&\geq \lambda_1(\eta_1(e))=\lambda_1(f). \qedhere   
\end{align*}
\end{proof}

\begin{LEM}\label{lem:decomp2model}
Let $\V_1$ and $\V_2$ be subspace arrangements of subspaces of $\F^r$. 
Let $(T,\L)$ be a branch-decomposition of $\V_1 \dot{\cup} \V_2$ and 
let $(T_i,\L_i)=(T,\L)|_{\V_i}$ for $i=1,2$.
If $\eta_1$ is a $T_1$-model in~$T$ and $\eta_2$ is a $T_2$-model in~$T$,
then $(\eta_1,\eta_2)$ is a $(T_1,T_2)$-model in~$T$.
\end{LEM}
\begin{proof}
Clearly, the set of all leaves $A(T)$ is a disjoint union of $\etabar_1(A(T_1))$ and $\etabar_2(A(T_2))$ since $\L$ is a bijection from $A(T)$ to $\V_{1}\dot{\cup}\V_{2}$. Hence, (i) in the definition of a $(T_1,T_2)$-model is satisfied. Note that $T$ does not have a node of degree 2, and each leaf $v$ is a leaf in exactly one of $\eta_1$ and $\eta_2$ due to (i). Therefore, (ii) is trivial. See Figure~\ref{fig:decomposition} for an illustration.
\end{proof}
\begin{figure}
  \centering
  \tikzstyle {v}=[draw,circle,fill=black,inner sep=1pt]
  \tikzstyle {v1}=[draw,fill=black,inner sep=2pt]
  \tikzstyle {v2}=[draw,circle,inner sep=2pt]
  \tikzstyle {h}=[inner sep=0pt]
  \begin{tikzpicture}
    \foreach \x in {1,2,3,5,7} {
      \node [v1] at (\x*36:1.5) (v\x) {};
    }
    \foreach \x in {4,6,8,9,0}{
      \node [v2] at (\x*36:1.5) (v\x) {};
    }
    \foreach \x in {0,3,5,8}{
      \node[v] at (\x*36+18:1) (w\x){};
      \pgfmathtruncatemacro{\y}{\x + 1}
      \draw (v\x)--(w\x)--(v\y);
    }
    \node [v] at (-18:0.5) (x9) {};
    \draw (w8)--(x9)--(w0);
    \node [v] at (180-18:0.5) (r){};
    \draw (w3)--(r);
    \node [v] at (18+36:0.5) (r2) {};
    \draw (v2)--(r2)--(x9);
    \draw (r2)--(r);
    \node [v] at (180+36+18:1) (r3){};
    \draw (w5)--(r3)--(v7);
    \draw (r)--(r3);
    \draw node [label=below:$T$] at (-90:1.5) {};
    \begin{scope}[xshift=4cm]
      \foreach \x in {1,2,3,5,7} {
        \node [v1] at (\x*36:1.5) (v\x) {};
      }
      \node [v] at (180-18:0.5) (r){};
      \draw (r)--(3*36+18:1)--(v3);
      \node [v] at (18+36:0.5) (r2) {};
      \draw (v2)--(r2)--(-18:0.5)--(18:1)--(v1);
      \draw (r2)--(r);
      \node [v] at (180+36+18:1) (r3){};
      \draw (v5)--(5*36+18:1)--(r3)--(v7);
      \draw (r)--(r3);
    \draw node [label=below:$T_1$] at (-90:1.5) {};
    \end{scope}
    \begin{scope}[xshift=8cm]
      \foreach \x in {4,6,8,9,0}{
        \node [v2] at (\x*36:1.5) (v\x) {};
      }
      \foreach \x in {8}{
        \node[v] at (\x*36+18:1) (w\x){};
        \pgfmathtruncatemacro{\y}{\x + 1}
        \draw (v\x)--(w\x)--(v\y);
      }
      \node [v] at (-18:0.5) (x9) {};
      \draw (w8)--(x9)--(18:1)--(v0);
      \node [v] at (180-18:0.5) (r){};
      \draw (v4)--(36*3+18:1)--(r);
      \draw (x9)--(18+36:0.5)--(r);
      \draw (v6)--(36*5+18:1)--(180+36+18:1)--(r);
    \draw node [label=below:$T_2$] at (-90:1.5) {};
    \end{scope}
  \end{tikzpicture}
  \caption[Constructing $T_1$ and $T_2$ in the proof of Proposition~\ref{lem:decomp2model}.]{Constructing $T_1$ and $T_2$ in the proof of Lemma~\ref{lem:decomp2model}.  \tikz \node [v1] {}; represents a leaf node mapped to an element of $\V_1$ by $\L$ and \tikz \node [v2] {}; represents a leaf node mapped to an element of $\V_2$ by $\L$.}
  \label{fig:decomposition}
\end{figure}

\begin{LEM}[{\cite[Lemma 3.17]{JKO2016}}]\label{lem:join-key}
Let $\V_1$ and $\V_2$ be subspace arrangements of subspaces of $\F^r$, and let $B$ be a subspace of $\F^r$.
If $(\spn{\V_1}+B)\cap (\spn{\V_2}+B)=B$, then \[(X\cap B)+ (Y\cap B)=(X+Y)\cap B\] 
for every subspace $X$ of $\spn{\V_1}$ and every subspace $Y$ of $\spn{\V_2}$.
\end{LEM}

\begin{LEM}\label{lem:decomp2sum}
Let $\V_1$ and $\V_2$ be subspace arrangements of subspaces of $\F^r$ and $B$ be a subspace of $\F^r$. 
Let $\Delta$ be the canonical $B$-namu of 
a branch-\decomposition{} $(T,\L)$ of $\V_1 \dot{\cup} \V_2$, and for $i=1,2$, 
let $\Delta_i$ be the canonical $B$-namu of $(T_i,\L_i)=(T,\L)|_{\V_{i}}$. 
If $(\spn{\V_1}+B)\cap (\spn{\V_2}+ B)=B$, 
then $\Delta$ is a sum of $\Delta_1$ and $\Delta_2$.
Furthermore, $\Delta = \Delta_1 +_{(\eta_1,\eta_2)} \Delta_2$ 
if $\eta_i$ is the $T_i$-model in~$T$ for $i=1,2$. 
\end{LEM}
\begin{proof}
Let $\Delta=(T,\alpha,\lambda,U)$ and $\Delta_i=(T_i,\alpha_i,\lambda_i,U_i)$ for $i=1,2$.
Let $\eta_i$ be a $T_i$-model in~$T$ for $i=1,2$.
By Lemma~\ref{lem:decomp2model}, $(\eta_1,\eta_2)$ is a $(T_1,T_2)$-model in~$T$.
Let $\Delta^+=(T,\alpha^+,\lambda^+,U^+)$ be a $B$-namu such that $(\eta_1,\eta_2)$ co-extends $\Delta_1$ and $\Delta_2$ to $\Delta^+$. That is,
\begin{align*}
\alpha^+&=\alpha_1\circ \vec\eta_1+\alpha_2\circ \vec\eta_2,\\
 U^+&=U_1+U_2,\\
\lambda^+(\emptyset)&=0, \qquad\text{ and }\\
\lambda^+(e)& =\lambda_1(\eta_1(e))+\lambda_2(\eta_2(e))+\dim (U_1\cap U_2)\\
&\quad  -\dim (\alpha_1(\vec\eta_1(v,e))\cap\alpha_2(\vec\eta_2(v,e))) -\dim (\alpha_1(\vec\eta_1(u,e))\cap\alpha_2(\vec\eta_2(u,e)))
\end{align*}
for all $e=uv\in E(T)$. %
The proof is completed by showing that $\Delta^+= \Delta$.
First, by Lemma~\ref{lem:join-key},
\[
 U=B\cap \sum_{X\in \V_1\dot\cup\V_2} X=B\cap \sum_{X\in \V_1}X+B\cap \sum_{X\in \V_2}X =U_1+U_2=U^+.
\]
Note that here we use the assumption that $(\spn{\V_1}+B)\cap (\spn{\V_2}+ B)=B$ to apply  Lemma~\ref{lem:join-key}.

Second, let us prove that $\alpha(v,e)=\alpha^+(v,e)$ for each incidence $(v,e)$ of~$T$.
Observe that
\begin{align*}
\alpha(v,e)&=B\cap \sum_{x\in A_v(T-e)}\L(x)=B\cap \Big( \sum_{\substack{x\in A_v(T-e)\\ \L(x)\in \V_1}}\L(x) +\sum_{\substack{x\in A_v(T-e)\\ \L(x)\in \V_2}} \L(x) \Big) \\
			&=B\cap \sum_{\substack{x\in A_v(T-e)\\ \L(x)\in \V_1}}\L(x)+B\cap \sum_{\substack{x\in A_v(T-e)\\ \L(x)\in \V_2}} \L(x) 
\qquad \text{by Lemma~\ref{lem:join-key}}.
\end{align*}
To complete the proof for this case, we claim that for each $i\in\{1,2\}$, 
\[\alpha_i(\eta_i(v,e))=B\cap \sum_{\substack{x\in A_v(T-e)\\ \L(x)\in \V_i}}\L(x).\]
Let $u$ be the end of $e$ other than $v$.
If $\vec\eta_i(v,e)=(0,\emptyset)$, then there is no $x\in A_v(T-e)$ such that $\L(x)\in \V_i$ and therefore $B\cap  \sum_{\substack{x\in A_v(T-e)\\ \L(x)\in \V_i}}\L(x)=\{0\}=\alpha_i(0,\emptyset)$.
If $\vec\eta_i(v,e)=(*,\emptyset)$, then 
$A_v(T-e)$ contains every leaf mapped to a member of $\V_i$ and so $B\cap  \sum_{\substack{x\in A_v(T-e)\\ \L(x)\in \V_i}}\L(x)=B\cap \sum_{X\in \V_i}X=U_i=\alpha(*,\emptyset)$.
If $\eta_i(e)\neq \emptyset$, then say  $\vec\eta_i(v,e)=(v',e')$. We can rewrite 
\[
\alpha_i(\vec\eta_i(v,e))=B\cap \sum_{x\in A_{v'}(T_i-e')}\L_i(x)=B\cap \sum_{x\in A_{v'}(T_i-e')}\L(x)=B\cap \sum_{\substack{x\in A_v(T-e)\\ \L(x)\in \V_i}}\L(x).
\]
This proves the claim and so we conclude that 
\[\alpha=\alpha_1\circ \vec\eta_1+\alpha_2\circ\vec\eta_2 =\alpha^+.\]

Finally, let us prove that $\lambda(e)=\lambda^+(e)$ for each edge $e$ of~$T$.
Suppose $\eta_1(e)\neq\emptyset$ and $\eta_2(e)\neq\emptyset$.
Let $u$, $v$ be the ends of $e$ and let $(v_i,e_i)=\vec\eta_i(v,e)$, $(u_i,e_i)=\vec\eta_i(u,e)$ for $i=1,2$
where $e_i$ is an edge of $T_i$ and $v_i,u_i$ are two ends of $e_i$. 
For $i=1,2$, let 
\[
  L_i=\sum_{\substack{x\in A_v(T-e)\\ \L(x)\in \V_i}} \L_i(x),~R_i=\sum_{\substack{x\in A_u(T-e)\\ \L(x)\in \V_i}} \L_i(x),\]
\[
L_i'=\sum_{x\in A_{v_i}(T_i-e_i)}\L_i(x),~ R_i'=\sum_{x\in A_{u_i}(T_i-e_i)}\L_i(x).
\]
Note that $L_i=L_i'$, $R_i=R_i'$, and $L_i,R_i\subseteq\spn{\V_i}$ for $i=1,2$.
Also, we have $L_1\cap L_2\subseteq(L_1+R_1)\cap(L_2+R_2)\subseteq(\spn{\V_1}+B)\cap(\spn{\V_2}+B)=B$.
Then 
by Lemma~\ref{lem:dim-join}
\begin{align*}
\lambda(e)&=\dim ((L_1+L_2)\cap  (R_1+R_2)) \\ 
          &=\dim (L_1\cap R_1) + \dim (L_2\cap R_2) - \dim (L_1 \cap L_2) - \dim (R_1\cap R_2)
  \\
  &\qquad+ \dim ((L_1+R_1)\cap (L_2+R_2)) \\
          &=\lambda_1(e_1)+\lambda_2(e_2)- \dim (L_1' \cap L_2' \cap B)- \dim (R_1'\cap R_2'\cap B)
  \\
  &\qquad
          + \dim ((L_1'+R_1')\cap (L_2'+R_2')\cap B) \\
          &=\lambda_1(e_1)+\lambda_2(e_2) - \dim (\alpha_1(v_1,e_1) \cap \alpha_2(v_2,e_2)) - \dim (\alpha_1(u_1,e_1) \cap \alpha_2(u_2,e_2))  \\
  &\qquad+ \dim (U_1\cap U_2)\\
&=\lambda^+(e).
\end{align*}

Suppose exactly one of $\eta_1(e)$ or $\eta_2(e)$ is $\emptyset$.
Without loss of generality, we assume that $\eta_1(e)=\emptyset$.
Let $\eta_1(v,e)=(*,\emptyset)$, $\eta_1(u,e)=(0,\emptyset)$, $\eta_2(v,e)=(v_2,e_2)$, and $\eta_2(u,e)=(u_2,e_2)$.
We define $L_1''=\sum_{x\in A(T)}\L(x)$, $R_1''=\emptyset$, and 
let $L_1$, $R_1$, $L_2$, $L_2'$, $R_2$, $R_2'$ as above
so that $L_1=L_1''$, $R_1=R_1''$, $L_2=L_2'$, and $R_2=R_2'$.
Note that $\lambda_1(\eta_1(e))=\lambda_1(\emptyset)=0$
$\alpha_1(\vec\eta_1(u,e))=\{0\}$
and $L_1'',R_1''\subseteq\spn{\V_1}$.
Then by Lemma~\ref{lem:dim-join}
\begin{align*}
\lambda(e)&=\dim ((L_1+L_2)\cap (R_1+R_2))\\
          &=\dim (L_1\cap R_1) + \dim (L_2\cap R_2) - \dim (L_1 \cap L_2) - \dim (R_1\cap R_2) \\
  &\qquad+ \dim ((L_1+R_1)\cap (L_2+R_2)) \\ 
          &=\dim (L_1''\cap R_1'') + \dim (L_2'\cap R_2') - \dim (L_1'' \cap L_2' )- \dim (R_1''\cap R_2')
  \\
  &\qquad+ \dim ((L_1''+R_1'')\cap (L_2'+R_2') )\\ 
&=\dim (L_2'\cap R_2') - \dim (L_1'' \cap L_2' \cap B)  + \dim ((L_1''+R_1'')\cap (L_2'+R_2') \cap B) \\ 
&=\lambda_2(e_2) - \dim (\alpha_1(v_1,e_1)\cap\alpha_2(v_2,e_2)) + \dim (U_1\cap U_2) \\
&= \lambda_1(\eta_1(e))+\lambda_2(\eta_2(e)) - \dim \alpha_1(\vec\eta_1(v,e)) \cap \alpha_2(\vec\eta_2(v,e)) \\
&\quad- \dim (\alpha_1(\vec\eta_1(u,e)) \cap \alpha_2(\vec\eta_2(u,e))) +\dim (U_1\cap U_2) \\
&=\lambda^+(e).
\end{align*}

Suppose $\eta_1(e)=\eta_2(e)=\emptyset$.
If $\eta_i(e)=\emptyset$ for some $i=1,2$, then 
all leaves $x$ of~$T$ with $\L(x)\in\V_i$ are in one component of $T-e$.
As $A(T)$ is the disjoint union of $\etabar_1(A(T_1))$ and $\etabar_2(A(T_2))$
for two ends $u$, $v$ of $e$, we may assume that 
all elements in $\V_1$ are mapped from the leaves in the component of $T-e$ containing~$u$ 
and all elements in $\V_2$ are mapped from the leaves in the component of $T-e$ containing~$v$.
Then $\lambda_1(\eta_1(e))=\lambda_2(\eta_2(e))=0$, and 
$\alpha(\vec\eta_1(v,e))=\alpha(\vec\eta_2(u,e))=\{0\}$.
Thus, $\lambda^+=\lambda_1(\eta_1(e))+\lambda_2(\eta_2(e)) 
- \dim (\alpha_1(\vec\eta_1(v,e)) \cap \alpha_2(\vec\eta_2(v,e))) - \dim (\alpha_1(\vec\eta_1(u,e)) \cap \alpha_2(\vec\eta_2(u,e))) +\dim (U_1\cap U_2) = \dim (U_1\cap U_2) = \lambda(e)$.
This completes the proof.
\end{proof}

\subsection{Comparing two $\boldsymbol B$-namus}
The following lemmas are trivial by definitions.

\begin{LEM}\label{lem:transitive}
  The binary relations $\le$ and $\tle$ on $B$-namus are transitive.
\end{LEM}

\begin{LEM}
If $\Gamma_1\tle\Gamma_2$,
then the width of $\Gamma_1$ is at most the width of $\Gamma_2$.
\end{LEM}

The following lemma is analogous to Lemma~3.9 of~\cite{BK1996}.
\begin{LEM}\label{lem:tautrim}
If $\Gamma$ is a $B$-namu, 
then $\tau(\Gamma)\tle\trim(\Gamma)$ and $\trim(\Gamma)\tle\tau(\Gamma)$. 
\end{LEM}

\begin{LEM}\label{lem:comparetrim}
If $\Gamma_1\tle\Gamma_2$, then $\trim(\Gamma_1)\tle\trim(\Gamma_2)$.
\end{LEM}
\begin{proof}
It is enough to prove that 
if $\Gamma_1^*=(T_1^*,\alpha_1^*,\lambda_1^*,U)$ and $\Gamma_2^*=(T_2^*,\alpha_2^*,\lambda_2^*,U)$ 
are subdivisions of $\Gamma_1$ and $\Gamma_2$, respectively, 
such that $\Gamma_1^*\le\Gamma_2^*$, then $\trim(\Gamma_1)\tle \trim(\Gamma_2)$.
It is easy to see that $\trim(\Gamma_1^*)\le\trim(\Gamma_2^*)$ because $T_1^*=T_2^*$ and $\alpha_1^*=\alpha_2^*$.

We will show that $\trim(\Gamma_1^*)$ is a subdivision of $\trim(\Gamma_1)$.
Let $\eta^*=T_1^*$ be a $T(\Gamma_1)$-model in $T_1^*$.
Let $\trim(\Gamma_1^*)=(T_1',\alpha_1',\lambda_1',U)$ and $\trim(\Gamma_1)=(T_1,\alpha_1,\lambda_1,U)$.
Let $\eta$~be a function 
defined on $E(T_1')\cup\{\emptyset\}$ %
by $\eta(e)=\eta^*(e)$.
It is enough to show that $\eta$ is a $T_1$-model in $T_1'$.
In other words, if $e$ is an edge of $T_1'$, then $\eta^*(e)$ is an edge of~$T_1$.
Suppose not. 
If $\eta^*(e)$ is not an edge of $T_1$, then it is removed by trimming~$\Gamma_1$.
Since $\Gamma_1^*$ is a subdivision of $\Gamma_1$, the edge $e$, as an edge of $T_1^*$, is also 
removed by trimming~$\Gamma_1^*$.
This contradicts the fact that 
$e$ is an edge of $T_1'$, which is the tree in $\trim(\Gamma_1^*)$.
Therefore, $\trim(\Gamma_1^*)$ is a subdivision of $\trim(\Gamma_1)$ 
and similarly $\trim(\Gamma_2^*)$ is a subdivision of $\trim(\Gamma_2)$.
So $\trim(\Gamma_1)\tle\trim(\Gamma_2)$ because $\trim(\Gamma_1^*)\le\trim(\Gamma_2^*)$.
\end{proof}

  Let us discuss the time complexity 
  when we check whether $\Delta\tle\Gamma$ for given two $B$-namus $\Delta$ and $\Gamma$.
  It can be checked 
  in time $f(\abs{V(T(\Delta))}, \abs{V(T(\Gamma))}, \dim(B), \abs{\F})$ 
  for some function $f$.
  The following lemma gives the upper bound of the size of a subdivision of $\Gamma$.

\begin{LEM}\label{lem:comptime}
  Let $\Gamma_1$ and $\Gamma_2$ be two $B$-namus.
  If $\Gamma_1\tle\Gamma_2$, 
  then there exist a $B$-namu $\Gamma_1^*$ isomorphic to a subdivision of $\Gamma_1$
  and a $B$-namu $\Gamma_2^*$ isomorphic to a subdivision of $\Gamma_2$
  satisfying that $\Gamma_1^*\le\Gamma_2^*$ and 
  \[
  \abs{V(T(\Gamma_1^*))}\le \abs{V(T(\Gamma_1))} + \abs{V(T(\Gamma_2))} + 2. 
  \]
\end{LEM}

\begin{proof}
    We may assume that both $T(\Gamma_1)$ and $T(\Gamma_2)$ have at least one edge.
    We may assume that there exist
	a subdivision $\Gamma_1^*=(T,\alpha_1,\lambda_1,U_1)$ of $\Gamma_1$ and  
	a subdivision $\Gamma_2^*=(T,\alpha_2,\lambda_2,U_2)$ of $\Gamma_2$ 
	such that $\Gamma_1^*\le\Gamma_2^*$.
	Assume that $\abs{E(T)}$ is chosen to be minimum.
	Let $T_1=T(\Gamma_1)$ and $T_2=T(\Gamma_2)$.  
	Let $\eta_1$ be a $T_1$-model in~$T$ extending $\Gamma_1$ to $\Gamma_1^*$
	and let $\eta_2$ be a $T_2$-model in~$T$ extending $\Gamma_2$ to $\Gamma_2^*$.
	
	We first claim that there exists no edge $e$ of~$T$ 
	such that $\eta_1(e)=\eta_2(e)=\emptyset$ and $e$~is incident with some leaf $\ell$ of~$T$. 
	Suppose not.
	Let $\Gamma_1'$ be a $B$-namu $(T-\ell,\alpha_1',\lambda_1',U_1)$ 
	where $\alpha_1'(v,f)=\alpha_1(v,f)$ for all incidences $(v,f)$ of $T-\ell$ 
	and $\lambda_1'(f)=\lambda_1(f)$ for all edges $f$ of $T-\ell$.
	Similarly, we define $\Gamma_2'$.
	Then $\Gamma_1'$ is a subdivision of $\Gamma_1$ and 
	$\Gamma_2'$ is a subdivisions of $\Gamma_2$
	satisfying that $\Gamma_1'\le\Gamma_2'$.
	It is a contradiction because $\abs{E(T)}$ is chosen to be minimum.
	Thus, we may assume that there is no such edge in~$T$.
	
	We also claim that there exist no two edges $e_1$, $e_2$ of~$T$ such that 
	$\eta_1(e_1)=\eta_1(e_2)$, $\eta_2(e_1)=\eta_2(e_2)$, and 
	$e_1,e_2$ share a degree-$2$ node.
	Suppose not.
	Let $\Gamma_1''$ be a $B$-namu $(T/e_2,\alpha_1'',\lambda_1'',U_1)$
	where $\alpha_1''(v,f)=\alpha_1(v,f)$ for all incidences $(v,f)$ of $T/e_2$ 
	and $\lambda_1''(f)=\lambda_1(f)$ for all edges $f$ of $T/e_2$.
	Similarly, we define $\Gamma_2''$.
	Then $\Gamma_1''$ is a subdivision of $\Gamma_1$ and 
	$\Gamma_2''$ is a subdivision of $\Gamma_2$ 
	satisfying that $\Gamma_1''\le\Gamma_2''$.
	It is a contradiction because $\abs{E(T)}$ is chosen to be minimum.
	Thus, we may assume that there are no such two edges in~$T$.
	
	By the first claim, 
	we deduce that $\abs{A(T)}\le\abs{A(T_1)}+\abs{A(T_2)}$.
	Note that for a subcubic tree $T$, the number of degree-$3$ nodes in~$T$
	is $\abs{A(T)}-2$.
	By the second claim, 
	every degree-$2$ node of~$T$ 
	is a branch node in $\eta_1$ or a branch node in $\eta_2$.
	Since the number of degree-$2$ nodes in $T_i$ is 
	$\abs{V(T_i)}-\abs{A(T_i)}-(\abs{A(T_i)}-2)$ for $i=1,2$,
	the number of degree-$2$ nodes in~$T$ is at most 
	\[
	\abs{V(T_1)}+\abs{V(T_2)}-2\abs{A(T_1)}-2\abs{A(T_2)}+4.
	\]
	Thus, the number of nodes in~$T$ is at most 
	\begin{multline*}
	(\abs{A(T_1)}+\abs{A(T_2)})      +(\abs{V(T_1)}+\abs{V(T_2)}-2\abs{A(T_1)}-2\abs{A(T_2)}+4) \\
      + (\abs{A(T_1)}+\abs{A(T_2)}-2) 
	= \abs{V(T_1)}+\abs{V(T_2)}+2.   \qedhere
	\end{multline*}
\end{proof}

For two $B$-namus $\Gamma_1$ and $\Gamma_2$ with $\Gamma_1\tle\Gamma_2$, 
we say that the tree $T(\Gamma_1^*)$ in Lemma~\ref{lem:comptime} \emph{ensures} $\Gamma_1\tle\Gamma_2$.
  Let $\F$ be a fixed finite field and $B$ be a subspace of $\F^r$.
  For given two $B$-namus $\Delta$ and $\Gamma$,
  in order to check whether $\Delta\tle\Gamma$, we need to find 
  a $B$-namu $\Delta^*$ isomorphic to a subdivision of $\Delta$ and 
  a $B$-namu $\Gamma^*$ isomorphic to a subdivision of $\Gamma$
  such that $\Delta^*\le \Gamma^*$.
  By Lemma~\ref{lem:comptime},
  it is enough to consider finitely many subdivisions of $\Delta$ and $\Gamma$.
  Thus, the number of comparison operations needed to check whether $\Delta\tle\Gamma$ is bounded 
  by some function of $\abs{V(T(\Delta))}$, $\abs{V(T(\Gamma))}$, $\dim(B)$, and $\abs{\F}$.

\begin{LEM}\label{lem:tletime}
Let $\F$ be a finite field, let $r$ be a positive integer, and let 
$B$ be a subspace of $\F^r$. %
For two $B$-namus $\Delta$ and $\Gamma$,
we can decide whether $\Delta\tle\Gamma$ 
by executing at most  $f(\abs{V(T(\Delta))}, \abs{V(T(\Gamma))}, \dim(B), \abs{\F})$
comparison operations (on integers and on subspaces of $B$) for some function $f$.
\end{LEM}

\begin{PROP}\label{prop:subdivcomposition}
Let $\Gamma_1,\Gamma_2$,$\Gamma'_1$ and $\Gamma'_2$ be $B$-namus. If $\Gamma'_1\tle \Gamma_1$ and $\Gamma'_2\tle \Gamma_2$, then for every $\Gamma\in \Gamma_1 \tplus \Gamma_2$, there exists $\Gamma'\in \Gamma'_1\tplus \Gamma'_2$ such that $\Gamma'\tle \Gamma$.
\end{PROP}

To prove this proposition, we use the following lemmas.
\begin{LEM}\label{lem:latticejoin}
  Let $\Gamma_1,\Gamma_1',\Gamma_2$ be $B$-namus such that $\Gamma_1'\le \Gamma_1$. 
  If $\Gamma_1+_{(\eta_1,\eta_2)} \Gamma_2$ is the sum of $\Gamma_1$ and $\Gamma_2$ by $(\eta_1,\eta_2)$,
  then $\Gamma_1'+_{(\eta_1,\eta_2)} \Gamma_2$ is well defined and 
  $\Gamma'_1+_{(\eta_1,\eta_2)} \Gamma_2 \le \Gamma_1+_{(\eta_1,\eta_2)}  \Gamma_2$.
\end{LEM}
\begin{proof}
Let $\Gamma_1=(T_1,\alpha_1,\lambda_1, U_1)$ and $\Gamma_1'=(T_1',\alpha_1',\lambda'_1,U_1)$. 
From $\Gamma_1\le \Gamma'_1$, we may assume that $T_1=T'_1$ and $\alpha_1=\alpha'_1$. 
Let $T_2=T(\Gamma_2)$. 
Then $(\eta_1,\eta_2)$ is a $(T_1',T_2)$-model in~$T$ as well.
Thus, $\Gamma_1'+_{(\eta_1,\eta_2)} \Gamma_2$ is well defined.

By Lemma~\ref{lem:sum} and the fact that $\lambda_1'(e)\le\lambda_1(e)$ for all edges $e$ of $T_1=T_1'$,
it is easy to check that 
  $\Gamma'_1+_{(\eta_1,\eta_2)} \Gamma_2 \le \Gamma_1+_{(\eta_1,\eta_2)}  \Gamma_2$.
\end{proof}

\begin{LEM}\label{lem:naturalext}
 Let $\Gamma_1,\Gamma_2$ be $B$-namus and $\Gamma'_1$ be a subdivision of $\Gamma_1$. 
 Then for every sum $\Gamma$ of $\Gamma_1$ and $\Gamma_2$,
 there is a sum $\Gamma'$ of $\Gamma'_1$ and $\Gamma_2$  
such that $\Gamma'$ is a subdivision of $\Gamma$.
\end{LEM}
\begin{proof}
Let $\Gamma_1=(T_1,\alpha_1,\lambda_1,U_1)$, $\Gamma_2=(T_2,\alpha_2,\lambda_2,U_2)$, 
$\Gamma_1'=(T_1',\alpha_1',\lambda_1',U_1')$. See Figure~\ref{fig:joinextension} for a diagram.  
We use induction on $\abs{V(T_1')}-\abs{V(T_1)}$.
If $\abs{V(T_1')}=\abs{V(T_1)}$, then it is trivial.
We may assume that $\Gamma_1'$ is obtained from $\Gamma_1$ by subdividing one edge $v_1v_3$ of $T_1$.
Let $v_2$ be the new node of $T_1'$ created by subdividing $v_1v_3$.
Let $\Gamma=\Gamma_1+_{(\eta_1,\eta_2)}\Gamma_2=(T,\alpha,\lambda,U)$.
Let $u_1u_3$ be an edge of~$T$ such that $\eta_1(u_1u_3)=v_1v_3$.
Let $T'$ be a subdivision of~$T$ obtained by subdividing $u_1u_3$ once, and denote by $u_2$ the new node.

Let us define functions $\rho_1:E(T')\cup\{\emptyset\}\rightarrow E(T_1')\cup\{\emptyset\}$ and $\rho_2:E(T')\cup\{\emptyset\}\rightarrow E(T_2')\cup\{\emptyset\}$ as follows.
\[
\rho_1(e)=
\begin{cases}
v_1v_2 &\text{if $e=u_1u_2$,} \\
v_2v_3 &\text{if $e=u_2u_3$,} \\
\eta_1(e) &\text{otherwise,}
\end{cases}
\quad\text{and}\quad
\rho_2(e)=
\begin{cases}
\eta_2(u_1u_3)&\text{if $e=u_1u_2$ or $e=u_2u_3$,} \\
\eta_2(e) &\text{otherwise.}
\end{cases}
\]
We claim that $(\rho_1,\rho_2)$ is a $(T_1',T_2)$-model in $T'$.
Since $(\eta_1,\eta_2)$ is a $(T_1,T_2)$-model in~$T$,
it is enough to check the third condition for $u_2$ whose degree is $2$ in $T'$.
It is also clear because $u_2$ is a branch node in $\eta_1'$ but not in $\eta_2'$.

The proof is completed by proving that 
the sum $\Gamma'=(T',\alpha',\lambda',U')$ of $\Gamma_1'$ and $\Gamma_2$ by $(\rho_1,\rho_2)$ is a subdivision of $\Gamma$.
It is clear that $\alpha'(v,e)=\alpha(v,e)$ and $\lambda'(e)=\lambda(e)$ if $e$ is neither $u_1u_2$ nor $u_2u_3$.
We claim that $\alpha'(u_1,u_1u_2)=\alpha'(u_2,u_2u_3)=\alpha(u_1,u_1u_3)$
and $\lambda'(u_1u_2)=\lambda'(u_2u_3)=\lambda(u_1u_3)$.
We observe that 
\begin{align*}
  \alpha'(u_1,u_1u_2)&=\alpha_1'(\vec\rho_1(u_1,u_1u_2))+\alpha_2(\vec\rho_2(u_1,u_1u_2))\\
  &
=\alpha_1'(v_1,v_1v_2)+\alpha_2(\vec\eta_2(u_1,u_1u_3))\\
&=\alpha_1(v_1,v_1v_3)+\alpha_2(\vec\eta_2(u_1,u_1u_3)), \\
  \alpha'(u_2,u_2u_3)&=\alpha_1'(\vec\rho_1(u_2,u_2u_3))+\alpha_2(\vec\rho_2(u_2,u_2u_3))\\
  &=\alpha_1'(v_2,v_2v_3)+\alpha_2(\vec\eta_2(u_1,u_1u_3))\\
&=\alpha_1(v_1,v_1v_3)+\alpha_2(\vec\eta_2(u_1,u_1u_3)), \quad\text{ and } \\
\alpha_1(v_1,v_1v_3)+ & \alpha_2(\vec\eta_2(u_1,u_1u_3)) 
=\alpha_1(\vec\eta_1(u_1,u_1u_3))+\alpha_2(\vec\eta_2(u_1,u_1u_3))
=\alpha(u_1,u_1u_3).
\end{align*}
Also, since
$\lambda_1'(\rho_1(u_1u_2))=\lambda_1'(v_1v_2)=\lambda_1(v_1v_3)=\lambda_1(\eta_1(u_1u_3))$,
we have
\begin{align*}
\lambda'(u_1u_2)
&=\lambda_1'(\rho_1(u_1u_2))+\lambda_2(\rho_2(u_1u_2)) 
-\dim (\alpha_1(\vec\rho_1(u_1,u_1u_2))\cap \alpha_2(\vec\rho_2(u_1,u_1u_2)))  \\
&\quad-\dim (\alpha_1(\vec\rho_1(u_2,u_1u_2))\cap \alpha_2(\vec\rho_2(u_2,u_1u_2)))  
+\dim (U_1'\cap U_2)\\
&=\lambda_1(\eta_1(u_1u_3))+\lambda_2(\eta_2(u_1u_3))
-\dim (\alpha_1(\vec\eta_1(u_1,u_1u_3))\cap \alpha_2(\vec\eta_2(u_1,u_1u_3)))  \\
&\quad-\dim (\alpha_1(\vec\eta_1(u_3,u_1u_3))\cap \alpha_2(\vec\eta_2(u_3,u_1u_3))) 
+\dim (U_1\cap U_2)\\
&=\lambda(u_1u_3). 
\end{align*}
Similarly, we have $\lambda'(u_2u_3)=\lambda(u_1u_3)$.
\end{proof}
\begin{figure}
  \centering
  \begin{tikzpicture}
    \node (g) at (1,3) {$\Gamma$};
    \node (g1) at (-0.5,4) {$\Gamma_1$};
    \node (g2) at (3,4) {$\Gamma_2$};
    \node (g1') at (-0.5,2) {$\Gamma_1'$};
    \node (g+) at (1,0.5) {$\Gamma'$};
    \draw [dashed] (g)  -- (g+) ;
    \draw (g1) -- (g1') ;
    \draw [dashed,<-] (g1') -- (g+)
    node [midway,label=left:$\rho_1$]{} ;
    \draw [dashed,<-] (g2) -- (g+)
    node [midway,label=right:$\rho_2$]{} ;
    \draw [<-] (g1)--(g)
    node[midway,label=above:$\eta_1$]{};
    \draw [<-] (g2)--(g)
    node[midway,label=above:$\eta_2$]{};
  \end{tikzpicture}
  \caption{A diagram for Lemma~\ref{lem:naturalext}. Solid lines represent given functions (models) and dashed lines represent functions to be constructed.}
  \label{fig:joinextension}
\end{figure}

\begin{LEM}\label{lem:joinext}
 Let $\Gamma_1,\Gamma_2$ be $B$-namus and $\Gamma'_1$ be a subdivision of $\Gamma_1$. 
 Then for every sum $\Gamma'$ of $\Gamma'_1$ and $\Gamma_2$, 
 there is a sum $\Gamma$ of $\Gamma_1$ and $\Gamma_2$ 
 such that $\Gamma'$ is a subdivision of $\Gamma$.
\end{LEM}
\begin{proof}
This lemma can be proved similarly to the proof of Lemma~\ref{lem:naturalext}.
\end{proof}

\begin{proof}[Proof of Proposition~\ref{prop:subdivcomposition}]
  Since   $\tle$ is transitive by Lemma~\ref{lem:transitive} and $\Gamma_1\tplus \Gamma_2=\Gamma_2\tplus \Gamma_1$, it is enough to prove this proposition for the case that $\Gamma_2'=\Gamma_2$.  Let $\Gamma_1''$ be a $B$-namu isomorphic to a subdivision of $\Gamma_1'$ and $\Gamma_1^*$ be a $B$-namu isomorphic to a subdivision of $\Gamma_1$ such that $\Gamma_1''\le \Gamma_1^*$.  
We may assume that $\Gamma_1''$ is a subdivision of $\Gamma_1'$
and $\Gamma_1^*$ is a subdivision of $\Gamma_1$.
  From Lemma~\ref{lem:naturalext}, there is a sum of $\Gamma_1^*$ and $\Gamma_2$, say, 
 $\Gamma^*_1+_{(\eta_1,\eta_2)}\Gamma_2$, which is a subdivision of $\Gamma$. 
 Note that $\Gamma_1^*+_{(\eta_1,\eta_2)} \Gamma_2\tle \Gamma$. By Lemma~\ref{lem:latticejoin}, $\Gamma_1''+_{(\eta_1,\eta_2)} \Gamma_2\le \Gamma_1^*+_{(\eta_1,\eta_2)} \Gamma_2$.  By Lemma~\ref{lem:joinext}, there is $\Gamma'\in \Gamma_1'\tplus\Gamma_2$ such that $\Gamma_1''+_{(\eta_1,\eta_2)} \Gamma_2$ is a subdivision of $\Gamma'$. Note that $\Gamma' \tle \Gamma_1''+_{(\eta_1,\eta_2)} \Gamma_2$. Now,
 \[\Gamma' \tle \Gamma_1''+_{(\eta_1,\eta_2)} \Gamma_2 \le \Gamma_1^*+_{(\eta_1,\eta_2)} \Gamma_2 \tle \Gamma. \qedhere \]
\end{proof}

\begin{LEM}\label{lem:compareproj}
Let $B$ be a subspace of $\F^r$ and $B'$ be a subspace of $B$.
Let $\Gamma_1$ and $\Gamma_2$ be $B$-namus.
If $\Gamma_1\tle\Gamma_2$, then $\Gamma_1|_{B'}\tle\Gamma_2|_{B'}$.
\end{LEM}
\begin{proof}
Let $\Gamma'_1=(T'_1,\alpha'_1,\lambda'_1,U'_1)$ and $\Gamma'_2=(T'_2,\alpha'_2,\lambda'_2,U'_2)$ 
be $B$-namus isomorphic to subdivisions of $\Gamma_1=(T_1,\alpha_1,\lambda_1,U_1)$ 
and $\Gamma_2=(T_2,\alpha_2,\lambda_2,U_2)$, respectively, such that $\Gamma'_1\le \Gamma'_2$. 
From $\Gamma'_1\le \Gamma'_2$, for every incidence $(v,e)$ of $T'_1=T'_2$ 
we have $\alpha'_1(v,e)\cap B' =\alpha'_2(v,e)\cap B'$ and $\lambda'_1(e)\le \lambda'_2(e)$. 
Clearly, we have $U'_1=U_1\cap B'=U_2\cap B' =U'_2$. 
Therefore, %
$\Gamma'_1|_{B'}\le \Gamma'_2|_{B'}$. 

It is easy to see that $\Gamma_1'|_{B'}$ and $\Gamma_2'|_{B'}$ are $B'$-namus isomorphic to 
subdivisions of $\Gamma_1|_{B'}$ and $\Gamma_2|_{B'}$. %
This completes the proof.
\end{proof}
By definition, we can deduce the following lemma. %

\begin{LEM}\label{lem:tleksafe}
Let $\Gamma_1$ and $\Gamma_2$ be $B$-namus such that $\Gamma_1\tle \Gamma_2$. 
If $\Gamma_2$ is a $k$-safe extension of $\trim(\Gamma_2)$, 
then $\Gamma_1$ is a $k$-safe extension of $\trim(\Gamma_1)$.
\end{LEM}
\begin{proof}
Let $\Gamma_1=(T_1,\alpha_1,\lambda_1,U)$ and $T_1'=T(\trim(\Gamma_1))$.
Let $\eta_1=T_1'$ be a $T_1'$-model in $T_1$ and 
let $(v_1,e_1)$ be an incidence of $T_1$ with $\vec\eta_1(v_1,e_1)=(*,\emptyset)$.
We will show that 
\[
\lambda_1(e_1) + \dim U - \dim\alpha_1(v_1,e_1) \le k, 
\]
which proves that $\Gamma_1$ is a $k$-safe extension of $\trim(\Gamma_1)$.

Since $\Gamma_1\tle\Gamma_2$, there exist $B$-namus $\Gamma_1'$ and $\Gamma_2'$ that are subdivisions of 
$\Gamma_1$ and $\Gamma_2=(T_2,\alpha_2,\lambda_2,U)$, respectively, such that $\Gamma_1'\le\Gamma_2'$. 
We may assume that $\Gamma_1'$ is a subdivision of $\Gamma_1$ and $\Gamma_2'$ is a subdivision of $\Gamma_2$.
Let $T=T(\Gamma_1')=T(\Gamma_2')$.
Let $\eta_2=T(\trim(\Gamma_2))$ be a $T(\trim(\Gamma_2))$-model in $T_2$.
Let $e$ be an edge of~$T$ that corresponds to $e_1$ of $T_1$
and let $e_2$ be the edge of $T_2$ corresponding to $e$.
By definition, since $e_1$ is removed by trimming $\Gamma_1$, 
we deduce that $e_2$ is removed by trimming $\Gamma_2$.
Thus, 
$\lambda_2(e_2)+\dim U - \dim \alpha_2(v_2,e_2)\le k$ 
for an end $v_2$ of $e_2$ with $\vec\eta_2(v_2,e_2)=(*,\emptyset)$.
Since $\lambda_1(e_1)\le\lambda_2(e_2)$ and $\alpha_1(v_1,e_1)=\alpha_2(v_2,e_2)$,
we have $\lambda_1(e_1)+\dim U - \dim\alpha_1(v_1,e_1) \le \nolinebreak k$.
\end{proof}

\subsection{Witnesses}\label{subsec:witness}

\paragraph{Labeling}
In order to describe a subcubic tree with no degree-$2$ node containing a subdivision of a subcubic tree $T$,
we define a \emph{labeling} of~$T$ 
as a mapping $\zeta$ defined on the union of the set of nodes of degree at most $2$ and the set of the incidences of~$T$ such that 
\begin{itemize}
\item for each node~$v$ of degree at most $2$ in~$T$, $\zeta(v)$ is 
  a rooted binary tree with the root $v$
  satisfying that 
  \[ \deg_{\zeta(v)}(v)=0 \text{ or }
    \deg_{\zeta(v)}(v)=3-\max(\deg_T(v),1),\]
\item for each incidence $(v,e)$ of~$T$, 
		$\zeta(v,e)$ is a sequence $(T_1,T_2,\ldots,T_{\ell_{v,e}})$ of
		rooted binary trees whose roots have degree $1$ where $\ell_{v,e}\ge0$.
\end{itemize}
We say that an edge or a node is \emph{in $\zeta$} if it is in $\zeta(u)$ 
or $\zeta(v,e)$ for some node $u$ or some incidence $(v,e)$ of~$T$. 

Before moving on, let us explain the intuition behind a labeling.
Suppose that 
a subcubic tree $G$ contains a subdivision $T'$ of a subcubic tree $T$ such that 
$E(T)\neq\emptyset$ and 
every node in $T'$ incident with an edge in $E(G)-E(T')$ has degree $3$ in $G$. 
We define $\zeta(v)$ to be the component of $G-E(T')$
seen as a rooted binary tree having $v$ as the root.
For each edge $uv$ of~$T$, $G$ has a path $P_{uv}$ from $u$ to $v$, as $T'$ is a subdivision of~$T$. We arbitrarily split the path $P_{uv}$ into two subpaths $P_{u,uv}$ and $P_{v,uv}$ containing~$u$ and $v$, respectively, so that they share exactly one edge.
Then we define $\zeta(u,uv)$ as the sequence of components of $G-E(T')$ having internal nodes in $P_{u,uv}$ as roots in which we ignore components having only one node.
Thus, roughly speaking, a labeling of~$T$ expresses how the  ``hidden'' rooted binary trees are attached to $T$.

We say that $v$ is a \emph{leaf} of a labeling $\zeta$ of~$T$
if it is
\begin{itemize}
\item 
 a  nonroot leaf of $\zeta(v)$ for some node $v\in V(T)$ of degree at most $2$, 
\item 
an isolated root of $\zeta(v)$ for some leaf $v$ of~$T$, or 
\item a nonroot leaf of a rooted binary tree in $\zeta(v,e)$ for some incidence $(v,e)$ of~$T$.
\end{itemize}

\paragraph{$\V$-labeling}
For a subspace arrangement $\V$, 
a \emph{$\V$-labeling} of a subcubic tree $T$ 
is a pair $(\zeta,\L)$ of a labeling $\zeta$ of~$T$
and a bijection from the set of leaves of $\zeta$ to $\V$.
Note that $\zeta(v,e)$ can be the empty sequence. %

For a $\V$-labeling $(\zeta,\L)$ of a subcubic tree $T$, 
its \emph{canonical branch-decomposition} $(\T(\zeta),\L)$ 
consists of a tree $\T(\zeta)$ 
and a bijection $\L$
where $\T(\zeta)$ is obtained from $T$ by 
\begin{enumerate}[(i)] 
\item subdividing each edge $uv$ of~$T$  $\abs{\zeta(u,uv)}+\abs{\zeta(v,uv)}$ times,
\item attaching the $i$th rooted binary tree in $\zeta(u,uv)$
  whose root is identified with the $i$th new subdividing internal node on the path from $u$ to $v$ 
  for each incidence $(u,uv)$ of~$T$ and $1\le i\le \abs{\zeta(u,uv)}$, 
\item attaching $\zeta(v)$ for all $v\in V(T)$ of degree at most $2$ 
  whose root is identified with~$v$, and 
\item smoothing $v$ for each $v\in V(T)$ with 
  $\deg_T(v)=2$ and $\abs{V(\zeta(v))}=1$.
\end{enumerate}
We remark that the set of leaves of $\T(\zeta)$ is equal to the set of
leaves of $\zeta$
and $(\T(\zeta),\L)$ is a branch-decomposition if $T$ has at least one edge and a rooted branch-decomposition if $T$ has no edges.
By (iv), if an edge $e=uv$ of $\T(\zeta)$ is not in $\zeta$, 
then $e$ corresponds to a path $P_e$ in~$T$ from $u$ to $v$ such that $\zeta(w)$ has only
one node for every internal node $w$ of $P_e$.
Furthermore for each $v\in V(T)$, $\zeta(v)$ is a subtree of $\T(\zeta)$, and for each incidence $(v,e)$ of~$T$, $\zeta(v,e)$ is a subtree of $\T(\zeta)$.

Let $(T^b,\L^b)$ be a rooted branch-decomposition of~$\V$ and $\V_0\subseteq\V$.
We say that a $\V_0$-labeling $(\zeta,\L)$ of a subcubic tree $T$ is \emph{totally pure with respect to $(T^b,\L^b)$} if its canonical branch-decomposition is totally pure with respect to $(T^b,\L^b)$.
For a node~$x$ of $T^b$ with $\V_x\subseteq \V_0$ and 
an edge $uv$ in a branch-decomposition $(T,\L)$ of $\V_0$, 
we say $(u,v)$ is \emph{$x$\=/protected} in $(T,\L)$ if
\begin{itemize}
\item $uv$ is $x$\=/blocked and $uv$ $x$\=/guards $v$, or 
\item
there exists a node~$z$ of $T^b$ with $z\le x$ such that 
$uv$ cuts $\V_z$, $(T,\L)$ is $z$-degenerate, and $\L(T,u,v)\subseteq \V_z$.
\end{itemize}
An edge $uv$ of~$T$ is \emph{$x$\=/protected} in $(T,\L)$ if $(u,v)$ or $(v,u)$ is $x$\=/protected in $(T,\L)$.
A node~$v$ of~$T$ is \emph{$x$\=/protected} in $(T,\L)$ if $(u,v)$ is $x$\=/protected in $(T,\L)$ 
for some neighbor~$u$ of $v$ in~$T$. 
See Figure~\ref{fig:protected} for an example. 

\begin{figure}
  \begin{center}
  \tikzstyle{v}=[circle,draw,fill=black,inner sep=0pt,minimum width=2pt]
\begin{tikzpicture}[grow cyclic,level distance=1cm,sibling angle=120,scale=.8]
  \node [v](u) {}
  child { 
    node [v,label=$6$] (z6) {}
  }
  child {
    node [v,label={\color{red}$1$}] (z1) {}
  }
  child {
    node [v] (v) {}
      child {
        node [v] (w) {}
          child {
            node [v,label={\color{red}$2$}] (z2) {}
          }
          child {
            node [v,label=left:$5$] (z5) {}
          }
      }
      child {
        node [v] (w2) {}
          child {
            node [v,label=$7$] (z7) {}
          }
          child {
            node [v] (w3) {}
              child {
                node [v,label={\color{red}$4$}] (z4){}
              }
              child {
                node [v,label={\color{red}$3$}] (z3){}
              }
          }
      }
  };
    \foreach \s/\t in {w2/z7,w3/z4,w3/z3,u/z1,w/z2}
      \draw [very thick,blue,->] (\s)--(\t);
    \foreach \s in {z3,z4,z7,z2,z1}
    \node [circle,draw,inner sep=0pt,minimum width=4pt] at (\s){};
  \end{tikzpicture}
\end{center}
\caption{$x$\=/protected edges and nodes from the example of Figure~\ref{fig:mixed2}.}\label{fig:protected}
\end{figure}

\begin{LEM}\label{lem:protectedconnected}
Let $x$ be a node of $T^b$ and let $(T,\L)$ be a branch-\decomposition{} of~$\V_0$ with $\V_x\subseteq\V_0$.
For two nodes $p$ and $q$ of~$T$,
if neither $p$ nor $q$ is $x$\=/protected and $v$~is a node on the unique path from $p$ to $q$,
then $v$ is not $x$\=/protected.
\end{LEM}
\begin{proof}
We may assume that $v\neq p$ and $v\neq q$.
Suppose that $(u,v)$ is $x$\=/protected in $(T,\L)$ for a neighbor $u$ of $v$.
As $v$ is an internal node of the path from $p$ to $q$, $(u,v)$ points toward $p$ or $q$. 
Assume that $(u,v)$ points toward~$q$.
Let $q'$ be the neighbor of~$q$ on~$P$.
If $uv$ is $x$\=/blocked and $uv$ $x$\=/guards $v$, 
then $q'q$ is $x$\=/blocked and $q'q$ $x$\=/guards~$q$ because $(u,v)$ points toward $q$.
If there exists a node~$z$ of $T^b$ with $z\le x$ such that $uv$ cuts $\V_z$, 
$(T,\L)$ is $z$-degenerate, and $\L(T,u,v)\subseteq\V_z$,
then $\L(T,q',q)\subseteq\V_z$ and $q'q$ cuts $\V_z$ because $(u,v)$ points toward $q$.
Therefore, $q$ is $x$\=/protected, which is a contradiction.
\end{proof}

For  convenience, let us define a few notations to denote a subspace associated to a node or an incidence.
For a $\V$-labeling $(\zeta,\L)$ of a subcubic tree $T$ and a node~$v$ of~$T$ of degree at most $2$, 
let \[\Sigma_\zeta(v)=\spn{ \{ \L(w): \text{$w$ is a
node in $\zeta(v)$ that is a leaf of $\zeta$} \} }.\]
For an incidence $(v,e)$ of~$T$
and $i\in\{1,\ldots,\abs{\zeta(v,e)}\}$, 
let \[\Sigma_\zeta(v,e,i)=\spn{ \{ \L(w): \text{$w$ is a 
nonroot leaf of the $i$th tree in $\zeta(v,e)$}\}}.\]
For an incidence $(u,uv)$ of~$T$
and $i\in\{0,1,\ldots,\abs{\zeta(u,uv)}\}$, 
let 
$L_{\zeta}(u,uv,i)$ be the sum of 
\begin{enumerate}[(i)]
\item 
all $\Sigma_\zeta(w)$ for all nodes $w$ in the component of $T-uv$ containing~$u$
with $\deg_T(w)\le 2$, 
\item 
all $\Sigma_\zeta(w,e,j)$ for all incidences $(w,e)$ in the component of $T-uv$ containing~$u$ 
and all $1\le j\le \abs{\zeta(w,e)}$,
and
\item 
all $\Sigma_{\zeta}(u,uv,j)$ for all $j\le i$,
\end{enumerate}
and let
$R_\zeta(u,uv,i)$ be the sum of 
\begin{enumerate}[(i)]
\item 
all $\Sigma_\zeta(u,uv,j)$ for all $i<j\le \abs{\zeta(u,uv)}$, and 
\item
$L_{\zeta}(v,uv,\abs{\zeta(v,uv)})$.
\end{enumerate}

\paragraph{Witness}
For a node~$x$ of $T^b$ and a vector space $B$ containing $B_x$, %
we say that 
a $\V_x$-labeling $(\zeta,\L)$ of a subcubic tree $T$ is a \emph{witness} of a $B$-namu $\Gamma=(T,\alpha,\lambda,U)$ in $B$
if all of the following hold.
\begin{enumerate}[(i)]
\item For each incidence $(u,uv)$ of~$T$ and $0\le i \le \abs{\zeta(u,uv)}$, 
 	\begin{align*}
 	\alpha(u,uv)&=L_{\zeta}(u,uv,0)\cap B=R_{\zeta}(v,uv,0)\cap B, \\
 	\lambda(uv) &\ge \dim (L_\zeta(u,uv,i) \cap R_{\zeta} (u,uv,i)).
 	\end{align*}
\item 
  $U=\spn{\V_x}\cap B$.
\item %
  $(\zeta,\L)$ is totally pure with respect to $(T^b,\L^b)$.
\item 
Every edge in $\zeta$ is $x$\=/protected in $(\T(\zeta),\L)$.
\end{enumerate}
A $\V_x$-labeling $(\zeta,\L)$ of~$T$ is said to be \emph{$k$-safe} 
if  
  \[
    \dim (\spn{L_t}\cap \spn{\V_x-L_t}) + \dim B_x - \dim (\spn{\V_x-L_t}\cap B_x) \le k
  \]
  for every nonroot node $t$ in $\zeta$,
  where 
  $L_t=\{ \L(w): w$ is a leaf that is a descendant of~$T$ or $w=t\}$.
  Let us call this inequality the \emph{$k$-safe inequality}.
  We note that if $\Gamma$ has width at most~$k$ and $(\zeta,\L)$  is a $k$-safe witness of $\Gamma$, 
  then $(\T(\zeta),\L)$ is a branch-decomposition of 
   width at most~$k$ because $\dim B_x - \dim (\spn{\V_x-L_t}\cap B_x) \ge 0$. 

\paragraph{Join}
For  two disjoint subspace arrangements $\V_1$ and $\V_2$,
let $(\zeta_1,\L_1)$ be a $\V_1$-labeling of a subcubic tree $T_1$ and let $(\zeta_2,\L_2)$ be a $\V_2$-labeling of a subcubic tree $T_2$.
Given a $(T_1,T_2)$-model $(\eta_1,\eta_2)$ in a tree $T$, we define the \emph{sum} $(\zeta,\L)$ of $(\zeta_1,\L_1)$ and $(\zeta_2,\L_2)$ by 
$(\eta_1,\eta_2)$, denoted as $(\zeta_1,\L_1)+_{(\eta_1,\eta_2)} (\zeta_2,\L_2)$, as follows. 
	For each node~$v$ of degree at most $2$ in~$T$, 
	let \[\zeta(v)=
	\begin{cases}
	\zeta_1(v_1) &\text{if $v=\bar\eta_1(v_1)$,}\\
	\zeta_2(v_2) &\text{if $v=\bar\eta_2(v_2)$.}
	\end{cases}\]
	For two sequences $s_1$ and $s_2$,
	let $s_1s_2$ denote the concatenation of $s_1$ and $s_2$.
	For each incidence $(v,e)$ in~$T$,
	let \[\zeta(v,e)=
	\begin{cases}
	\zeta_1(\vec \eta_1(v,e))\zeta_2(\vec\eta_2(v,e))
			&\text{if $v$ is a branch node in both $\eta_1$ and $\eta_2$,}\\
	\zeta_1(\vec\eta_1(v,e)) &\text{if $v$ is a branch node in only $\eta_1$,}\\
	\zeta_2(\vec\eta_2(v,e)) &\text{if $v$ is a branch node in only $\eta_2$,}\\
	\text{the empty sequence} &\text{otherwise.} %
	\end{cases}\]	
We define $\L$ as a function 
from the set of all leaves of $\zeta$ 
to $\V_1\cup \V_2$ 
such that
\[
\L(\ell)=
\begin{cases}
\L_1(\ell) & \text{if $\ell$ is a leaf of $\zeta_1$},\\
\L_2(\ell) & \text{if $\ell$ is a leaf of $\zeta_2$}.
\end{cases}
\]
Trivially  $(\zeta,\L)$ is a $\V_1\cup\V_2$-labeling of~$T$
and the set of edges in $\zeta$ is precisely the union of the set of edges in $\zeta_1$ or $\zeta_2$.

\begin{LEM}\label{lem:purejoinwitness}
Let $(T^b,\L^b)$ be a rooted branch-decomposition of a subspace arrangement $\V$ 
	and let $x$ be a node of $T^b$ with two children $x_1$ and $x_2$.
        Let $T_1$, $T_2$, $T$ be subcubic trees such that there exists a $(T_1,T_2)$-model $(\eta_1,\eta_2)$ in~$T$.
        Let $(\zeta_i,\L_i)$ be a $\V_{x_i}$-labeling of $T_i$ for $i=1,2$.
	If
	\begin{itemize}
	\item  $(\zeta_i,\L_i)$  is totally pure with respect to $(T^b,\L^b)$ for every $i=1,2$,
	\item every $y$\=/protected edge in $(\T(\zeta_i),\L_i)$ is in $\zeta_i$ for every $i=1,2$ and every node $y\leq x_i$ of $T^b$,
	\end{itemize} 
then the following hold.
\begin{enumerate}[(i)]
\item 
The sum $(\zeta,\L)$ of $(\zeta_1,\L_1)$ and $(\zeta_2,\L_2)$ by $(\eta_1,\eta_2)$ 
	is totally pure with respect to $(T^b,\L^b)$.
\item Every $y$\=/protected edge in $(\T(\zeta),\L)$ is in $\zeta$ for every node $y<x$ of $T^b$.
\end{enumerate}
\end{LEM}
\begin{proof}
For each $y<x$, every $y$\=/protected edge is in $\zeta_1$ or $\zeta_2$ and by the construction of $\zeta$, such an edge is in $\zeta$ as well. 

It remains to prove that $(\T(\zeta),\L)$ is totally pure with respect to $(T^b,\L^b)$. 
Since  $(\T(\zeta),\L)$ is trivially $x$\=/pure, it suffices to prove that $(\T(\zeta),\L)$ 
is $y$\=/pure for each node $y<x$ of $T^b$.
We may assume $y\leq x_1$ by symmetry.
We rely on the fact that $(\T(\zeta),\L)$ is $y$\=/degenerate if and only if $(\T(\zeta_1),\L_1)$ is 
$y$\=/degenerate by the definition of $y$\=/degenerate branch-decompositions.

Let us first consider the case that $(\T(\zeta),\L)$ is not
$y$\=/degenerate. Let $uv$ be a $y$\=/blocked edge of $(\T(\zeta),\L)$. We
may assume that $uv$ $y$\=/guards $v$. 
If $uv$ is in $\zeta$, then it is in $\zeta_1$ or $\zeta_2$ and so $(u,v)$
is not $y$\=/mixed because both $(\T(\zeta_1),\L_1)$ and
$(\T(\zeta_2),\L_2)$ are totally pure.
Thus we may assume that $uv$ is not in $\zeta$. 
Then $uv$ corresponds to a path $P$ from $u$ to $v$ in~$T$.
Let $e$ be an edge incident with $u$ in $P$ and 
let $u'v'=\eta_1(e)\in E(\T(\zeta_1))$ where $\vec\eta_1(u,e)=(u',u'v')$.
Because $(\T(\zeta_1),\L_1)$ is totally pure with respect 
to $(T^b,\L^b)$
and $(u',v')$ is $y$\=/protected in $(\T(\zeta_1),\L_1)$, 
it follows that $u'v'$ is in $\zeta_1$ by the assumption. However,
this contradicts our assumption that $uv$ is not in $\zeta$.

Now it remains to consider the case that $(\T(\zeta),\L)$ is $y$\=/degenerate. 
Then $(\T(\zeta_1),\L_1)$ is $y$\=/degenerate and so it is $y$\=/disjoint.
Then every edge of $T(\zeta_1)$ cutting $\V_z$ is $y$\=/protected and therefore it is in $\zeta_1$ by the assumption. So, $(\T(\zeta),\L)$ is $y$\=/disjoint by its construction.
\end{proof}
\begin{LEM}\label{lem:purejoinksafe}
Let $(T^b,\L^b)$ be a rooted branch-decomposition of a subspace arrangement $\V$ 
	and let $x$ be a node of $T^b$ with two children $x_1$ and $x_2$.
        Let $T_1$, $T_2$, $T$ be subcubic trees such that there exists a $(T_1,T_2)$-model $(\eta_1,\eta_2)$ in~$T$.
        Let $(\zeta_i,\L_i)$ be a $\V_{x_i}$-labeling of $T_i$ for $i=1,2$.
	If
  $(\zeta_1,\L_1)$ and $(\zeta_2,\L_2)$ are $k$-safe,
  then so is the sum $(\zeta,\L)$ of $(\zeta_1,\L_1)$ and $(\zeta_2,\L_2)$ by $(\eta_1,\eta_2)$.
\end{LEM}
\begin{proof}
  By symmetry between $\V_{x_1}$ and $\V_{x_2}$, it is enough to prove that
  for each $L\subseteq \V_{x_1} \subseteq \V_{x}$,
  \begin{multline}\label{eq:ksafeeq}
    \dim (\spn{L}\cap \spn{\V_{x_1}-L})
    +\dim B_{x_1} -\dim (\spn{\V_{x_1}-L} \cap B_{x_1}) \\
    =
    \dim (\spn{L}\cap \spn{\V_{x}-L})
    +\dim B_x -\dim (\spn{\V_x-L} \cap B_x).
  \end{multline}
  First observe that $\spn{\V_{x_1}-L}\cap B_{x_1}=\spn{\V_{x_1}-L}\cap\spn{\V-\V_{x_1}}$
  and $\spn{\V_{x}-L}\cap B_x= \spn{\V_x-L}\cap \spn{\V-\V_x}$.
  Now, we claim that both sides of \eqref{eq:ksafeeq} are equal to 
  \[\dim (\spn{L}\cap \spn{\V-L}).\]
  To see that, we use Lemma~\ref{lem:dim-join}.
  First we take $\spn{\V-L}=\spn{\V_{x_1}-L}+\spn{\V-\V_{x_1}}$. Then
  by Lemma~\ref{lem:dim-join}, 
  \begin{align*}
      \dim (\spn{L}\cap \spn{\V-L})
      &=
                                 \dim ((\spn{L}+0)\cap (\spn{\V_{x_1}-L}+\spn{\V-\V_{x_1}}))
                                 \\
    &=\dim (\spn{L}\cap \spn{\V_{x_1}-L})   - \dim (\spn{\V_{x_1}-L}\cap
      \spn{\V-\V_{x_1}})
      \\
      &\quad
      + \dim ((\spn{L}+\spn{\V_{x_1}-L})\cap \spn{\V-\V_{x_1}})\\
    &= \dim (\spn{L}\cap \spn{\V_{x_1}-L})  
      +\dim B_{x_1} - \dim (\spn{\V_{x_1}-L} \cap B_{x_1}).
  \end{align*}
  Now we take $\spn{\V-L}=\spn{\V_x-L}+\spn{\V-\V_x}$. By the same
  method as above with Lemma~\ref{lem:dim-join}, we deduce that 
  \begin{align*}
      \dim (\spn{L}\cap \spn{\V-L})
      &=
                                 \dim ((\spn{L}+0)\cap (\spn{\V_{x}-L}+\spn{\V-\V_{x}}))
                                 \\
    &=\dim (\spn{L}\cap \spn{\V_{x}-L})   - \dim (\spn{\V_{x}-L}\cap
      \spn{\V-\V_{x}})
      \\
      &\quad
      + \dim ((\spn{L}+\spn{\V_{x}-L})\cap \spn{\V-\V_{x}})\\
    &= \dim (\spn{L}\cap \spn{\V_{x}-L})  
      +\dim B_{x} - \dim (\spn{\V_{x}-L} \cap B_{x}).\qedhere  
  \end{align*}
\end{proof}

We prove a lemma on boundary spaces.
\begin{LEM}\label{lem:boundaryspace}
	Let $(T^b,\L^b)$ be a rooted branch-decomposition of a subspace arrangement $\V$ 
	and let $x$ be a node of $T^b$ with two children $x_1$ and $x_2$.
Let $B_{x_i}=\spn{\V_{x_i}}\cap\spn{\V-\V_{x_i}}$ for $i=1,2$.
	Then
\begin{align*}
\spn{\V_{x_1}} \cap (B_{x_1}+B_{x_2}) &= B_{x_1}, \\
\spn{\V_{x_2}} \cap (B_{x_1}+B_{x_2}) &= B_{x_2}, \\
(\spn{\V_{x_1}}+B_{x_1}+B_{x_2})\cap (\spn{\V_{x_2}}+B_{x_1}+B_{x_2})&= B_{x_1}+B_{x_2}.
\end{align*}
\end{LEM}
\begin{proof}

Let us first prove that $\spn{\V_{x_i}}\cap(B_{x_1}+B_{x_2})=B_{x_i}$. By symmetry, it is enough to prove for $i=1$.
Since $B_{x_1}$ is a subspace of $\spn{\V_{x_1}}$, 
\[B_{x_1}=\spn{\V_{x_1}}\cap B_{x_1} \subseteq  \spn{\V_{x_1}}\cap(B_{x_1}+B_{x_2}).\]
As $B_{x_1}\subseteq \spn{\V-\V_{x_1}}$ and $B_{x_2}\subseteq \spn{\V_{x_2}}\subseteq \spn{\V-\V_{x_1}}$, we have 
\[
\spn{\V_{x_1}}\cap(B_{x_1}+B_{x_2})
\subseteq 
\spn{\V_{x_1}}\cap \spn{\V-\V_{x_1}}=B_{x_1}.
\]
This proves that $\spn{\V_{x_i}}\cap(B_{x_1}+B_{x_2})=B_{x_i}$ for $i=1,2$.

Now let us prove the last equation.
Since $(\spn{\V_{x_1}}+B_{x_1}+B_{x_2})\cap (\spn{\V_{x_2}}+B_{x_1}+B_{x_2})\supseteq (B_{x_1}+B_{x_2})$,
the proof is completed by showing that the dimensions of both sides are equal.
Note that $ B_{x_1} \cap B_{x_2} =  \spn{\V_{x_1}} \cap \spn{\V-\V_{x_1}} \cap \spn{\V_{x_2}} \cap \spn{\V-\V_{x_2}}= \spn{\V_{x_1}}\cap \spn{\V_{x_2}}$ 
and therefore
\begin{align*}
&\dim ((\spn{\V_{x_1}}+B_{x_1}+B_{x_2})\cap (\spn{\V_{x_2}}+B_{x_1}+B_{x_2}))\\
&= \dim ((\spn{\V_{x_1}}+B_{x_2})\cap (\spn{\V_{x_2}}+B_{x_1})) \\
&=  \dim \spn{\V_{x_1}} + \dim B_{x_2} - \dim (\spn{\V_{x_1}}\cap B_{x_2}) \\
&\quad +  \dim \spn{\V_{x_2}} + \dim B_{x_1} - \dim (\spn{\V_{x_2}}\cap B_{x_1}) 
-\dim (\spn{\V_{x_1}}+\spn{\V_{x_2}})\\
&= \dim \spn{\V_{x_1}} + \dim \spn{\V_{x_2}} -\dim (\spn{\V_{x_1}}\cap\spn{\V_{x_2}})
- \dim(\spn{\V_{x_1}}+\spn{\V_{x_2}})  \\
&\quad+ \dim B_{x_1} + \dim B_{x_2} - \dim (\spn{\V_{x_1}}\cap\spn{\V_{x_2}}) \\
&=\dim (B_{x_1} + B_{x_2}). \qedhere
\end{align*}
\end{proof}

\begin{PROP}\label{prop:joinwitness} 
  Let $(T^b,\L^b)$ be a branch-decomposition 
  of a subspace arrangement $\V$ 
  and let $x$ be a node of $T^b$ with two children $x_1$ and $x_2$.
  For each $i=1,2$, let 
  $\Gamma_i=(T_i,\alpha_i,\lambda_i,U_i)$ be a $B_{x_i}$-namu
  and 
  $(\zeta_i,\L_i)$ be a $\V_{x_i}$-labeling of $T_i$
  witnessing $\Gamma_i$ in~$B_{x_i}$.
  Let $\Gamma$ be a $(B_{x_1}+B_{x_2})$-namu that is the 
  sum of $\Gamma_1$ and $\Gamma_2$ by a $(T_1,T_2)$-model $(\eta_1,\eta_2)$.
  Let $(\zeta,\L)$ be the sum of $(\zeta_1,\L_1)$ and $(\zeta_2,\L_2)$ by $(\eta_1,\eta_2)$.

  If every $y$\=/protected edge in $(\T(\zeta_i),\L_i)$ is in $\zeta_i$ for each $i=1,2$ and 
  every node $y\le x_i$ of $T^b$,
  then 
  $(\zeta,\L)$
  is a witness of $\Gamma$ in $B_{x_1}+B_{x_2}$
  and 
  every $y$\=/protected edge in $(\T(\zeta),\L)$ is in $\zeta$ for every
  node $y<x$ of $T^b$.
  In addition, if both $(\zeta_1,\L_1)$ and $(\zeta_2,\L_2)$ are
  $k$-safe, then so is $(\zeta,\L)$.
\end{PROP}
\begin{proof}
  Let $\Gamma=(T,\alpha,\lambda,U)$  and  $B=B_{x_1}+B_{x_2}$.
  Lemma~\ref{lem:purejoinksafe} already showed that if both $(\zeta_1,\L_1)$ and $(\zeta_2,\L_2)$ are
  $k$-safe, then so is $(\zeta,\L)$. 
  Lemma~\ref{lem:purejoinwitness}
  implies that $(\zeta,\L)$  is totally pure with respect to
  $(T^b,\L^b)$ 
  and  every $y$\=/protected edge of $(\T(\zeta),\L)$ is in $\zeta$ for every
  node $y<x$ of $T^b$.

  Now it remains to show that $(\zeta,\L)$ is a witness of $\Gamma$ in
  $B_{x_1}+B_{x_2}$.
  We need to check four properties (i)--(iv) in the definition of a witness.
  First of all, we have seen that (iii) holds.
  For (ii), by Lemmas~\ref{lem:boundaryspace} and~\ref{lem:join-key}, %
  \begin{align*}
    U&=U_1+U_2=(\spn{\V_{x_1}}\cap B_{x_1})+(\spn{\V_{x_2}}\cap B_{x_2}) \\
     &=(\spn{\V_{x_1}}+\spn{\V_{x_2}})\cap (B_{x_1}+B_{x_2})=\spn{\V_{x}}\cap B,
  \end{align*}
  and so  $(\zeta,\L)$ and $\Gamma$ satisfy (ii).

	To see (i), 
        observe that  for each incidence $(u,e)$ of~$T$, 
        \[ 
          L_{\zeta} (u,e,0)= L_{\zeta_1}(\vec \eta_1(u,e), 0) +L_{\zeta_2}(\vec \eta_2(u,e),0) .\]
        By Lemmas~\ref{lem:boundaryspace} and \ref{lem:join-key}, we deduce that 
        \[
          L_{\zeta} (u,e,0)\cap B
          = L_{\zeta_1}(\vec \eta_1(u,e), 0)\cap B
          +L_{\zeta_2}(\vec \eta_2(u,e),0)\cap B.\]
        As $L_{\zeta_i}(\vec \eta_i(u,e), 0)\subseteq \spn{\V_{x_i}}$
        and $\spn{\V_{x_i}}\cap B=B_{x_i}$ for each $i=1,2$, we deduce that 
        \[
          L_{\zeta} (u,e,0)\cap B
          = L_{\zeta_1}(\vec \eta_1(u,e), 0)\cap B_{x_1}
          +L_{\zeta_2}(\vec \eta_2(u,e),0)\cap B_{x_2},\]
        and since $(\zeta_i,\L_i)$ is a witness of $\Gamma_i$ for $i=1,2$ and 
        $\Gamma=\Gamma_1+_{(\eta_1,\eta_2)} \Gamma_2$, we have $L_{\zeta}(u,e,0)\cap B=  \alpha_1(\vec \eta_1(u,e)) + \alpha_2(\vec \eta_2(u,e)) =\alpha(u,e)$.
	Similarly, we have $R_{\zeta} (v,e,0)\cap B = \alpha(u,e)$.
	
	For $i\in\{0,\ldots,\abs{\zeta(u,e)}\}$, by the construction,
	there exist $i_1\in\{0,\ldots,\abs{\zeta_1(\vec \eta_1(u,e))}\}$ 
	and $i_2\in \{0,\ldots,\abs{\zeta_2(\vec \eta_2(u,e))}\}$ such that 
	$L_\zeta (u,e,i) = L_{\zeta_1} (\vec \eta_1(u,e),i_1) + L_{\zeta_2} (\vec \eta_2(u,e),i_2)$, and
	$R_\zeta (u,e,i) = R_{\zeta_1} (\vec \eta_1(u,e),i_1) + R_{\zeta_2} (\vec \eta_2(u,e),i_2)$.
        For $j=1,2$, let 
        \[
          L_j=L_{\zeta_j} (\vec \eta_j(u,e),i_j)
        \text{ and }R_j= R_{\zeta_j} (\vec \eta_j(u,e),i_j).
        \]
        By Lemma~\ref{lem:dim-join}, 
        \begin{align*}
          \lefteqn{\dim (L_\zeta (u,e,i)\cap R_\zeta (u,e,i)) 
          =\dim((L_1+L_2)\cap (R_1+R_2))}\\
          &= \dim (L_1\cap R_1)+\dim (L_2\cap R_2)
          -\dim (L_1\cap L_2) -\dim (R_1\cap R_2)
          \\&\quad
          +\dim ((L_1+R_1)\cap (L_2+R_2)).
        \end{align*}
        For $j=1,2$, as $(\zeta_j,\L_j)$ is a witness of $\Gamma_j$, we have 
        $\dim (L_j\cap R_j)\le \lambda_j(\eta_j(e))$ and 
        $U_j=\spn{\V_{x_j}}\cap B_{x_j}$.
        Also note that $L_1\cap L_2=(L_1\cap B_{x_1})\cap (L_2\cap B_{x_2})= \alpha_1(\vec\eta_1(u,e))\cap \alpha_2(\vec\eta_2(u,e))$
        and similarly $R_1\cap R_2=\alpha_1(\vec\eta_1(v,e))\cap \alpha_2(\vec\eta_2(v,e))$
        where $v$ is an end of $e$ other than $u$.
        Furthermore $(L_1+R_1)\cap (L_2+R_2)\subseteq B_{x_j}$ for each $j=1,2$,
        and so $(L_1+R_1)\cap (L_2+R_2)= (L_1+R_1)\cap B_{x_1}
        \cap (L_2+R_2)\cap B_{x_2} = U_1\cap U_2$.
        Thus, we deduce that 
        $
         \dim (L_\zeta (u,e,i)\cap R_\zeta (u,e,i))\le 
         \lambda_1(\eta_1(e))+\lambda_2(\eta_2(e)) -\dim (\alpha_1(\vec \eta_1(u,e))\cap\alpha_2(\vec \eta_2(u,e)))
         -\dim (\alpha_1(\vec\eta_1(v,e))\cap\alpha_2(\vec\eta_2(v,e))) + \dim (U_1 \cap U_2)
         = \lambda(e)$.
      This proves (i).

	Finally let us  verify (iv). We claim that every edge in $\zeta$ is an $x$\=/protected edge in $(\T(\zeta),\L)$. If $(\T(\zeta),\L)$ is $x$\=/degenerate, every edge of $(\T(\zeta),\L)$ 
	is $x$\=/protected and the claim trivially holds. Hence we may assume that $(\T(\zeta),\L)$ is not $x$\=/degenerate. 
	Note that every edge $e$ in $\zeta$ is either in $\zeta_1$ or 
	$\zeta_2$ and thus $e$ is $x_i$-protected in $(\T(\zeta_i),\L)$ for $i=1$ or $2$.
	By symmetry, we may assume that $i=1$. 

        If $e$ is $x_1$-blocked by an $x_1$-blocking path or an $x_1$-guarding edge 
	in $(\T(\zeta_1),\L_1)$, %
	then by construction of $(\T(\zeta),\L)$, it can be easily seen that $e$ is either $x$\=/blocked by an $x$\=/blocking path or an $x$\=/guarding edge in $(\T(\zeta),\L)$ by Lemma~\ref{lem:x-to-y}.
	Therefore, $e$ is $x$\=/protected.

	If there exists $y\le x_1$ such that $(T(\zeta),\L)$ is $y$\=/degenerate, $e=uv$ cuts $\V_y$ in $(\T(\zeta_1),\L_1)$ and $\L_1(\T(\zeta_1),u,v)\subseteq \V_y$, then it is clear that 
	$e=uv$ cuts $\V_y$ in $(\T(\zeta),\L)$ and $\L(\T(\zeta),u,v)\subseteq \V_y$ from the construction of $(\T(\zeta),\L)$ because $(\T(\zeta),\L)$ is totally pure with respect to $(T^b,\L^b)$. Hence, $e$ is $x$\=/protected. 
        This completes the proof that $(\zeta,\L)$ is a witness of $\Gamma$ in $B_{x_1}+B_{x_2}$.
\end{proof}

\paragraph{Shrink} 
It is easy to obtain a witness of $\Gamma|_{B'}$ in $B'$ if we have a witness of a $B$-namu $\Gamma$ in $B$ from the following lemma.
\begin{PROP}\label{prop:shrinkwitness}
  Let $(T^b,\L^b)$ be a branch-decomposition of a subspace arrangement $\V$ and let $x$ be a node of $T^b$.
  Let $(\zeta,\L)$ be a witness of a $B$-namu $\Gamma$ in $B$. 
  If $B'\subseteq B$, then 
  $(\zeta,\L)$ is also a witness of $\Gamma|_{B'}$ in $B'$
\end{PROP}
\begin{proof}
  We only need to check (i) and (ii) in the definition of a witness and it is trivial.
\end{proof}

\paragraph{Trim}
Let $(\zeta,\L)$ be a $\V_0$-labeling of a tree $T$ having at least one edge.

For an edge $e$ of $\T(\zeta)$ that is not in $\zeta$, %
we say $(\zeta',\L)$ is the \emph{$\V_0$-labeling induced by an edge $e$ from $(\zeta,\L)$}
if there is a tree $T'$ with exactly one node~$v$ such that
$\zeta'$ is a labeling of $T'$ where 
$\zeta'(v)$ is the rooted binary tree obtained from $\T(\zeta)$ by subdividing~$e$ 
to create a new degree-$2$ node~$v$, which will be the root of $\zeta'(v)$.

For a subtree $T'$ of~$T$ with $\abs{E(T')}\ge 1$, 
we say that $(\zeta',\L)$ is the \emph{$\V_0$-labeling induced by $T'$ from $(\zeta,\L)$}
if it satisfies all of the following.
\begin{itemize}
\item For each incidence $(v,e)$ of $T'$, 
  $\zeta'(v,e)=\zeta(v,e)$.
\item
  For each node~$v$ of $T'$ whose degree is at most~$2$,
  we define $\zeta'(v)=\zeta(v)$ if $\zeta(v)$ contains only one node, 
  and otherwise, 
  we define $\zeta'(v)$
  as the component of $\T(\zeta)-E'$ containing~$v$,
  where $E'$ is the set of all edges of $\T(\zeta)$ corresponding to an edge in $T'$. The set $E'$ is well defined because 
  $\T(\zeta)$ contains a subdivision of~$T$ and 
  $T'$ is a subtree of~$T$. The graph $\T(\zeta)-E'$ has a node~$v$ because $\zeta(v)$ has at least two nodes and so $v$ was not smoothed in the construction of $\T(\zeta)$.

\end{itemize}

Observe that the $\V_0$-labeling induced by $T'$ when $E(T')=\{e\}$ is different from the $\V_0$-labeling induced by an edge $e$. 
The latter will be applied in order to represent an $x$\=/degenerate branch-decomposition $(T,\L)$ with a $\V_0$-labeling of a single node. 
By definition, if $(\zeta',\L)$ is induced by a subtree $T'$ from a $\V_0$-labeling $(\zeta,\L)$, then 
$(\zeta',\L)$ is a $\V_0$-labeling of $T'$
and
$\T(\zeta)=\T(\zeta')$.
Furthermore, 
for each node~$v$ of~$T$ of degree at most $2$, 
$\zeta(v)$ is a subtree of some $\zeta'(w)$ for some $w\in V(T')$
and 
for each incidence $(v,e)$ of~$T$, 
either there exists a node $w$ of $T'$ 
such that 
all rooted binary trees in $\zeta(v,e)$
are subtrees of some $\zeta'(w)$
or 
there exists an incidence $(w,f)$ of $T'$
such that 
all rooted binary trees in $\zeta(v,e)$
are subtrees of rooted binary trees in $\zeta'(w,f)$.

\begin{PROP}\label{prop:trimwitness}
	Let $(T^b,\L^b)$ be a branch-decomposition of a subspace arrangement $\V$ and let $x$ be a node of $T^b$.
        Let $(\zeta,\L)$ be a witness of a $B_x$-namu $\Gamma$ in $B_x$ such that 
	all $y$\=/protected edges of $(\T(\zeta),\L)$ are in $\zeta$ for every node $y<x$ of $T^b$.
        Let $(\zeta',\L)$ be the $\V_x$-labeling induced by $T(\trim(\Gamma))$ if $\Gamma$ has no degenerate edge
        and a $\V_x$-labeling induced by an $x$\=/degenerate edge of $(\T(\zeta),\L)$ that is not in $\zeta$ if $\Gamma$ has a degenerate edge.

        Then $(\zeta',\L)$
	is a witness of $\trim(\Gamma)$ in $B_x$
	and 
	all $y$\=/protected edges of $(\T(\zeta'),\L)$ are in $\zeta'$ for every node $y$ of $T^b$ with $y\le x$.
        Furthermore, if $(\zeta,\L)$ is $k$-safe and $\Gamma$ is a $k$-safe extension of $\trim(\Gamma)$, then $(\zeta',\L)$ is $k$-safe.
\end{PROP}
\begin{proof}
        First suppose that $\Gamma$ has a degenerate edge and
        $(\zeta',\L)$ is a $\V_x$-labeling of $T(\trim(\Gamma))$ induced by $e$ from $(\zeta,\L)$ for 
        an $x$\=/degenerate edge $e$ of $(\T(\zeta),\L)$ that is not in $\zeta$.
        Note that such $e$ exists because $\Gamma$ has a degenerate edge and $(\zeta,\L)$ is a witness of $\Gamma$ in $B_x$.
        Then $(\T(\zeta'),\L)$ is a rooted branch-decomposition
        obtained from $(\T(\zeta),\L)$ by subdividing~$e$.
        Let $e_1$, $e_2$ be the edges of $\T(\zeta')$ obtained from $e$ by subdividing~$e$.
        As $T(\trim(\Gamma))$ has only one node, (i) in the definition of a witness trivially holds for $(\zeta',\L)$. Also (ii) is trivial.
        As $(\zeta,\L)$ is totally pure with respect to $(T^b,\L^b)$,
        so is $(\zeta',\L)$, proving (iii) in the definition of a witness.
        We claim that $(\T(\zeta'),\L)$ is $x$\=/degenerate. If not, then in $(\T(\zeta),\L)$, $e$ cuts $\V_y$ for some $y<x$ where $(\T(\zeta),\L)$ is $y$\=/degenerate.
        However, this means that $e$ is $y$\=/protected because 
        $(\T(\zeta),\L)$ is $y$\=/pure. This contradicts the assumption that every $y$\=/protected edge in $(\T(\zeta),\L)$ is 
        in $\zeta$.
        This proves the claim and we conclude that every edge of $\T(\zeta')$ is $x$\=/protected, and so (iv) holds
and furthermore every $x$\=/protected edge of $(T(\zeta'),\L)$ is in $\zeta$ trivially.
        Thus, we proved that if $\Gamma$ has a degenerate edge, then $(\zeta',\L)$ is a witness of $\trim(\Gamma)$ and every $y$\=/protected edge of $(\T(\zeta'),\L)$ is in $\zeta'$ for all $y\le x$.

        Now let us assume that $\Gamma$ has no degenerate edge.
        We claim that $(\zeta',\L)$ is a witness of $\trim(\Gamma)$.
        By definition, $T(\zeta)=T(\zeta')$ and so 
        $(\zeta',\L)$ is totally pure, proving (iii) for $(\zeta',\L)$.
        By the definition of the trim operation, (i) and (ii) hold.

	Every edge $e$ in $\zeta'$ either is in $\zeta$
	or is in~$T$ that is removed by trimming $\Gamma$.
	For the former case, 
        it is $x$\=/protected because $(\zeta,\L)$ is a witness of $\Gamma$. 
	For the latter case, a removed edge is blocked in $\Gamma$ and so 
        it is $x$\=/protected. 
        This proves that $(\zeta',\L)$ is a witness of $\trim(\Gamma)$ in $B_x$.
        
        It remains to show that if $(\zeta,\L)$ is $k$-safe and $\Gamma$ is a $k$-safe extension of $\trim(\Gamma)$, then 
        $(\zeta',\L)$ is $k$-safe.
        This follows trivially from the definition of $k$-safe extensions.
\end{proof}

\paragraph{Comparison}
For a $\V_0$-labeling $(\zeta,\L)$ of~$T$, 
a subdivision $T'$ of~$T$ induces a $\V_0$-labeling $(\zeta',\L)$ as follows.
For each node~$v$ of $T'$ having degree at most $2$, 
we define $\zeta'(v)=\zeta(v)$ if $v$ is present in~$T$
and otherwise $\zeta'(v)$ as the tree having only one vertex~$v$.
For each edge $uv$ of~$T$,
let $P=v_0v_1\cdots v_p$ be a path from $u$ to $v$ in $T'$.
We define
$\zeta'(v_0,v_0v_1)= \zeta(u,uv)$
and $\zeta'(v_p,v_{p-1}v_p)= \zeta(v,uv)$
and for all other incidences $(w,e)$ in $P$, $\zeta'(w,e)$ is the empty sequence.
We simply say that such a $\V_0$-labeling $(\zeta',\L)$ is a \emph{subdivision} of $(\zeta,\L)$ induced by $T'$.

From the definition of a subdivision, it is easy to observe the following.
\begin{LEM}\label{lem:subdivisionwitness}
  Let $(\zeta,\L)$ be a $\V_0$-labeling of~$T$
  and let $T'$ be a subdivision of~$T$.
  If $(\zeta',\L)$ is a subdivision of $(\zeta,\L)$ induced by $T'$, 
  then $(\T(\zeta'),\L)=(\T(\zeta),\L)$
  and the set of edges in $\zeta$ is equal to the set of edges in $\zeta'$.
\end{LEM}

Conversely, 
if $(\zeta',\L)$ is a $\V_0$-labeling of $T'$
and $T'$ is a subdivision of $T''$,
then $(\zeta',\L)$ induces a $\V_0$-labeling $(\zeta'',\L)$ by $T''$ as follows.
First we orient edges of $T''$ arbitrarily. 
For each node~$v$ of $T''$, we define $\zeta''(v)=\zeta'(v)$.
For each edge $uv$ of $T''$, if $uv$ is oriented toward $v$, then we define $\zeta''(v,uv)$ to be the empty sequence
and we define $\zeta''(u,uv)$ as follows.
Let $P$ be the path $v_0v_1\cdots v_p$ in $T'$ from $u$ to $v$. 
Then $\zeta''(u,uv)$ is the concatenation of the following for $i=1,2,\ldots,p$ in the order:
\begin{itemize}
\item $\zeta'(v_{i-1},v_{i-1}v_i)(\zeta'(v_i,v_{i-1}v_i))^{-1}$,
\item a sequence $(\zeta'(v_i))$ of length $1$ 
  if $i<p$ and $\zeta'(v_i)$ has more than one node.
\end{itemize}
Here we use the notation $(\zeta'(v_i,v_{i-1}v_i))^{-1}$ to denote the inverted sequence.
Again, we can easily observe the following lemma.

\begin{LEM}\label{lem:subdivisionwitness2}
  Let $(\zeta,\L)$ be a $\V_0$-labeling of $T'$
  and let $T'$ be a subdivision of $T''$.
  If $(\zeta',\L)$ induces a $\V_0$-labeling $(\zeta,\L)$ by $T''$,
  then $(\T(\zeta'),\L)=(\T(\zeta),\L)$
  and the set of edges in $\zeta$ is equal to the set of edges in $\zeta'$.
\end{LEM}

If a tree $T$ ensures $\Gamma_1\tle\Gamma_2$ 
for two $B$-namus $\Gamma_1=(T_1,\alpha_1,\lambda_1,U)$ and 
$\Gamma_2=(T_2,\alpha_2,\lambda_2,U)$, then 
$T$ is isomorphic to a subdivision of~$T_1$
and also isomorphic to a subdivision of~$T_2$. 
By the above construction, a $\V_0$-labeling $(\zeta,\L)$ of $T_1$
induces a $\V_0$-labeling $(\zeta',\L)$ of~$T$,
which in turn induces a $\V_0$-labeling $(\zeta'',\L)$ of $T_2$.
We say that $(\zeta'',\L)$ is \emph{induced} by $T$ from $(\zeta,\L)$.

\begin{PROP}\label{prop:comparewitness}
	Let $(T^b,\L^b)$ be a branch-decomposition of a subspace arrangement $\V$ and let $x$ be a node of $T^b$. 
	Let $\Gamma_1$ and $\Gamma_2$ be $B_x$-namus such that $\Gamma_1\tle\Gamma_2$. 
        Let $T$ be a tree ensuring $\Gamma_1\tle\Gamma_2$.
        
	If a $\V_x$-labeling $(\zeta_1,\L)$ is a witness of $\Gamma_1$ in $B_x$
	and 
	all $y$\=/protected edges of $(\T(\zeta_1),\L)$ are in $\zeta_1$ 
        for every node $y$ of $T^b$ with $y\le x$,
	then 
	a $\V_x$-labeling $(\zeta_2,\L)$ of $T(\Gamma_2)$ induced by $T$ from  $(\zeta_1,\L)$
	is a witness of $\Gamma_2$ in $B_x$
	and 
	all $y$\=/protected edges of $(\T(\zeta_2),\L)$ are in $\zeta_2$
        for every node $y$ of $T^b$ with $y\le x$.
        In addition, if $(\zeta_1,\L)$ is $k$-safe, then so is $(\zeta_2,\L)$.
\end{PROP}
\begin{proof}
  Let $\Gamma_1=(T_1,\alpha_1,\lambda_1,U)$ and $\Gamma_2=(T_2,\alpha_2,\lambda_2,U)$.
  We may assume, by applying the isomorphism, that
  $T$ is a subdivision of $T_1$ 
  as well as a subdivision of $T_2$.
  Let $(\zeta,\L)$ be the subdivision of $(\zeta_1,\L)$ induced by $T$
  so that $(\zeta,\L)$ induces $(\zeta_2,\L)$.
  Let $\Gamma=(T,\alpha,\lambda,U)$ be the subdivision of $\Gamma_1$. 

  It is easy to see from the definition
  that $(\zeta,\L)$ is a witness of $\Gamma$ in $B_x$.
  By Lemma~\ref{lem:subdivisionwitness},
  the set of edges in $\zeta_1$ is equal to the set of edges in $\zeta$
  and $(\T(\zeta),\L)=(\T(\zeta_1),\L)$.
  In particular, this implies that 
  for each node $y\le x$ of $T^b$, 
  every $y$\=/protected edge of $(\T(\zeta),\L)$ is in $\zeta$
  and 
  if $(\zeta_1,\L)$ is $k$-safe, then 
  so is $(\zeta,\L)$.

	Let us check that $(\zeta_2,\L)$ satisfies the conditions of the definition of witnesses.
        For each incidence $(u,uv)$ of $T_2$, 
        $T$ has a corresponding incidence $(u',u'v')$ 
        such that $L_{\zeta_2}(u,uv,0)=L_\zeta (u',u'v',0)$.
        As $\alpha(u',u'v')=\alpha_2(u,uv)$ by $\Gamma\tle \Gamma_2$
        and $\alpha(u',u'v')=L_\zeta(u',u'v',0)\cap B_x$,
        we deduce that 
        $\alpha_2(u,uv)=L_{\zeta_2}(u,uv,0)\cap B_x$.
        Similarly 
        $\alpha_2(u,uv)=R_{\zeta_2}(v,uv,0)\cap B_x$.
        
        For each incidence $(u,uv)$ of $T_2$
        and $0\le i\le \abs{\zeta_2(u,uv)}$,
	there exists an edge $u'v'$ of~$T$ and $i'\in\{0,\ldots, \abs{\zeta(u',u'v')}\}$
	such that 
        $(u,uv)$ corresponds to $(u',u'v')$ in~$T$,
        $L_{\zeta_2}(u,uv,i)= L_\zeta(u',u'v',i')$,
        and 
        $R_{\zeta_2}(u,uv,i)= R_\zeta(u',u'v',i')$.
        As $\Gamma\tle \Gamma_2$, we have 
        \begin{multline*}
          \lambda_2(uv)\ge \lambda(u'v')
        \ge \dim (L_{\zeta}(u',u'v',i')\cap R_\zeta(u',u'v',i'))\\
        = \dim (L_{\zeta_2}(u,uv,i)\cap R_{\zeta_2}(u,uv,i)).
      \end{multline*}
	
      By Lemma~\ref{lem:subdivisionwitness2}, 
      $(\T(\zeta),\L)=(\T(\zeta_2),\L)$ and 
	the set of all edges in $\zeta$ is equal to the set of all edges in $\zeta_2$.
	This implies that
	(iii) and (iv) in the definition of a witness hold, 
	all $y$\=/protected edges of $(\T(\zeta_2),\L)$ are in $\zeta_2$ 
        for every node $y$ of $T^b$ with $y\le x$, and
	if $(\zeta,\L)$ is $k$-safe, then so is 
        $(\zeta_2,\L)$.
        This completes the proof.
\end{proof}

\section{Correctness of Algorithm~\ref{alg:fullset}} \label{sec:fullset}

For a $B$-namu $\Gamma=(T,\alpha,\lambda,U)$
and a subtree $T'$ of~$T$,
we say that  a $B$-namu $\Gamma'=(T',\alpha',\lambda',U)$ is \emph{induced by $T'$ from $\Gamma$}
if 
$\alpha'(v,e)=\alpha(v,e)$ and $\lambda'(e)=\lambda(e)$
for every incidence $(v,e)$ of $T'$.

For a finite field $\F$ and a positive integer $r$,
let $\V$ be a subspace arrangement of subspaces of $\F^r$.
Let $(T^b,\L^b)$ be a rooted branch-decomposition of~$\V$
and let $k$ be a nonnegative integer.
For a node~$x$ of $T^b$ and a branch-decomposition $(T,\L)$ of $\V_x$, %
the \emph{reduced $B_x$-namu of $(T,\L)$} is the $B_x$-namu
induced by a subtree $T'$ of~$T$ from the canonical $B_x$-namu of $(T,\L)$
where $T'$ is obtained by the following rule.
\begin{itemize}
\item If $(T,\L)$ is $x$\=/degenerate, then $T'$ is a subtree having only one node of~$T$.
\item If $(T,\L)$ is not $x$\=/degenerate, then 
$T'$ is a subtree of~$T$ obtained by deleting all nodes $v$ 
such that $(u,v)$ is $x$\=/protected in $(T,\L)$.
\end{itemize}

For a node~$x$ of $T^b$ and a branch-decomposition $(T,\L)$ of $\V_x$,
we say that $(T,\L)$ is \emph{$k$-safe with respect to $x$} if
for every edge $uv$ of~$T$ such that $(u,v)$ is $x$\=/protected,
\[
\dim \left(\sum_{s\in A_v(T-uv)}
\hspace{-1.5em}
\L(s) \cap\sum_{t\in A_u(T-uv)}
\hspace{-1.5em}
\L(t)\right)+ \dim B_x - \dim \left(B_x \cap \sum_{t\in A_u(T-uv)}
\hspace{-1.5em}
\L(t)\right) \le k.
\]
For each node~$x$ of $T^b$,
we define \emph{the full set at $x$ of width~$k$ with respect to $(T^b,\L^b)$}
as the set of all $\Gamma$ in $U_k(B_x)$ 
such that $\Delta\tle\Gamma$ for the reduced $B_x$-namu $\Delta$ of some branch-\decomposition{} $(T,\L)$ of $\V_x$ having width at most~$k$ which is 
$k$-safe with respect to $x$, and 
totally pure with respect to $(T^b,\L^b)$.
We write it $\FS(x;T^b,\L^b)$ or $\FS(x)$ for brevity.
In this section, we aim to show that the full set is precisely the set computed in Algorithm~\ref{alg:fullset}. 
\begin{PROPfullset}
  For every node~$x$ of $T^b$,
  $\FS(x)$ is equal to $\FF_x$ in Algorithm~\ref{alg:fullset}.
\end{PROPfullset}

Proposition~\ref{prop:fullset} combined with 
Proposition~\ref{prop:pure} allows us to show that Algorithm~\ref{alg:fullset} is correct.

\begin{PROP}\label{prop:FSroot}
 Let $k$ be a nonnegative integer. 
 Let $(T^b,L^b)$ be a rooted branch-decomposition of a subspace arrangement $\V$ and for each node~$x$ of $T^b$, let $\FF_x$ be the set computed by Algorithm~\ref{alg:fullset}.
 The branch-width of~$\V$ is at most~$k$ 
 if and only if $\FF_{r}\neq \emptyset$ at the root node $r$ of $T^b$.
\end{PROP}
\begin{proof}
By Proposition~\ref{prop:pure}, if the branch-width of~$\V$ is at most~$k$, 
then $\V$ has a branch-\decomposition{} $(T,\L)$ of width at most~$k$
that is totally pure with respect to $(T^b,\L^b)$.
Note that $(T,\L)$ is $k$-safe with respect to the root $r$ of $T^b$ because $B_{r}=\{0\}$.
Let $\Delta$ be the reduced $B_r$-namu of $(T,\L)$.
Then $\tau(\Delta)$ is in $\FS(r)$ because $\Delta\tle\tau(\Delta)$.
By Proposition~\ref{prop:fullset}, $\FS(r)=\FF_r$ and thus $\FF_r$ is nonempty.

If $\Gamma\in\FF_r$, then $\Gamma\in\FS(r)$ by Proposition~\ref{prop:fullset}.
So there exists a branch-decomposition of $\V_r=\V$ whose width is at most~$k$.
\end{proof}

The following proposition gives a relation between full sets and witnesses.

\begin{PROP}\label{prop:witnessexist}
Let $(T^b,\L^b)$ be a branch-decomposition of a subspace arrangement $\V$ and let $x$ be a node of $T^b$
and $\Gamma \in U_k(B_x)$.
Then $\Gamma\in \FS(x)$
if and only if 
there exists a $k$-safe $\V_x$-labeling $(\zeta,\L)$ of $T(\Gamma)$ 
that is a witness of $\Gamma$ in $B_x$
such that 
every $y$\=/protected edge in $(\T(\zeta),\L)$ is in $\zeta$ for every node $y\le x$ of~$T^b$.
\end{PROP}
\begin{proof}
For the forward direction, 
suppose a $B_x$-namu $\Gamma$ is in $\FS(x)$.
Then, by definition,
there exists a branch-decomposition $(T,\L)$ of $\V_x$ such that 
\begin{itemize}
\item the width of $(T,\L)$ is at most~$k$,
\item $(T,\L)$ is $k$-safe with respect to $x$,
\item $(T,\L)$ is totally pure with respect to $(T^b,\L^b)$, and
\item $\Delta\tle\Gamma$ for the reduced $B_x$-namu $\Delta$ of $(T,\L)$.
\end{itemize}
Let $T'=T(\Delta)$ be a subtree of~$T$.
For every incidence $(v,e)$ of $T'$, let $\zeta(v,e)$ be the empty sequence,
and for every node~$v$ of degree at most $2$ in $T'$, 
let $\zeta(v)$ be the component of $T-E(T')$ containing~$v$,
viewed as a rooted binary tree with the root $v$.
Then, it is clear that $(\zeta,\L)$ is a $\V_x$-labeling of $T'$.
Furthermore $T=\T(\zeta)$ if $(T,\L)$ is not $x$\=/degenerate
and $\T(\zeta)$ is obtained from $T$ by subdividing an improper $x$\=/degenerate edge 
if $(T,\L)$ is $x$\=/degenerate.
By definition, it is easy to see that $(\zeta,\L)$ is a witness of $\Delta$ in $B_x$.
Because $(T,\L)$ is $k$-safe with respect to $x$, by definition, the $k$-safe inequality holds 
for every edge in $\zeta$,
and so $(\zeta,\L)$ is $k$-safe.

For a node $y$ of $T^b$ with $y\le x$ and a $y$\=/protected edge $uv$ in $(\T(\zeta),\L)$, 
if $uv$ is $y$\=/blocked and $uv$ $y$\=/guards $v$,
then by Lemma~\ref{lem:x-to-y}, $uv$ is $x$\=/blocked and $uv$ $x$\=/guards~$v$
because $(\T(\zeta),\L)$ is $y$\=/pure. 
So $(u,v)$ is $x$\=/protected in $(\T(\zeta),\L)$.
If there exists a node~$z$ of $T^b$ with $z\le y$ such that 
$uv$ cuts $\V_z$, $(T,\L)$ is $z$-degenerate, and $\L(\T(\zeta),u,v)\subseteq\V_z$,
then 
$(u,v)$ is $x$\=/protected in $(\T(\zeta),\L)$ because $z\le y\le x$.
Thus, %
every $y$\=/protected edge in $(\T(\zeta),\L)$ is 
$x$\=/protected in $(\T(\zeta),\L)$ and so it is in $\zeta$
because of the definitions of $T'$ and $\zeta$.

As $\Delta\tle\Gamma$, by Proposition~\ref{prop:comparewitness},
there exists a $k$-safe $\V_x$-labeling $(\zeta',\L)$ of $T(\Gamma)$ 
that is a witness of $\Gamma$ in $B_x$.
This completes the proof of the forward direction.

\smallskip
It remains to prove the backward direction.
Suppose a $B_x$-namu $\Gamma=(T,\alpha,\lambda,U)$ is in $U_k(B_x)$ and 
a $k$-safe $\V_x$-labeling $(\zeta,\L)$ of~$T$ is a witness of $\Gamma$ in $B_x$.
By the definition of witnesses,
$(\T(\zeta),\L)$ is $k$-safe with respect to $x$ 
and totally pure with respect to $(T^b,\L^b)$.
Since $(\zeta,\L)$ is $k$-safe and $\Gamma$ has width at most~$k$, 
the width of $(\T(\zeta),\L)$ is at most~$k$.

Let $\Delta$ be the reduced $B_x$-namu of $(\T(\zeta),\L)$. %
It remains to show that $\Delta\tle\Gamma$.
If $(\T(\zeta),\L)$ is $x$\=/degenerate, then both $T(\Delta)$ and $T$ have only one node 
and so $\Delta\tle\Gamma$.
Suppose $(\T(\zeta),\L)$ is not $x$\=/degenerate. 
Then, by definition, $T(\Delta)$ is a subtree of $\T(\zeta)$ obtained by deleting all nodes $v$ 
such that $(u,v)$ is $x$\=/protected in $(\T(\zeta),\L)$.
Since every $x$\=/protected edge in $(\T(\zeta),\L)$ is in $\zeta$ and 
the condition (iv) of witnesses holds, %
an edge is $x$\=/protected in $(\T(\zeta),\L)$ if and only if the edge is in $\zeta$.
Recall that 
$\T(\zeta)$ is obtained from $T$ by 
\begin{enumerate}[(i)] 
\item subdividing each edge $uv$ of~$T$  $\abs{\zeta(u,uv)}+\abs{\zeta(v,uv)}$ times,
\item attaching the $i$th rooted binary tree in $\zeta(u,uv)$
  whose root is identified with the $i$th new subdividing internal node on the path from $u$ to $v$ 
  for each incidence $(u,uv)$ of~$T$ and $1\le i\le \abs{\zeta(u,uv)}$, 
\item attaching $\zeta(v)$ for all $v\in V(T)$ of degree at most $2$ 
  whose root is identified with~$v$, and 
\item smoothing $v$ for each $v\in V(T)$ with 
  $\deg_T(v)=2$ and $\abs{V(\zeta(v))}=1$.
\end{enumerate}
Therefore, 
$T(\Delta)$ is a tree obtained from $T$ 
by subdividing each edge $uv$ of~$T$  $\abs{\zeta(u,uv)}+\abs{\zeta(v,uv)}$ times 
and smoothing $w$ for each node $w\in V(T)$ with $\deg_T(w)=2$ and $\abs{V(\zeta(w))}=1$.
Note that 
if a node $w$ of~$T$ has exactly two neighbors $w_1$ and $w_2$ in~$T$ and 
$\abs{V(\zeta(w))}=1$, then we have $\alpha(w_1,w_1w)=\alpha(w,ww_2)$ and $\alpha(w,w_1w)=\alpha(w_2,ww_2)$.
Since both $\Delta$ and $\Gamma$ are induced from the canonical $B_x$-namu of $(\T(\zeta),\L)$,
we observe that $\Delta\tle\Gamma$. This completes the proof.
\end{proof}

\subsection{At a leaf node}\label{subsec:dpleaf}

\begin{PROP}\label{prop:single}
Let $\ell$ be a leaf of $T^b$ with $V=\L^b(\ell)$ and let $k$ be a nonnegative integer.
Let $\Delta_\ell=(T,\alpha,\lambda,B_{\ell})$ be the $B_\ell$-namu such that 
$T$ is a tree with $V(T)=\{1\}$. 
Then $\FS(\ell)=\{\Delta_\ell\}$.
\end{PROP}
\begin{proof}
Trivially, $\Delta_\ell$ is in $\FS(\ell)$.

If $\Gamma\in \FS(\ell)$,
then 
there exists a branch-decomposition $(T,\L)$ of $\V_\ell=\{V\}$ 
such that 
$\Delta\tle\Gamma$ 
for the reduced $B_x$-namu $\Delta$ of $(T,\L)$.
Since $\Gamma\in U_k(B_\ell)$, $V(T(\Gamma))=\{1\}$ and so $V(T)=\{1\}$. 
Thus, $(T,\L)$ is the unique branch-decomposition of $\{V\}$,
which implies that $\Delta=\Delta_\ell$.
As $T(\Delta_\ell)$ has no edge, 
$\Delta_\ell\tle\Gamma$ implies that $\Delta_\ell=\Gamma$.
Therefore, 
$\FS(\ell)=\{\Delta_\ell\}$.
\end{proof}

\subsection{At an internal node}\label{subsec:dpinternal}

For a subspace $B$ of $\F^r$
and a set $\mathcal{R}$ of $B$-namus,
the set $\up(\mathcal{R},B)$ is the collection of all $B$-namus $\Gamma\in U_k(B)$ 
with $\trim(\Gamma')\tle\Gamma$ 
for some  $\Gamma'\in\mathcal{R}$ such that $\Gamma'$ is a $k$-safe extension of $\trim(\Gamma')$.
In the following two propositions, we prove that
\[\FS(x)=\up((\FS(x_1)\oplus \FS(x_2))|_{B_x},B_x)\]
if $x$ has two children $x_1$ and $x_2$ in $T^b$.
First, we prove the inclusion in one direction.

\begin{PROP}\label{prop:fullsetwitness}
Let $x_1$ and $x_2$ be two children of a node~$x$ in $T^b$
and let $k$ be a nonnegative integer.
If $\Gamma\in\up((\FS(x_1)\fplus\FS(x_2))|_{B_x},B_x)$,
then $\Gamma\in\FS(x)$.
\end{PROP}
\begin{proof}
Since $\Gamma\in\up((\FS(x_1)\fplus\FS(x_2))|_{B_x},B_x)$,
there exist $\Gamma_1\in\FS(x_1)$, $\Gamma_2\in\FS(x_2)$, and their sum $\Gamma^+$
such that $\trim(\Gamma^+|_{B_x})\tle \Gamma$
and $\Gamma^+|_{B_x}$ is a $k$-safe extension of $\trim(\Gamma^+|_{B_x})$.
By Proposition~\ref{prop:witnessexist}, for each $i=1,2$, 
there exists a $k$-safe $\V_{x_i}$-labeling $(\zeta_i,\L_i)$
that is a witness of $\Gamma_i$ in $B_{x_i}$
such that
every $y$\=/protected edge in $(\T(\zeta_i),\L_i)$ is in $\zeta_i$ for every node $y\le x_i$ of $T^b$.
Then, by Propositions~\ref{prop:joinwitness},~\ref{prop:shrinkwitness},~\ref{prop:trimwitness},
and~\ref{prop:comparewitness}, 
we can find a $k$-safe $\V_x$-labeling $(\zeta,\L)$ such that it is a witness of $\Gamma$ in $B_x$
and $(\T(\zeta),\L)$ is in $\zeta$ for every node $y\le x$ of $T^b$.
By Proposition~\ref{prop:witnessexist}, $\Gamma$ is in the full set at $x$ of width~$k$ with respect to $(T^b,\L^b)$.
\end{proof}

For the other inclusion, we will prove the following proposition.

\begin{PROP}\label{prop:divide}
Let $x_1$ and $x_2$ be two children of a node~$x$ in $T^b$
and let $k$ be a nonnegative integer.
If $\Gamma\in\FS(x)$,
then $\Gamma\in\up((\FS(x_1)\fplus\FS(x_2))|_{B_x},B_x)$.
\end{PROP}

In order to prove Proposition~\ref{prop:divide},
we need some lemmas.

\begin{LEM}\label{lem:expand}
Let $B$, $B'$ be subspaces of $\F^r$. 
Let $(T,\L)$ be a branch-decomposition of a subspace arrangement $\V$.
Let $\Gamma$ and $\Gamma'$ be the canonical $B$-namu and $B'$-namu of $(T,\L)$, respectively.
If $\spn{\V}\cap B'\subseteq B$ and $\spn{\V}\cap B\subseteq B'$, 
then $\Gamma= \Gamma'$.
\end{LEM}

\begin{proof}
Let $\Gamma=(T,\alpha,\lambda,U)$ and $\Gamma'=(T,\alpha',\lambda',U')$.
Let $e=uv$ be an edge of~$T$. %
Then by definition
\[
\lambda(uv)=\dim\left(\sum_{x\in A_u(T-e)}\L(x)\cap\sum_{y\in A_v(T-e)}\L(y)\right) = \lambda'(uv).
\]

By the symmetry of $B$ and $B'$, it is enough to show that for each incidence $(v,e)$ of~$T$,
$\alpha(v,e)\subseteq \alpha'(v,e)$, and $U\subseteq U'$.
Because $\L(x)$ is in $\V$ for every leaf $x$ of~$T$, we have
\begin{align*}
  \alpha(v,e)&=B\cap\sum_{x\in A_v(T-e)}\L(x)= B\cap \spn{\V} \cap \sum_{x\in A_v(T-e)} \L(x) \\
  &
\subseteq B'\cap \sum_{x\in A_v(T-e)} \L(x) =\alpha'(v,e),\\
U&=B\cap\sum_{x\in A(T)}\L(x)= B\cap\spn{\V}\cap\sum_{x\in A(T)}\L(x)\subseteq B'\cap\sum_{x\in A(T)}\L(x)=U',
\end{align*}
which complete the proof.
\end{proof}

\begin{LEM}\label{lem:assumption}
Let $x_1$ and $x_2$ be two children of a node~$x$ in $T^b$. 
Let $(T,\L)$ be a branch-\decomposition{} of $\V_x$ 
and let $(T_i,\L_i)=(T,\L)|_{\V_{x_i}}$ for every $i=1,2$.
If $\Gamma^+$ is the canonical $(B_{x_1}+B_{x_2})$-namu of $(T,\L)$ and 
$\Gamma_i$ is the canonical $B_{x_i}$-namu of $(T_i,\L_i)$ for every $i=1,2$, then
\begin{itemize}
\item $\Gamma^+|_{B_x}$ is the canonical $B_x$-namu of $(T,\L)$,  
\item $\Gamma_i$ is the canonical $(B_{x_1}+B_{x_2})$-namu of $(T_i,\L_i)$ for every $i=1,2$, and 
\item $\Gamma^+$ is a sum of $\Gamma_1$ and $\Gamma_2$.
\end{itemize}
\end{LEM}
\begin{proof}
It is easy to check by definition that $\Gamma^+|_{B_x}$ is the canonical $B_x$-namu of $(T,\L)$.

By Lemma~\ref{lem:boundaryspace}, 
$\spn{\V_{x_1}}\cap(B_{x_1}+B_{x_2}) \subseteq B_{x_1}$.
Then, by Lemma~\ref{lem:expand},
$\Gamma_1$ is the canonical $(B_{x_1}+B_{x_2})$-namu of $(T_1,\L_1)$ and 
$\Gamma_2$ is the canonical $(B_{x_1}+B_{x_2})$-namu of $(T_2,\L_2)$.

We have 
$(\spn{\V_{x_1}}+B_{x_1}+B_{x_2})\cap (\spn{\V_{x_2}}+B_{x_1}+B_{x_2})= B_{x_1}+B_{x_2}$
by Lemma~\ref{lem:boundaryspace}.
Thus, by Lemma~\ref{lem:decomp2sum}, $\Gamma^+$ is a sum of $\Gamma_1$ and $\Gamma_2$.
\end{proof}

\begin{LEM}\label{lem:sumsafeeachsafe}
Let $x_1$ and $x_2$ be two children of a node~$x$ in $T^b$. 
Let $(T,\L)$ be a branch-\decomposition{} of $\V_x$ that is totally pure with respect to $(T^b,\L^b)$
and let $(T_i,\L_i)=(T,\L)|_{\V_{x_i}}$ for every $i=1,2$.
If $(T,\L)$ is $k$-safe with respect to $x$, 
then $(T_i,\L_i)$ is $k$-safe with respect to $x_i$ for every $i=1,2$.
\end{LEM}
\begin{proof}
Let $uv$ be an edge of $T_1$ such that $(u,v)$ is $x_1$-protected in $(T_1,\L_1)$.
We claim that there exists an edge $u'v'$ of~$T$ 
corresponding to $uv$ 
such that $(u',v')$ is $x$\=/protected in $(T,\L)$.
If there exists a node~$z$ of $T^b$ such that $z\le x_1$, $(T_1,\L_1)$ is $z$-degenerate, 
$uv$ cuts $\V_z$, and $\L(T_1,u,v)\subseteq\V_z$,
then $(u,v)$ itself is $x$\=/protected in $(T,\L)$ because $z\le x$, $(T,\L)$ is $z$-degenerate, and
$\L(T,u,v)\subseteq\V_z$.
If $uv$ is $x_1$-blocked and $uv$ $x_1$-guards $v$,
then there exists an edge $u'v'$ of~$T$ such that 
$u'v'$ $x_1$-guards $v'$ and $\L(T,u',v')\subseteq\V_{x_1}$
because $(T,\L)$ is $x_1$-pure.
By Lemma~\ref{lem:x-to-y}, $(u',v')$ is $x$\=/protected in $(T,\L)$, which completes the proof of the claim.

Let $\Gamma^+=(T,\alpha,\lambda,U)$ be the canonical $(B_{x_1}+B_{x_2})$-namu of $(T,\L)$
and let $\Gamma_i=(T_i,\alpha_i,\lambda_i,U_i)$ be the canonical $B_{x_i}$-namu of $(T_i,\L_i)$
for every $i=1,2$.
By Lemma~\ref{lem:assumption}, $\Gamma^+=\Gamma_1+_{(\eta_1,\eta_2)}\Gamma_2$.
If $e=uv$ is an edge of $T_1$ such that $(u,v)$ is $x_1$-protected in $(T_1,\L_1)$
and $e'=u'v'$ is the edge corresponding to $uv$ as above, then 
$\alpha_2(\vec\eta_2(u',e'))=U_2$, $\alpha_2(\vec\eta_2(v',e'))=\{0\}$, 
$U_1=B_{x_1}$, $U_2=B_{x_2}$, and $\alpha_1(u,e)\subseteq B_{x_1}\subseteq B_{x_2}+B_x$.
Therefore,
\begin{align*}
&\dim \bigl(\sum_{s\in A_{v'}(T-e')} 
\hspace{-1.5em}
\L(s) \cap\sum_{t\in A_{u'}(T-e')} 
\hspace{-1.5em}
 \L(t)\bigr)
+ \dim (\spn{\V_x}\cap B_x) - \dim \bigl(B_x \cap \sum_{t\in A_{u'}(T-e')}
\hspace{-1.5em}
\L(t)\bigr) \\
&=\lambda(e') +\dim (U\cap B_x) - \dim (\alpha(u',e') \cap B_x) \\
  &=\lambda_1(\eta_1(e'))+\lambda_2(\eta_2(e'))+ \dim ((U_1+U_2) \cap B_x)\\
  &\quad
    -\dim \Bigl(\bigl(\alpha_1(\vec\eta_1(u',e'))+ \alpha_2(\vec\eta_2(u',e'))\bigr)\cap B_x\Bigr)\\
  &\quad +\dim (U_1\cap U_2) -\dim (\alpha_1(\vec\eta_1(v',e'))\cap \alpha_2(\vec\eta_2(v',e')))\\
 &\quad
-\dim (\alpha_1(\vec\eta_1(u',e'))\cap \alpha_2(\vec\eta_2(u',e')))\\
&=\lambda_1(e) - \dim \alpha_1(u,e) - \dim U_2 + \dim(\alpha_1(u,e)+U_2) +\dim U_1 +\dim U_2 \\
&\quad -\dim (U_1+U_2) - \dim((\alpha_1(u,e)+U_2) \cap B_x) + \dim ((U_1+U_2)\cap B_x)\\
&=\lambda_1(e) - \dim \alpha_1(u,e) + \dim U_1 -\dim(\alpha_1(u,e)+U_2 + B_x) + \dim (U_1+U_2+ B_x) \\
&= \lambda_1(e) - \dim \alpha_1(u,e) + \dim U_1. 
\end{align*}
As $(T,\L)$ is $k$-safe with respect to $x$, we have 
\begin{align*}
  k &\ge \dim \Bigl(\sum_{s\in A_{v'}(T-e')} 
  \hspace{-1.5em}
  \L(s) \cap\sum_{t\in A_{u'}(T-e')} 
  \hspace{-1.5em}
  \L(t)\Bigr)
+ \dim (\spn{\V_x}\cap B_x) - \dim \Bigl(B_x \cap \sum_{t\in A_{u'}(T-e')}
\hspace{-1.5em}
\L(t)\Bigr) \\
&=\lambda_1(e) - \dim \alpha_1(u,e) + \dim U_1.
 \end{align*}
This implies that $(T_1,\L_1)$ is $k$-safe with respect to $x_1$.
Similarly, $(T_2,\L_2)$ is $k$-safe with respect to $x_2$.
\end{proof}

\begin{LEM}\label{lem:eachprotected}
Let $x_1$ be a child of a node~$x$ in $T^b$. 
Let $(T,\L)$ be a branch-de\-com\-po\-si\-tion of $\V_x$ that is totally pure with respect to $(T^b,\L^b)$
and let $(T_1,\L_1)=(T,\L)|_{\V_{x_1}}$. 
For a node~$v$ of $T_1$, 
if $v$ is $x_1$-protected in $(T_1,\L_1)$,
then $v$ is also $x_1$-protected in $(T,\L)$.
\end{LEM}
\begin{proof}
Let $uv$ be an edge of $T_1$ such that $(u,v)$ is $x_1$-protected in $(T_1,\L_1)$.
Then, by the construction of $T_1$, $u$ and $v$ are nodes of~$T$. 
If $u$, $v$ are adjacent in~$T$, then it is obvious that $(u,v)$ is $x_1$-protected in $(T,\L)$.
Thus, we may assume that $u$, $v$ are not adjacent.
Let $p$ be the neighbor of $v$ closer to $u$ and 
$q$ be the neighbor of $p$ not on the path from $u$ to $v$.

If $uv$ is $x_1$-blocked and $x_1$-guards $v$ in $(T_1,\L_1)$,
then $pv$ is $x_1$-blocked and $x_1$-guards $v$ in $(T,\L)$.
Thus, $v$ is $x_1$-protected in $(T,\L)$.

We may assume that 
there exists a node~$z$ of $T^b$ with $z\le x_1$ 
such that $uv$ cuts $\V_z$ in $(T_1,\L_1)$, $(T_1,\L_1)$ is $z$-degenerate, and $\L_1(T_1,u,v)\subseteq\V_z$.
As $(T_1,\L_1)$ is $z$-degenerate, so is $(T,\L)$.
Since $(T,\L)$ is $z$-pure, it is $z$-disjoint and thus $T$ has an edge $ww'$ such that 
$\L(T,w',w)=\V_{z}$.
If $p\neq w$, then 
$\L(T,p,q)\subseteq\V_z\subseteq\V_{x_1}$ and so $p$ is a node of $T_1$.
This means that $uv$ is not an edge of $T_1$ and thus it is a contradiction.
If $p=w$, 
then an edge $wv=pv$ of~$T$ cuts $\V_z$ in $(T,\L)$, $(T,\L)$ is $z$-degenerate,
and $\L(T,p,v)\subseteq\V_z\subseteq\V_{x_1}$.
So, $(p,v)$ is $x_1$-protected in $(T,\L)$.
\end{proof}

\begin{LEM}\label{lem:trimshrink}
Let $x_1$ and $x_2$ be two children of a node~$x$ in $T^b$. 
Let $(T,\L)$ be a branch-decomposition of $\V_x$ that is totally pure with respect to $(T^b,\L^b)$
and let $(T_i,\L_i)=(T,\L)|_{\V_{x_i}}$ for every $i=1,2$.
For every $i=1,2$, let $P_i$ be the set of all $x_i$-protected nodes $v$ in $(T_i,\L_i)$.
If $T'$ is the subtree of~$T$ obtained by removing all nodes in $P_1\cup P_2$ %
and $\Gamma'$ is the $(B_{x_1}+B_{x_2})$-namu induced by $T'$ from 
the canonical $(B_{x_1}+B_{x_2})$-namu of $(T,\L)$, 
then $\trim(\Gamma'|_{B_x})$ is equal to the reduced $B_x$-namu of $(T,\L)$.
\end{LEM}
\begin{proof}
Note that by Lemmas~\ref{lem:protectedconnected} and~\ref{lem:eachprotected},
a graph $T'$ of~$T$ obtained by removing all nodes in $P_1\cup P_2$ 
is a tree.
Let $\Delta$ be the reduced $B_x$-namu of $(T,\L)$.
By Lemma~\ref{lem:assumption}, $\Gamma'|_{B_x}$ is induced by $T'$ from the canonical $B_x$-namu of $(T,\L)$.
So it is enough to show that $T(\trim(\Gamma'|_{B_x}))$ is equal to $T(\Delta)$. 

If $(T,\L)$ is $x$\=/degenerate, then it is trivial because both trees have only one node.
Therefore, we may assume that $(T,\L)$ is not $x$\=/degenerate.
Recall that $V(T)-V(T(\Delta))$ is the set of all $x$\=/protected nodes $v$ in $(T,\L)$.
By definition, $V(T)-V(T(\Gamma'|_{B_x}))$ is the set of all nodes $v$ 
that is $x_1$-protected in $(T_1,\L_1)$ 
or $x_2$-protected in $(T_2,\L_2)$,
and $V(T(\Gamma'|_{B_x}))-V(T(\trim(\Gamma'|_{B_x})))$ is 
the set of all nodes $v$ such that $v$ is blocked in $\Gamma'|_{B_x}$.
Since $(T,\L)$ is not $x$\=/degenerate,
for a node~$v$ of~$T$,
$v$ is $x$\=/protected in $(T,\L)$ 
if and only if 
\begin{itemize}
\item $v$ is $x_1$-protected in $(T_1,\L_1)$,  or
\item $v$ is $x_2$-protected in $(T_2,\L_2)$,  or 
\item an edge $uv$ of~$T$ is $x$\=/blocked and $x$\=/guards $v$.
\end{itemize}
Therefore, 
$V(T)-V(T(\Delta))$ is equal to $V(T)-V(T(\trim(\Gamma'|_{B_x})))$.
This completes the proof.
\end{proof}

\begin{LEM}\label{lem:pure2trim}
Let $x_1$ and $x_2$ be two children of a node~$x$ in $T^b$. 
Let $(T,\L)$ be a branch-decomposition of $\V_x$ that is totally pure with respect to $(T^b,\L^b)$
and $k$-safe with respect to $x$
and let $(T_i,\L_i)=(T,\L)|_{\V_{x_i}}$ for every $i=1,2$.
Let $\Delta_i$ be the reduced $B_{x_i}$-namu of $(T_i,\L_i)$ for every $i=1,2$.
If $\Gamma'$ is the $(B_{x_1}+B_{x_2})$-namu defined in Lemma~\ref{lem:trimshrink}, %
then
there exists $\Delta^+\in\tau(\Delta_1)\tplus\tau(\Delta_2)$ such that $\Delta^+\tle\Gamma'$ and
$\Delta^+|_{B_x}$ is a $k$-safe extension of $\trim(\Delta^+|_{B_x})$.
\end{LEM}
\begin{proof}
We will first show that $\Gamma'$ is a sum of $\Delta_1$ and $\Delta_2$.
By Lemma~\ref{lem:assumption},
the canonical $(B_{x_1}+B_{x_2})$-namu $\Gamma^+$ of $(T,\L)$ is a sum of 
the canonical $B_{x_1}$-namu $\Gamma_1$ of $(T_1,\L_1)$ 
and the canonical $B_{x_2}$-namu $\Gamma_2$ of $(T_2,\L_2)$.
We remark that, by definition, $\Gamma'$ is induced from $\Gamma^+$ and 
$\Delta_i$ is induced from $\Gamma_i$ for each $i=1,2$.
By Lemma~\ref{lem:eachprotected}, for a node~$v$ of $T_1$, 
if $v$ is $x_1$-protected in $(T_1,\L_1)$, 
then $v$ is $x_1$-protected in $(T,\L)$. 
And for an edge $uv$ of~$T$, if $(u,v)$ is $x_1$-protected in $(T,\L)$, 
then it is clear that $uv$ does not cut $\V_{x_2}$ in $(T,\L)$. 
Thus, $V(T)-V(T(\Gamma'))$ is the disjoint union of $V(T_1)-V(T(\Delta_1))$ and $V(T_2)-V(T(\Delta_2))$.
So we conclude that $\Gamma'$ is a sum of $\Delta_1$ and $\Delta_2$.

Since $\Delta_1$ and $\Delta_2$ are trimmed, 
by Lemma~\ref{lem:tautrim} and Proposition~\ref{prop:subdivcomposition}, 
there exists 
$\Delta^+\in\tau(\Delta_1)\tplus\tau(\Delta_2)$ such that $\Delta^+\tle\Gamma'$.

Now we will show that $\Gamma'|_{B_x}=(T',\alpha',\lambda',U')$ 
is a $k$-safe extension of $\trim(\Gamma'|_{B_x})$.
By Lemma~\ref{lem:assumption}, $\Gamma'|_{B_x}$ is induced from the canonical $B_x$-namu of $(T,\L)$.
If $e$ is an edge of $T(\Gamma'|_{B_x})$ that is removed by trimming $\Gamma'|_{B_x}$,
then $e$ is $x$\=/protected in $(T,\L)$.
If $u$, $v$ are the ends of $e$ and $(u,v)$ is $x$\=/protected,
then 
\[
\lambda'(e)+\dim U' - \dim \alpha'(u,e) \le k
\]
because $(T,\L)$ is $k$-safe.
Thus, $\Gamma'|_{B_x}$ is a $k$-safe extension of $\trim(\Gamma'|_{B_x})$.

Since $\Delta^+\tle\Gamma'$,
by Lemmas~\ref{lem:compareproj} and~\ref{lem:tleksafe},
we conclude that 
$\Delta^+|_{B_x}$ is a $k$-safe extension of $\trim(\Delta^+|_{B_x})$.
\end{proof}

Here we give the proof of Proposition~\ref{prop:divide}.

\begin{proof}[Proof of Proposition~\ref{prop:divide}]
By definition,
there exists a branch-decomposition $(T,\L)$ of $\V_x$ such that 
\begin{itemize}
\item the width of $(T,\L)$ is at most~$k$, 
\item $(T,\L)$ is $k$-safe with respect to $x$, 
\item $(T,\L)$ is totally pure with respect to $(T^b,\L^b)$, and 
\item $\Delta\tle\Gamma$ for the reduced $B_x$-namu $\Delta$ of $(T,\L)$.
\end{itemize}
For every $i=1,2$, let $(T_i,\L_i)=(T,\L)|_{\V_{x_i}}$
and let $\Delta_i$ be the reduced $B_{x_i}$-namu of $(T_i,\L_i)$.
We first claim that $\tau(\Delta_1)\in\FS(x_1)$.
Since the width of $(T_1,\L_1)$ is at most the width of $(T,\L)$,  
the width of $(T_1,\L_1)$ is at most~$k$.
By Lemma~\ref{lem:pure}, $(T_1,\L_1)$ is totally pure with respect to $(T^b,\L^b)$.
Also, by Lemma~\ref{lem:sumsafeeachsafe}, $(T_1,\L_1)$ is $k$-safe with respect to $x_1$.
It is clear that %
$\Delta_1=\trim(\Delta_1)\tle\tau(\Delta_1)$ by Lemma~\ref{lem:tautrim}.
Since $\tau(\Delta_1)\in U_k(B_{x_1})$, we conclude that $\tau(\Delta_1)\in\FS(x_1)$.
Similarly, $\tau(\Delta_2)\in\FS(x_2)$.

Let $T'$ be the subgraph of~$T$ obtained by removing all nodes $v$ that  
is $x_1$-protected in $(T_1,\L_1)$ 
or is $x_2$-protected in $(T_2,\L_2)$. %
Note that such $v$ is a node of~$T$ by Lemma~\ref{lem:eachprotected}
and $T'$ is a tree by Lemma~\ref{lem:protectedconnected}.
Let $\Gamma'$ be the $(B_{x_1}+B_{x_2})$-namu induced by $T'$ from 
the canonical $(B_{x_1}+B_{x_2})$-namu of $(T,\L)$.
By Lemma~\ref{lem:pure2trim}, there exists $\Delta^+\in\tau(\Delta_1)\tplus\tau(\Delta_2)$
such that $\Delta^+\tle\Gamma'$ %
and $\Delta^+|_{B_x}$ is a $k$-safe extension of $\trim(\Delta^+|_{B_x})$.
By Lemmas~\ref{lem:comparetrim} and~\ref{lem:compareproj}, 
$\Delta^+\tle\Gamma'$ implies that 
$\trim(\Delta^+|_{B_x})\tle\trim(\Gamma'|_{B_x})$.
By Lemma~\ref{lem:trimshrink}, $\trim(\Gamma'|_{B_x})=\Delta$.
Then, by Lemma~\ref{lem:transitive} together with $\Delta\tle\Gamma$, 
we conclude that $\trim(\Delta^+|_{B_x})\tle\Gamma$.
Therefore, since $\Delta^+|_{B_x}$ is a $k$-safe extension of $\trim(\Delta^+|_{B_x})$, 
$\Gamma$ is in $\up((\FS(x_1)\fplus\FS(x_2))|_{B_x},B_x)$.
\end{proof}

Now we are ready to prove Proposition~\ref{prop:fullset}.
\begin{PROP}\label{prop:fullset}
  For every node~$x$ of $T^b$,
  $\FS(x)$ is equal to $\FF_x$ in Algorithm~\ref{alg:fullset}.
\end{PROP}
 \begin{proof}
  We use induction on the number of steps executed 
  by \textsc{full-set}$(\V,k,\allowbreak (T^b,\L^b), \allowbreak \{T_v\}_{v\in V(T^b)})$.
  For a leaf node, it is trivial by Proposition~\ref{prop:single}.
  From line~\ref{line:join} %
  to line~\ref{line:compare},
  we observe that if $x$ is an internal node with two children $x_1$ and $x_2$ in $T^b$,
  then $\FF_x=\up((\FF_{x_1}\fplus\FF_{x_2})|_{B_x},B_x)$.
  By the induction hypothesis and by Propositions~\ref{prop:fullsetwitness} and~\ref{prop:divide},
  we conclude that $\FF_x=\FS(x)$.
  \end{proof}

\section{The algorithm}\label{sec:algo}
In this section we present an algorithm to perform the following task.
  \begin{listspec}[Branch-Width]
  \item [Parameters:] A finite field $\F$ and an integer $k$.
  \item [Input:] An $r\times m$ matrix~$M$ over $\F$ with an ordered partition $\mathcal I=\{I_1,I_2,\ldots,I_n\}$ of $\{1,2,\ldots,m\}$ and an integer $k$.
  \item [Output:] A branch-decomposition $(T,\L)$ of width at most~$k$ of a subspace arrangement $\V$ consisting of the column space of $M[I_i]$ for each $i$
or a confirmation that the branch-width of~$\V$ is larger than $k$.
  \end{listspec}

If such a branch-decomposition exists, then we say that $(M,\mathcal{I},k)$ is a YES instance.
Otherwise, it is a NO instance.
\subsection{Basic operations}\label{subsec:operations}

Let $\F$ be a field and $r$ be a positive integer.
The \emph{sum} of two subspaces $X$ and $Y$ of $\F^r$ is denoted by $X+Y$ and is defined as $\{\mathbf{x}+\mathbf{y}: \mathbf{x}\in X, \mathbf{y}\in Y\}$.
Let $B$ and $B'$ be two subspaces of $\F^r$ such that $B'\subseteq B$, where $B$ comes with a fixed ordered basis $\B$. 
We say that a $\dim{B}\times \dim{B'}$ matrix $M_{B,B'}$ over $\F$ %
\emph{represents} $B'$ with respect to $\B$ if the $i$th column of $M_{B,B'}$ is the coordinate vector of the $i$th element  of \emph{some} fixed ordered basis $\B'$ of $B'$ with respect to $\B$. 
Let ${\mathfrak C}$ be a basis of a subspace $C$ such that $\B\subseteq {\mathfrak C}$. Then there exists the unique $\abs{\mathfrak C}\times \abs{\B}$ matrix $P$ over
$\F$
such that 
\[
P [x]_{\B}=[x]_{{\mathfrak C}}
\]
for all $x\in \spn{\B}$.
(For a vector $x$ in a vector space with a basis $\B$ over a field $\F$, 
$[x]_\B$ denotes the coordinate
vector with respect to the basis $\B$, which is a $\abs{\B}\times 1$ matrix over $\F$.)
This matrix $P$ is called the \emph{transition matrix}  from $\B$ to ${\mathfrak C}$.

Given two matrices $M_{B,B'}$ and $M_{B,B''}$ representing two subspaces $B'$ and $B''$ of~$B$ with respect to $\B$, we can do the following operations; let $d:=\dim B$:

\begin{itemize}
\item We can decide whether $B'\subseteq B''$ with $O(d^3)$ elementary field operations on $\F$. Note that $B'\subseteq B''$ if and only if the column space of $M_{B,B'}$ is contained in the column space of $M_{B,B''}$, or equivalently, the basis of $B'$ is contained in the column space of $M_{B,B''}$. The latter holds if and only if the equation $M_{B,B''} X =M_{B,B'}$ is feasible. The equation can be solved with $O(d^3)$ elementary field operations on $\F$.
\item We can compute the sum $B'+B''$, whose output is a matrix representing it, with $O(d^3)$ elementary field operations on $\F$. This can be done by computing the basis of the column vectors in the matrix $\left [M_{B,B'} \quad M_{B,B''} \right ]$. 
\item We can compute the matrix representing $B'\cap B''$ with  $O(d^3)$ elementary field operations on $\F$. The basis $\begin{bmatrix} M_x \\ M_y \end{bmatrix}$ of the null space of $\left[ M_{B,B'} \quad -M_{B,B''} \right ]$ can be computed using Gaussian elimination so that $M_{B,B'}M_x- M_{B,B''}M_y=0$. Then the matrix representing $B'\cap B''$ with respect to $\B$ can be obtained from $M_{B,B'}M_x$ by repeatedly removing a column if it is spanned by other columns. The last step can be done by finding a maximal set of linearly independent column vectors of $M_{B,B'}M_x$, which can be done by reducing it to reduced row echelon form and taking columns with leading $1$ of  some row.
  
  In particular, if $\B'$ is a basis of $B'$ and $\B'\subseteq
  \B$, then 
  the matrix representing $B'\cap B''$ with respect to $\B'$ can be
  easily obtained as follows:
  First, we obtain a matrix representation of $B'\cap B''$ with respect
  to $\B$.
  Second, we discard rows corresponding to $\B-\B'$ in that matrix representation.

\item If $M_{B,B'}$ is a matrix representing the subspace $B'$ of $\spn{\B}$ with respect to $\B$, and $P$ is a transition matrix from $\B$ to ${\mathfrak C}$, 
$PM_{B,B'}$ equals the matrix $M_{C,B'}$ representing $B'$ with respect to ${\mathfrak C}$.
  
\end{itemize}

\subsection{Preprocessing}\label{subsec:preprocessing}
We will first describe the preprocessing steps to reduce the input size. 
The subspace arrangement of $n$ subspaces
is given by an $r\times m$ matrix where $r$ and $m$ could be arbitrary large.
Our aim here is to reduce $r$ and $m$ or confirm that branch-width is larger than $k$. 
Eventually we will convert the input into a smaller one with $r\le m\le kn$, described in the proof of Theorem~\ref{thm:summary-brw} by using the following two lemmas.
Furthermore we will convert $M$ into the reduced row echelon form, which is crucial for our algorithm 
for computing the transcript of a branch-decomposition
in Theorem~\ref{thm:computebases}.

\begin{LEM}[row reduction lemma]\label{lem:RRL}
  Let $\F$ be a finite field. Given an $r\times m$ matrix~$M$ over $\F$ 
  with an ordered partition $\mathcal{I}=\{I_1,I_2,\ldots, I_n\}$ of $\{1,2,\ldots,m\}$, 
  let $M'$ be the matrix obtained from the reduced row echelon form of $M$ by removing zero rows. Let $\V=\{\col(M[I_1]),\ldots,\col(M[I_n])\}$ and  $\V'=\{\col(M'[I_1]),\ldots,\col(M'[I_n])\}$.

  Then $(M,\mathcal I,k)$ is a YES instance with a branch-decomposition $(T,\L)$ if and only if $(M',\mathcal I,k)$ is a YES instance with $(T,\L')$ where $\L'$ maps a leaf $v$ to $\col(M'[I_i])$ whenever $\L$ maps $v$ to $\col(M[I_i])$. 
  Moreover  in time $O(rm^2)$, we can find $M'$ and a subset $B$ of $\{1,2,\ldots,m\}$ representing all columns of $M'$ having the leading $1$ entry of some row such that $M'[B]$ is the identity matrix.
\end{LEM}
\begin{proof}
  Let $M''$ be the reduced echelon form of $M$. Then $M''=EM$ for some invertible matrix $E$. Let $T_E:\F^r\to\F^r$ be a linear transformation such that $T_E(x)=Ex$. Then $T_E$ is one-to-one and onto, and so $T_E(M[X])=M''[X]$ for all $X\subseteq\{1,2,\ldots,m\}$.
  So it is clear that $(M,\mathcal I,k)$ is a YES instance with a branch-decomposition $(T,\L)$ if and only if $(M'',\mathcal I,k)$ is a YES instance with the corresponding $(T,\L')$. We can remove zero rows in $M''$ to obtain $M'$ without changing any outcome. The set $B$ is indeed the set of column indices to apply pivot operation for finding the reduced row echelon form, thus is easy to be obtained in the algorithm. It is well known that we can get $M''$ and $B$ in time $O(rm^2)$.
\end{proof}

  \begin{LEM}[column reduction lemma]\label{lem:CRL}
  Let $\F$ be a finite field and let $k$ be a nonnegative integer.  Let $n\ge 2$.
  Let $M$ be an $r\times m$ matrix over $\F$ in reduced row echelon form 
  with an ordered partition $\mathcal{I}=\{I_1,I_2,\ldots,I_n\}$ of $\{1,2,\ldots,m\}$.
  Let $V_i$ be the column space of $M[I_i]$ for every $i$.
  
  In time $O((k+1) r m n)$, we can either find an $r\times m'$ matrix $M'$ over $\F$ and 
  an ordered partition $\mathcal{I}'=\{I_1',I_2',\ldots,I_n'\}$ of $\{1,2,\ldots,m'\}$ such that 
  \begin{enumerate}[(i)]
  \item the column vectors of $M'[I_i']$ are linearly independent for every $i\in\{1,2,\ldots,n\}$,
  \item $\abs{I_i'}\le \min(\abs{I_i},k)$ for every $i\in\{1,2,\ldots,n\}$,
  \item $\col (M'[I_i'])\subseteq \col ( M'[\{1,2,\ldots,m'\}-I_i'])$ for every $i\in\{1,2,\ldots,n\}$,
  \item for all $k$, $(M,\mathcal{I},k)$ is a YES instance with a branch-decomposition $(T,\L)$ 
    if and only if $(M',\mathcal{I}',k)$ is a YES instance with a branch-decomposition $(T,\L')$, where $\L'$ maps a leaf $v$ to $\col(M'[I_i'])$ whenever $\L$ maps $v$ to $\col(M[I_i])$ for every $i\in\{1,2,\ldots,n\}$,
  \end{enumerate}
  or find $i\in\{1,2,\ldots,n\}$ such that $\dim (V_i\cap (\sum_{j\neq i} V_i))>k$.
\end{LEM}
\begin{proof}
  We may assume that $M$ has no zero rows by deleting such rows.
  Let $B$ be the set of indices of columns having the leading $1$ of some row. Since $M$ is in reduced row echelon form with no zero rows,  $M[B]$ is an identity matrix.

  We aim to construct a new $r\times m'$ matrix $M'$ with an ordered partition $\mathcal I'=\{I_1',I_2',\ldots,I_n'\}$ such that for each $i=1,2,\ldots,n$, the set of the column vectors in~$M'[I_i']$ is a basis of $V_i'=V_i\cap \left(\sum_{j\neq i} V_j\right)$.
  \begin{itemize}
  \item If we find such a matrix $M'$ with $\mathcal I'$, then (iii) and  (iv) hold. 
    To see this, we claim that for all nonempty disjoint subsets $X$ and $Y$ of $\{1,2,\ldots,n\}$, 
    $(\sum_{i\in X}V_i)\cap (\sum_{j\in Y}V_j) = 
    (\sum_{i\in X}V_i')\cap (\sum_{j\in Y}V_j')$. This is because 
    if $a\in (\sum_{i\in X}V_i)\cap (\sum_{j\in Y}V_j)$, then $a=\sum_{i\in X} a_i$ for some $a_i\in V_i$ and yet $a\in \sum_{j\in Y} V_j$ and therefore for each $i\in X$, $a_i=a-\sum_{j\in X, j\neq i} a_j\in \sum_{j\neq i} V_j$. Thus,  $a_i\in V_i'$ and so $a\in  (\sum_{i\in X}V_i')$. Similarly $a\in  (\sum_{j\in Y}V_j')$. This proves the claim and it immediately implies (iv).
    In addition, by taking $X=\{i\}$ and $Y=\{1,2,\ldots,n\}-\{i\}$, we obtain $V_i'\subseteq  V_i\cap (\sum_{j\neq i} V_j) =V_i'\cap \sum_{j\neq i} V_j'\subseteq \sum_{j\neq i} V_j'$ and so (iii) holds.
  \item 
    If we find an index $i$ such that $\dim (V_i\cap (\sum_{j\neq i} V_i))>k$, then we can stop. This means that (ii) holds if $\dim (V_i\cap (\sum_{j\neq i} V_i))\le k$ for all $i$.
  \end{itemize}
  Let $E=\{1,2,\ldots,m\}$. 
  For each $i$,  a basis of $\col (M[I_i])\cap \col (M[E-I_i])$ can be found by computing a column basis of $M[B-I_i, I_i-B]$ and a column basis of $M[B\cap I_i, E-(I_i\cup B)]$ by Lemma~\ref{lem:basislemma}. Since we can always stop if those column bases have more than $k$ vectors, these two column bases can be found in time $O((k+1)rm)$ for each $i=1,2,\ldots,n$.  Thus the total running time is $O((k+1)rmn)$.
\end{proof}
\subsection{Data structure}\label{subsec:datastructure}
Given an input matrix, we will apply the preprocessing in Subsection~\ref{subsec:preprocessing} as follows: We first apply the algorithm in Lemma~\ref{lem:RRL}  and then apply the algorithm in Lemma~\ref{lem:CRL}, and then finally reapply the algorithm in Lemma~\ref{lem:RRL}.
After these steps, we may assume that the subspace arrangement $\V=\{V_1,V_2,\ldots , V_n\}$ is represented as an $r\times m$ matrix~$M$ and an ordered partition $\mathcal{I}=\{I_1,I_2,\ldots , I_n\}$ of $\{1,2,\ldots , m\}$ such that
\begin{itemize}
\item $r\leq m\le kn$,
\item $M$ is in reduced row echelon form with no zero rows,
\item for each $i\in\{1,2,\ldots,n\}$, the column vectors of $M[I_i]$ are linearly independent, $V_i$ is a column space of $M[I_i]$, and $\abs{I_i}\leq k$.
\end{itemize}
In a branch-decomposition $(T,\L)$ of~$\V$, although $\L$ is a mapping from the set of all leaves to $\V$, we implement $\L$ as a mapping to $\{1,2,\ldots , n\}$ so that $\L$ refers the index of the corresponding subspace in $\V$. When we carry out a computation on the actual subspaces $V_i$ of~$\V$, we  read the basis of $V_i$ stored as the column vectors of~$M[I_i]$.
Throughout the entire algorithm, we read $M[I_i]$ only in two cases: one for preprocessing, the other for computing the transcript. 
Therefore, %
we assume that $(T,\L)$ is represented using $O(n\log n)$ space.

For a subspace $B$ over $\F$ and a $B$-namu $(T,\alpha,\lambda,U)$, 
we represent $B'=\alpha(v,e)$ as a matrix $M_{B,B'}$ having $\dim B$ rows and $\dim B'$ columns;
see Subsection~\ref{subsec:operations} for details. 
The subspace $U$ of $B$ is represented in the same way. 
Therefore, every $B$-namu $(T,\alpha,\lambda,U)$ of width at most~$k$ admits a representation of 
size \[O\left(\abs{V(T)} \left((\dim B)^2\log\abs{\F} + \log(k+1) \right) \right).\]

\subsection{Computing a transcript of  a branch-decomposition}\label{subsec:bases}
Recall that Algorithm~\ref{alg:fullset} requests a given branch-decomposition $(T^b,\L^b)$ and a transcript $\Lambda$ of  $(T^b,\L^b)$; see the definition of transcripts in Subsection~\ref{subsec:transcript}. 
In this subsection, we describe how to compute a transcript $\Lambda$ in time $O(\theta^3n^2)$ when a branch-decomposition $(T^b,\L^b)$ of width at most $\theta$ is given.

Every transcript satisfies the properties in the following proposition.
\begin{PROP}\label{prop:transcript}
Let $(T^b,\L^b)$ be a rooted branch-decomposition of a subspace arrangement $\V$.
Let $\Lambda=(\{\B_v\}_{v\in V(T^b)},\{\B_v'\}_{v\in V(T^b)})$ be a transcript of $(T^b,\L^b)$.
Then the following hold.
\begin{enumerate}[(1)]

\item If a node $w$ is a parent of a node~$v$ of $T^b$, 
then
$\spn{\B_v}\subseteq \spn{\B_w'}$.

\item For each node~$v$ of $T^b$, %
\[\spn{\B_v}=\spn{\V_{v}}\cap\spn{\V-\V_{v}}\subseteq \spn{\B_v'},\]
and if $\rt$ is the root of $T^b$, then $\B_{\rt}=\emptyset$. 

\item For each node~$v$ with two children $w_1$ and $w_2$, 
  \begin{align*}
    \spn{\V_{w_i}}\cap \spn{\B_v'} &= \spn{\B_{w_i}}
                              \text{ for all $i=1,2$,}\\
    (\spn{\V_{w_1}}+\spn{\B_v'})\cap (\spn{\V_{w_2}}+\spn{\B_v'})&=\spn{\B_v'}.
  \end{align*}
\end{enumerate}
\end{PROP}

\begin{proof}
By definition, (1) and (2) are trivial. 
By Lemma~\ref{lem:boundaryspace}, (3) holds.
\end{proof}

Let $M$ be an $r\times m$ matrix of rank $r$ such that column vectors are indexed by~$E$ and there exists a subset $B$ of $E$ such that the submatrix of $M$ induced by the columns in $B$ is the identity matrix. We may also assume that the rows of $M$ are indexed by~$B$ as well. 
Let us write $M[X,Y]$ to denote a submatrix of $M$ induced by rows in~$X$ and columns in $Y$.
For simplicity, we also write $M[X]$ to denote $M[B,X]$.
For a matrix~$M$, let $\col (M)$ be the column space of $M$, that is, the span of all column vectors of $M$.
For a subset $X$ of $B$, let $I_{X}$ be the $r\times r$ diagonal matrix  whose rows and columns are indexed by~$B$ such that the diagonal entry at $v\in B$ is $1$ if $v\in X$ and $0$ otherwise.

We aim to construct a quadratic time algorithm to find a basis of $\col (M[X])\cap \col (M[Y])$ for all partitions $(X,Y)$ appearing in the branch decomposition.
First, observe the following lemma, which is folklore. Cunningham and Geelen~\cite{CG2007} presented essentially the same claim in Section 3 with a proof but unfortunately the statement of their claim is incorrect, which is probably a typo.
\begin{LEM}\label{lem:basislemma}
  Let $M$, $B$ be given as above.
  If $(X,Y)$ is a partition of $E$, then 
  \[ \dim (\col(M[X])\cap\col(M[Y])) = 
    \rank M[B\cap Y, X-B]+ \rank M[B\cap X, Y-B].\]
  Furthermore if $P\subseteq X-B$ is a column basis of $M[B\cap Y,X-B]$ and $Q\subseteq Y-B$ is a column basis of $M[B\cap X, Y-B]$, then 
  the disjoint union of the column vectors of $I_{B\cap Y} M[P]$ and the column vectors of $I_{B\cap X} M[Q]$ form a basis of $\col(M[X])\cap \col(M[Y])$.
\end{LEM}
\begin{proof}
  After rearranging columns and rows, we may write  \[
    M=\bordermatrix {
      & B\cap X & B\cap Y & X-B & Y-B \cr
      B\cap X &
      {\begin{matrix}
        1 & \\
        & 1 \\
        &&\ddots \\
        &&&1
      \end{matrix}}
      & 0 & M_{11} & M_{12} \cr
      B\cap Y& 0 & 
      {\begin{matrix}
        1 & \\
        & 1 \\
        &&\ddots \\
        &&&1
      \end{matrix}}
      & M_{21} & M_{22} \cr
    }.
  \]
  We may assume $M_{11}=0$ and $M_{22}=0$ because $\col(M[X])$ and $\col(M[Y])$ are preserved.
  Then it follows easily that $\col (M[X])\cap \col (M[Y])$ is spanned by $M[(X\cup Y)-B]$ and its dimension is equal to the rank of $
  \left(\begin{smallmatrix}
      0 & M_{12}\\ M_{21}&0
    \end{smallmatrix}\right)$, that is, the sum of $\rank (M_{21})$ and $\rank (M_{12})$.
  If $P$ is a column basis of $M_{21}$ and $Q$ is a column basis of $M_{12}$, then the basis of $M[(X\cup Y)-B]$ is the column vectors of $M[P\cup Q]$ (assuming that $M_{11}=0$ and $M_{22}=0$). This completes the proof.
\end{proof}

\begin{LEM}\label{lem:column2row}
  Let $M$, $B$ be given as above and let  $(X,Y)$ be a partition of $E$.
  If $R$~is a row basis of $M[B\cap X, Y-B]$
  and $Q$~is a column basis of $M[R,Y-B]$, 
  then 
  $Q$ is a column basis of $M[B\cap X, Y-B]$.

  In addition, if $\abs{R}\le \theta$, then $Q$ can be computed in time $O(\theta^2m)$.
\end{LEM}
\begin{proof}
  Columns in~$Q$ of $M[R, Y-B]$ span all column vectors of $M[R,Y-B]$. 
  As rows in~$R$ of $M[B\cap X,Y-B]$ span all row vectors of $M[B\cap X,Y-B]$, 
  columns in~$Q$ of $M[B\cap X, Y-B]$ also span all column vectors of $M[B\cap X,Y-B]$.
  As $\abs{R}\le \theta$, a column basis of $M[R,Y-B]$ can be computed in time $O(\theta^2m)$. 
\end{proof}

\begin{LEM}\label{lem:basisleaf}
  Let $M$, $B$ be given as above. Let $(X,Y)$ be a partition of $E$ such that $\abs{X}\le\theta$. Then in time $O(\theta^2m)$, we can find a column basis $P$ of $M[B\cap Y,X-B]$ and a row basis $R$ of $M[B\cap X, Y-B]$.
\end{LEM}
\begin{proof}
  As $M[B\cap Y, X-B]$ has at most~$k$ columns, 
  its column basis $P$ can be found in time $O(\theta^2 r)$.
  Similarly, $M[B\cap X, Y-B]$ has at most $\theta$ rows and so its row basis can be found in time $O(\theta^2m)$.  Note that $r\le m$.
\end{proof} 
\begin{LEM}\label{lem:basisinternal}
  Let $M$, $B$ be given as above.
  Let $(X_1\cup X_2,Y)$ be a partition of $E$. Suppose that 
  \begin{itemize}
  \item 
  $P_1\subseteq X_1-B$ is a column basis of $M[B\cap (Y\cup X_2),X_1-B]$ and $R_1\subseteq B\cap X_1$ is a row basis of $M[B\cap X_1, (Y\cup X_2)-B]$,
  \item 
  $P_2\subseteq X_2-B$ is a column basis of $M[B\cap (Y\cup X_1),X_2-B]$
  and $R_2\subseteq B\cap X_2$ is a row basis of $M[B\cap X_2, (Y\cup X_1)-B]$.
  \end{itemize}
  Let $P$ be a column basis of $M[B\cap Y, P_1\cup P_2]$ and 
  let $R$ be a row basis of $M[ R_1\cup R_2, Y-B]$.
  Then 
  $P$ is a column basis of $M[B\cap Y, (X_1\cup X_2)-B]$
  and $R$ is a row basis of $M[B\cap (X_1\cup X_2), Y-B]$.

  In addition, if $\abs{P_i}+\abs{R_i}\le \theta$ for each $i=1,2$, then 
  in time $O(\theta^2m)$, we can either find $P$ and $R$
  or confirm that $\abs{P}+\abs{R}>\theta$.
\end{LEM}
\begin{proof}
  The columns in $P_i$ of $M[B\cap Y,X_i-B]$ span all column vectors of $M[B\cap Y,X_i-B]$ for each $i=1,2$. Therefore 
  the columns in $P_1\cup P_2$ of $M[B\cap Y,(X_1\cup X_2)-B]$ span all column vectors of $M[B\cap Y,(X_1\cup X_2)-B]$. 
  Similarly, the rows in $R_i$ of $M[B\cap X_i,Y-B]$ span all row vectors of $M[B\cap X_i,Y-B]$ for each $i=1,2$. So, the rows in $R_1\cup R_2$ of $M[B\cap (X_1\cup X_2),Y-B]$ span all row vectors of $M[B\cap (X_1\cup X_2), Y-B]$.
  The conclusion on $P$ and $R$ follows easily.

  As $\abs{B\cap Y}\le r$ and $\abs{P_1\cup P_2}\le 2\theta$, $P$ can be computed in time $O(\theta^2r)$. Similarly as $\abs{R_1\cup R_2}\le 2\theta$ and $\abs{Y-B}\le m$, $R$ can be computed in time $O(\theta^2 m)$. Note that $r\le m$.
\end{proof}

Let $\V$ be a subspace arrangement of $\F^r$. 
Let $M$ be an $r\times m$ matrix in reduced row echelon form
with a partition $\I$ of the column index set $E$ of $M$
such that 
$\abs{\V}=\abs{\I}$, $m=\sum_{V\in \V} \dim V$, and 
for each $V\in\V$, there exists the unique $I\in \I$ with $V=\col(M[I])$.
We assume that $M$ has no zero rows by deleting such rows.
Let $B\subseteq E$ be the indices such that $M[B]$ is the identity matrix. As we did in the beginning of this section, we assume that rows of $M$ are indexed by~$B$.

Let $(T^b,\L^b)$ be a rooted branch-decomposition of~$\V$. 
For each node~$v$ of~$T$, let $E_v$ be the subset of $E$ 
indexing every column vector that corresponds to a vector in a basis of some $V\in \V_v$.
For each node~$v$ of $T^b$, let $E_v$ be the union of all $I\in \I$ satisfying that 
$\col(M[I])=V$ for some $V\in\V_v$.

Given a rooted branch-decomposition $(T^b,\L^b)$ of~$\V$, we want to construct a transcript $\Lambda=(\{\B_v\},\{\B_v'\})$ with its transition matrices $\{T_v\}$.
First we need to compute  a basis of $\spn{\V_v}\cap \spn{\V-\V_v}$, that is, 
\[
  \col( M[E_v] ) \cap \col (M[ E-E_v] )
\]
for each node~$v$ of $T^b$. 
For this step, 
Theorem~\ref{thm:computebases} shows that 
Algorithm~\ref{alg:bases} correctly computes a basis of 
$  \col( M[E_v] ) \cap \col (M[ E-E_v] )$
for all nodes $v$ of $T^b$ if $(T^b,\L^b)$ has width at most $\theta$.
Cunningham and Geelen~\cite[Section 3]{CG2007} presented an algorithm to do essentially the same task but its running time is cubic. Our algorithm runs in
quadratic time.

\begin{algorithm}

  \caption{Finding a basis for each boundary space in a rooted branch-\decomposition{} $(T^b,\L^b)$ of a subspace arrangement $\V$ of subspaces of dimension at most $\theta$ 
  or confirming that the branch-width is larger than $\theta$}
  \label{alg:bases}
  \begin{algorithmic}[1]
    \Ensure $M$ is an $r\times m$ matrix representing $\V$ in reduced row echelon form, and $E$ and its subset $B$ are given as in the beginning of Section~\ref{subsec:bases}
    \Procedure{bases}{$(T^b,\L^b)$, $\V$, $\theta$}
    \Repeat
    \State choose an unmarked node~$v$ of $T^b$ farthest from the root
    \State let $E_v$ be the subset of $E$, indexing all columns corresponding to some $V\in\V_v$
    \If {$v$ is a leaf} 
    \State compute a column basis  $P_v$  of $M[B-E_v,E_v-B]$
    \State compute  a row basis $R_v$  of $M[E_v\cap B, E-(E_v\cup B)]$
    \Else
    \State let $w_1,w_2$ be two children of $v$ in $T^b$
    \State compute a column basis $P_v$  of $M[ B-E_v, P_{w_1} \cup P_{w_2}]$
    \State compute a row basis  $R_v$ of $M[ R_{w_1}\cup R_{w_2}, E-(E_v\cup B)]$
    \EndIf
    \If {$\abs{P_v}+\abs{R_v}>\theta$}
    \State \textbf{stop}
    \EndIf
    \State compute a column basis $Q_v$ of $M[R_v,E-(E_v\cup B)]$
    \State let $B_v$ be the set of the column vectors in $I_{B-E_v} M[P_v]$ and $I_{B\cap E_v} M[Q_v]$
\State mark $v$
\Until{all nodes in $T^b$ are marked}
\EndProcedure
  \end{algorithmic}
\end{algorithm}

\begin{THM}\label{thm:computebases}
  Let $\V$ be a subspace arrangement of $\F^r$ represented by an $r\times m$ matrix~$M$ in reduced row echelon form with no zero rows
  such that each $V\in \V$ has dimension at most $\theta$.
  Let $n=\abs{\V}$.  
  Given a rooted branch-decomposition $(T^b,\L^b)$ of~$\V$,
  in time $O(\theta^3n^2)$, Algorithm~\ref{alg:bases} correctly computes a basis of 
  $ \spn{\V_v}\cap \spn{\V-\V_v}$
  for all nodes $v$ of $T^b$ 
  or confirms that $(T^b,\L^b)$ has width larger than $\theta$.
  In addition, if $(T^b,\L^b)$ has width at most $\theta$, then  we can compute the transcript $\Lambda=(\{\B_v\},\{\B_v'\})$ of $(T^b,\L^b)$ with its transition matrices in time $O(\theta^3n^2)$.
\end{THM}

To prove Theorem~\ref{thm:computebases}, we need the following lemma giving us a subroutine for finding bases.
\begin{LEM}\label{lem:computetranscript}
  Let $B$, $B_1$, $B_2$ be subspaces of $\F^r$ given by their ordered bases $\B$, $\B_1$, $\B_2$, respectively, such that $B\subseteq B_1+B_2$ and $\dim B\le \theta$, $\dim B_1\le \theta$, $ \dim B_2\le \theta$.
  Then in time $O(\theta^2 r)$, we can find an ordered basis $\B'$ of $B_1+B_2$ that extends $\B$ 
  and  matrices $T_1$ and $T_2$ (transition matrices) such that 
  $T_1 [x]_{\B_1} = [x]_{\B'}$ for all $x\in \B_{1}$ and $T_2[y]_{\B_2}=[y]_{\B'}$ for all $y\in \B_{2}$.
\end{LEM}
\begin{proof}
  Consider an $r\times (\abs{\B}+\abs{\B_1}+\abs{\B_2})$ matrix~$M$ consisting of column vectors in $\B$, $\B_1$, and $\B_2$. We find a column basis $\B'$ of $M$ extending $\B$ in time $O(\theta^2r)$.

  By symmetry, it is enough to show how to compute
  the transition matrices $T_1$ such that \[ T_1 [x]_{\B_1} = [x]_{\B'}\] for all $x\in \B_{1}$. 
This step is included in the proof of \cite[Proposition 5.3]{JKO2016} but we include its proof for the completeness.
Let $\B_1=\{b_1,b_2,\ldots,b_\ell\}\subseteq \F^r$. If $x=b_j$, then $T_1[b_j]_{\B_1}=T_1 e_j$ is the $j$th column vector of $T_1$, which is equal to the coordinate of $b_j$ with respect to $\B'$. So, \[(\B') T_1  = (\B_1),\] where $(\B')$ is an $r\times \abs{\B'}$ matrix whose column vectors are $\B'$
and similarly $(\B_1)$ is an $r\times \abs{\B_1}$ matrix corresponding to $\B_1$. This matrix equation can be solved in time $O(\theta^2 r)$. %
\end{proof}
\begin{proof}[Proof of Theorem~\ref{thm:computebases}]
  Let $E$ be the set indexing the columns of $M$ and let $B\subseteq E$ be a subset such that $M[B]$ is an $r\times r$ identity matrix.  
  Note that $r\le m$ because $r$ is equal to the rank of $M$
  and $m=\sum_{V\in \V}\dim V\le \theta n$ because each $V\in \V$ has dimension at most~$\theta$.
  Let $\B_v$ be a basis of $\spn{\V_v}\cap \spn{\V-\V_v}$.
  
  For each leaf $v$ of $T^b$, we can compute $P_v$ and $R_v$ in time $O(\theta^2m)\le O(\theta^3n)$ by Lemma~\ref{lem:basisleaf}.
  For each internal node~$v$ of $T^b$, we can compute $P_v$ and $R_v$ in time $O(\theta^2 m)\le O(\theta^3n)$ by Lemma~\ref{lem:basisinternal}.
  For each node~$v$, after computing $P_v$ and $R_v$, we can compute $Q_v$ in time $O(\theta^2m)\le O(\theta^3n)$ by Lemma~\ref{lem:column2row}
  and using that we can output $\B_v$ in time $O(\theta r)\le O(\theta^2 n)$ by Lemma~\ref{lem:basislemma}.
  Since there are $O(n)$ nodes of $T^b$, it takes $O(\theta^3n^2)$  time to compute $\B_v$ for all nodes $v$ of $T^b$.

  When we are given a basis $\B_v$ for each node~$v$ of $T^b$, 
  computing the transcript with its transition matrices can be done in time $O(\theta^2rn)\le O(\theta^3n^2)$ by  Lemma~\ref{lem:computetranscript}.
\end{proof}

\subsection{Running time analysis of Algorithm~\ref{alg:fullset}}
We recall how our algorithm works at each operation in \textsc{full-set}$(\V,k,(T^b,\L^b))$ of Algorithm~\ref{alg:fullset}  
and discuss the time complexity. 

\begin{description}
\item [Initialization:]
At line~\ref{line:initial1}, for a leaf node~$x$, 
the set $\FF_x^L$ is $\{\Delta_x\}$ as in Proposition~\ref{prop:single}.

\item [Join:]
Note that for each $i=1,2$, the size of $\FF_{x_i}$ is at most some function of $k$, $\theta$, $\abs{\F}$ 
by Lemma~\ref{lem:ukbtime} because $\FF_{x_i}$ is a subset of $U_k(B_{x_i})$. 
At line~\ref{line:join}, by Lemma~\ref{lem:generatesum}, the set $\FF_{x_1}\fplus \FF_{x_2}$ can be 
computed in time $f_1(k,\theta)$ elements for some function~$f_1$.

\item [Shrink:]
At line~\ref{line:shrink}, 
computing $\FF_x^+|_{B_x}$ takes at most $f_2(k,\theta)$ steps for 
some function $f_2$ because the number of all elements in $\FF_x^+$ is at most 
some function of $k$,~$\theta$.

\item[Trim:]
At line~\ref{line:trim},
for every $\Gamma\in\FF_x^S$, we test whether $\Gamma$ is a $k$-safe extension of $\trim(\Gamma)$
and if so put $\trim(\Gamma)$ into $\FF_x^T$.
Note that testing whether a $B$-namu is a $k$-safe extension of its trim
and trimming $\Gamma'$ can be done in time 
$f_3(k,\theta,\abs{\F})$ for some function $f_3$.

\item[Compare:]
At line~\ref{line:compare},
a $B_x$-namu $\Gamma$ is in $\FF_x$
if and only if $\Gamma\in U_k(B_x)$ and 
$\trim(\Gamma')\tle\Gamma$ for some $B_x$-namu $\Gamma'\in\FF_x^T$.
Thus, for each $\Gamma'\in\FF_x^T$ and each $\Gamma\in U_k(B_x)$,
we test whether $\trim(\Gamma')\tle\Gamma$.
Note that $\abs{\FF_x^T}$ is at most some function of $k$ and $\theta$, 
and for $\Gamma'\in\FF_x^T$, the number of edges in $T(\Gamma')$ bounded by some function of $k$ and 
$\abs{\F}$ by Lemma~\ref{lem:sumsize} 
because it comes from a sum of two compact $B$-namus of width at most~$k$.
By Lemmas~\ref{lem:ukbtime} and~\ref{lem:tletime},
$\FF_x$ can be computed in time $f_4(k,\theta,\abs{\F})$ for some function $f_4$.

\end{description}

Now we summarize the observations as follows.
Note that the number of nodes in $T^b$ is~$O(\abs{\V})$.

\begin{PROP}\label{prop:FSruntime}
Let $(T^b,\L^b)$ be a rooted branch-decomposition of~$\V$ of width at most $\theta$ for some nonnegative integer $\theta$.
In Algorithm~\ref{alg:fullset}, the procedure 
\textsc{full-set}$(\V,k,(T^b,\L^b),\{T_v\}_{v\in V(T^b)})$ runs in time
$f(k,\theta,\abs{\F})\cdot \abs{\V}$ 
for some function $f$. %
\end{PROP}

\subsection{Constructing a branch-decomposition of~$\boldsymbol\V$}\label{subsec:backtracking}
\mbox{}
We illustrate how to construct a branch-decomposition of~$\V$
when we have a nonempty full set at the root node of $T^b$ in Algorithm~\ref{alg:fullset}.

{\bf Composition trees.}
For each node~$x$ of $T^b$,
Algorithm~\ref{alg:fullset} constructs $\FF_x$ if $x$ is a leaf and  
constructs $\FF_x^+$, $\FF_x^S$, $\FF_x^T$, and $\FF_x$ otherwise.
In order to describe the $B$-namus in $\FF\in\{\FF_x,\FF_x^+, \FF_x^S, \FF_x^T\}$, 
we will use a \emph{composition tree} $T^c$ obtained from $T^b$.
The concept of a composition tree was introduced in~\cite{JKO2016},
which we modified slightly for our problem.
Note that Algorithm~\ref{alg:fullset} runs along a given rooted binary tree $T^b$ 
whose every internal node has degree $3$ except the root.
Here is the definition of a composition tree $T^c$.
(See Figure~\ref{fig:comptree}.)
\begin{itemize}
\item The node set of $T^c$ is the disjoint union of $\{v^J,v^S,v^T,v\}$ for each internal node~$v$ of $T^b$ 
and $\{v\}$ for each leaf $v$ of $T^b$.
\item For each internal node~$v$ of $T^b$ and its two children $v_1$ and $v_2$ in $T^b$, 
$v_1v^J$, $v_2v^J$, $v^Jv^S$, $v^Sv^T$, $v^Tv$ are the edges in a subtree of $T^c$ 
induced by a node set $\{v_1,v_2,v^J,v^S,v^T,v\}$.
\end{itemize}
\begin{figure}
  \centering
  \tikzstyle {u}=[draw,circle,fill=black,inner sep=1.2pt]
  \begin{tikzpicture}
  	\node [u,label=left:$v_1$] at (-1,2) (v1){};
  	\node [u,label=left:$v_2$] at (-1,0) (v2){};
  	\node [u,label=above:$v^J$] at (0,1) (vJ){};
  	\node [u,label=above:$v^S$] at (1,1) (vS){};
  	\node [u,label=above:$v^T$] at (2,1) (vT){};
  	\node [u,label=above:$v$] at (3,1) (v){};
  	\draw (v1)--(vJ)--(v2);
  	\draw (vJ)--(vS)--(vT)--(v);
  \end{tikzpicture}
  \caption{A subtree of $T^c$ for an internal node~$v$ with two children $v_1$ and $v_2$.}
  \label{fig:comptree}
\end{figure}
We remark that $T^c$ can be obtained from $T^b$ by subdividing 
every edge that is not incident with a leaf twice.
Each node of $T^c$
corresponds to each operation in Algorithm~\ref{alg:fullset}.
There are five kinds of nodes in $T^c$.
\begin{enumerate}[(1)]
\item \emph{Compare nodes}: nodes $v$ of $T^c$ for all internal nodes $v$ of $T^b$. 
								Note that the root is a compare node.
\item \emph{Trim nodes}: nodes $v^T$ of $T^c$ for all internal nodes $v$ of $T^b$.
\item \emph{Shrink nodes}: nodes $v^S$ of $T^c$ for all internal nodes $v$ of $T^b$.
\item \emph{Join nodes}: nodes $v^J$ of $T^c$ for all internal nodes $v$ of $T^b$.
\item \emph{Leaf nodes}: leaves $v$ of $T^c$. %
\end{enumerate}

\smallskip{\bf Evidence.}
Each operation in Algorithm~\ref{alg:fullset}
constructs the set $\FF\in\{\FF_x,\FF_x^+, \FF_x^S,\allowbreak \FF_x^T\}$.
Each $B$-namu $\Gamma$ in $\FF$ is associated with \emph{evidence} %
that certifies why $\Gamma$ is in~$\FF$.
When we run Algorithm~\ref{alg:fullset}, 
we store the evidence of each $B$-namu $\Gamma\in\FF$
at each operation as follows. %
\begin{enumerate}[\bf {Case} 1:]
\item (Initialization) \emph{$\FF=\FF_x$ for a leaf $x$ of $T^b$.}

	We store an index $i$ such that $V_i=\L^b(x)$.

\item (Join) \emph{$\FF=\FF_x^+$ for an internal node~$x$ of $T^b$ with two children $x_1$ and $x_2$.
Note that $x_1$ and $x_2$ are compare nodes or leaf nodes.}

	For every $\Gamma\in\FF_x^+$, we store $\Gamma_1\in\FF_{x_1}$, $\Gamma_2\in\FF_{x_2}$, 
	and a $(T(\Gamma_1),T(\Gamma_2))$-model $(\eta_1,\eta_2)$ in $T(\Gamma)$
	such that $\Gamma=\Gamma_1+_{(\eta_1,\eta_2)}\Gamma_2$. 

\item (Shrink) \emph{$\FF=\FF_x^S$ for an internal node~$x$ of $T^b$.}

	For every $\Gamma\in\FF_x^S$, we store $\Gamma'\in\FF_x^+$ such that $\Gamma=\Gamma'|_{B_x}$.
	
\item (Trim) \emph{$\FF=\FF_x^T$ for an internal node~$x$ of $T^b$.}

	For every $\Gamma\in\FF_x^T$, we store $\Gamma'\in\FF_x^S$ such that $\Gamma=\trim(\Gamma')$.
	
\item (Compare) \emph{$\FF=\FF_x$ for an internal node~$x$ of $T^b$.}

	For every $\Gamma\in\FF_x$, we store $\Gamma'\in\FF_x^T$ with $\Gamma'\tle\Gamma$,
	a tree $T$ ensuring $\Gamma'\tle\Gamma$, a $T(\Gamma')$-model in~$T$, and a $T(\Gamma)$-model in~$T$.
\end{enumerate}
We remark that all evidence  has bounded size by Lemmas~\ref{lem:ukbtime} and~\ref{lem:generatesum}.

Let $\V=\{V_1,V_2,\ldots,V_{\abs{\V}}\}$ be a subspace arrangement and 
let $(T^b,\L^b)$ be its rooted branch-decomposition of width at most $\theta$, 
which are inputs of Algorithm~\ref{alg:fullset}.
If $r$ is the root of $T^b$ and $\FF_r$ is nonempty after running Algorithm~\ref{alg:fullset},
then let $\Gamma_r\in\FF_r$. 
Afterward, 
we can find, in time $O(\abs{V(T^b)})$, the evidence corresponding to $\Gamma_r$
at each node of the composition tree $T^c$ by backtracking from the root $r$ to the leaves.
The evidence corresponding to $\Gamma_r$ is $B$-namus, and 
they are stored at all nodes of $T^c$ so that 
$\Gamma_r$ can be constructed by joining, shrinking, trimming, and comparing using them.
These will be used for Algorithm~\ref{alg:print}.

\smallskip{\bf Postorder traversal.}
We will use the \emph{postorder traversal} to express a branch-decomposition of a subspace arrangement.
For a branch-decomposition $(T,\L)$ of a subspace arrangement,
if the postorder traversal visits a leaf of~$T$, then we print (the index of) the corresponding subspace,
and if it visits an internal node of~$T$, then we print $*$.
For example, if $T$ is a rooted binary tree in Figure~\ref{fig:postorder} and 
$\L(\ell_i)=V_i$ for all $i=1,2,3$, 
then $1$, $2$, $3$, $*$, $*$ represents $(T,\L)$.
\begin{figure}
  \centering
  \tikzstyle {u}=[draw,circle,fill=black,inner sep=1.2pt]
  \begin{tikzpicture}
  	\node [u,label=below:$\ell_1$] at (-1,0) (v1){};
  	\node [u,label=below:$\ell_2$] at (0,0) (v2){};
  	\node [u,label=below:$\ell_3$] at (1,0) (v3){};
  	\node [u] at (0,2) (u1){};
  	\node [u] at (0.5,1) (u2){};
  	\draw (v1)--(u1)--(u2)--(v3);
  	\draw (v2)--(u2);
  \end{tikzpicture}
  \caption{A rooted binary tree $T$.}
  \label{fig:postorder}
\end{figure}

\smallskip{\noindent\bf Correctness.}
Let $\Gamma_x$ be the $B$-namu that is the evidence at a node~$x$ of $T^c$ 
with respect to $\Gamma_r\in\FF_r$.
Here we describe the five functions \textsc{printnode}, \textsc{printinc}, 
\textsc{printincrev}, \textsc{printsubtree}, and \textsc{printsubmid}, 
which are used in Algorithm~\ref{alg:print}.
\begin{enumerate}[(1)]
\item \textsc{printnode}$(x,v)$
\begin{description}
\item[Input] 
A node~$x$ of $T^c$ and a node~$v$ of $T(\Gamma_x)$.
\item[Output] 
A sequence of rooted binary trees whose leaves are labeled by indices of~$\V_x$.
\item[Return]
The number of output rooted binary trees.
\end{description}

\item \textsc{printinc}$(x,v,e)$, \textsc{printincrev}$(x,v,e)$
\begin{description}
\item[Input] 
A node~$x$ of $T^c$ and an incidence $(v,e)$ of $T(\Gamma_x)$.
\item[Output] 
A sequence of rooted binary trees whose leaves are labeled by indices of~$\V_x$.
\item[Return]
The number of output rooted binary trees.
\end{description}

\item \textsc{printsubtree}$(x,u,v)$, \textsc{printsubtreemid}$(x,u,v)$
\begin{description}
\item[Input] 
A node~$x$ of $T^c$ and two adjacent nodes $u$, $v$ of $T(\Gamma_x)$.
\item[Output] 
A rooted binary tree whose leaves are labeled by indices of $\V_x$.
\item[Return]
$1$ if the output is nonempty, and $0$ otherwise.
\end{description}
\end{enumerate}

Based on a composition tree $T^c$ and the evidence $\Gamma_x$ at each node~$x$ of $T^c$
with respect to the element $\Gamma_r$ in $\FF_{r}$, 
we will run Algorithm~\ref{alg:print} 
to print a branch-decomposition of~$\V$ of width at most~$k$ (using a postorder traversal).
Proposition~\ref{prop:correctness} will show that 
Algorithm~\ref{alg:print} can generate a branch-decomposition of~$\V$ correctly
using the concept of witnesses introduced in Subsection~\ref{subsec:witness}.

For each node~$x$ of $T^c$ and the corresponding node $x'$ of $T^b$,
let $\V_x=\V_{x'}$ and we define $B_x=B_{x_1'}+B_{x_2'}$ if $x$ is a join node in $T^c$ and $x_1'$, $x_2'$
are two children of $x'$ in~$T^b$,
and $B_x=B_{x'}$ otherwise.
We define $(\zeta_x,\L_x)$ as follows.
\begin{itemize}
\item
For each node~$v$ of $T(\Gamma_x)$ having degree at most $2$,
we define $\zeta_x(v)$ as 
the tree rooted at $v$ obtained by taking the disjoint union of $v$ and  
the trees output by \textsc{printnode($x,v$)}
and adding edges between $v$ and roots of trees in \textsc{printnode($x,v$)}.

\item
For each incidence $(v,e)$ of $T(\Gamma_x)$, we define $\zeta_x(v,e)$ as the sequence of 
disjoint trees output by \textsc{printinc($x,v,e$)}.

\item
Let $\L_x$ be the canonical mapping from the set of leaves of $\zeta_x$ to $\V_x$.
\end{itemize}

It is not difficult to see that $(\zeta_x,\L_x)$ is a $\V_{x}$-labeling  of $T(\Gamma_x)$. 
Lemma~\ref{lem:printzeta} shows that 
$(\zeta_x,\L_x)$ is a $k$-safe witness of $\Gamma_x$ in~$B_{x}$, from which Proposition~\ref{prop:correctness} follows.

\begin{LEM}\label{lem:printzeta}
For every node~$x$ of a composition tree $T^c$ and the evidence $\Gamma_x$ at $x$, 
$(\zeta_x,\L_x)$ is a $\V_x$-labeling of $T(\Gamma_x)$ and 
is a $k$-safe witness of $\Gamma_x$ in $B_x$.
Furthermore, for the node $x'$ of $T^b$ corresponding to $x$,
\begin{enumerate}[(1)]
\item if $x$ is a join node or a shrink node, 
then every $z$-protected edge in $(\T(\zeta_x),\L_x)$ is in $\zeta_x$ for every node $z<x'$ of $T^b$,

\item if $x$ is a trim node or compare node,
then every $z$-protected edge in $(\T(\zeta_x),\L_x)$ is in $\zeta_x$ for every node $z \le x'$ of $T^b$.
\end{enumerate}
\end{LEM}
\begin{proof}
We use induction on the number of descendants of a node~$x$ in $T^c$.
Let $\Gamma_x$ be the evidence at $x$ and let $T_x=T(\Gamma_x)$.

First, if $x$ is a leaf node of $T^c$,
then $T_x$ has a single node 
and  $(\zeta_x,\L_x)$ is a $\V_x$-labeling of $T_x$ and is a $k$-safe witness of $\Gamma_x$ in $B_x$.

If $x$ is a join node with two children $y_1$ and $y_2$ in $T^c$,
then $y_i$ is a leaf node or a compare node for every $i=1,2$.
Let $\Gamma_x=\Gamma_{y_1}+_{(\eta_1,\eta_2)}\Gamma_{y_2}$ 
for the evidence $\Gamma_{y_1}$ and $\Gamma_{y_2}$ at $y_1$ and $y_2$, respectively.
For a node~$v$ of degree $2$ in $T_x$, by definition, 
$v$ is a branch node in exactly one of $\eta_1$ and $\eta_2$.
Thus, 
\textsc{printnode($x,u$)} is either 
\textsc{printnode($y_1,u_1$)} if $u=\etabar_1(u_1)$
or \textsc{printnode($y_2,u_2$)} if $u=\etabar_2(u_2)$.
For an incidence $(v,e)$ of $T_x$,
\textsc{printinc($x,v,e$)} calls 
\textsc{printinc($y_1,\vec\eta_1(v,e)$)} or \textsc{printinc($y_2,\vec\eta_2(v,e)$)} 
(it may call none of them and return $0$)
from line~\ref{line:joinincstart} to line~\ref{line:joinincend}.
Then, we observe that 
$(\zeta_x,\L_x)$ is the sum of $(\zeta_{y_1},\L_{y_1})$ and $(\zeta_{y_2},\L_{y_2})$ by $(\eta_1,\eta_2)$.
Trivially, $(\zeta_x,\L_x)$ is a $\V_x$-labeling of $T_x$.
By the induction hypothesis and Proposition~\ref{prop:joinwitness},
$(\zeta_x,\L_x)$ is a $k$-safe witness of $\Gamma_x$ in $B_x$
and every $z$-protected edge in $(\T(\zeta_x),\L_x)$ is in $\zeta_x$ for every node $z< x'$ of~$T^b$.

Let $y$ be the child of $x$ in $T^c$ and $\Gamma_y$ be the evidence at $y$. Let $T_y=T(\Gamma_y)$.

If $x$ is a shrink node,
then for each node $u$ of degree at most $2$ in $T_x$ and each incidence $(v,e)$ of $T_x$,
\[\text{\textsc{printnode($x,u$)}$=$\textsc{printnode($y,u$)} and 
\textsc{printinc($x,v,e$)}$=$\textsc{printinc($y,v,e$)}.}\]
Clearly, we have $(\zeta_x,\L_x)=(\zeta_y,\L_y)$ and $\L_x$ is a $\V_x$-labeling of $T_x$.
Thus, by the induction hypothesis and Proposition~\ref{prop:shrinkwitness}, 
$(\zeta_x,\L_x)$ is a $k$-safe witness of $\Gamma_x$ in $B_x$
and since $y$ is a join node, every $z$-protected edge in $(\T(\zeta_x),\L_x)=(\T(\zeta_y),\L_y)$
is in $\zeta_x=\zeta_y$ for every node $z < x'$ of $T^b$.

If $x$ is a trim node, 
then 
for each incidence $(v,e)$ of $T_x$,
\[\text{\textsc{printinc($x,v,e$)}$=$\textsc{printinc($y,v,e$)}.}\]
Note that $T_x$ is a subtree of $T(\Gamma_y)$.
If $T_x$ has at least two nodes,
then for a node~$v$ with $\deg_{T_y}(v)=\deg_{T_x}(v)$,
\textsc{printnode($x,v$)}$=$\textsc{printnode($y,v$)}.
If $\deg_{T_y}(v)\neq\deg_{T_x}(v)=1$ and $w$ is the neighbor of $v$ in $T_y$, but not a neighbor of $v$ in $T_x$,
then \textsc{printnode($x,v$)} calls \textsc{printsubtree($y,v,w$)},
and if $\deg_{T_y}(v)\neq\deg_{T_x}(v)=2$ and $w$, $w'$ are the neighbors of $v$ in $T_y$ that are not neighbors of $v$ in $T_x$, then 
\textsc{printnode($x,v$)} calls both \textsc{printsubtree($y,v,w$)} and \textsc{printsubtree($y,v,w'$)}.
In Algorithm~\ref{alg:print}, 
it is easy to observe that the output of \textsc{printincrev($y,v,e$)} is the reversed sequence of
the output of \textsc{printinc($y,v,e$)}
by comparing two functions \textsc{printinc($y,v,e$)} and \textsc{printincrev($y,v,e$)}.
Thus, we observe that, 
from line~\ref{line:subtreestart} to line~\ref{line:subtreeend} of Algorithm~\ref{alg:print}, 
the function \textsc{printsubtree($y,v,w$)} outputs
\[\T(\zeta_y)[S\cup V(\zeta(v,vw))\cup V(\zeta(w,vw))],\] where $S$ is the set of all nodes in the component 
$\T(\zeta_y)-v$ containing $w$.
This implies that, by the induction hypothesis, 
$(\zeta_x,\L_x)$ is a $\V_x$-labeling induced by $T_x$ from $(\zeta_y,\L_y)$.
By Proposition~\ref{prop:trimwitness}, $(\zeta_x,\L_x)$ is a $k$-safe witness of $\Gamma_x$ in $B_x$
and every $z$-protected edge in $(\T(\zeta_x),\L_x)$ is in $\zeta_x$ for every node $z \le x'$ of $T^b$.

If $T_x$ has only one node,
then 
\textsc{printnode($x,v$)} calls \textsc{printsubtreemid($y,u,v$)} and \textsc{print\-subtreemid($y,v,u$)}
for an improper $y$\=/degenerate edge $uv$ in $\Gamma_y$.
In Algorithm~\ref{alg:print},
we can easily 
observe that the function \textsc{printsubtreemid($y,u,v$)} outputs
$\T(\zeta_y)[S'\cup V(\zeta(v,uv))]$ where $S'$ is the set of all nodes in the component 
$\T(\zeta_y)-u$ containing~$v$.
So \textsc{printnode($x,v$)} will output a rooted binary tree obtained from $T(\zeta_y)$ 
by subdividing $uv$.
Then, by the induction hypothesis, 
$(\zeta_x,\L_x)$ is a $\V_x$-labeling induced by $uv$ from $(\zeta_y,\L_y)$.
By Proposition~\ref{prop:trimwitness}, $(\zeta_x,\L_x)$ is a $k$-safe witness of $\Gamma_x$ in $B_x$
and every $z$-protected edge in $(\T(\zeta_x),\L_x)$ is in $\zeta_x$ for every node $z \le x'$ of $T^b$.

We assume that $x$ is a compare node.
For a node~$v$ of $T_x$ of degree at most $2$,
\textsc{printnode($x,v$)} calls \textsc{printnode($y,v'$)} if there exists a node $v'$ of $T_y$
corresponding to $v$ and returns $0$ otherwise.
From line~\ref{line:comparestart} to line~\ref{line:compareend}, for each incidence $(v,e)$ of $T_x$, 
\textsc{printinc($x,v,e$)} will output a sequence of rooted binary trees in $\zeta_y$ 
that correspond to $(v,e)$ with respect to $\Gamma_y\tle\Gamma_x$.
More precisely, it matches with the construction of a $\V_x$-labeling of $T(\Gamma_x)$ 
induced by $T$ from $(\zeta_y,\L_y)$ where $T$ is a tree ensuring $\Gamma_y\tle\Gamma_x$.
Then, by the induction hypothesis and Proposition~\ref{prop:comparewitness},
$(\zeta_x,\L_x)$ is a $\V_x$-labeling that is $k$-safe witness of $\Gamma_x$ in $B_x$ 
and every $z$-protected edge in $(\T(\zeta_x),\L_x)$ is in $\zeta_x$ for every node $z < x'$ of $T^b$.
This completes the proof.
\end{proof}

\begin{PROP}\label{prop:correctness}
Let $\V$ be a subspace arrangement and $k$ be a nonnegative integer. 
Let $(T^b,\L^b)$ be a given rooted branch-decomposition of width $\theta$ with the root $r$
and let $\{T_v\}_{v\in V(T^b)}$ be the given set of transition matrices. 
If $\FF_{r}$ in Algorithm~\ref{alg:fullset} with 
an input $(\V,k,(T^b,\L^b),\{T_v\}_{v\in V(T^b)})$ is nonempty,
then a (rooted) branch-decomposition of~$\V$ of width at most~$k$ 
can be constructed by Algorithm~\ref{alg:print}
in time $f(k,\theta,\abs{\F})\cdot \abs{\V}$ for some function $f$.
\end{PROP}
\begin{proof}
From line~\ref{line:root} to line~\ref{line:rootend},
Algorithm~\ref{alg:print} outputs 
the tree rooted at $v$ obtained by taking the disjoint union of $v$ and  
the trees output by \textsc{printnode($r,v$)} 
and adding edges between $v$ and roots of trees in \textsc{printnode($r,v$)},
where $v$ is the unique node of $T(\Gamma_r)$.
The tree is nonempty because $\FF_{r}$ is nonempty.
By definition, it is equal to $\zeta_r(v)$ and by Lemma~\ref{lem:printzeta},
the $\V_r$-labeling $(\zeta_r,\L_r)$ is a $k$-safe witness of $\Gamma_r$ in~$B_r$.
Thus, $(\T(\zeta_r),\L_r)$ is a branch-decomposition of $\V_r$ whose width is at most~$k$.
As $T(\Gamma_r)$ has only one node~$v$ and $\V_r=\V$, 
$(\zeta_r(v),\L_r)$ itself is a (rooted) branch-decomposition of~$\V$ having width at most~$k$,
which means that Algorithm~\ref{alg:print} prints a (rooted) branch-decomposition of~$\V$ 
of width at most~$k$.

The running time of Algorithm~\ref{alg:print} is 
$f(k,\theta,\abs{\F})\cdot \abs{\V}$ for some function $f$
because 
the number of nodes of $T^c$ is $O(\abs{V(T^b)})=O(\abs{\V})$ 
and the size of the evidence at each node of $T^c$ is at most
$h(k,\theta,\abs{\F})$ for some function $h$ by Lemmas~\ref{lem:ukbtime} and~\ref{lem:generatesum}.
\end{proof}

\subsection{Summary}
The following theorem summarizes the procedure and the time complexity of 
our main algorithm for a subspace arrangement.

\begin{THMMAIN}%
Let $\F$ be a finite field, let $r$ be a positive integer, and let $k$ be a nonnegative integer. 
Let $\V=\{V_1,V_2,\ldots,V_n\}$ be a subspace arrangement in $\F^r$
where each $V_i$ is given by its spanning set of $d_i$ vectors and $m=\sum_{i=1}^n d_i$.
In time $O(rm^2+(k+1)rmn + k^3n^3 + f(\abs{\F},k) n^2)$ for some function $f$, one can either
\begin{itemize}
\item find a branch-decomposition of~$\V$ having width at most~$k$ or
\item confirm that no such branch-decomposition exists.
\end{itemize}
\end{THMMAIN}
\begin{figure}
\footnotesize
\centering
\tikzstyle{startstop} = [rectangle, rounded corners, minimum width=5em, minimum height=0.5cm,text centered, draw=black, fill=red!30]
\tikzstyle{io} = [trapezium, trapezium left angle=70, trapezium right angle=110, minimum width=3cm, minimum height=1cm, text centered, text width=3cm, draw=black, fill=blue!30]
\tikzstyle{io3} = [trapezium, trapezium left angle=70, trapezium right angle=110, minimum width=5em, minimum height=1cm, text centered, text width=3cm, draw=black, fill=blue!30]
\tikzstyle{io2} = [trapezium, trapezium left angle=70, trapezium right angle=110, minimum width=3cm, minimum height=1cm, text centered, text width=3cm, draw=black, fill=blue!30]
\tikzstyle{io1} = [trapezium, trapezium left angle=70, trapezium right angle=110, minimum width=5em, minimum height=0.6cm, text centered, draw=black, fill=blue!30]
\tikzstyle{process} = [rectangle, minimum width=5em, minimum height=0.6cm, text centered, draw=black, fill=orange!30]
\tikzstyle{process3} = [rectangle, minimum width=3cm, minimum height=1cm, text centered, text width=9.4cm, draw=black, fill=orange!30]
\tikzstyle{process2} = [rectangle, minimum width=3cm, minimum height=1cm, text centered, text width=9.5cm, draw=black, fill=orange!30]
\tikzstyle{process4} = [rectangle, minimum width=1cm, minimum height=1cm, text centered, text width=1.5cm, draw=black, fill=orange!30]
\tikzstyle{process1} = [rectangle, minimum width=3cm, minimum height=0.6cm, text centered, text width=8.3cm, draw=black, fill=orange!30]
\tikzstyle{decision} = [diamond, aspect=5, minimum width=2cm, minimum height=0.3cm, text centered, text width=3.5cm, draw=black, fill=green!30]
\tikzstyle{decision1} = [diamond, aspect=5, minimum width=2cm, minimum height=0.5cm, text centered, text width=2cm, draw=black, fill=green!30]
\tikzstyle{decision2} = [diamond, aspect=4, minimum width=2cm, minimum height=0.5cm, text centered, text width=5.6cm, draw=black, fill=green!30]
\tikzstyle{decision3} = [diamond, aspect=5, minimum width=2cm, minimum height=0.5cm, text centered, text width=4.7cm, draw=black, fill=green!30]
\tikzstyle{arrow} = [thick,->,>=stealth]

\begin{tikzpicture}[node distance=2cm]
\node (start) [startstop] {Start};
\node (in1) [io1, below of=start, yshift=1.2cm] {Input a subspace arrangement in $\F^r$.};
\node (pro1) [process, below of=in1, yshift=1.2cm] {Run the preprocessing described in Subsection~\ref{subsec:preprocessing}.};
\node (dec3) [decision2, below of=pro1,yshift=0.2cm] {Does the preprocessing confirm that the branch-width is greater than $k$?};
\node (pro7) [process1, below of=dec3] {
Let $\V=\{V_1,V_2,\ldots,V_n\}$ be the resulting subspace arrangement.};
\node (dec4) [decision1, below of=pro7, yshift=0.6cm] {Is $\abs{\V}\le2$?};
\node (out4) [io3, right of=dec4, xshift=3cm] {Output an arbitrary branch-de\-com\-po\-si\-tion of the input.};
\node (pro2) [process1, below of=dec4] {
Let $\ell=3$. Let $(T_2,\L_2)$ be the unique
branch-decomposition of $\{V_1,V_2\}$.};
\node (pro5) [process2, below of=pro2, yshift=0.4cm] {Convert a given branch-decomposition $(T_{\ell-1},\L_{\ell-1})$ of $\{V_i\}_{i=1}^{\ell-1}$ to a branch-decomposition $(T'_\ell,\L'_\ell)$ of $\{V_i\}_{i=1}^{\ell}$ of width at most $2k$. %
Run the algorithm in Theorem~\ref{thm:computebases} to have a transcript.
Run Algorithm~\ref{alg:fullset} with $\{V_i\}_{i=1}^{\ell}$, $k$, and $(T'_\ell,\L'_\ell)$.};
\node (dec1) [decision, below of=pro5,yshift=0.2cm] {Is $\mathcal{F}_{root}$ an empty set?};
\node (pro4) [process3, below of=dec1, yshift=0.5cm] {Run Algorithm~\ref{alg:print} with an element of $\mathcal{F}_{root}$ to obtain a branch-decomposition $(T_\ell,\L_\ell)$ of $\{V_i\}_{i=1}^{\ell}$ of width at most~$k$.};
\node (dec2) [decision1, below of=pro4, yshift=0.8cm] {Is $\ell=n$?};
\node (out1) [io2, below of=dec2, yshift=0.4cm] {Output a branch-decomposition $(T_n,\L_n)$ of the input of width at most~$k$.};
\node (stop) [startstop, below of=out1, yshift=0.5cm] {Stop};
\node (out2) [io, right of=out1, xshift=2.3cm] {Confirm that the branch-width of the input is greater than $k$.};
\node (pro6) [process4, right of=dec2, xshift=2cm] {Increase $\ell$ by $1$.};

\draw [arrow] (start) -- (in1);
\draw [arrow] (in1) -- (pro1);
\draw [arrow] (pro1) -- (dec3);
\draw [arrow] (dec3) -- node[anchor=west] {No} (pro7);
\draw [arrow] (pro7) -- (dec4);
\draw [arrow] (dec4) -- node[anchor=west] {No} (pro2);
\draw [arrow] (dec3.east) node[above,xshift=1cm] {Yes}
-|([xshift=3cm,yshift=0.5cm]out2.north)-| ([xshift=2cm]out2);
\draw [arrow] (pro2) -- (pro5);
\draw [arrow] (pro5) -- (dec1);
\draw [arrow] (dec1) -- node[anchor=west] {No} (pro4);
\draw [arrow] (pro4) -- (dec2);
\draw [arrow] (dec2) -- node[anchor=west] {Yes} (out1);
\draw [arrow] (out1) -- (stop);
\draw [arrow] (out4.south)--++(0,-1cm)--++(2cm,0) |- ([yshift=-6mm]stop);
\draw [arrow] (dec4) -- node[anchor=south] {Yes} (out4);
\draw [arrow] (dec1) -| node[anchor=south, pos=0.1] {Yes} ([xshift=1cm]out2.north);
\draw [arrow] (out2) |- ([yshift=6mm]stop);
\draw [arrow] (dec2) ->node [above] {No} (pro6.west);
\draw [arrow] (pro6.east) -- ++(0.8cm,0)|- (pro5);
\end{tikzpicture}
\caption[A flowchart for an iterative compression.]
{A flowchart for the algorithm in Theorem~\ref{thm:summary-brw}.}
\label{fig:flowchart}
\end{figure}

\begin{proof}
  We may assume that $n>1$ because otherwise there is no branch-de\-com\-po\-si\-tion and the branch-width is $0$ by convention.
We preprocess the input, given as an $r\times m$ matrix~$M$ over a fixed finite field $\F$ with an ordered partition $\mathcal{I}=\{I_1,I_2,\ldots,I_n\}$ of $\{1,2,\ldots , m\}$ and an integer $k$ such that $V_i=\col(M[I_i])$ for every $i$, as follows.
We apply the reductions of Lemmas~\ref{lem:RRL} and \ref{lem:CRL} and then the reduction of Lemma~\ref{lem:RRL} again. Let $M'$ be the resulting $r'\times m'$ matrix and $\mathcal{I}'$ be the resulting ordered partition $\mathcal I'=\{I_1',I_2',\ldots,I_n'\}$ of $\{1,2, \ldots , m'\}$. By Lemmas~\ref{lem:RRL} and~\ref{lem:CRL}, it is straightforward to verify that
\begin{enumerate}[(i)]
\item $r'\leq m'\le kn$,
\item $M'$ is in reduced row echelon form with no zero rows,
\item for each $i$, the column vectors of $M'[I_i']$ are linearly independent and $\abs{I_i'}\leq k$,
\item $\col (M'[I_i']) \subseteq \col (M'[\{1,2,\ldots,m'\}-I_i'])$, %
\item   $(M,\mathcal I,k)$ is a YES instance with a branch-decomposition $(T,\L)$ if and only if $(M',\mathcal I',k)$ is a YES instance with $(T,\L')$ where $\L'$ maps a leaf $v$ to $\col(M'[I_i'])$ whenever $\L$ maps $v$ to $\col(M[I_i])$,
\end{enumerate}
and otherwise, we confirm that the branch-width of~$\V$ exceeds $k$. 
We may assume that $k>0$ 
because if $I_i'=\emptyset$ for all $i$, then every branch-decomposition has width~$0$.
We notice that the running time of preprocessing requires
$O(rm^2+(k+1)rmn)$. 
Henceforth, we assume that the $M=M'$, $\mathcal I=\mathcal I'$, $V_i=\col (M'[I_i'])$ to simplify notations. 

We may also assume that $\dim V_i\neq 0$ for all $i$ because otherwise we delete all such $V_i$ and later we can extend a branch-decomposition of $\V-\{V_i\}$ to that of~$\V$ of the same width easily. 
After the preprocessing, if $n=1$,
then an arbitrary branch-decomposition has width $0$ and so we simply output 
an arbitrary branch-decomposition of~$\V$. 
If $n=2$, then the branch-width is at most~$k$
because $\dim V_i\le k$ for each $i=1,2$ by (iii).
Thus we may assume that  $n\geq 3$.

We will apply iterative compression %
on $\V_i=\{V_1,\ldots , V_i\}$ for $i=3,\ldots , n$. 
We initially start with a trivial branch-decomposition $(T_2,\L_2)$ of $\V_2=\{V_1,V_2\}$ having width at most~$k$ such that $T_2$ is a tree with two nodes.
We  carry out a \textsc{compression step} for each  $i=3,\ldots , n$ iteratively as follows.
(See Figure~\ref{fig:flowchart} for a flowchart.)
\begin{enumerate}[(1)]
\item By adding a new leaf $v$ to $T_{i-1}$ and extending $\L_{i-1}$ to map $v$ to $V_i$, we create a branch-decomposition $(T_i',\L_i')$ of $\V_i$. %
Note that the width of $(T_i',\L_i')$ is at most $2k$ because $(T_{i-1},\L_{i-1})$ has width at most~$k$ and $V_i$ has dimension at most~$k$.
\item We use the algorithm in Theorem~\ref{thm:computebases} to compute 
transition matrices $\{T_v\}_{v\in V(T_i')}$ of the transcript corresponding to $(T_i',\L_i')$
in time $O(k^3 n^2)$. Note that the submatrix $M[I_1\cup I_2\cup\cdots\cup I_i]$ is already in reduced row echelon form and so we can apply Theorem~\ref{thm:computebases} by ignoring zero rows.
\item We compute the full set for $\V_i$ by using Algorithm~\ref{alg:fullset} 
  with an input $(\V_i,k,(T_i',\L_i'),\allowbreak
  \{T_v\}_{v\in V(T_i')})$
in time $f_1(k,2k,\abs{\F}) i$ for some function $f_1$ by Proposition~\ref{prop:FSruntime}.
  If the full set at the root is empty, 
  then the branch-width of $\V_i$ is larger than $k$.
  If so, we conclude that the branch-width of~$\V$ is larger than $k$ and stop.
\item If the full set at the root is nonempty, then we %
run Algorithm~\ref{alg:print} 
to obtain a branch-decomposition $(T_i,\L_i)$ of width at most~$k$ 
in time $f_2(k,2k,\abs{\F}) i$ for some function $f_2$
by Proposition~\ref{prop:correctness}. 
\end{enumerate}
If this algorithm finds $(T_n, \L_n)$, then $(T_n, \L_n)$ is a branch-decomposition of~$\V$ 
having width at most~$k$. 
For each $i$, (1)--(4) runs in at most $O(k^3n^2)+f(k,\abs{\F}) n$ time for some function $f$ and therefore the total running time of this step is $O(k^3n^3)+f(k, \abs{\F}) n^2$.
\end{proof}

\captionsetup[algorithm]{style=algori}
\captionof{algorithm}{Print a branch-decomposition of width at most~$k$}\label{alg:print}
	\begin{algorithmic}[1] 
	\Procedure{decomposition}{$r$}
		\State{$\alpha\leftarrow$\textsc{printnode}($r$, the unique node of $T(\Gamma_r)$)} \label{line:root}
		\If{$\alpha\ge2$}
			\State{print `$*$' $(\alpha-1)$ times}
		\EndIf \label{line:rootend}
	\EndProcedure

	\Function{printnode}{$x,v$}
		\If{$x$ is a leaf node}
                \State{print $i\in\{1,2,\ldots,\abs{\V}\}$ such that $x$ corresponds to $V_i$}	\label{line:i}
                \State{}\Return 1
		\ElsIf{$x$ is a join node}
			\State{let $y_1$ and $y_2$ be two children of $x$ in $T^c$}
			\If{the degree of $v$ in $T(\Gamma_x)$ is at most $2$}
				\If{$v$ is a branch node in $\eta_1$}
					\State{let $v'$ be a node of $T(\Gamma_y)$ such that $\etabar_1(v')=v$}
					\State{}\Return \textsc{printnode}($y_1,v'$) %
				\Else
					\State{let $v'$ be a node of $T(\Gamma_y)$ such that $\etabar_2(v')=v$}
					\State{}\Return \textsc{printnode}($y_2,v'$) %
				\EndIf		
			\EndIf
                        \State{}\Return 0
		\EndIf
		\State{let $y$ be the child of $x$ in $T^c$} 
		\If{$x$ is a compare node}
			\If{there exists the node $v'$ of $T(\Gamma_y)$ corresponding to $v$}
				\State{}\Return \textsc{printnode}($y,v'$)
			\EndIf
				\State{}\Return 0
		\ElsIf{$x$ is a shrink node} 
			\State{}\Return \textsc{printnode}($y,v$)	
		\ElsIf{$x$ is a trim node and $T(\Gamma_x)$ has at least two nodes}
			\State{$\alpha\gets 0$}
			\If{$\deg_{T(\Gamma_y)}(v)=\deg_{T(\Gamma_x)}(v)$}
				\State{$\alpha\leftarrow$ $\alpha$ $+$ \textsc{printnode}($y,v$)}
			\Else
				\For{each edge $vw$ incident with~$v$ in $T(\Gamma_y)$}
					\If{$vw$ is not in $T(\Gamma_x)$}
						\State{$\alpha\leftarrow$ $\alpha$ $+$ \textsc{printsubtree}($y,v,w$)}
					\EndIf
				\EndFor
			\EndIf
			\State{}\Return $\alpha$	
		\ElsIf{$x$ is a trim node and $T(\Gamma_x)$ has only one node}
			\State{let $uv$ be a degenerate edge of $\Gamma_y$}
			\State{\Return  \textsc{printsubtreemid}($y,u,v$)$+$\textsc{printsubtreemid}($y,v,u$)}
		\EndIf
	\EndFunction

	\Function{printinc}{$x,v,e$}
		\If{$x$ is a join node}
			\State{let $y_1$ and $y_2$ be two children of $x$ in $T^c$}
			\State{$\alpha\gets0$}
			\For{$i=1$ to $2$} \label{line:joinincstart}
				\If{$v$ is a branch-node in $\eta_i$ and $e'=\eta_i(e)$ is an edge of $T(\Gamma_{y_i})$} %
					\State{let $v'$ be a node of $T(\Gamma_{y_i})$ such that $\etabar_i(v')=v$}
					\State{$\alpha\leftarrow$ $\alpha$ $+$ \textsc{printinc}($y_i,v',e'$)}
				\EndIf
			\EndFor  \label{line:joinincend}
			\State{}\Return $\alpha$	
		\EndIf
		\State{let $y$ be the child of $x$ in $T^c$}
		\If{$x$ is a compare node} 
			\State{let $w$ be the end of $e$ other than $v$}	\label{line:comparestart}
			\If{$v<w$ (according to the canonical linear ordering of $V(T(\Gamma_x))$)}
				\State{let $v_0v_1\cdots v_\ell$ be the path from $v$ to $w$
							in the tree ensuring $\Gamma_x\tle\Gamma_y$} 
				\State{$\alpha\gets0$}
				\For{$i=1$ to $\ell$}
					\If{$v_{i-1}$ is a branch node in the subdivision of $T(\Gamma_y)$}
						\State{let $v_{i-1}'$ be the node of $T(\Gamma_y)$ corresponding to $v_{i-1}$}
						\State{let $e'$ be the edge incident with $v_{i-1}'$ that corresponds to $v_{i-1}v_i$}
						\State{$\alpha\leftarrow$ $\alpha$ $+$ \textsc{printinc}($y,v_{i-1}',e'$)}
					\EndIf
					\If{$v_{i}$ is a branch node in the subdivision of $T(\Gamma_y)$}
						\State{let $v_{i}'$ be the node of $T(\Gamma_y)$ corresponding to $v_{i}$}
						\State{let $e'$ be the edge incident with $v_{i}'$ that corresponds to $v_{i-1}v_i$}
						\State{$\alpha\leftarrow$ $\alpha$ $+$ \textsc{printincrev}($y,v_{i}',e'$)}
						\If{$i<\ell$}
							\State{$\alpha\leftarrow$ $\alpha$ $+$ \textsc{printnode}($y,v_{i}'$)}
						\EndIf
					\EndIf
				\EndFor
				\State{}\Return $\alpha$
			\EndIf		\label{line:compareend}
		\EndIf
		\If{$x$ is a trim node or a shrink node}
			\State{}\Return \textsc{printinc}($y,v,e$)
		\EndIf			
	\EndFunction

	\Function{printincrev}{$x,v,e$}
		\If{$x$ is a join node}
			\State{let $y_1$ and $y_2$ be two children of $x$ in $T^c$}
			\State{$\alpha\gets0$}
			\For{$i=2$ to $1$}
				\If{$v$ is a branch-node in $\eta_i$ and $e'=\eta_i(e)$ is an edge of $T(\Gamma_{y_i})$} %
					\State{let $v'$ be a node of $T(\Gamma_{y_i})$ such that $\etabar_i(v')=v$}
					\State{$\alpha\leftarrow$ $\alpha$ $+$ \textsc{printincrev}($y_i,v',e'$)}
				\EndIf
			\EndFor
			\State{}\Return $\alpha$
		\EndIf
		\State{let $y$ be the child of $x$ in $T^c$}
		\If{$x$ is a compare node} 
			\State{let $w$ be the end of $e$ other than $v$} 
			\If{$v<w$ (according to the canonical linear ordering of $V(T(\Gamma_x))$)}
				\State{let $v_0v_1\cdots v_\ell$ be the path from $v$ to $w$
							in the tree ensuring $\Gamma_x\tle\Gamma_y$}
				\State{$\alpha\gets0$}
				\For{$i=\ell$ to $1$}
					\If{$v_{i}$ is a branch node in the subdivision of $T(\Gamma_y)$}
						\If{$i<\ell$}
							\State{$\alpha\leftarrow$ $\alpha$ $+$ \textsc{printnode}($y,v_{i}'$)}
						\EndIf
						\State{let $v_{i}'$ be the node of $T(\Gamma_y)$ corresponding to $v_{i}$}
						\State{let $e'$ be the edge incident with $v_{i}'$ that corresponds to $v_{i-1}v_i$}
						\State{$\alpha\leftarrow$ $\alpha$ $+$ \textsc{printinc}($y,v_{i}',e'$)}
					\EndIf
					\If{$v_{i-1}$ is a branch node in the subdivision of $T(\Gamma_y)$}
						\State{let $v_{i-1}'$ be the node of $T(\Gamma_y)$ corresponding to $v_{i-1}$}
						\State{let $e'$ be the edge incident with $v_{i-1}'$ that corresponds to $v_{i-1}v_i$}
						\State{$\alpha\leftarrow$ $\alpha$ $+$ \textsc{printincrev}($y,v_{i-1}',e'$)}
					\EndIf
				\EndFor
				\State{}\Return $\alpha$
			\EndIf	 
		\EndIf
		\If{$x$ is a trim node or a shrink node}
			\State{}\Return \textsc{printincrev}($y,v,e$)
		\EndIf			
	\EndFunction

	\Function{printsubtree}{$x,u,v$} \label{line:subtreestart}
		\State{$\alpha\leftarrow$ \textsc{printinc}($x,u,uv$)}
		\State{$\alpha\leftarrow$ $\alpha$ $+$ \textsc{printincrev}($x,v,uv$)}
		\For{each neighbor $w$ of $v$ with $w\neq u$}
			\State{$\alpha\leftarrow$ $\alpha$ $+$ \textsc{printsubtree}($x,v,w$)}
		\EndFor
		\State{$\alpha\leftarrow$ $\alpha$ $+$ \textsc{printnode}($x,v$)}
		\If{$\alpha\ge1$}
			\State{print `$*$' $(\alpha-1)$ times}
			\State{}\Return $1$
		\Else
			\State{}\Return $0$
		\EndIf
	\EndFunction	\label{line:subtreeend}

	\Function{printsubtreemid}{$x,u,v$}	\label{line:subtreemidstart}
		\State{$\alpha\leftarrow$ \textsc{printincrev}($x,v,uv$)}
		\For{each neighbor $w$ of $v$ with $w\neq u$}
			\State{$\alpha\leftarrow$ $\alpha$ $+$ \textsc{printsubtree}($x,v,w$)}
		\EndFor
		\State{$\alpha\leftarrow$ $\alpha$ $+$ \textsc{printnode}($x,v$)}
		\If{$\alpha\ge1$}
			\State{print `$*$' $(\alpha-1)$ times}
			\State{}\Return $1$
		\Else
			\State{}\Return $0$
		\EndIf
	\EndFunction 	\label{line:subtreemidend}
	\end{algorithmic}

  \providecommand{\bysame}{\leavevmode\hbox to3em{\hrulefill}\thinspace}
  \providecommand{\MR}{\relax\ifhmode\unskip\space\fi MR }
  \providecommand{\MRhref}[2]{%
    \href{http://www.ams.org/mathscinet-getitem?mr=#1}{#2}
  }
  \providecommand{\href}[2]{#2}

\end{document}